\documentclass[acmsmall,10pt,screen]{acmart}  % review
\acmJournal{PACMPL}
\acmVolume{1}
\acmNumber{CONF} % CONF = POPL or ICFP or OOPSLA
\acmArticle{1}
\acmYear{2021}
\acmMonth{1}
\acmDOI{} % \acmDOI{10.1145/nnnnnnn.nnnnnnn}
\startPage{1}

\setcopyright{none}
\usepackage[utf8]{inputenc}
\usepackage{amsmath}
\usepackage{amssymb}

\let\subparagraph\paragraph
\usepackage[compact]{titlesec}
\usepackage{thmtools}
\usepackage[english]{babel}
\usepackage{pdfpages}
\usepackage{graphicx}
\usepackage{subcaption}
\usepackage[labelfont=bf]{caption}
\usepackage{longtable}
\usepackage{booktabs}
\usepackage{tabu}
\usepackage[referable]{threeparttablex}
\usepackage{varwidth}
\usepackage{multirow}
\usepackage{wrapfig}

\usepackage{float}

\usepackage{longtable}

\usepackage{pifont}

\usepackage[noend, ruled, linesnumbered]{algorithm2e}

\usepackage{hyperref}
\usepackage{thm-restate}
\usepackage[capitalise]{cleveref}

\crefname{figure}{Figure}{Figure}
\crefformat{footnote}{#2\footnotemark[#1]#3}
\usepackage{tikz}
\usetikzlibrary{calc}
\usepackage{comment}

\usepackage{todonotes}
\usepackage{tikz}
\usetikzlibrary{arrows,automata,shapes,decorations,decorations.markings,calc, matrix,decorations.pathmorphing, patterns,backgrounds}

\usepackage{bm}
\usepackage{pgfplots}
\usepackage{pbox}

\usepackage{appendix}
\usepackage{calc}
\usepackage{fp}
\usepackage{multirow}
\usepackage{pbox}
\usepackage[clock]{ifsym}

\usepackage{booktabs}   %% For formal tables:
\usepackage{subcaption} %% For complex figures with subfigures/subcaptions
\usepackage{enumerate,enumitem}
\setlist{nosep}
\DeclareMathAlphabet{\mathpzc}{OT1}{pzc}{m}{it}

\usepackage{graphicx}

\newcommand{\Paragraph}[1]{\paragraph{\bf #1}}  % to strictly adhere to 'always-indent' rule of ACM style
\newcommand{\SubParagraph}[1]{\paragraph{\it #1}}  % to strictly adhere to 'always-indent' rule of ACM style

\newcommand{\W}[1]{\operatorname{W[#1]}}

\newcommand{\bad}[1]{\mkern 1.5mu\overline{\mkern-1.5mu#1\mkern-1.5mu}\mkern 1.5mu}

\newcommand{\RF}[1]{\mathsf{RF}_{#1}}
\newcommand{\schedules}{\mathsf{schedules}}
\newcommand{\pre}[2]{\mathsf{pre}_{#1}(#2)}
\newcommand{\mutations}{\mathsf{mutations}}
\newcommand{\linearize}{\operatorname{Witness}}

\newcommand{\VTSOrm}{\operatorname{VTSO-rf}}
\newcommand{\VPSOrm}{\operatorname{VPSO-rf}}
\newcommand{\VSCrm}{\operatorname{VSC-rf}}

\newcommand{\AlgoTSO}{\operatorname{VerifyTSO}}
\newcommand{\AlgoPSO}{\operatorname{VerifyPSO}}
\newcommand{\NaiveAlgoTSO}{\operatorname{NaiveVerifyTSO}}
\newcommand{\NaiveAlgoPSO}{\operatorname{NaiveVerifyPSO}}

\newcommand{\Sequence}{\kappa}
\newcommand{\SequenceVar}{\mu}
\newcommand{\FenceMap}{\operatorname{FMap}}
\newcommand{\SFenceMap}{\operatorname{SFMap}}
\newcommand{\NEBMap}{\operatorname{NEBMap}}
\newcommand{\Threads}{\operatorname{Threads}}

\newcommand{\DPOR}{\operatorname{DPOR}}

\mathchardef\mhyphen="2D % Define a "math hyphen"

\newcommand{\Worklist}{\mathcal{S}}
\newcommand{\DoneSet}{\mathsf{Done}}

\newcommand{\ov}{\overline}
\newcommand{\Source}{\mathsf{Source}}

\newcommand{\ReadsFrom}{\mathsf{rfsc}}

\newcommand{\DCTSOPSOM}{\operatorname{RF-SMC}} % ExO-RF
\newcommand{\TSO}{\operatorname{TSO}}
\newcommand{\PSO}{\operatorname{PSO}}
\newcommand{\SC}{\operatorname{SC}}
\newcommand{\MemoryModel}{\mathcal{M}}

\newcommand{\NegativeMarked}{\mathsf{mrk}}

\newcommand{\Seq}{\tau}
\newcommand{\LocalTrace}{\rho}
\newcommand{\Trace}{\sigma}

\newcommand{\SysReads}{\mathcal{R}}
\newcommand{\SysWrites}{\mathcal{W}}
\newcommand{\SysWritesB}{\mathcal{W}^B}
\newcommand{\SysWritesM}{\mathcal{W}^M}
\newcommand{\SysFences}{\mathcal{F}}
\newcommand{\LocalEvents}[1]{\mathcal{L}(#1)}

\newcommand{\NumVariables}{d}
\newcommand{\Read}{\operatorname{r}}
\newcommand{\Fence}{\operatorname{fnc}}
\newcommand{\SFence}{\operatorname{storefnc}}
\newcommand{\TO}{\mathsf{PO}}
\newcommand{\Write}{w}
\newcommand{\WriteB}{\operatorname{wB}}
\newcommand{\WriteM}{\operatorname{wM}}
\newcommand{\Wpair}[1]{\mathbf{#1}}
\newcommand{\RMW}{\operatorname{rmw}}
\newcommand{\CAS}{\operatorname{cas}}

\newcommand{\Event}{e}

\newcommand{\Unordered}[3]{#1\parallel_{#2} #3}
\usepackage{centernot}
\newcommand{\Ordered}[3]{#1 \centernot \parallel_{\hspace{-1.4mm}#2}\hspace{0.8mm} #3}
\newcommand{\Refines}{\sqsubseteq}

\newcommand{\Confl}[2]{#1 \Join #2}

\newcommand{\Enabled}{\mathsf{enabled}}

\newcommand{\Concat}{\circ}
\newcommand{\Events}[1]{\SysEvents(#1)}
\newcommand{\Reads}[1]{\SysReads(#1)}
\newcommand{\Writes}[1]{\SysWrites(#1)}
\newcommand{\WritesB}[1]{\SysWritesB(#1)}
\newcommand{\WritesM}[1]{\SysWritesM(#1)}

\newcommand{\SpuriousWritesM}[1]{\mathcal{SW}^M(#1)}

\newcommand{\PendingWritesM}[1]{\mathcal{PW}^M(#1)}

\newcommand{\Project}{|}
\newcommand{\Domain}{\mathsf{dom}}

\newcommand{\Image}{\mathsf{img}}

\newcommand{\True}{\mathsf{True}}
\newcommand{\False}{\mathsf{False}}

\newcommand{\System}{\mathscr{P}}
\newcommand{\ObsE}{\sim_{\Observation}}%\mathsf{DC}

\newcommand{\Process}{\mathsf{thr}}
\newcommand{\Proc}[1]{\Process(#1)}

\newcommand{\Globals}{\mathcal{G}}

\newcommand{\Aquire}{\mathsf{acquire}}
\newcommand{\Release}{\mathsf{release}}
\newcommand{\Location}[1]{\mathsf{var}(#1)}

\newcommand{\SysEvents}{\mathcal{E}}

\newcommand{\TraceSpaceMax}{\mathcal{T}}

\newcommand{\NP}{NP}

\newcommand{\Observation}{\mathsf{RF}}

\newcommand{\Upd}{\mathsf{Upd}}

\newcommand{\myblue}{blue!80!black}
\newcommand{\myred}{red!80!black}

\usetikzlibrary{arrows.meta,calc,decorations.markings,math,arrows.meta}

\makeatletter
\g@addto@macro\bfseries{\boldmath}
\makeatother

\pgfdeclarelayer{bg}    % declare background layer
\pgfsetlayers{bg,main}  % set the order of the layers (main is the standard layer)

\def \darkred {black!20!red}

\SetKw{Continue}{continue}

\SetCommentSty{mycommfont}

\newtheorem{remark}{Remark}
\setlist{leftmargin=4mm}
\begin{document}

%% Author information
%% Contents and number of authors suppressed with 'anonymous'.
%% Each author should be introduced by \author, followed by
%% \authornote (optional), \orcid (optional), \affiliation, and
%% \email.
%% An author may have multiple affiliations and/or emails; repeat the
%% appropriate command.
%% Many elements are not rendered, but should be provided for metadata
%% extraction tools.

%% Author with single affiliation.
\title{The Reads-From Equivalence for the TSO and PSO Memory Models}

\author{Truc Lam Bui}
%\orcid{nnnn-nnnn-nnnn-nnnn}             %% \orcid is optional
\authornote{Work done while the author was an intern at IST Austria.}
\affiliation{
  %\position{Position1}
  %\department{Department of Computer Science}              %% \department is recommended
  \institution{Comenius University}            %% \institution is required
  \streetaddress{Mlynsk\'{a} dolina}
  \city{Bratislava}
  %\state{State1}
  \postcode{842 48} % 818 06
  \country{Slovakia}                    %% \country is recommended
}
\email{bujtuclam@gmail.com}          %% \email is recommended

\author{Krishnendu Chatterjee}
%\orcid{nnnn-nnnn-nnnn-nnnn}             %% \orcid is optional
\affiliation{
  %\position{Position1}
  %\department{Department1}              %% \department is recommended
  \institution{IST Austria}            %% \institution is required
  \streetaddress{Am Campus 1}
  \city{Klosterneuburg}
  %\state{State1}
  \postcode{3400}
  \country{Austria}                    %% \country is recommended
}
\email{krishnendu.chatterjee@ist.ac.at}          %% \email is recommended

\author{Tushar Gautam}
\authornotemark[1]    % uncomment this only once we are no longer anonymous
%\authornote{Work done while the author was an intern at IST Austria.}
%\orcid{nnnn-nnnn-nnnn-nnnn}             %% \orcid is optional
\affiliation{
  %\position{Position1}
  %\department{Department1}              %% \department is recommended
  \institution{IIT Bombay}            %% \institution is required
  \streetaddress{Main Gate Rd, IIT Area, Powai}
  \city{Mumbai}
  %\state{State1}
  \postcode{400076}
  \country{India}                    %% \country is recommended
}
\email{tushargautam.gautam@gmail.com}          %% \email is recommended

\author{Andreas Pavlogiannis}
%\orcid{nnnn-nnnn-nnnn-nnnn}             %% \orcid is optional
\affiliation{
  %\position{Position1}
  %\department{Department1}              %% \department is recommended
  \institution{Aarhus University}            %% \institution is required
  \streetaddress{Nordre Ringgade 1}
  \city{Aarhus}
  %\state{State1}
  \postcode{8000}
  \country{Denmark}                    %% \country is recommended
}
\email{pavlogiannis@cs.au.dk}          %% \email is recommended

\author{Viktor Toman}
%\orcid{nnnn-nnnn-nnnn-nnnn}             %% \orcid is optional
\affiliation{
  %\position{Position1}
  %\department{Department1}              %% \department is recommended
  \institution{IST Austria}            %% \institution is required
  \streetaddress{Am Campus 1}
  \city{Klosterneuburg}
  %\state{State1}
  \postcode{3400}
  \country{Austria}                    %% \country is recommended
}
\email{viktor.toman@ist.ac.at}          %% \email is recommended

\begin{abstract}
The verification of concurrent programs remains an open challenge due to the non-determinism in inter-process communication.
One recurring algorithmic problem in this challenge is the consistency verification of concurrent executions.
In particular, consistency verification under a reads-from map allows to compute the \emph{reads-from (RF) equivalence} between concurrent traces, with direct applications to areas such as Stateless Model Checking (SMC).
Importantly, the RF equivalence was recently shown to be coarser than the standard Mazurkiewicz equivalence, leading to impressive scalability improvements for SMC under $\SC$ (sequential consistency).
However, for the \emph{relaxed memory} models of $\TSO$ and $\PSO$ (total/partial store order), the algorithmic problem of deciding the RF equivalence, as well as its impact on SMC, has been elusive.

%the standard equivalence has been Shasha--Snir traces, seen as an extension of the classic Mazurkiewicz equivalence from $\SC$ (sequential consistency) to $\TSO$ and $\PSO$.

In this work we solve the algorithmic problem of consistency verification for the $\TSO$ and $\PSO$ memory models given a reads-from map, denoted $\VTSOrm$ and $\VPSOrm$, respectively.
For an execution of $n$ events over $k$ threads and $\NumVariables$ variables,
we establish novel bounds that scale as $n^{k+1}$ for $\TSO$ and as $n^{k+1}\cdot \min(n^{k^2}, 2^{k\cdot \NumVariables})$ for $\PSO$.
Moreover, based on our solution to these problems, we develop an SMC algorithm under $\TSO$ and $\PSO$ that uses the RF equivalence.
The algorithm is \emph{exploration-optimal}, in the sense that it is guaranteed to explore each class of the RF partitioning exactly once, and spends polynomial time per class when $k$ is bounded.
Finally, we implement all our algorithms in the SMC tool Nidhugg, and perform a large number of experiments over benchmarks from existing literature.
Our experimental results show that our algorithms for $\VTSOrm$ and $\VPSOrm$ provide significant scalability improvements over standard alternatives.
%, often by orders of magnitude.
Moreover, when used for SMC, the RF partitioning is often much coarser than the standard Shasha--Snir partitioning for $\TSO/\PSO$, which yields a significant speedup in the model checking task.
\end{abstract}

\ccsdesc[500]{Theory of computation~Verification by model checking}
\ccsdesc[500]{Software and its engineering~Formal software verification}

%% Keywords
%% comma separated list
\keywords{concurrency, relaxed memory models, execution-consistency verification, stateless model checking} %, partial-order reduction}  %% \keywords are mandatory in final camera-ready submission

%% \maketitle
%% Note: \maketitle command must come after title commands, author
%% commands, abstract environment, Computing Classification System
%% environment and commands, and keywords command.
\maketitle

\section{INTRODUCTION}\label{sec:intro}
% {Introduction}

The formal analysis of concurrent programs is a key problem in program analysis and verification.
Scheduling non-determinism makes programs both hard to write correctly, and to analyze formally,
as both the programmer and the model checker need to account for all possible communication patterns among threads.
This non-determinism incurs an exponential blow-up in the state space of the program,
which in turn yields a significant computational cost on the verification task.

Traditional verification has focused on concurrent programs adhering to sequential consistency~\cite{Lamport79}.
Programs operating under relaxed memory semantics exhibit additional behavior compared to sequential consistency.
This makes it exceptionally hard to reason about correctness, as, besides scheduling subtleties, the formal reasoning needs to account for buffer/caching mechanisms.
Two of the most standard operational relaxed memory models in the literature are \emph{Total Store Order }($\TSO$)  and \emph{Partial Store Order }($\PSO$)~\cite{Adve96,SPARCInternational94,Sewell09,Sewell10,Alglave17,Alglave10}.

On the operational level, both models introduce subtle mechanisms via which write operations become visible to the shared memory
and thus to the whole system.
Under $\TSO$, every thread is equipped with its own buffer.
Every write to a shared variable is pushed into the buffer, and thus remains hidden from the other threads.
The buffer is flushed non-deterministically to the shared memory, at which point the writes become visible to the other threads.
The semantics under $\PSO$ are even more involved, as now every thread has one buffer \emph{per shared variable},
and non-determinism now governs not only when a thread flushes its buffers, but also which buffers are flushed.

\begin{figure}
\centering
\small
\begin{subfigure}[t]{0.07\textwidth}
\begin{align*}
%\text{Th}&\text{read}~\Process_{1}\\
\text{Th}&\text{read}_{1}\\
\hline\\[-1em]
1.~& \textcolor{\myblue}{\Write(x)}\\  % \Write_1
\\[-1.9em]
%2.~& \textcolor{\myblue}{\Write'(x)}\\  % Write_2
%\\[-1.9em]
2.~& \textcolor{\myred}{\Read(y)}\\  % Read_1
%2.~& \textcolor{\myred}{\Write_{y,1}}\\
\end{align*}
\end{subfigure}
\quad
\begin{subfigure}[t]{0.07\textwidth}
\begin{align*}
%\text{Th}&\text{read}~\Process_{2}\\
\text{Th}&\text{read}_{2}\\
\hline\\[-1em]
1.~& \textcolor{\myred}{\Write'(y)}\\
\\[-1.9em]
%2.~& \textcolor{\myred}{\Write_4(y)}\\
%\\[-1.9em]
2.~& \textcolor{\myblue}{\Read'(x)}\\  % \Read_2
\end{align*}
\end{subfigure}
\qquad\qquad\qquad\qquad
\begin{subfigure}[t]{0.07\textwidth}
\begin{align*}
%\text{Th}&\text{read}~\Process_{1}\\
\text{Th}&\text{read}_{1}\\
\hline\\[-1em]
1.~& \textcolor{\myblue}{\Write(x)}\\
\\[-1.9em]
%2.~& \textcolor{\myblue}{\Write'(x)}\\
%\\[-1.9em]
2.~& \textcolor{\myred}{\Write'(y)}\\
\end{align*}
\end{subfigure}
\quad
\begin{subfigure}[t]{0.07\textwidth}
\begin{align*}
%\text{Th}&\text{read}~\Process_{2}\\
\text{Th}&\text{read}_{2}\\
\hline\\[-1em]
1.~& \textcolor{\myred}{\Read(y)}\\
\\[-1.9em]
2.~& \textcolor{\myblue}{\Read'(x)}\\
\end{align*}
\end{subfigure}
\vspace{-3mm}
\caption{A TSO example (left) and a PSO example (right).}
\label{fig:motivating}
\end{figure}

To illustrate the intricacies under $\TSO$ and $\PSO$, consider the examples in \cref{fig:motivating}.
On the left, under $\SC$, in every execution at least one of
% $\Read_1(y)$ and $\Read_2(x)$
$\Read(y)$ and $\Read'(x)$
will observe the corresponding
% $\Write_2(y)$ and $\Write_1(x)$. % $\Write_4(y)$ and $\Write_2(x)$.
$\Write'(y)$ and $\Write(x)$.
Under $\TSO$, however, the write events may become visible on the
shared memory only after the read events have executed,
and hence both write events go unobserved.
%(such execution is illustrated in~\cref{subfig:intro2}).
Executions under $\PSO$ are even more involved, see \cref{fig:motivating} right.
Under either $\SC$ or $\TSO$, if
% $\Read_1(y)$ observes $\Write_2(y)$,
$\Read(y)$ observes $\Write'(y)$,
then % $\Read_1(y)$ observes $\Write_3(y)$
% $\Read_2(x)$ must observe $\Write_1(x)$, % $\Read_2(x)$ must observe $\Write_2(x)$
$\Read'(x)$ must observe $\Write(x)$,
% as $\Write_1(x)$ becomes visible on the shared memory before $\Write_2(y)$. % as $\Write_2(x)$ becomes visible on the shared memory before $\Write_3(y)$.
as $\Write(x)$ becomes visible on the shared memory before $\Write'(y)$.
Under $\PSO$, however, there is a single local buffer for each variable.
Hence the order in which
% $\Write_1(x)$ and $\Write_2(y)$
$\Write(x)$ and $\Write'(y)$
become visible % $\Write_2(x)$ and $\Write_3(y)$
in the shared memory can be reversed, allowing
% $\Read_1(y)$ to observe $\Write_2(y)$
$\Read(y)$ to observe $\Write'(y)$
while % $\Write_3(y)$
% $\Read_2(x)$ does not observe $\Write_1(x)$. % $\Write_2(x)$
$\Read'(x)$ does not observe $\Write(x)$.
%(such execution is illustrated in~~\cref{subfig:intro3}).

The great challenge in verification under relaxed memory is to systematically,
yet efficiently, explore all such extra behaviors of the system,
i.e., account for the additional non-determinism that comes from the buffers.
In this work we tackle this challenge for two verification tasks under $\TSO$ and $\PSO$,
namely,
(A)~for verifying the consistency of executions, and
(B)~for stateless model checking.

\Paragraph{A.~Verifying execution consistency with a reads-from function.}
One of the most basic problems for a given memory model is the verification of the consistency of program executions
with respect to the given model~\cite{ChiniS20}.
The input is a set of thread executions, where each execution performs operations accessing the shared memory.
The task is to verify whether the thread executions can be interleaved to a concurrent execution, which has the property that every read observes a specific value written by some write~\cite{Gibbons97}.
The problem is of foundational importance to concurrency, and has been studied heavily under $\SC$ %(e.g.,~)
\cite{Chen09,Cain02,Hu12}.
%while, on the technical level, it shares similarities with serializability in concurrent databases \cite{Papadimitriou79,Papadimitriou86}.

The input is often enhanced with a \emph{reads-from (RF) map}, which further specifies for each read access the write access that the former should observe.
Under sequential consistency, the corresponding problem $\VSCrm$ was shown to be $\NP$-hard in the landmark work of~\citet{Gibbons97}, while it was recently shown $\W{1}$-hard~\cite{Mathur20}.
The problem lies at the heart of many verification tasks in concurrency, such as % $\VSCrm$
dynamic analyses~\cite{Smaragdakis12,Kini17,Pavlogiannis20,Mathur20,Roemer20,Mathur21}, %  for race detection
linearizability and transactional consistency~\cite{Herlihy90,BiswasE19},
as well as SMC~\cite{Abdulla19,Chalupa17,Kokologiannakis19}.

%%%%% The above paragraph already discusses RF and its relevance, but only in context of SC. RF-SMC in other models
%is~\cite{AbdullaAJN18,Kokologiannakis19,Kokologiannakis20} but we should not cite it here, we cite it in the 'SMC under relaxed memory'
%paragraph. So in intro RF relevance is pretty scattered, maybe we unify and reiterate relevance later in the paper? Snippets to use:
%We further refer to Se ction 9: 'Why The Reads-From Equivalence Matters For SMC' of~\cite{}
%(with subsection on SMC of Approximate Data Structures).
%The RF equivalence is also very natural in related topics of concurrency, such as dynamic race detection [6,7,8](written below) and~\cite{Mathur20}. Our main technical results on RF have direct applications to that setting under TSO and PSO.
%[6] Y Smaragdakis, J Evans, C Sadowski, J Yi, and C Flanagan. 2012. Sound Predictive Race Detection in Polynomial Time. POPL 12

\SubParagraph{Executions under relaxed memory.}
The natural extension of verifying execution consistency with an RF map is from $\SC$ to relaxed memory models
such as $\TSO$ and $\PSO$,
we denote the respective problems by $\VTSOrm$ and $\VPSOrm$. % and
Given the importance of $\VSCrm$ for $\SC$, and the success in establishing both upper and lower bounds,
the complexity of $\VTSOrm$ and $\VPSOrm$ is a very natural question and of equal importance. % the problems
The verification problem is known to be $\NP$-hard for most memory models~\cite{Furbach15},
including $\TSO$ and $\PSO$, however, no other bounds are known.
Some heuristics have been developed for $\VTSOrm$~\cite{Manovit06,ZennouBEE19}, while other works study $\TSO$ executions that are also sequentially consistent~\cite{Bouajjani11,Bouajjani13}.
%While $\NP$-hardness of $\VTSOrm$ and $\VPSOrm$ is well-known~\cite{Furbach15},
%in this work we provide tigher lower bounds, as well as close-to-optimal upper bounds.
%%%%%  Unify throughout: close-to-optimal or nearly optimal or optimal?

% \SubParagraph{Example.}
% Consider the example of \cref{fig:intro}.
% \cref{subfig:intro1} shows an instance of $\VSCrm$, where solid arrows represent the order of thread executions, and dotted edges show the write that should be observed by the corresponding read.
% The input has a $\SC$-consistent execution, shown in the figure legend.
% \cref{subfig:intro2} shows an instance of $\VTSOrm$.
% Under $\TSO$, each thread is equipped with a \emph{local buffer}, and a write event $\WriteB_i$ writes to the buffer instead to the shared memory.
% At some point later, the corresponding $\WriteM_i$ event signifies the transfer of the write from the buffer to the memory.
% The instance is realizable by a $\TSO$-consistent execution, but there is no $\SC$-consistent execution.
% \cref{subfig:intro3} shows an instance of $\VPSOrm$.
% Under $\PSO$ each thread is equipped with a local buffer \emph{per variable}, that operates similarly to the $\TSO$ case.
% Notice that the instance is realizable by a $\PSO$-consistent execution, but there is no $\SC$-consistent, or even $\TSO$-consistent execution.

% \input{figures/intro}

% Overall, the intricacies introduced by the extra buffers complicate the consistency check significantly.
% In this work, we address this challenge by establishing upper and lower bounds for $\VTSOrm$ and $\VPSOrm$.

\Paragraph{B.~Stateless Model Checking.}
The most standard solution to the space-explosion problem is \emph{stateless model checking}~\cite{G96}.
Stateless model-checking methods typically explore traces rather than states of the analyzed program.
The depth-first nature of the exploration enables it to be both systematic and memory-efficient,
by storing only a few traces at any given time.
Stateless model-checking techniques have been employed successfully in several well-established
tools, e.g., VeriSoft~\cite{Godefroid97,Godefroid05} and {\sc CHESS}~\cite{Musuvathi07b}.

As there are exponentially many interleavings, a trace-based exploration
typically has to explore exponentially many traces, which is intractable in practice.
%There exist various techniques for reducing the number of explored interleavings,
%such as context bounding~\cite{QadeerR05,Lal09,TorreMP09,Musuvathi07,Chini17}.
One standard approach is the partitioning of the trace space into equivalence classes,
and then attempting to explore every class via a single representative trace.
The most successful adoption of this technique is in \emph{dynamic partial order reduction (DPOR)} techniques~\cite{Clarke99,G96,Peled93,Flanagan05}.
The great advantage of DPOR is that it handles indirect memory accesses precisely without introducing spurious interleavings.
The foundation underpinning DPOR is the famous Mazurkiewicz equivalence, which constructs equivalence classes based on the order in which traces execute conflicting memory access events.
This idea has led to a rich body of work, with improvements using symbolic techniques~\cite{Kahlon09}, context-sensitivity~\cite{Elvira17}, unfoldings~\cite{Sousa15}, effective lock handling~\cite{Kokologiannakis19b}, and others~\cite{Aronis18,Elvira18,Chatterjee19}.
The work of~\citet{Abdulla14} developed an SMC algorithm that is exploration-optimal for the Mazurkiewicz equivalence, in the sense that it explores each class of the underlying partitioning exactly once.
%This approach was later refined with the notion of observers in~\cite{Aronis18}. % Cited two lines above instead.
Finally, techniques based on SAT/SMT solvers have been used to construct
even coarser partitionings~\cite{Demsky15,HUANG15,Huang017}.
%Finally, heavyweight techniques have also been used to construct even coarser partitionings based on SAT/SMT solvers \cite{Demsky15,HUANG15,Huang017}.

\SubParagraph{The reads-from equivalence for SMC.}
A new direction of SMC techniques has been recently developed using the \emph{reads-from (RF)} equivalence to partition the trace space.
The key principle is to classify traces as equivalent based on whether read accesses observe the same write accesses.
The idea was initially explored for acyclic communication topologies~\cite{Chalupa17}, and has been recently extended to all topologies~\cite{Abdulla19}.
As the RF partitioning is guaranteed to be (even exponentially) coarser than the Mazurkiewicz partitioning, SMC based on RF has shown remarkable scalability potential~\cite{Abdulla19,AbdullaAJN18,Kokologiannakis19,Kokologiannakis20}.
The key technical component for SMC using RF is the verification of execution consistency, as presented in the previous section.
The success of SMC using RF under $\SC$ has thus rested upon new efficient methods for the problem $\VSCrm$.

\SubParagraph{SMC under relaxed memory.}
The SMC literature has taken up the challenge of model checking concurrent programs under relaxed memory.
Extensions to SMC for $\TSO/\PSO$ have been considered
by~\citet{Zhang15} using shadow threads to model memory buffers,
as well as by~\citet{Abdulla2015} using chronological traces to represent the Shasha--Snir notion of trace under relaxed memory~\cite{Shasha88}.
Chronological/Shasha--Snir traces are the generalization of Mazurkiewicz traces to $\TSO/\PSO$.
Further extensions have also been made to other memory models, namely
by~\citet{AbdullaAJN18} for the release-acquire fragment of C++11, %they do reads-from
\citet{Kokologiannakis17,Kokologiannakis19} %first does Mazurkiewicz, second reads-from
for the RC11 model~\cite{LahavVKHD17}, and
\citet{Kokologiannakis20}
for the IMM model~\cite{PodkopaevLV19},
but notably none for $\TSO$ and $\PSO$ using the RF equivalence.
Given the advantages of the RF equivalence for SMC under
$\SC$~\cite{Abdulla19},
release-acquire~\cite{AbdullaAJN18},
RC11~\cite{Kokologiannakis19} and
IMM~\cite{Kokologiannakis20},
a very natural standing question is whether RF can be used for effective SMC under $\TSO$ and $\PSO$.
Here we tackle this challenge.
%model-checking also under different relaxed memory models,
%and whether its benefits under $\SC$/release-acquire/RC11/IMM
%remain also in the new relaxed settings.
%In this work, we address this question for the standard
%memory models
%Further extensions to other memory models have also been made by~\citet{Kokologiannakis17,Abdulla19,Kokologiannakis19}.

\subsection{Our Contributions}

Here we outline the main results of our work.
We refer to \cref{sec:summary} for a formal presentation.

\Paragraph{A. Verifying execution consistency for $\TSO$ and $\PSO$.}
Our first set of results and the main contribution of this paper is on the problems $\VTSOrm$ and $\VPSOrm$ for verifying $\TSO$- and $\PSO$-consistent executions, respectively.
Consider an input to the corresponding problem that consists of $k$ threads and $n$ operations, where each thread executes write and read operations, as well as \emph{fence} operations that flush each thread-local buffer to the main memory.
Our results are as follows.
\begin{enumerate}[noitemsep,topsep=0pt,partopsep=0px]
\item\label{item:resA1} We present an algorithm that solves $\VTSOrm$ in $O(k\cdot n^{k+1})$ time.
The case of $\VSCrm$ is solvable in $O(k\cdot n^{k})$ time~\cite{Abdulla19,BiswasE19,Mathur20}.
%%%%% Is the Mathur citation supposed to be here? If yes, can we also claim that naive extension of this work is $n^{2k}$ TSO not-poly PSO?
%, which stems from having $k$ event chains (i.e., subsets of totally ordered events), one for each thread.
%Interestingly,
Although for $\TSO$ %the number of chains increases to $2\cdot k$ (as there are $k$ additional chains, one for each buffer),
there are $k$ additional buffers,
our result shows that the complexity is only minorly impacted by an additional factor $n$, as opposed to $n^k$.
\item\label{item:resA2} We present an algorithm that solves $\VPSOrm$ in $O(k\cdot n^{k+1}\cdot \min(n^{k\cdot (k-1)}, 2^{k\cdot \NumVariables}))$ time, where $\NumVariables$ is the number of variables.
%Note that under $\PSO$ %the number of chains increases from $k$ to $k\cdot (\NumVariables+1)$ (as we have one buffer per thread \emph{per variable}),
%we have $k\cdot \NumVariables$ buffers,
%and thus naive extensions of the methods by~\citet{Abdulla19,BiswasE19,Mathur20} lead to solutions superpolynomial
%even when the number of threads is bounded.
%%%%%  Check not-poly sentence above, and yes-poly in brackets below
%(we might delete yes-poly as it is repeated in below paragraph). Also check summary.
Note that even though there are $k\cdot \NumVariables$ buffers,
one of our two bounds is independent of $\NumVariables$ and thus yields polynomial time when the number of threads is bounded.
%(i)~our first bound  bounded, and
%(ii)~our second bound only exponentiates a constant to the product $\NumVariables\cdot \NumThreadsFences$.
Moreover, our bound collapses to $O(k\cdot n^{k+1})$ when there are no fences, and hence this case is no more difficult that $\VTSOrm$.
% \item\label{item:resA3} Finally, we study the hardness of $\VTSOrm$ and $\VPSOrm$.
% When fences are present, placing a fence after each write operation forces a sequentially consistent execution.
% Since $\VSCrm$ is known to be $\W{1}$-hard parameterized by $k$~\cite{Mathur20}, the same lower bound holds for $\VTSOrm$ and $\VPSOrm$ in the presence of fences.
% Here we show a stronger result, namely that $\W{1}$-hardness for $\VTSOrm$ and $\VPSOrm$ holds even without fences.
\end{enumerate}

%Note that our results in \cref{item:resA1} and \cref{item:resA2} yield polynomial upper bounds for a constant number of threads.
%We establish a lower bound demonstrating that a dependence on $k$ in the exponent of $n$ is necessary, and thus our upper bounds are nearly optimal.
% Moreover, our lower bound in \cref{item:resA3} shows that a dependence on $k$ in the exponent of $n$ is necessary, and thus our upper bounds are nearly optimal.

\Paragraph{B. Stateless model checking for $\TSO$ and $\PSO$ using the reads-from equivalence (RF).}
Our second contribution is an algorithm $\DCTSOPSOM$ for SMC under $\TSO$ and $\PSO$ using the RF equivalence.
The algorithm is based on the reads-from algorithm for %Sequential Consistency
$\SC$~\cite{Abdulla19} and uses our solutions to $\VTSOrm$ and $\VPSOrm$ for visiting each class of the respective partitioning.
Moreover, $\DCTSOPSOM$ is \emph{exploration-optimal}, in the sense that it explores only maximal traces and further it is guaranteed to explore each class of the RF partitioning exactly once.
%Our second algorithm, $\DCTSOPSOP$, is \emph{space-optimal}, in the sense that it is guaranteed to use polynomial space throughout the exploration, regardless of the (possibly exponential) size of the partitioning.
%The algorithm $\DCTSOPSOM$ uses our solutions to $\VTSOrm$ and $\VPSOrm$ for visiting each class of the partitioning.
%%%%%  In fact, any algorithm or tool requiring to solve queries of $\VTSOrm$ and $\VPSOrm$
%(e.g.,~\cite{Kokologiannakis19,Kokologiannakis20,Chalupa17}) can immediately utilize our solutions to $\VTSOrm$ and $\VPSOrm$ as a black box.
For the complexity statements, let $k$ be the total number of threads and $n$ be the number of events of the longest trace.
The time spent by $\DCTSOPSOM$ per class of the RF partitioning is
%, for each of the two algorithms,
%polynomial in $n$ when $k$ is bounded, specifically
\begin{enumerate}[noitemsep,topsep=0pt,partopsep=0px]
\item $n^{O(k)}$ time, for the case of $\TSO$, and
\item $n^{O(k^2)}$ time, for the case of $\PSO$.
\end{enumerate}
%Note that the complexity for $\PSO$ with no fence instructions equals the one for $\TSO$.
%Note that when no threads have fence instructions, the complexity for $\PSO$ equals the one for $\TSO$.
Note that the time complexity per class is polynomial in $n$ when $k$ is bounded.

\Paragraph{C. Implementation and experiments.}
We have implemented $\DCTSOPSOM$ in the stateless model checker Nidhugg~\cite{Abdulla2015},
and performed an evaluation on an extensive set of benchmarks from the recent literature. % experimental
Our results show that our algorithms for $\VTSOrm$ and $\VPSOrm$ provide significant scalability improvements over standard alternatives, often by orders of magnitude.
Moreover, when used for SMC, the RF partitioning is often much coarser than the standard Shasha--Snir partitioning for $\TSO/\PSO$, which yields a significant speedup in the model checking task.

\section{PRELIMINARIES}\label{sec:prel}
% {Preliminaries}

\smallskip\noindent{\bf General notation.}
Given a natural number $i\geq 1$, we let $[i]$ be the set $\{ 1,2,\dots, i \}$.
Given a map $f\colon X\to Y$, we let $\Domain(f)=X$ and $\Image(f)=Y$ denote the domain and image of $f$, respectively.
We represent maps $f$ as sets of tuples $\{ (x, f(x))\}_x$.
Given two maps $f_1, f_2$ over the same domain $X$, we write $f_1=f_2$ if  for every $x\in X$
we have $f_1(x)=f_2(x)$.
Given a set $X'\subset X$, we denote by $f\Project X'$ the restriction of $f$ to $X'$.
A binary relation $\sim$ on a set $X$ is an {\em equivalence} iff $\sim$ is reflexive, symmetric and transitive.
% and some $x\in X$, we denote by
%$\Class{x}{E}=\{y\in X\ |\ x\sim_E y\}$ the equivalence class of $x$ under $\sim_E$.
We denote by $X/\sim$ the \emph{quotient} (i.e., the set of all equivalence classes) of $X$ under $\sim$.
%We denote by $X/\sim$ the \emph{quotient set} of $X$ under $\sim$, i.e., the set of all equivalence classes of $X$ under $\sim$.

\subsection{Concurrent Model under $\TSO$/$\PSO$}\label{subsec:model}

Here we describe the computational model of concurrent programs with shared memory under the Total Store Order ($\TSO$) and Partial Store Order ($\PSO$) memory models.
We follow a standard exposition, similarly to~\citet{Abdulla2015,Huang16}.
We first describe $\TSO$ and then extend our description to $\PSO$.

\Paragraph{Concurrent program with Total Store Order.}
We consider a concurrent program $\System=\{ \Process_i \}_{i=1}^k$ of $k$ threads.
The threads communicate over a shared memory %that consists of a set $\Globals$ of global variables.
$\Globals$ of global variables.
Each thread additionally owns a \emph{store buffer}, which is a FIFO queue for storing updates of variables to the shared memory.
Threads execute \emph{events} of the following types.
\begin{enumerate}[noitemsep,topsep=0pt,partopsep=0px]
\item A \emph{buffer-write event} $\WriteB$ enqueues into the local store buffer an update
that wants to write a value $v$ to a global variable $x\in \Globals$.
\item A \emph{read event} $\Read$ reads the value $v$ of a global variable $x\in \Globals$.
The value $v$ is the value of the most recent local buffer-write event, if one still exists in the buffer, otherwise $v$ is the value of $x$ in the shared memory.
\end{enumerate}
Additionally, whenever a store buffer of some thread is nonempty, the respective thread can execute the following. % non-deterministically
\begin{enumerate}[noitemsep,topsep=0pt,partopsep=0px]
\setcounter{enumi}{2}
\item A \emph{memory-write event} $\WriteM$ that dequeues the oldest update from the local buffer and performs the corresponding write-update on the shared memory.
\end{enumerate}
Threads can also flush their local buffers into the memory using fences.
\begin{enumerate}[noitemsep,topsep=0pt,partopsep=0px]
\setcounter{enumi}{3}
\item A \emph{fence event} $\Fence$ blocks the corresponding thread until its store buffer is empty.
\end{enumerate}
Finally, threads can execute local events that are not modeled explicitly, as usual.
We refer to all non-memory-write events as \emph{thread events}.
Following the typical setting of stateless model
checking~\cite{Flanagan05,Abdulla14,Abdulla2015,Chalupa17},
each thread of the program $\System$ is deterministic, and further $\System$ is bounded,
meaning that all executions of $\System$ are finite and the number of events of $\System$'s
longest execution is a parameter of the input.
%For simplicity of presentation, we do not consider lock events in our model.
%Later, in \cref{rem:locks} we show how locks can be naturally handled by our approach.

Given an event $\Event$, we denote by $\Proc{\Event}$ its thread and by $\Location{\Event}$ its global variable.
We denote by $\SysEvents$ the set of all events,
by $\SysReads$ the set of read events, by $\SysWritesB$ the set of buffer-write events,
by $\SysWritesM$ the set of memory-write events, and
by $\SysFences$ the set of fence events.
Given a buffer-write event $\WriteB \in \SysWritesB$ and its corresponding memory-write $\WriteM \in \SysWritesM$, we let
$\Wpair{\Write} = (\WriteB,\WriteM)$ be the two-phase write event, and we denote
$\Proc{\Wpair{\Write}} = \Proc{\WriteB} = \Proc{\WriteM}$ and $\Location{\Wpair{\Write}} = \Location{\WriteB} = \Location{\WriteM}$.
We denote by $\SysWrites$ the set of all such two-phase write events.
Given two events $\Event_1, \Event_2\in \SysReads\cup \SysWritesM$, we say that they \emph{conflict}, denoted $\Confl{\Event_1}{\Event_2}$,
if they access the same global variable and at least one of them is a memory-write event.

\Paragraph{Proper event sets.}
Given a set of events $X\subseteq \SysEvents$,
we write $\Reads{X}=X\cap \SysReads$ for the set of read events of $X$, and similarly
$\WritesB{X}=X\cap \SysWritesB$ and $\WritesM{X}=X\cap \SysWritesM$ for the buffer-write and memory-write events of $X$, respectively.
We also denote by $\LocalEvents{X}=X\setminus\WritesM{X}$ the
thread events (i.e., the non-memory-write events) of $X$.
We write $\Writes{X} = (X \times X)\cap \SysWrites$ for the set of two-phase write events in $X$.
We call $X$ \emph{proper} if $\WriteB \in X$ iff $\WriteM \in X$ for each $(\WriteB,\WriteM) \in \SysWrites$.
Finally, given a set of events $X\subseteq \SysEvents$ and a thread $\Process$,
we denote by $X_{\Process}$ and $X_{\neq\Process}$ the events of $\Process$,
and the events of all other threads in $X$, respectively.

\Paragraph{Sequences and Traces.}
Given a sequence of events $\Seq=\Event_1,\dots,\Event_j$, we denote by
$\Events{\Seq}$ the set of events that appear in $\Seq$.
We further denote $\Reads{\Seq} = \Reads{\Events{\Seq}}$,
$\WritesB{\Seq} = \WritesB{\Events{\Seq}}$,
$\WritesM{\Seq} = \WritesM{\Events{\Seq}}$, and
$\Writes{\Seq} = \Writes{\Events{\Seq}}$.
Finally we denote by $\epsilon$ an empty sequence.

Given a sequence $\Seq$ and two events
$\Event_1, \Event_2 \in \Events{\Seq}$, we write
$\Event_1 <_\Seq \Event_2$ when $\Event_1$ appears before $\Event_2$
in $\Seq$, and $\Event_1 \leq_\Seq \Event_2$ to denote that
$\Event_1 <_\Seq \Event_2$ or $\Event_1 = \Event_2$.
Given a sequence $\Seq$ and a set of events $A$, we denote by
$\Seq \Project A$ the \emph{projection} of $\Seq$ on $A$, which is
the unique sub-sequence of $\Seq$ that contains all events of
$A \cap \Events{\Seq}$, and only those.
Given a sequence $\Seq$ and an event $\Event \in \Events{\Seq}$,
we denote by $\pre{\Seq}{\Event}$ the prefix up until and including
$\Event$, formally
$\Seq \Project \{ \Event' \in \Events{\Seq}\, |\, \Event' \leq_\Seq \Event\}$.
Given two sequences $\Seq_1$ and $\Seq_2$, we denote by
$\Seq_1 \Concat \Seq_2$ the sequence that results in appending
$\Seq_2$ after $\Seq_1$.

A (concrete, concurrent) \emph{trace} is a sequence of events $\Trace$ that
corresponds to a concrete valid execution of $\System$ under standard
semantics~\cite{Shasha88}. We let $\Enabled(\Trace)$ be the set of enabled
events after $\Trace$ is executed, and call $\Trace$ \emph{maximal} if
$\Enabled(\Trace)=\emptyset$.
A concrete \emph{local trace} $\LocalTrace$ is a sequence of thread events of the same thread.

\Paragraph{Reads-from functions.}
Given a proper event set $X\subseteq \SysEvents$, a \emph{reads-from function} % set of events
over $X$ is a function that maps each read event of $X$ to some two-phase
write event of $X$ accessing the same global variable.
Formally, $\Observation\colon\Reads{X}\to \Writes{X}$, where
$\Location{\Read} = \Location{\Observation(\Read)}$
for all $\Read \in \Reads{X}$.
% With a small abuse of notation,
Given a buffer-write event $\WriteB$ (resp. a memory-write event $\WriteM$),
we write $\Observation(\Read)= (\WriteB, \_)$ (resp. $\Observation(\Read)= (\_, \WriteM)$)
to denote that $\Observation(\Read)$ is a two-phase write for which $\WriteB$ (resp. $\WriteM$)
is the corresponding buffer-write (resp. memory-write) event.

Given a sequence of events $\Seq$
where the set $\Events{\Seq}$ is proper, we define the
\emph{reads-from function} of $\Seq$, denoted
$\RF{\Seq}\colon\Reads{\Seq}\to \Writes{\Seq}$, as follows.
Given a read event $\Read \in \Reads{\Seq}$, consider the set $\Upd$
of enqueued conflicting updates in the same thread that have not yet
been dequeued, i.e.,
$\Upd = \{ (\WriteB,\WriteM) \in (\Writes{\Seq})_{\Proc{\Read}}\ |\ \Confl{\WriteM}{\Read},\; \WriteB <_\Seq \Read <_\Seq \WriteM\}$.
Then, $\RF{\Seq}(\Read) = (\WriteB',\WriteM')$,
where one of the two cases happens:
\begin{itemize}[noitemsep,topsep=0pt,partopsep=0px]
\item $\Upd \neq \emptyset$,
and $(\WriteB',\WriteM') \in \Upd$ is the latest in
$\Seq$, i.e., for each $(\WriteB'',\WriteM'') \in \Upd$ we have
$\WriteB'' \leq_\Seq \WriteB'$.
\item $\Upd = \emptyset$, and
$\WriteM' \in \WritesM{\Seq},\, \Confl{\WriteM'}{\Read},\, \WriteM' <_\Seq \Read$
is the latest memory-write (of any thread) conflicting with $\Read$ and
occurring before $\Read$ in $\Seq$, i.e.,
for each $\WriteM'' \in \WritesM{\Seq}$ such that
$\Confl{\WriteM''}{\Read}$ and $\WriteM'' <_\Seq \Read$,
we have $\WriteM'' \leq_\Seq \WriteM'$.
\end{itemize}
Notice how relaxed memory comes into play in the above definition,
as $\RF{\Seq}(\Read)$ does not record which of the two above
%aforementioned
cases actually happened.

\Paragraph{Partial Store Order and Sequential Consistency.}
The memory model of Partial Store Order ($\PSO$) is more relaxed than $\TSO$. % strictly
On the operational level, each thread is equipped with a store buffer for each global variable, % separate
rather than a single buffer for all global variables. Then, at any point during execution,
a thread can non-deterministically dequeue and perform the oldest update from any of its nonempty store buffers.
The notions of events, traces and reads-from functions remain the same for $\PSO$ as defined for $\TSO$. % Crucially,
%When presenting our model-checking techniques for TSO, we describe how they are adapted to handle $\PSO$.\footnote{Although we have not
The Sequential Consistency ($\SC$) memory model can be simply thought % Finally,
of as a model where each thread flushes its buffer immediately
after a write event, e.g., by using a fence.

\Paragraph{Concurrent program semantics.}
The semantics of $\System$ are defined by means of a transition system over a state space of global states.
A global state consists of (i) a memory function that maps every global variable to a value,
(ii) a local state for each thread, which contains the values of the local variables of the thread, and
(iii) a local state for each store buffer, which captures the contents of the queue.
We consider the standard setting with the $\TSO/\PSO$ memory model, and refer to~\citet{Abdulla2015} for formal details.
As usual in stateless model checking, we focus on concurrent programs with
%deterministic threads and
acyclic state spaces.

%\Paragraph{Trace Spaces.}
\Paragraph{Reads-from trace partitioning.}
Given a concurrent program $\System$ and a memory model
$\MemoryModel\in \{\SC, \TSO, \PSO\}$, we denote by
$\TraceSpaceMax_{\MemoryModel}$
%$\TraceSpace_{\MemoryModel}$  (resp., $\TraceSpaceMax_{\MemoryModel}$)
%$\TraceSpace$  (resp., $\TraceSpaceMax$)
%the set of traces (resp., maximal traces)
the set of maximal traces of the %concurrent program
program
$\System$ under the respective memory model.
We call two traces $\Trace_1$ and $\Trace_2$ \emph{reads-from equivalent}
if $\Events{\Trace_1} = \Events{\Trace_2}$ and
$\RF{\Trace_1} = \RF{\Trace_2}$. The corresponding
\emph{reads-from equivalence} $\ObsE$ partitions the trace space into
equivalence classes $\TraceSpaceMax_{\MemoryModel}/\ObsE$ and
we call this the \emph{reads-from partitioning} (or RF partitioning).
Traces in the same class of the RF partitioning visit the same
set of local states in each thread, and thus the RF partitioning
is a sound partitioning for local state
reachability~\cite{Abdulla19,Chalupa17,Kokologiannakis19}.

\subsection{Partial Orders}\label{subsec:partialorders}

Here we present relevant notation around partial orders.

\Paragraph{Partial orders.}
Given a set of events $X\subseteq \SysEvents$, a \emph{(strict) partial order} $P$ over $X$ is an irreflexive,
antisymmetric and transitive relation over $X$ (i.e., $<_{P}\,\subseteq X\times X$).
Given two events $\Event_1,\Event_2\in X$, we write $\Event_1\leq_P \Event_2$ to denote that $\Event_1<_P\Event_2$ or $\Event_1=\Event_2$.
Two distinct events $\Event_1,\Event_2\in X$ are \emph{unordered} by $P$, denoted $\Unordered{\Event_1}{P}{\Event_2}$,
if neither $\Event_1<_{P}\Event_2$ nor $\Event_2<_{P} \Event_1$, and \emph{ordered} (denoted $\Ordered{\Event_1}{P}{\Event_2}$) otherwise.
% THIS IS UNUSED?
% An event $\Event_1$ is \emph{maximal} in $P$ if there is no $\Event_2$ such that $\Event_1<_{P}\Event_2$.
%
Given a set $Y\subseteq X$, we denote by $P\Project Y$ the \emph{projection} of $P$ on the set $Y$,
where for every pair of events $\Event_1, \Event_2\in Y$,
we have that $\Event_1<_{P\Project Y} \Event_2$ iff $\Event_1<_{P} \Event_2$.
Given two partial orders $P$ and $Q$ over a common set $X$, we say that $Q$ \emph{refines} $P$, denoted by $Q\Refines P$, if
for every pair of events $\Event_1, \Event_2\in X$, if $\Event_1<_{P}\Event_2$ then $\Event_1<_{Q}\Event_2$.
A \emph{linearization} of $P$ is a total order that refines $P$.

\Paragraph{Lower  sets.}
Given a pair $(X,P)$, where $X$ is a set of events and $P$ is a partial order over $X$,
a \emph{lower set} of $(X,P)$ is a set $Y\subseteq X$ such that
for every event $\Event_1\in Y$ and event $\Event_2\in X$ such that $\Event_2\leq_{P} \Event_1$, we have $\Event_2\in Y$.

\Paragraph{The program order $\TO$.}
The \emph{program order} $\TO$ of $\System$ is a partial order
$<_{\TO}\subseteq \SysEvents\times \SysEvents$ that defines
a fixed order between some pairs of events of the same thread.
Given any (concrete) trace $\Trace$ and thread $\Process$,
the buffer-writes, reads, and fences of $\Process$ that appear in $\Trace$
are fully ordered in $\TO$ the same way as they are ordered in $\Trace$.
Further, for each thread $\Process$, the program order $\TO$ satisfies the
following conditions:
\begin{itemize}[noitemsep,topsep=0pt,partopsep=0px]
\item $\WriteB<_\TO \WriteM$ for each $(\WriteB,\WriteM) \in \SysWrites_\Process$.
\item $\WriteB<_\TO \Fence$ iff $\WriteM<_\TO \Fence$ for each
$(\WriteB,\WriteM) \in \SysWrites_\Process$ and fence event $\Fence \in \SysFences_\Process$.
\item $\WriteB_1<_\TO \WriteB_2$ iff $\WriteM_1<_\TO \WriteM_2$ for each
$(\WriteB_i,\WriteM_i) \in \SysWrites_\Process, i \in \{1,2\}$.
In PSO, this condition is enforced only when
$\Location{(\WriteB_1,\WriteM_1)} = \Location{(\WriteB_2,\WriteM_2)}$.
\end{itemize}

A sequence $\Seq$ is \emph{well-formed} if it respects the program
order, i.e., $\Seq\Refines \TO\Project \Events{\Seq}$. % formally,
Naturally, every trace $\Trace$ is well-formed, as
it corresponds to a concrete valid program execution.

\section{SUMMARY OF RESULTS}\label{sec:summary}% and Novelty
% {Summary of Results}

%In this section
Here we present formally the main results of this paper. %and highlight the novelty
In later sections we present the details, algorithms and examples. %i.e., proofs
Due to space restrictions, proofs appear in the appendix.

\Paragraph{A. Verifying execution consistency for $\TSO$ and $\PSO$.}
%\subsection{Verifying execution consistency for $\TSO$/$\PSO$}\label{subsec:sum_verify}
Our first set of results and the main contribution of this paper is on the problems $\VTSOrm$ and $\VPSOrm$ for verifying $\TSO$- and $\PSO$-consistent executions, respectively.
% In each case, the input is a pair $(X, \Observation)$, where $X$ is a proper set of events of $\System$,
% and $\Observation\colon \Reads{X} \to \Writes{X}$ is a reads-from function.
% The task is to decide whether there exists a linearization $\Trace$ of $(X, \TO)$
% such that $\Observation_{\Trace}=\Observation$,
% where $\Observation_{\Trace}$ is wrt $\TSO/\PSO$ memory semantics.
The corresponding problem $\VSCrm$ for Sequential Consistency ($\SC$) was recently shown to be in polynomial time for a constant number of threads~\cite{Abdulla19,BiswasE19}.
The solution for $\SC$ is obtained by essentially enumerating all the $n^k$ possible lower sets of the program order $(X,\TO)$, where $k$ is the number of threads,
and hence yields a polynomial when $k=O(1)$.
For $\TSO$, the number of possible lower sets is $n^{2\cdot k}$, since there are $k$ threads and $k$ buffers
(one for each thread). %which yields $2\cdot k$ chains of events of $X$.
For $\PSO$, the number of possible lower sets is $n^{k\cdot (\NumVariables+1)}$, where $\NumVariables$ is the number of variables, since there are $k$ threads and
$k\cdot \NumVariables$ buffers
($\NumVariables$ buffers for each thread). %, and that yields $k\cdot (\NumVariables+1)$ chains of events of $X$.
Hence, following an approach similar to \citet{Abdulla19,BiswasE19} would yield
a running time of a polynomial with degree $2\cdot k$ for $\TSO$,
and with degree $k\cdot (\NumVariables+1)$ for $\PSO$
(thus the solution for $\PSO$ is not polynomial-time even when the number of threads is bounded).
In this work we show that both problems can be solved significantly faster.

%\Paragraph{Results.}
%Our first result establishes an upper-bound for $\VTSOrm$.
%Our second result establishes two upper-bounds for $\VPSOrm$.
%Finally, our third result establishes a lower-bound for both $\VTSOrm$ and $\VPSOrm$.

%\smallskip
%\vspace{-1mm}
\begin{restatable}{theorem}{themvtso}\label{them:vtso}
$\VTSOrm$ for $n$ events and $k$ threads is solvable in $O(k\cdot n^{k+1})$ time.
\end{restatable}
%\vspace{-3mm}

%Our second result establishes two upper-bounds for $\VPSOrm$.
%\smallskip
\begin{restatable}{theorem}{themvpso}\label{them:vpso}
$\VPSOrm$ for $n$ events, $k$ threads and $d$ variables is solvable in $O(k\cdot n^{k+1}\cdot  \min(n^{k\cdot (k-1)}, 2^{k\cdot \NumVariables}))$.
Moreover, if there are no fences, the problem is solvable in $O(k\cdot n^{k+1})$ time.
\end{restatable}
%\vspace{-1mm}

\Paragraph{Novelty.}
For $\TSO$, \cref{them:vtso} yields an improvement of order $n^{k-1}$ compared to the naive $n^{2\cdot k}$ bound.
For $\PSO$, perhaps surprisingly, the first upper-bound of~\cref{them:vpso} does not depend on the number of variables.
Moreover, when there are no fences, the cost for $\PSO$ is the same as for $\TSO$ (with or without fences).

\Paragraph{B. Stateless Model Checking for $\TSO$ and $\PSO$.}
%\subsection{Stateless Model Checking for $\TSO$ and $\PSO$}\label{subsec:sum_dpor}
Our second result concerns stateless model checking (SMC)
%of concurrent programs
under $\TSO$ and $\PSO$. % set of results
We introduce an SMC algorithm $\DCTSOPSOM$ that explores the RF partitioning in the $\TSO$ and $\PSO$ settings, as stated in the following theorem.

%\Paragraph{Results.}
%We introduce two algorithms for reads-from SMC,
%namely $\DCTSOPSOM$ and $\DCTSOPSOP$, both of which handle both $\TSO$ and $\PSO$ memory models.
%In particular, we obtain the following theorems.

%\smallskip
%\vspace{-2mm}
\begin{restatable}{theorem}{themdctsopsomaximal}\label{them:dctsopso_maximal}
%Consider a bounded
Consider a concurrent program $\System$ with $k$ threads and $d$ variables, under a memory model $\MemoryModel\in \{\TSO,\PSO\}$
with trace space $\TraceSpaceMax_{\MemoryModel}$ and $n$ being the number of events of the longest trace in $\TraceSpaceMax_{\MemoryModel}$.
$\DCTSOPSOM$ is a sound, complete and exploration-optimal algorithm
for local state reachability in $\System$, i.e., it explores only maximal traces and visits
each class of the RF partitioning exactly once. The time complexity is
$O\left(\alpha \cdot \left|\TraceSpaceMax_{\MemoryModel}/\sim_{\Observation}\right|\right)$, where
\begin{enumerate}[noitemsep,topsep=0pt,partopsep=0px]
\item $\alpha=n^{O(k)}$ under $\MemoryModel=\TSO$, and
\item $\alpha=n^{O(k^2)}$ under $\MemoryModel=\PSO$.
\end{enumerate}
\end{restatable}
%\vspace{-2mm}

% \smallskip
% \begin{restatable}{theorem}{themdctsopsopartial}\label{them:dctsopso_partial}
%Consider a bounded
% Consider a concurrent program $\System$ with $k$ threads and $d$ variables, under a memory model $\MemoryModel\in \{\TSO,\PSO\}$
% with trace space $\TraceSpaceMax_{\MemoryModel}$ and $n$ being the length of the longest trace in $\TraceSpaceMax_{\MemoryModel}$.
% $\DCTSOPSOP$ is a sound and complete algorithm for local state reachability in $\System$, with $\Poly(n)$ space complexity, and with time complexity
% $O\left(\alpha \cdot \left|\TraceSpaceMax_{\MemoryModel}/\sim_{\Observation}\right|\right)$, where
% \begin{enumerate}
% \item $\alpha=n^{O(k)}$ under $\MemoryModel=\TSO$, and
% \item $\alpha=n^{O(k^2)}$ under $\MemoryModel=\PSO$.
% \end{enumerate}
% \end{restatable}

An algorithm with RF exploration-optimality in $\SC$ is presented by~\citet{Abdulla19}.
Our $\DCTSOPSOM$ algorithm generalizes the above approach to achieve RF exploration-optimality
in the relaxed memory models $\TSO$ and $\PSO$.
Further, the time complexity of $\DCTSOPSOM$ per class of RF partitioning
is equal between $\PSO$ and $\TSO$ for programs with no fence instructions.

$\DCTSOPSOM$ uses the verification algorithms developed in \cref{them:vtso} and \cref{them:vpso} as black-boxes to
decide whether any specific class of the RF partitioning is $\TSO$- or $\PSO$-consistent, respectively.
We remark that these theorems can potentially be used as black-boxes to other SMC algorithms that explore the RF partitioning (e.g.,~\citet{Chalupa17,Kokologiannakis19,Kokologiannakis20}).

\section{VERIFYING TSO AND PSO EXECUTIONS WITH A READS-FROM FUNCTION}\label{sec:verifyingtsopso}
% {Verifying TSO and PSO Executions with a Reads-From Function}

In this section we tackle the verification problems $\VTSOrm$ and $\VPSOrm$.
In each case, the input is a pair $(X, \Observation)$, where $X$ is a proper set of events of $\System$,
and $\Observation\colon \Reads{X} \to \Writes{X}$ is a reads-from function.
The task is to decide whether there exists a %witness
trace $\Trace$ that is a linearization of $(X, \TO)$ with % such that
$\Observation_{\Trace}=\Observation$,
where $\Observation_{\Trace}$ is wrt $\TSO/\PSO$ memory semantics.
%In the following,
In case such $\Trace$ exists, we say that $(X, \Observation)$ is \emph{realizable}
and $\Trace$ is its \emph{witness} trace.
We first define some relevant notation, and then
%prove \cref{them:vtso}, \cref{them:vpso} and \cref{them:vtspso_w1hard}.
establish upper bounds for $\VTSOrm$ and $\VPSOrm$, i.e.,
\cref{them:vtso} and \cref{them:vpso}.
%We establish lower bounds to show that our upper bounds
%are close-to-optimal, see~\cref{sec:app_hardness}.
%Due to space restrictions we relegate the presentation of lower
%bounds to~\cref{sec:app_hardness}.

\Paragraph{Held variables.}
Given a trace $\Trace$ and a memory-write $\WriteM \in \WritesM{\Trace}$ present in the trace,
we say that $\WriteM$ \emph{holds} variable $x=\Location{\WriteM}$ in $\Trace$ if the following hold.
\begin{enumerate}[noitemsep,topsep=0pt,partopsep=0px]
\item $\WriteM$ is the last memory-write event of $\Trace$ on variable $x$.
\item There exists a read event $\Read\in X\setminus\Events{\Trace}$ such that $\Observation(\Read)=(\_, \WriteM)$. % event
\end{enumerate}
We similarly say that the thread $\Proc{\WriteM}$ \emph{holds} $x$ in $\Trace$.
Finally, a variable $x$ is \emph{held} in $\Trace$ if it is held by some thread in $\Trace$.
Intuitively, $\WriteM$ holds $x$ until % $\Location{\WriteM}$
all reads that need to read-from $\WriteM$ get executed.

\Paragraph{Witness prefixes.}
Throughout this section, we use the notion of witness prefixes.
Formally, a \emph{witness prefix} is a trace $\Trace$ that can be extended to a trace $\Trace^*$ that realizes $(X, \Observation)$, under the respective memory model.
Our algorithms for $\VTSOrm$ and $\VPSOrm$ operate by constructing
traces $\Trace$ such that if $(X,\Observation)$ is realizable, then
$\Trace$ is a witness prefix that can be extended with the remaining
events and finally realize $(X,\Observation)$.

%In the following we establish upper and lower bounds for $\VTSOrm$ and $\VPSOrm$.
Throughout, we assume wlog that whenever $\Observation(\Read)=(\WriteB,\WriteM)$ with $\Proc{\Read}=\Proc{\WriteB}$,
then $\WriteB$ is the last buffer-write on $\Location{\WriteB}$ before $\Read$ in their respective thread.
Clearly, if this condition does not hold, then the corresponding pair $(X,\Observation)$ is not realizable
in TSO nor PSO.
%in any of these memory models.

\subsection{Verifying TSO Executions}\label{subsec:verifyingtso}

In this section we establish \cref{them:vtso}, i.e.,
we present an algorithm $\AlgoTSO$ that solves $\VTSOrm$ in $O(k\cdot n^{k+1})$ time.
The algorithm relies crucially on the notion of $\TSO$-executable events, defined below.
Throughout this section we consider fixed an instance $(X,\Observation)$ of $\VTSOrm$,
and all traces $\Trace$ considered in this section are such that $\Events{\Trace}\subseteq X$.

\Paragraph{$\TSO$-executable events.}
Consider a trace $\Trace$.
An event $\Event\in X\setminus \Events{\Trace}$ is \emph{$\TSO$-executable} (or executable for short)
in $\Trace$ if $\Events{\Trace}\cup \{ \Event \}$ is a lower set of $(X,\TO)$ and the following conditions hold.

\begin{enumerate}[noitemsep,topsep=0pt,partopsep=0px]
\item \emph{If $\Event$ is a read event $\Read$}, let $\Observation(\Read)=(\WriteB, \WriteM)$.
If $\Proc{\Read}\neq \Proc{\WriteM}$, then $\WriteM\in \Trace$.
\item \emph{If $\Event$ is a memory-write event $\WriteM$} then the following hold.
\begin{enumerate}[noitemsep,topsep=0pt,partopsep=0px]
\item\label{item:tso_execw1} Variable $\Location{\WriteM}$ is not held in $\Trace$.
\item\label{item:tso_execw2}
Let $\Read\in \Reads{X}$ be an arbitrary read with $\Observation(\Read)=(\WriteB, \WriteM)$
and $\Proc{\Read}\neq \Proc{\WriteM}$. For each two-phase write $(\WriteB', \WriteM')$
with $\Location{\Read} = \Location{\WriteB'}$ and $\WriteB'<_{\TO}\Read$, we have $\WriteM'\in \Trace$.
%There is no read  $\Read\in \Reads{X}$ with $\Observation(\Read)=(\WriteB, \WriteM)$ for which
%there exists a two-phase write  $(\WriteB', \WriteM')$  with
%(i)~$\Location{\Read} = \Location{\WriteB'}$,
%(ii)~$\WriteB'<_{\TO}\Read$,
%(iii)~$\WriteM'\not \in \Trace$.
\end{enumerate}
\end{enumerate}

\begin{figure*}[!h]
\begin{subfigure}[b]{0.42\textwidth}
\centering
\begin{tikzpicture}[thick,
pre/.style={<-,shorten >= 2pt, shorten <=2pt, very thick},
post/.style={->,shorten >= 2pt, shorten <=2pt,  very thick},
seqtrace/.style={->, line width=2},
aux_seqtrace/.style={->, line width=1, draw=gray},
und/.style={very thick, draw=gray},
event/.style={rectangle, minimum height=3.5mm, draw=black, fill=white, minimum width=8mm,   line width=1pt, inner sep=2, font={\footnotesize}},
aux_event/.style={event, draw=gray},
virt/.style={circle,draw=black!50,fill=black!20, opacity=0},
bad/.style={preaction={fill, white}, pattern color=red!40, pattern=north east lines},
good/.style={preaction={fill, white}, pattern color=green!60, pattern=north west lines},
isLabel/.style={rectangle, fill opacity=0.5, fill=white, text opacity=1}
]

\newcommand{\xstep}{1.7}
\newcommand{\ystep}{0.6}%70909}
\newcommand{\yaux}{0.17727}

\fill[gray!20]
(-0.5*\xstep, 0.32*\ystep) to
(-0.5*\xstep, -0.5*\ystep) to
(1*\xstep, -1.5*\ystep - \yaux) to
(1.75*\xstep, -1.5*\ystep - \yaux) to
(2.25*\xstep, -2.5*\ystep) to
(3*\xstep, -2.5*\ystep) to
(3*\xstep, 0.32*\ystep) to cycle;

\node[]       (t1m_0)   at (0*\xstep, 0*\ystep) {\small$\LocalTrace_1$};
\node[]       (t1m_end) at (0*\xstep, -5*\ystep) {};

\node[]       (t2m_0)   at (0.5*\xstep, 0*\ystep) {\small$\LocalTrace_2$};
\node[]       (t2m_end) at (0.5*\xstep, -5*\ystep) {};

\node[]       (t3a_0)   at (1.5*\xstep, 0*\ystep) {\small$\LocalTrace_3'$};
\node[]       (t3a_end) at (1.5*\xstep, -5*\ystep) {};

\node[]       (t3m_0)   at (2.5*\xstep, 0*\ystep) {\small$\LocalTrace_3$};
\node[]       (t3m_end) at (2.5*\xstep, -5*\ystep) {};

\draw[seqtrace] (t1m_0) to (t1m_end);
\draw[seqtrace] (t2m_0) to (t2m_end);
\draw[aux_seqtrace] (t3a_0) to (t3a_end);
\draw[seqtrace] (t3m_0) to (t3m_end);

\node[event, good]  (t1m_r1)   at (0*\xstep, -2*\ystep) {$\Read_1(x)$};
\node[event, bad]  (t1m_r2)   at (0*\xstep, -3*\ystep) {$\Read_2(x)$};
\node[event, bad]  (t2m_r)   at (0.5*\xstep, -4*\ystep) {$\Read_3(y)$};
\node[aux_event]  (t3a_w1)  at (1.5*\xstep, -1*\ystep - \yaux) {$\WriteM_1(x)$};
\node[aux_event]  (t3a_w2)  at (1.5*\xstep, -2*\ystep - \yaux) {$\WriteM_2(y)$};
\node[event]  (t3m_w1)  at (2.5*\xstep, -1*\ystep) {$\WriteB_1(x)$};
\node[event]  (t3m_w2)  at (2.5*\xstep, -2*\ystep) {$\WriteB_2(y)$};
\node[event, good]  (t3m_r3)  at (2.5*\xstep, -3*\ystep) {$\Read_4(y)$};

\draw[post, draw=gray]   (t3m_w1) to (t3a_w1);
\draw[post, draw=gray]   (t3m_w2) to (t3a_w2);
\draw[post, \darkred, dashed]   (t3a_w1) to (t1m_r1);
\draw[post, \darkred, dashed]   (t3a_w1) to (t1m_r2);
\draw[post, \darkred, dashed]   (t3a_w2) -- (t2m_r) node [midway, xshift=7pt, yshift=-7pt, rotate=35, isLabel] {\small read by};
\draw[post, \darkred, dashed]   (t3m_w2) to[out=0, in=0, distance=1cm] (t3m_r3);

\end{tikzpicture}
\vspace{-2mm}
\caption{
The reads $\Read_1$ and $\Read_4$ are TSO-executable.
The read $\Read_2$ is not TSO-executable,
because $\Events{\Trace} \cup \{ \Read_2 \}$
is not a lower set;
neither is the read $\Read_3$, because
$\Observation(\Read_3) = (\_, \WriteM_2)$ has not been executed yet.
}
\label{subfig:TSO_exec_reads}
\end{subfigure}
%%%%%%%%%%%%%%%%%%%%%%%%%%%%%%%%%%%%%%%%%%%%%%%%%%%%%%%%%%%%%%%%%%%%%%
\hspace{0.5cm}
\begin{subfigure}[b]{0.53\textwidth}
\centering
\begin{tikzpicture}[thick,
pre/.style={<-,shorten >= 2pt, shorten <=2pt, very thick},
post/.style={->,shorten >= 2pt, shorten <=2pt,  very thick},
seqtrace/.style={->, line width=2},
aux_seqtrace/.style={->, line width=1, draw=gray},
und/.style={very thick, draw=gray},
event/.style={rectangle, minimum height=3.5mm, draw=black, fill=white, minimum width=8mm,   line width=1pt, inner sep=2, font={\footnotesize}},
aux_event/.style={event, draw=gray},
virt/.style={circle,draw=black!50,fill=black!20, opacity=0},
bad/.style={preaction={fill, white}, pattern color=red!40, pattern=north east lines},
good/.style={preaction={fill, white}, pattern color=green!60, pattern=north west lines},
isLabel/.style={rectangle, fill opacity=0.5, fill=white, text opacity=1}
]

\newcommand{\xstep}{1.7}
\newcommand{\ystep}{0.57}%6
\newcommand{\yaux}{0.15}

\fill[gray!20]
(-0.5*\xstep, 0.32*\ystep) to
(-0.5*\xstep, -2.5*\ystep) to
(0.5*\xstep, -2.5*\ystep) to
(1*\xstep, -1.5*\ystep) to
(1.5*\xstep, -1.5*\ystep) to
(1.75*\xstep, -1*\ystep) to
(2.35*\xstep, -1*\ystep) to
(2.45*\xstep, -2.55*\ystep) to
(2.9*\xstep, -2.55*\ystep) to
(3.0*\xstep, -3.5*\ystep) to
(3.5*\xstep, -3.5*\ystep) to
(3.5*\xstep, 0.32*\ystep) to cycle;

\node[] (t1a_0)   at (0*\xstep, 0*\ystep) {\small$\LocalTrace_1'$};
\node[] (t1a_end) at (0*\xstep, -6*\ystep) {};

\node[] (t2a_0)   at (0.6*\xstep, -3*\ystep) {\small$\LocalTrace_2'$};
\node[] (t2a_end) at (0.6*\xstep, -6*\ystep) {};

\node[] (t2m_0)   at (1.2*\xstep, 0*\ystep) {\small$\LocalTrace_2$};
\node[] (t2m_end) at (1.2*\xstep, -6*\ystep) {};

\node[] (t3a_0)   at (2*\xstep, 0*\ystep) {\small$\LocalTrace_3'$};
\node[] (t3a_end) at (2*\xstep, -6*\ystep) {};

\node[] (t4a_0)   at (2.55*\xstep, 0*\ystep) {\small$\LocalTrace_4'$};
\node[] (t4a_end) at (2.55*\xstep, -6*\ystep) {};

\node[] (t5a_0)   at (3.1*\xstep, 0*\ystep) {\small$\LocalTrace_5'$};
\node[] (t5a_end) at (3.1*\xstep, -6*\ystep) {};

\draw[aux_seqtrace] (t1a_0) to (t1a_end);
\draw[aux_seqtrace] (t2a_0) to (t2a_end);
\draw[seqtrace] (t2m_0) to (t2m_end);
\draw[aux_seqtrace] (t3a_0) to (t3a_end);
\draw[aux_seqtrace] (t4a_0) to (t4a_end);
\draw[aux_seqtrace] (t5a_0) to (t5a_end);

\node[aux_event] (t1a_1) at (0*\xstep, -1*\ystep) {$\WriteM_1(x)$};
\node[aux_event] (t1a_2) at (0*\xstep, -2*\ystep) {$\WriteM_2(y)$};

\node[aux_event, bad] (t2a_1) at (0.6*\xstep, -4.2*\ystep) {$\WriteM_3(z)$};

\node[event] (t2m_1) at (1.2*\xstep, -1*\ystep) {$\Read_1(x)$};
\node[event] (t2m_2) at (1.2*\xstep, -2*\ystep) {$\Read_2(y)$};
\node[event] (t2m_3) at (1.2*\xstep, -3*\ystep) {$\WriteB_3(z)$};
\node[event] (t2m_4) at (1.2*\xstep, -5*\ystep) {$\Read_3(z)$};

\node[aux_event, good] (t3a_1) at (2*\xstep, -1.5*\ystep) {$\WriteM_4(x)$};
\node[aux_event, bad] (t4a_1) at (2.55*\xstep, -3*\ystep) {$\WriteM_5(y)$};
\node[aux_event, bad] (t5a_1) at (3.1*\xstep, -4*\ystep) {$\WriteM_6(z)$};

\draw[post, draw=gray]   (t2m_3) to (t2a_1);
\draw[post, \darkred, dashed] (t1a_1) to (t2m_1);
\draw[post, \darkred, dashed] (t1a_2) to (t2m_2);
\draw[post, \darkred, dashed] (t5a_1) to[out=-150, in=0] (t2m_4);

\draw[post, color=black!20!blue, dotted] (t2m_2) to (t4a_1);
\draw[post, color=black!20!blue, dotted] (t2a_1) -- (t5a_1) node [midway, above, isLabel] {\small blocks};

\end{tikzpicture}
\vspace{-2mm}
\caption{
The memory-write $\WriteM_4$ is TSO-executable. The other
memory-writes are not;
$\Events{\Trace} \cup \{ \WriteM_3 \}$ is not a lower set,
for $\WriteM_5$ resp. $\WriteM_6$, the blue dotted arrows show
the events that they have to wait for, because of
\cref{item:tso_execw1} resp. \cref{item:tso_execw2}
(some buffer-writes are not displayed here for brevity).
}
\label{subfig:TSO_exec_mwrites}
\end{subfigure}

\caption{
$\TSO$-executability. The already executed events (i.e., $\Events{\Trace}$)
are in the gray zone, the remaining events are outside the gray zone.
The buffer threads are gray and thin, the main threads are black and thick.
}
\label{fig:TSO_exec}
%\vspace{-2mm}
\end{figure*}

Intuitively, the conditions of executable events ensure that executing
an event does not immediately create an invalid witness prefix.
The lower-set condition ensures that the program order $\TO$ is
respected. This is a sufficient condition for a buffer-write
or a fence (in particular, for a fence this implies that the
respective buffer is currently empty). The extra condition for a read
ensures that its reads-from constraint is satisfied.
The extra conditions for a memory-write prevent it from
causing some reads-from constraint to become unsatisfiable.

\cref{fig:TSO_exec} illustrates the notion of $\TSO$-executability
on several examples.
Observe that if $\Trace$ is a valid trace, extending $\Trace$ with an executable event (i.e., $\Trace \Concat \Event$)
also yields a valid trace that is well-formed,
as, by definition, $\Events{\Trace}\cup \{ \Event \}$ is a lower set of $(X,\TO)$.

\Paragraph{Algorithm $\AlgoTSO$.}
We are now ready to describe our algorithm $\AlgoTSO$ for the problem $\VTSOrm$.
At a high level, the algorithm enumerates all lower sets of $(\WritesM{X},\TO)$ by constructing a trace $\Trace$ with $\WritesM{\Trace}=Y$ for every lower set $Y$ of $(\WritesM{X},\TO)$.
The crux of the algorithm is to maintain the following.
Each constructed trace $\Trace$ is
\emph{maximal} in the set of thread events,
among all witness prefixes with the same set of memory-writes.
That is, for every witness prefix $\Trace'$ with $\WritesM{\Trace'}=\WritesM{\Trace}$, we have that $\LocalEvents{\Trace}\supseteq \LocalEvents{\Trace'}$.
Thus, the algorithm will only explore $n^{k}$ traces, as opposed to $n^{2\cdot k}$ from a naive enumeration of all lower sets of $(X,\TO)$.

% See \cref{algo:verifytso} for a formal description.
The formal description of $\AlgoTSO$ is in \cref{algo:verifytso}.
The algorithm maintains a worklist $\Worklist$ of prefixes and a set
$\DoneSet$ of already-explored lower sets of $(\WritesM{X},\TO)$. %throughout.
In each iteration, the \cref{algo:tso_local_extend_loop} loop
makes the prefix maximal in the thread events, then
\cref{algo:tso_test_done} checks if we are done,
otherwise the loop in~\cref{algo:tso_mem_execute} enumerates the
executable memory-writes to extend the prefix with.

%\smallskip
\vspace{1mm}
\begin{algorithm}%[H]
\small
\SetInd{0.4em}{0.4em}
\DontPrintSemicolon
%\setstretch{1.05}
\caption{$\AlgoTSO$}\label{algo:verifytso}
\KwIn{
An event set $X$ and a reads-from function $\Observation\colon \Reads{X}\to \Writes{X}$
}
\KwOut{
A witness $\Trace$ that realizes $(X, \Observation)$ if $(X, \Observation)$ is realizable under $\TSO$, else $\bot$
}
\BlankLine
$\Worklist\gets \{\epsilon\}$; $\DoneSet\gets \{\emptyset\}$\\
\While{$\Worklist\neq \emptyset$}{\label{algo:tso_main_while}
Extract a trace $\Trace$ from $\Worklist$\label{algo:tso_extract_worklist}\\
\While%(\tcp*[f]{Execute new thread events})
{$\exists$ thread event $\Event$ $\TSO$-executable in $\Trace$}{\label{algo:tso_local_extend_loop}
$\Trace\gets \Trace \Concat \Event$\tcp*[f]{Execute the thread event $\Event$}\label{algo:tso_local_extend}\\
}
\lIf(\tcp*[f]{Witness found}){$\Events{\Trace}=X$}{\label{algo:tso_test_done}%All reads realized, witness found
\Return{$\Trace$}
}
\ForEach{memory-write $\WriteM$ that is $\TSO$-executable in $\Trace$}{\label{algo:tso_mem_execute}
$\Trace_{\WriteM}\gets \Trace\Concat \WriteM$\tcp*[f]{Execute $\WriteM$}\label{algo:tso_mem_extend}\\
\uIf%(\tcp*[f]{First-time discovered this set of memory writes})
{$\not\exists\Trace'\in \DoneSet$ s.t. $\WritesM{\Trace_{\WriteM}}=\WritesM{\Trace'}$}{\label{algo:tso_if_new}
Insert $\Trace_{\WriteM}$ in $\Worklist$ and in $\DoneSet$\tcp*[f]{Continue from $\Trace_{\WriteM}$}\label{algo:tso_insert_worklist}\\%Continue exploration from $\Trace_{\WriteM}$
}
}
}
\Return{$\bot$}
\end{algorithm}

%\Paragraph{Correctness.}
We now provide the insights behind the correctness of $\AlgoTSO$.
% This figure is already referred to later on, where it's more appropriate::::  , and refer to \cref{fig:TSO_max_cmp} for an illustration.
The correctness proof has two components: (i) soundness and (ii) completeness, which we present below.

\Paragraph{Soundness.}
The soundness follows directly from the definition of $\TSO$-executable events.
%We now describe in detail the soundness argument.
In particular, when the algorithm extends
a trace $\Trace$ with a read $\Read$, where $\Observation(\Read)=(\WriteB, \WriteM)$,
the following hold.
%In particular, when the algorithm extends a trace $\Trace$ with a read $\Read$, where $\Observation(\Read)=(\WriteB, \WriteM)$, the following hold

\begin{enumerate}[noitemsep,topsep=0pt,partopsep=0px]
\item If $\Proc{\Read}\neq \Proc{\WriteB}$, then $\WriteM\in \Trace$, since $\Read$ became executable.
Moreover, when $\WriteM$ appeared in $\Trace$, the variable $x=\Location{\WriteM}$ became held by $\WriteM$, and remained held at least until the current step where $\Read$ is executed.
Hence, no other memory-write $\WriteM'$ with $\Location{\WriteM'}=x$ could have become executable in the meantime, to violate the observation of $\Read$.
Moreover, $\Read$ cannot read-from a local buffer write $\WriteB'$ with $\Location{\WriteB'}=x$, as by definition, when $\WriteM$ became executable, all buffer-writes on $x$ that are local to $\Read$ and precede $\Read$ must have been flushed to the main memory
(i.e., $\WriteM'$ must have also appeared in the trace).
\item If $\Proc{\Read}=\Proc{\WriteB}$, then either $\WriteM$ has not appeared already in $\Trace$, in which case $\Read$ reads-from $\WriteB$ from its local buffer, or $\WriteM$ has appeared in the trace and held its variable until $\Read$ is executed, as in the previous item.
\end{enumerate}

\Paragraph{Completeness.} %{Completeness: Main idea.}
Let $\Trace'$ be an arbitrary witness prefix, %, let $Y' = \Events{\Trace'}$ be its set of events.
$\AlgoTSO$ constructs a trace $\Trace$ % with events $Y =$
such that
$\WritesM{\Trace} = \WritesM{\Trace'}$ and
$\LocalEvents{\Trace}\supseteq \LocalEvents{\Trace'}$.
% This immediately implies completeness, and it follows
% This follows from the fact that
This is because
$\AlgoTSO$ constructs for every lower set $Y$ of $(\WritesM{X}, \TO)$
% attempts to construct IS MORE PRECISE BUT MIGHT UNNECCESARILY CONFUSE REVIEWERS
% IT IS MORE PRECISE because for some lower set no such trace or witness prefix might even exist
a single representative trace $\Trace$ with % $Y$ as the memory-write events (i.e.,
$\WritesM{\Trace} = Y$.
The key % insight
is to make $\Trace$ \emph{maximal} on the thread events, i.e.,
$\LocalEvents{\Trace}\supseteq \LocalEvents{\Trace'}$
% NOT for any TRACE!!! You can take trace that blatantly dissatisfies some reads-from, and hence has more thread events
for any witness prefix $\Trace'$ with $\WritesM{\Trace'}=\WritesM{\Trace}$,
and thus % therefore,
any memory-write $\WriteM$ that is executable in $\Trace'$ is also executable in $\Trace$.
% The algorithm constructs a representative trace $\Trace$ with $\WritesM{\Trace}=Y$
% for every lower set $Y$ of $(\WritesM{X}, \TO)$ such that
% $(Y', \Observation\Project Y')$ is realizable,
% where $Y'$ is the largest lower set of $(X, \TO)$ with $\WritesM{Y'}=Y$
% (recall that $\Observation\Project Y'$ denotes the projection of $\Observation$ on the set $Y'$).

We now present the above insight in detail.
Indeed, if $\WriteM$ is not executable in $\Trace$, one of the following holds.
Let $\Location{\WriteM}=x$.
\begin{enumerate}[noitemsep,topsep=0pt,partopsep=0px]
\item $x$ is already held in $\Trace$. But since $\WritesM{\Trace'}=\WritesM{\Trace}$
and any read of $\Trace'$ also appears in $\Trace$, the variable $x$ is also held in $\Trace'$, thus $\WriteM$ is not executable in $\Trace'$ either.
\item There is a later read $\Read\not \in \Trace$ that must read-from $\WriteM$,
but $\Read$ is preceded by a local write $(\WriteB', \WriteM')$
(i.e., $\WriteB'<_{\TO}\Read$)
also on $x$, for which $\WriteM' \not\in \Trace$.
Since $\LocalEvents{\Trace}\supseteq\LocalEvents{\Trace'}$, we have $\Read\not \in \Trace'$, and as $\WritesM{\Trace'}=\WritesM{\Trace}$, also $\WriteM'\not \in \Trace'$.
Thus $\WriteM$ is also not executable in $\Trace'$.
\end{enumerate}
The final insight is on how the algorithm maintains the maximality invariant as it extends $\Trace$ with new events.
This holds because read events become executable as soon as their corresponding remote observation $\WriteM$ appears in the trace, and hence all such reads are executable for a given lower set of $(\WritesM{X},\TO)$.
All other thread events are executable without any further conditions.
\cref{fig:TSO_max_cmp} illustrates the intuition behind the maximality invariant.
The following lemma states the formal correctness,
which together with the complexity argument gives us \cref{them:vtso}.
%We refer to \cref{subsec:app_verifyingtso} for the proof.

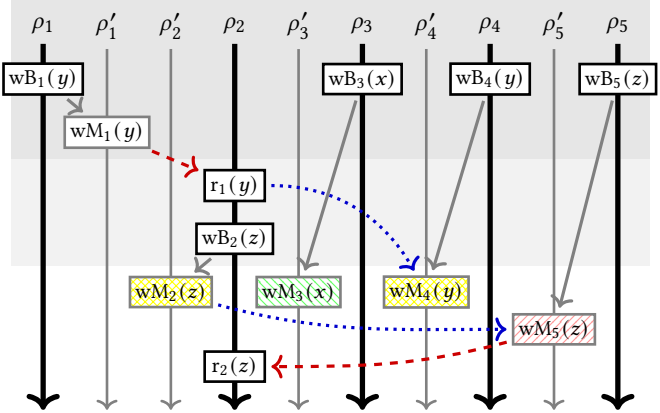
\begin{figure}[h]
\begin{minipage}[c]{0.645\textwidth}
\raggedright
\begin{tikzpicture}[thick,
pre/.style={<-,shorten >= 2pt, shorten <=2pt, very thick},
post/.style={->,shorten >= 2pt, shorten <=2pt,  very thick},
seqtrace/.style={->, line width=2},
aux_seqtrace/.style={->, line width=1, draw=gray},
und/.style={very thick, draw=gray},
event/.style={rectangle, minimum height=3.5mm, draw=black, fill=white, minimum width=8mm,   line width=1pt, inner sep=2, font={\footnotesize}},
aux_event/.style={event, draw=gray},
virt/.style={circle,draw=black!50,fill=black!20, opacity=0},
bad/.style={preaction={fill, white}, pattern color=red!40, pattern=north east lines},
good/.style={preaction={fill, white}, pattern color=green!60, pattern=north west lines},
neutral/.style={preaction={fill, white}, pattern color=yellow, pattern=crosshatch}
]

\newcommand{\xstep}{0.85}
\newcommand{\ystep}{0.7}
\newcommand{\yaux}{0.2}

\fill[gray!10]
(-0.5*\xstep, 0.5*\ystep) to
(-0.5*\xstep, -4.5*\ystep) to
(9.5*\xstep, -4.5*\ystep) to
(9.5*\xstep, 0.5*\ystep) to cycle;

\fill[gray!20]
(-0.5*\xstep, 0.5*\ystep) to
(-0.5*\xstep, -2.5*\ystep) to
(9.5*\xstep, -2.5*\ystep) to
(9.5*\xstep, 0.5*\ystep) to cycle;

\node[] (t1a_0)   at (1*\xstep, 0*\ystep) {\small$\LocalTrace_1'$};
\node[] (t1a_end) at (1*\xstep, -7.4*\ystep) {};

\node[] (t1m_0)   at (0*\xstep, 0*\ystep) {\small$\LocalTrace_1$};
\node[] (t1m_end) at (0*\xstep, -7.4*\ystep) {};

\node[] (t2a_0)   at (2*\xstep, 0*\ystep) {\small$\LocalTrace_2'$};
\node[] (t2a_end) at (2*\xstep, -7.4*\ystep) {};

\node[] (t2m_0)   at (3*\xstep, 0*\ystep) {\small$\LocalTrace_2$};
\node[] (t2m_end) at (3*\xstep, -7.4*\ystep) {};

\node[] (t3a_0)   at (4*\xstep, 0*\ystep) {\small$\LocalTrace_3'$};
\node[] (t3a_end) at (4*\xstep, -7.4*\ystep) {};

\node[] (t3m_0)   at (5*\xstep, 0*\ystep) {\small$\LocalTrace_3$};
\node[] (t3m_end) at (5*\xstep, -7.4*\ystep) {};

\node[] (t4a_0)   at (6*\xstep, 0*\ystep) {\small$\LocalTrace_4'$};
\node[] (t4a_end) at (6*\xstep, -7.4*\ystep) {};

\node[] (t4m_0)   at (7*\xstep, 0*\ystep) {\small$\LocalTrace_4$};
\node[] (t4m_end) at (7*\xstep, -7.4*\ystep) {};

\node[] (t5a_0)   at (8*\xstep, 0*\ystep) {\small$\LocalTrace_5'$};
\node[] (t5a_end) at (8*\xstep, -7.4*\ystep) {};

\node[] (t5m_0)   at (9*\xstep, 0*\ystep) {\small$\LocalTrace_5$};
\node[] (t5m_end) at (9*\xstep, -7.4*\ystep) {};

\draw[aux_seqtrace] (t1a_0) to (t1a_end);
\draw[aux_seqtrace] (t2a_0) to (t2a_end);
\draw[aux_seqtrace] (t3a_0) to (t3a_end);
\draw[aux_seqtrace] (t4a_0) to (t4a_end);
\draw[aux_seqtrace] (t5a_0) to (t5a_end);

\draw[seqtrace] (t1m_0) to (t1m_end);
\draw[seqtrace] (t2m_0) to (t2m_end);
\draw[seqtrace] (t3m_0) to (t3m_end);
\draw[seqtrace] (t4m_0) to (t4m_end);
\draw[seqtrace] (t5m_0) to (t5m_end);

\node[aux_event]  (t1a_1) at (1*\xstep, -2*\ystep) {$\WriteM_1(y)$};
\node[event]      (t1m_1) at (0*\xstep, -1*\ystep) {$\WriteB_1(y)$};
\draw[post, draw=gray]   (t1m_1) to (t1a_1);

\node[aux_event, neutral] (t2a_1) at (2*\xstep, -5*\ystep) {$\WriteM_2(z)$};
\node[event] (t2m_1) at (3*\xstep, -3*\ystep) {$\Read_1(y)$};
\node[event] (t2m_2) at (3*\xstep, -4*\ystep) {$\WriteB_2(z)$};
\node[event] (t2m_3) at (3*\xstep, -6.4*\ystep) {$\Read_2(z)$};
\draw[post, draw=gray]   (t2m_2) to (t2a_1);

\node[aux_event, good] (t3a_1) at (4*\xstep, -5*\ystep) {$\WriteM_3(x)$};
\node[event]                    (t3m_1) at (5*\xstep, -1*\ystep) {$\WriteB_3(x)$};
\draw[post, draw=gray]   (t3m_1) to (t3a_1);

\node[aux_event, neutral] (t4a_1) at (6*\xstep, -5*\ystep) {$\WriteM_4(y)$};
\node[event]                  (t4m_1) at (7*\xstep, -1*\ystep) {$\WriteB_4(y)$};
\draw[post, draw=gray]   (t4m_1) to (t4a_1);

\node[aux_event, bad] (t5a_1) at (8*\xstep, -5.7*\ystep) {$\WriteM_5(z)$};
\node[event]                  (t5m_1) at (9*\xstep, -1*\ystep) {$\WriteB_5(z)$};
\draw[post, draw=gray]   (t5m_1) to (t5a_1);

\draw[post, \darkred, dashed] (t1a_1) to (t2m_1);
\draw[post, \darkred, dashed] (t5a_1) to[out=-165, in=0] (t2m_3);

\draw[post, color=black!20!blue, dotted] (t2m_1) to[out=0, in=120] (t4a_1);
\draw[post, color=black!20!blue, dotted] (t2a_1) to[out=-15, in=180] (t5a_1);

\end{tikzpicture}
\end{minipage}
%\vspace{-6mm}
\begin{minipage}[c]{0.33\textwidth}
%\vspace{-2mm}
\caption{
$\AlgoTSO$ maximality invariant.
The gray zone shows the events of some witness prefix $\Trace'$;
the lighter gray shows the events of the corresponding trace $\Trace$, constructed by the algorithm, which is maximal on thread events.
Yellow writes ($\WriteM_2$ and $\WriteM_4$) are those that are $\TSO$-executable in $\Trace$ but not in $\Trace'$.
Green writes ($\WriteM_3$) and red writes ($\WriteM_5$) are $\TSO$-executable and non $\TSO$-executable, respectively.
}
\label{fig:TSO_max_cmp}
\end{minipage}
\end{figure}

%\smallskip
\begin{restatable}{lemma}{lemverifyingtsocorrectness}\label{lem:verifyingtso_correctness}
$(X, \Observation)$ is realizable under $\TSO$ iff $\AlgoTSO$ returns a trace $\Trace\neq \epsilon$.
\end{restatable}

%Thus $\AlgoTSO$ establishes the result of \cref{them:vtso},
%and the proof of \cref{them:vtso} is in \cref{subsec:app_verifyingtso}.
%We conclude this section with the proof of \cref{them:vtso}.

% \smallskip
% \begin{proof}[Proof of \cref{them:vtso}.]
% \cref{lem:verifyingtso_correctness} establishes the correctness, so here we focus on the complexity.
% Since there are $k$ threads, there exist at most $n^k$ distinct traces $\Trace_1, \Trace_2$ with $\WritesM{\Trace_1}\neq \WritesM{\Trace_2}$.
% Hence, the main loop in \cref{algo:tso_main_while} is executed at  most $ n^k$ times.
% For each of the $\leq n^k$ traces $\Trace$ inserted in $\Worklist$ in \cref{algo:tso_insert_worklist}, there exist at most $k-1$ traces that are not inserted in $\Worklist$ because $\WritesM{\Trace}=\WritesM{\Trace'}$ (hence the test  in \cref{algo:tso_if_new} fails).
% Hence, the algorithm handles $O(k\cdot n^{k})$ traces in total, while each trace is constructed in $O(n)$ time.
% Thus, the complexity of $\AlgoTSO$ is $O(k\cdot n^{k+1})$.
% The desired result follows.
% \end{proof}

\subsection{Verifying PSO Executions}\label{subsec:verifyingpso}

In this section we show \cref{them:vpso}, i.e.,
we present an algorithm $\AlgoPSO$ that solves $\VPSOrm$ in $O(k\cdot n^{k+1}\cdot \min(n^{k\cdot (k-1)}, 2^{k\cdot \NumVariables}))$ time, while the bound becomes $O(k\cdot n^{k+1})$ when there are no fences.
Similarly to the case of $\TSO$, the algorithm relies on the notion of $\PSO$-executable events, defined below.
We first introduce some relevant notation that makes our presentation simpler.

\Paragraph{Spurious and pending writes.}
Consider a trace $\Trace$ with $\Events{\Trace}\subseteq X$.
A memory-write $\WriteM\in \WritesM{X}$ is called \emph{spurious} in $\Trace$ if the following conditions hold.
\begin{enumerate}[noitemsep,topsep=0pt,partopsep=0px]
\item There is no read  $\Read\in \Reads{X}\setminus\Trace$ with $\Observation(\Read)= (\_, \WriteM)$
\\(informally, no remaining read wants to read-from $\WriteM$).
\item If $\WriteM\in \Trace$, then for every read $\Read\in \Trace$
with $\Observation_{\Trace}(\Read)= (\_, \WriteM)$
we have $\Read <_{\Trace} \WriteM$
\\(informally, reads in $\Trace$ that read-from this write read it from the local buffer).
% we have that $\Read$ reads-from $\WriteB$ in $\Trace$,
%where $\WriteB$ is the buffer-write corresponding to $\WriteM$.
\end{enumerate}
Note that if $\WriteM$ is a spurious memory-write in $\Trace$ then $\WriteM$ is spurious in all extensions of $\Trace$.
We denote by  $\SpuriousWritesM{\Trace}$ the set of memory-writes of $\Trace$ that are spurious in $\Trace$.
A memory-write $\WriteM$ is \emph{pending} in $\Trace$ if $\WriteB \in \Trace$ and $\WriteM\not \in \Trace$, where $\WriteB$ is the corresponding buffer-write of $\WriteM$.
We denote by $\PendingWritesM{\Trace, \Process}$ the set of all pending memory-writes $\WriteM$ in $\Trace$ with $\Proc{\WriteM}=\Process$.
See \cref{fig:PSO_sp} for an intuitive illustration of spurious and pending memory-writes.

\begin{figure}[h]
\vspace{2mm}
\begin{subfigure}[b]{0.52\textwidth}
\centering
\begin{tikzpicture}[thick,
pre/.style={<-,shorten >= 2pt, shorten <=2pt, very thick},
post/.style={->,shorten >= 2pt, shorten <=2pt,  very thick},
seqtrace/.style={->, line width=2},
aux_seqtrace/.style={->, line width=1, draw=gray},
und/.style={very thick, draw=gray},
event/.style={rectangle, minimum height=3.5mm, draw=black, fill=white, minimum width=8mm, line width=1pt, inner sep=2, font={\footnotesize}},
aux_event/.style={event, draw=gray},
virt/.style={circle,draw=black!50,fill=black!20, opacity=0}]

\newcommand{\xs}{1.4}
\newcommand{\ys}{0.52}%6

% fill in the backgrounds of prefixes
\fill[gray!10]
(-0.5*\xs, 0.5*\ys) to
(-0.5*\xs, -4.5*\ys) to
(1.5*\xs, -4.5*\ys) to
(1.5*\xs, 0.5*\ys) to cycle;

\fill[gray!20]
(-0.5*\xs, 0.5*\ys) to
(-0.5*\xs, -3.5*\ys) to
(1.5*\xs, -3.5*\ys) to
(1.5*\xs, 0.5*\ys) to cycle;

\fill[gray!30]
(-0.5*\xs, 0.5*\ys) to
(-0.5*\xs, -2.5*\ys) to
(1.5*\xs, -2.5*\ys) to
(1.5*\xs, 0.5*\ys) to cycle;

\fill[gray!40]
(-0.5*\xs, 0.5*\ys) to
(-0.5*\xs, -1.5*\ys) to
(1.5*\xs, -1.5*\ys) to
(1.5*\xs, 0.5*\ys) to cycle;

% prefix labels
%\node[] (sigma0_label) at (1.75*\xs, -1*\ys) {$\Trace_0$};
%\node[] (sigma1_label) at (1.75*\xs, -2*\ys) {$\Trace_1$};
%\node[] (sigma2_label) at (1.75*\xs, -3*\ys) {$\Trace_2$};
%\node[] (sigma3_label) at (1.75*\xs, -4*\ys) {$\Trace_3$};

% spurious/pending table
\node[] (spurious_label)  at (2.2*\xs, 0*\ys) {\footnotesize$\SpuriousWritesM{\Trace}$};%$\SpuriousWritesM{(\Trace)}$}; \SpuriousWritesMnoarg
\node[] (pending_label)   at (3.7*\xs, 0*\ys) {\footnotesize$\PendingWritesM{\Trace, \Process_1}$};%$\PendingWritesM{(\Trace, \Process_1)}$}; \PendingWritesMnoarg

\foreach \j in {0.5, -0.5, -1.5, -2.5, -3.5, -4.5}
\draw[line width=0.5] (1.5*\xs, \j*\ys) to (4.5*\xs, \j*\ys);
\foreach \i in {1.5, 2.9, 4.5}
\draw[line width=0.5] (\i*\xs, 0.5*\ys) to (\i*\xs, -4.5*\ys);

\node[] at (2.2*\xs, -1*\ys) {\footnotesize$\emptyset$}; % \footnotesize$(\Trace)$
\node[] at (2.2*\xs, -2*\ys) {\footnotesize$\emptyset$};
\node[] at (2.2*\xs, -3*\ys) {\footnotesize$\WriteM_1(x)$};
\node[] at (2.2*\xs, -4*\ys) {\footnotesize$\WriteM_1(x)$};

\node[] at (3.7*\xs, -1*\ys) {\footnotesize$\emptyset$}; % \footnotesize$(\Trace, \Process_1)$
\node[] at (3.7*\xs, -2*\ys) {\footnotesize$\WriteM_1(x)$};
\node[] at (3.7*\xs, -3*\ys) {\footnotesize$\WriteM_1(x)$};
\node[] at (3.7*\xs, -4*\ys) {\footnotesize$\emptyset$};

\node[] (t1m_0)   at (0*\xs, 0*\ys) {\small$\LocalTrace_1$};
\node[] (t1m_end) at (0*\xs, -5.2*\ys) {};

\node[] (t1a_0)   at (1*\xs, 0*\ys) {\small$\LocalTrace_1'(x)$};
\node[] (t1a_end) at (1*\xs, -5.2*\ys) {};

\draw[seqtrace] (t1m_0) to (t1m_end);
\draw[aux_seqtrace] (t1a_0) to (t1a_end);

\node[event] (t1m_1) at (0*\xs, -2*\ys) {$\WriteB_1(x)$};
\node[event] (t1m_2) at (0*\xs, -3*\ys) {$\Read_1(x)$};
\node[aux_event] (t1a_1) at (1*\xs, -4*\ys) {$\WriteM_1(x)$};

\draw[post, draw=gray] (t1m_1) to (t1a_1);
\draw[post, \darkred, dashed] (t1m_1.east) to[out=0, in=0, distance=0.5cm] (t1m_2.east);

\end{tikzpicture}
\vspace{-2mm}
\caption{
Linearization where $\WriteM_1$ is spurious.
The table shows the spurious and pending writes after each step.
}
\label{subfig:PSO_sp1}
\end{subfigure}
%%%%%%%%%%%%%%%%%%%%%%%%%%%%%%%%%%%%%%%%%%%%%%%%%%%%%%%%%%%%%%%%%%%%%%
\qquad%\quad
\begin{subfigure}[b]{0.4\textwidth}
\centering
\begin{tikzpicture}[thick,
pre/.style={<-,shorten >= 2pt, shorten <=2pt, very thick},
post/.style={->,shorten >= 2pt, shorten <=2pt,  very thick},
seqtrace/.style={->, line width=2},
aux_seqtrace/.style={->, line width=1, draw=gray},
und/.style={very thick, draw=gray},
event/.style={rectangle, minimum height=3.5mm, draw=black, fill=white, minimum width=8mm, line width=1pt, inner sep=2, font={\footnotesize}},
aux_event/.style={event, draw=gray},
virt/.style={circle,draw=black!50,fill=black!20, opacity=0},
bad/.style={preaction={fill, white}, pattern color=red!40, pattern=north east lines},
good/.style={preaction={fill, white}, pattern color=green!60, pattern=north west lines},
]

\newcommand{\xs}{1.6}
\newcommand{\ys}{0.52}%6

\node[] (t1m_0)   at (0*\xs, 0*\ys) {\small$\LocalTrace_1$};
\node[] (t1m_end) at (0*\xs, -5.2*\ys) {};

\node[] (t1a_0)   at (1*\xs, 0*\ys) {\small$\LocalTrace_1'(x)$};
\node[] (t1a_end) at (1*\xs, -5.2*\ys) {};

\draw[seqtrace] (t1m_0) to (t1m_end);
\draw[aux_seqtrace] (t1a_0) to (t1a_end);

\node[event] (t1m_1) at (0*\xs, -2*\ys) {$\WriteB_1(x)$};
\node[event] (t1m_2) at (0*\xs, -4*\ys) {$\Read_1(x)$};
\node[aux_event] (t1a_1) at (1*\xs, -3*\ys) {$\WriteM_1(x)$};

\draw[post, draw=gray] (t1m_1) to (t1a_1);
\draw[post, \darkred, dashed] (t1a_1) to (t1m_2);

\end{tikzpicture}
\vspace{-2mm}
\caption{Linearization where $\WriteM_1$ is \emph{not} spurious;
% as
here $\Observation_{\Trace}(\Read_1)= (\_, \WriteM_1)$ and $\WriteM_1 <_{\Trace} \Read_1$.}
\label{subfig:PSO_sp2}
\end{subfigure}

\caption{Illustration of spurious and pending writes.}
\label{fig:PSO_sp}
\end{figure}

\Paragraph{$\PSO$-executable events.}
Similarly to the case of $\VTSOrm$, we define the notion of $\PSO$-executable events (executable for short).
%In particular,
An event $\Event\in X\setminus \Events{\Trace}$ is \emph{$\PSO$-executable} in $\Trace$ if the following
conditions hold.
\begin{enumerate}[noitemsep,topsep=0pt,partopsep=0px]
\item \emph{If $\Event$ is a buffer-write or a memory-write}, then the same conditions apply as for $\TSO$-executable.
\item \emph{If $\Event$ is a fence $\Fence$}, then every pending memory-write %$\WriteM$
from $\Proc{\Fence}$ is $\PSO$-executable in $\Trace$,
\\and these memory-writes together with $\Fence$ and $\Events{\Trace}$ form a lower set of $(X,\TO)$.
\item \emph{If $\Event$ is a read $\Read$}, let $\Observation(\Read)=(\WriteB, \WriteM)$.
We have $\WriteB\in \Trace$, and the following conditions.
% THE COMMENTED BELOW is missing the lower-set info
% if $\Proc{\Read}\neq \Proc{\WriteB}$
% then either $\WriteM\in \Trace$ or $\WriteM$ is $\PSO$-executable in $\Trace$.
\begin{enumerate}
\item if $\Proc{\Read} = \Proc{\WriteB}$, then $\Events{\Trace} \cup \{\Read\}$ is a lower set of $(X,\TO)$.
\item if $\Proc{\Read} \neq \Proc{\WriteB}$, then $\Events{\Trace} \cup \{\WriteM, \Read\}$ is a lower set of $(X,\TO)$
\\and further either $\WriteM\in \Trace$ or $\WriteM$ is $\PSO$-executable in $\Trace$.
\end{enumerate}
%\begin{enumerate}
%\item Either $\WriteM\in \Trace$ or $\WriteM$ is $\PSO$-executable in $\Trace$.
%\item If $\WriteM$ is executable in $\Trace$, for every memory-write $\WriteM'$ such that $\WriteM'<_{\TO}\WriteM$, we have $\WriteM'\in \Trace$.
%\end{enumerate}
\end{enumerate}

\cref{fig:PSO_execs} illustrates several examples of $\PSO$-(un)executable events.
Similarly to the case of $\TSO$, the $\PSO$-executable conditions
ensure that we do not execute events creating an invalid witness prefix.
The executability conditions for $\PSO$ are different (e.g.,
there are extra conditions for a fence), since our approach for
$\VPSOrm$ fundamentally differs from the approach for $\VTSOrm$.

\begin{figure*}[!h]
%\begin{minipage}[c]{0.645\textwidth}
%\raggedright
\centering
\begin{tikzpicture}[thick,
pre/.style={<-,shorten >= 2pt, shorten <=2pt, very thick},
post/.style={->,shorten >= 2pt, shorten <=2pt,  very thick},
seqtrace/.style={->, line width=2},
aux_seqtrace/.style={->, line width=1, draw=gray},
und/.style={very thick, draw=gray},
event/.style={rectangle, minimum height=3.5mm, draw=black, fill=white, minimum width=8mm,   line width=1pt, inner sep=2, font={\footnotesize}},
aux_event/.style={event, draw=gray},
virt/.style={circle,draw=black!50,fill=black!20, opacity=0},
bad/.style={preaction={fill, white}, pattern color=red!40, pattern=north east lines},
good/.style={preaction={fill, white}, pattern color=green!60, pattern=north west lines},
isLabel/.style={rectangle, fill opacity=0.5, fill=white, text opacity=1}
]

\newcommand{\xs}{1.7}%1.7
\newcommand{\ys}{0.63}%6
\newcommand{\yaux}{0.2}

\fill[gray!20]
(-0.3*\xs, 0.5*\ys) to
(-0.3*\xs, -2.5*\ys) to
(7.3*\xs, -2.5*\ys) to
(7.3*\xs, 0.5*\ys) to cycle;

\node[] (t1m_0)   at (2*\xs, 0*\ys) {\small$\LocalTrace_2$};
\node[] (t1m_end) at (2*\xs, -5*\ys) {};

\node[] (t1ax_0)   at (1*\xs, 0*\ys) {\small$\LocalTrace_2'(x)$};
\node[] (t1ax_end) at (1*\xs, -5*\ys) {};

\node[] (t1ay_0)   at (3*\xs, 0*\ys) {\small$\LocalTrace_2'(y)$};
\node[] (t1ay_end) at (3*\xs, -5*\ys) {};

\node[] (t2m_0)   at (0*\xs, 0*\ys) {\small$\LocalTrace_1$};
\node[] (t2m_end) at (0*\xs, -5*\ys) {};

\node[] (t3m_0)   at (4*\xs, 0*\ys) {\small$\LocalTrace_3$};
\node[] (t3m_end) at (4*\xs, -5*\ys) {};

\node[] (t4m_0)   at (5*\xs, 0*\ys) {\small$\LocalTrace_4$};
\node[] (t4m_end) at (5*\xs, -5*\ys) {};

\node[] (t4ay_0)   at (6*\xs, 0*\ys) {\small$\LocalTrace_4'(y)$};
\node[] (t4ay_end) at (6*\xs, -5*\ys) {};

\node[] (t5m_0)   at (7*\xs, 0*\ys) {\small$\LocalTrace_5$};
\node[] (t5m_end) at (7*\xs, -5*\ys) {};

\draw[seqtrace] (t1m_0) to (t1m_end);
\draw[seqtrace] (t2m_0) to (t2m_end);
\draw[seqtrace] (t3m_0) to (t3m_end);
\draw[seqtrace] (t4m_0) to (t4m_end);
\draw[seqtrace] (t5m_0) to (t5m_end);

\draw[aux_seqtrace] (t1ax_0) to (t1ax_end);
\draw[aux_seqtrace] (t1ay_0) to (t1ay_end);
\draw[aux_seqtrace] (t4ay_0) to (t4ay_end);

\node[event] (t1m_1) at (2*\xs, -1*\ys) {$\WriteB_1(y)$};
\node[event] (t1m_2) at (2*\xs, -2*\ys) {$\WriteB_2(x)$};
\node[event, good] (t1m_3) at (2*\xs, -4*\ys) {$\Fence_1$};
\node[aux_event, good] (t1ax_1) at (1*\xs, -3*\ys) {$\WriteM_2(x)$};
% \node[aux_event, bad] (t1ax_2) at (1*\xs, -4*\ys) {$\ldots$};
\node[aux_event] (t1ay_1) at (3*\xs, -1.5*\ys) {$\WriteM_1(y)$};

\node[event, good] (t2m_1) at (0*\xs, -3*\ys) {$\Read_1(x)$};
\node[event, good] (t3m_1) at (4*\xs, -3*\ys) {$\Read_2(y)$};
% \node[event, bad] (t3m_2) at (4*\xs, -4*\ys) {$\ldots$};

\node[event] (t4m_1) at (5*\xs, -1*\ys) {$\WriteB_3(y)$};
\node[event, bad] (t4m_2) at (5*\xs, -4*\ys) {$\Fence_2$};
\node[aux_event, bad] (t4ay_1) at (6*\xs, -3*\ys) {$\WriteM_3(y)$};
% \node[aux_event, bad] (t4ay_2) at (6*\xs, -4*\ys) {$\ldots$};

\node[event, bad] (t5m_1) at (7*\xs, -4*\ys) {$\Read_3(y)$};

\draw[post, gray] (t1m_1) to (t1ay_1);
\draw[post, gray] (t1m_2) to (t1ax_1);
\draw[post, gray] (t4m_1) to (t4ay_1);

\draw[post, \darkred, dashed] (t1ax_1) to (t2m_1);
\draw[post, \darkred, dashed] (t1ay_1) to (t3m_1);
\draw[post, \darkred, dashed] (t4ay_1) to (t5m_1);

\draw[post, gray, densely dotted] (t1ax_1) to (t1m_3);
\draw[post, gray, densely dotted] (t1ay_1) -- (t1m_3) node [midway, rotate=46.5, yshift=-9pt] {\small flushed by};
\draw[post, gray, densely dotted] (t4ay_1) to (t4m_2);

\draw[post, color=black!20!blue, dotted] (t3m_1) to (t4ay_1);

\end{tikzpicture}
%\end{minipage}
%\vspace{-2mm}
%\begin{minipage}[c]{0.33\textwidth}
\vspace{-1mm}
\caption{
%Illustration of
$\PSO$-executability.
The green events are $\PSO$-executable; the red events are not.
The memory-write $\WriteM_2(x)$ is executable, and thus so are $\Read_1(x)$ and $\Fence_1$.
The memory-write $\WriteM_3(y)$ is not executable, as the variable $y$ is held by $\WriteM_1(y)$ until $\Read_2(y)$ is executed.
Consequently, $\Fence_2$ and $\Read_3(y)$ are not executable.
}
\label{fig:PSO_execs}
%\end{minipage}
%\vspace{-2mm}
\end{figure*}

\Paragraph{Fence maps.}
We define a \emph{fence map} as a function
$\FenceMap_{\Trace}\colon \Threads\times \Threads \to [n]$ as follows.
First, $\FenceMap_{\Trace}(\Process,\Process)=0$ for all $\Process\in \Threads$.
In addition, if $\Process$ does not have a fence unexecuted in $\Trace$ (i.e., a fence
$\Fence\in (X_{\Process}\setminus \Events{\Trace})$),
then $\FenceMap_{\Trace}(\Process,\Process')=0$ for all $\Process'\in \Threads$.
Otherwise, consider the set of all reads $A_{\Process, \Process'}$ such that every $\Read\in A_{\Process, \Process'}$ with $\Observation(\Read)=(\WriteB, \WriteM)$ satisfies the following conditions.
\begin{enumerate}[noitemsep,topsep=0pt,partopsep=0px]
\item $\Proc{\Read}=\Process'$ and $\Read\not \in \Trace$.
\item $\Proc{\WriteB}\not\in \{\Process, \Process'\}$, and $\Location{\Read}$ is held by $\WriteM$ in $\Trace$, % \Location{\Read}=v
and there is a pending memory write $\WriteM'$ in $\Trace$ with $\Proc{\WriteM'}=\Process$ and $\Location{\WriteM'}=\Location{\Read}$. % =v
\end{enumerate}
If $A_{\Process, \Process'}=\emptyset$ then we let $\FenceMap_{\Trace}(\Process, \Process')=0$, otherwise $\FenceMap_{\Trace}(\Process,\Process')$ is the largest index of a read in $A_{\Process, \Process'}$.
Given two traces $\Trace_1, \Trace_2$, % we write
$\FenceMap_{\Trace_1}\leq \FenceMap_{\Trace_2}$ % to denote that
denotes that
$\FenceMap_{\Trace_1}(\Process,\Process')\leq \FenceMap_{\Trace_2}(\Process,\Process')$ for all $\Process, \Process'\in [k]$.

The intuition behind fence maps is as follows.
Given a trace $\Trace$, the index $\,\FenceMap_{\Trace}(\Process, \Process')\,$ points to the \emph{latest} (wrt $\TO$) read $\Read$ of $\Process'$
that must be executed in any extension of $\Trace$
before $\Process$ can execute its next fence.
This occurs because the following hold in $\Trace$.
\begin{enumerate}[noitemsep,topsep=0pt,partopsep=0px]
\item The variable $\Location{\Read}$ is held by the memory-write $\WriteM \in \Trace$ with $\Observation(\Read)= (\_, \WriteM)$.
\item Thread $\Process$ has executed some buffer-write $\WriteB'\in \Trace$ with $\Location{\WriteB'}=\Location{\Read}=\Location{\WriteM}$, but the corresponding memory-write $\WriteM'$ has not yet been executed in $\Trace$.
Hence, $\Process$ cannot flush its buffers in any extension of $\Trace$ that does not contain $\Read$
(as $\WriteM'$ will not become executable until $\Read$ gets executed).
\end{enumerate}

The following lemmas state two key monotonicity properties of fence maps. %, with proofs in \cref{subsec:app_verifyingpso}.

\vspace{-1mm}
%\smallskip
\begin{restatable}{lemma}{lemfencemapmonotonicityone}\label{lem:fencemap_monotonicity1}
Consider two witness prefixes $\Trace_1,\Trace_2$ such that $\Trace_2=\Trace_1\Concat \WriteM$ for some memory-write $\WriteM$ executable in $\Trace_1$.
We have $\FenceMap_{\Trace_1}\leq \FenceMap_{\Trace_2}$.
Moreover, if $\WriteM$ is a spurious memory-write in $\Trace_1$, then $\FenceMap_{\Trace_1}=\FenceMap_{\Trace_2}$.
\end{restatable}
\vspace{-2mm}

%\smallskip
\begin{restatable}{lemma}{lemfencemapmonotonicitytwo}\label{lem:fencemap_monotonicity2}
Consider two witness prefixes $\Trace_1, \Trace_2$ such that
(i)~$\LocalEvents{\Trace_1}=\LocalEvents{\Trace_2}$,
(ii)~$\FenceMap_{\Trace_1}\leq \FenceMap_{\Trace_2}$, and
(iii)~$\WritesM{\Trace_1}\setminus \SpuriousWritesM{\Trace_1}\subseteq \WritesM{\Trace_2}$.
Let $\Event\in \LocalEvents{X}$ be a thread event that is
executable in $\Trace_i$ for each $i\in[2]$,
and let $\Trace'_i=\Trace_i\Concat\Event$, for each $i\in[2]$.
Then $\FenceMap_{\Trace'_1}\leq \FenceMap_{\Trace'_2}$.
\end{restatable}
\vspace{-1mm}
%
% more precise (but also more verbose) is that (i)+(ii)+(iii) for sigmas implies
% (i)+(ii)+(iii) for sigma-primes, but (i)+(iii) in the consequent is trivial so we just mention (ii)
%
% (i)(iii) together do NOT imply (ii), otherwise we would have an n^k algo for PSO
% the issue is with spurious writes, we can have several traces with (i)(iii) but with different fence maps
%

Note that there exist in total at most $n^{k \cdot k}$ different fence maps.
Further, the following lemma gives a bound on the number of different
fence maps among witness prefixes that contain the same thread events.
%We refer to \cref{subsec:app_verifyingpso} for the proof.

%\smallskip
\begin{restatable}{lemma}{lemfencemap}\label{lem:fencemap_size}
Let $\NumVariables$ be the number of variables.
There exist at most $2^{k\cdot \NumVariables}$ distinct witness prefixes $\Trace_1, \Trace_2$ such that
$\LocalEvents{\Trace_1}=\LocalEvents{\Trace_2}$ and $\FenceMap_{\Trace_1}\neq \FenceMap_{\Trace_2}$.
\end{restatable}

\Paragraph{Algorithm $\AlgoPSO$.}
We are now ready to describe our algorithm $\AlgoPSO$ for the problem $\VPSOrm$.
In high level, the algorithm enumerates all lower sets of
$(\LocalEvents{X}, \TO)$, i.e., the lower sets of the thread events.
The crux of the algorithm is to guarantee that for every witness-prefix $\Trace'$, the algorithm constructs a trace $\Trace$ such that
(i)~$\LocalEvents{\Trace}= \LocalEvents{\Trace'}$,
(ii)~$\WritesM{\Trace}\setminus \SpuriousWritesM{\Trace}\subseteq \WritesM{\Trace'}$, and
(iii)~$\FenceMap_{\Trace}\leq \FenceMap_{\Trace'}$.
To achieve this,
for a given lower set $Y$ of $(\LocalEvents{X}, \TO)$,
the algorithm examines at most as many traces $\Trace$ with $\LocalEvents{\Trace}=Y$
as the number of different fence maps of witness prefixes with the same set of thread events.
Hence, the algorithm examines significantly fewer traces than the $n^{k\cdot (\NumVariables+1)}$ lower sets of $(X,\TO)$.

% In high level, the algorithm enumerates all lower sets of $(\LocalEvents{X}, \TO)$
% by constructing a trace $\Trace$ with $\LocalEvents{\Trace}=Y$ for every lower set $Y$ of $(\LocalEvents{X}, \TO)$.
%maintain the following properties.
%\begin{enumerate}
%\item\label{item:algopso_inv1} For any two traces $\Trace_1, \Trace_2$ constructed by the algorithm we have $\LocalEvents{\Trace_1}\neq \LocalEvents{\Trace_2}$ or $\FenceMap_{\Trace_1}\neq \FenceMap_{\Trace_2}$.
%\item\label{item:algopso_inv2} For every trace $\Trace$ constructed, $\Trace$ is \emph{minimal} in the non-spurious memory-writes, among all witness prefixes with the same set of thread events and fence map.
%That is, for every witness prefix $\Trace'$ such that $\LocalEvents{\Trace}=\LocalEvents{\Trace'}$ and $\FenceMap_{\Trace}=\FenceMap_{\Trace'}$, we have $\WritesM{\Trace}\setminus \SpuriousWritesM{\Trace}\subseteq \WritesM{\Trace'}$.
%\end{enumerate}

\cref{algo:verifypso} presents a formal description of $\AlgoPSO$.
The algorithm maintains a worklist $\Worklist$ of prefixes, and a set
$\DoneSet$ of explored pairs ``(thread events, fence map)''.
Consider an iteration of the main loop in \cref{algo:pso_main_while}.
First in the loop of \cref{algo:pso_while_spurious} all spurious executable
memory-writes are executed. Then \cref{algo:pso_test_done} checks
whether the witness is complete. In case it is not complete, the loop
in \cref{algo:pso_ifextend} enumerates the possibilities to extend
with a thread event.
Crucially, the condition in \cref{algo:pso_if_new} ensures that there
are no duplicates with the same pair ``(thread events, fence map)''.

%\smallskip
\vspace{1mm}
\begin{algorithm}%[H]
\small
\SetInd{0.4em}{0.4em}
\DontPrintSemicolon
%\setstretch{1.05}
\caption{$\AlgoPSO$}\label{algo:verifypso}
\KwIn{
An event set $X$ and a reads-from function $\Observation\colon \Reads{X}\to \Writes{X}$
}
\KwOut{
A witness $\Trace$ that realizes $(X, \Observation)$ if $(X, \Observation)$ is realizable under $\PSO$, else $\Trace=\bot$
}
\BlankLine
$\Worklist\gets \{\epsilon\}$; $\DoneSet\gets \{\emptyset\}$\\
\While{$\Worklist\neq \emptyset$}{\label{algo:pso_main_while}
Extract a trace $\Trace$ from $\Worklist$\label{algo:pso_extract_worklist}\\
\While%(\tcp*[f]{Flush spurious memory-writes})
{$\exists$ spurious $\WriteM$ $\PSO$-executable in $\Trace$}{\label{algo:pso_while_spurious}%exists a spurious memory-write
$\Trace\gets \Trace\Concat \WriteM$\tcp*[f]{Flush spurious memory-write $\WriteM$}\\
}
\lIf(\tcp*[f]{Witness found}){$\Events{\Trace}=X$}{\label{algo:pso_test_done}%All reads realized, witness found
\Return{$\Trace$}
}
\ForEach{thread event $\Event$ $\PSO$-executable in $\Trace$}{\label{algo:pso_ifextend}
Let $\Trace_{\Event}\gets \Trace$\\
\uIf{$\Event$ is a read event with $\Observation(\Read)=(\WriteB,\WriteM)$}{\label{algo:pso_is_read}
\uIf%(\tcp*[f]{$\Event$ observes remotely, observation pending})
{$\Proc{\Read}\neq \Proc{\WriteB}$ and $\WriteM\not \in \Trace_{\Event}$}{\label{algo:pso_read_needs_its_mw}
$\Trace_{\Event}\gets \Trace_{\Event} \Concat\WriteM$\tcp*[f]{Execute the reads-from of $\Event$}\label{algo:pso_execute_observation}\\
}
}
\ElseIf{$\Event$ is a fence event}{\label{algo:pso_is_fence}
Let $\SequenceVar\gets$ any linearization of $(\PendingWritesM{\Trace, \Proc{\Event}}, \TO)$\label{algo:pso_linearize_pending_writes}\\%\tcp*[f]{Linearize pending memory writes }\label{algo:pso_linearize_pending_writes}\\
$\Trace_{\Event} \gets \Trace_{\Event} \Concat \SequenceVar$\tcp*[f]{Execute pending memory writes}\label{algo:pso_execute_pending_writes}\\
}
$\Trace_{\Event}\gets\Trace_{\Event} \Concat \Event$\tcp*[f]{Finally, execute $\Event$}\label{algo:pso_execute_event}\\
\uIf{$\not \exists\Trace'\in \DoneSet$ s.t. $\LocalEvents{\Trace_{\Event}}=\LocalEvents{\Trace'}$ and $\FenceMap_{\Trace_{\Event}}=\FenceMap_{\Trace'}$}{\label{algo:pso_if_new}
Insert $\Trace_{\Event}$ in $\Worklist$ and in $\DoneSet$\tcp*[f]{Continue from $\Trace_{\Event}$}\label{algo:pso_insert_worklist}\\
}
}
}
\Return{$\bot$}
\end{algorithm}
\vspace{1mm}

%We now provide the insights behind the correctness of $\AlgoPSO$.
%The correctness argument has two components: (i) soundness and (ii) completeness. %, which we present below.

%\Paragraph{Correctness.}
%We now describe the correctness of $\AlgoPSO$.
\Paragraph{Soundness.}
The soundness % argument
of $\AlgoPSO$ follows directly from the definition
of $\PSO$-executable events, % directly $\PSO$-executable events $\PSO$-executability
and is similar to the case of $\AlgoTSO$.

%Here we present the main insights towards the completeness of $\AlgoPSO$.
% Recall the three invariants maintained by $\AlgoPSO$ as described in the previous paragraph.
% (i)~$\LocalEvents{\Trace}= \LocalEvents{\Trace'}$,
% (ii)~$\WritesM{\Trace}\setminus\SpuriousWritesM{\Trace}\subseteq \WritesM{\Trace'}$, and
% (iii)~$\FenceMap_{\Trace}\leq \FenceMap_{\Trace'}$,

\Paragraph{Completeness.} %{Completeness: Main idea.}
For each witness prefix $\Trace'$, algorithm $\AlgoPSO$ generates
a trace $\Trace$ with (i)~$\LocalEvents{\Trace}= \LocalEvents{\Trace'}$,
(ii)~$\WritesM{\Trace}\setminus\SpuriousWritesM{\Trace}\subseteq \WritesM{\Trace'}$, and
(iii)~$\FenceMap_{\Trace}\leq \FenceMap_{\Trace'}$. This fact directly
implies completeness, and it is achieved by the following key invariant.
Consider that the algorithm has constructed a trace $\Trace$, and is attempting to extend $\Trace$ with a thread event $\Event$.
Further, let $\Trace'$ be an arbitrary witness prefix with
(i)~$\LocalEvents{\Trace}= \LocalEvents{\Trace'}$,
(ii)~$\WritesM{\Trace}\setminus\SpuriousWritesM{\Trace}\subseteq \WritesM{\Trace'}$, and
(iii)~$\FenceMap_{\Trace}\leq \FenceMap_{\Trace'}$.
% The key insight is that
If $\Trace'$ can be extended so that the next thread event is $\Event$,
then $\Event$ is also executable in $\Trace$,  % and further,
and
(by \cref{lem:fencemap_monotonicity1} and \cref{lem:fencemap_monotonicity2})
the extension of $\Trace$ with $\Event$ maintains the invariant.
In \cref{fig:PSO_algo} we provide an intuitive illustration
of the completeness idea.
%
% To be absolutely precise, the invariant is maintained in our
% extension with thread event $\Event$ only after we subsequently
% play maximally the spurious memory-writes.
%
%
% Points (i)+(ii) together automatically imply (iii)?
% NO, because of spurious writes (otherwise we would have $n^k$ algo).
%

\begin{figure*}[!h]
%\begin{minipage}[c]{0.6\textwidth}
\begin{tikzpicture}[thick,
pre/.style={<-,shorten >= 2pt, shorten <=2pt, very thick},
post/.style={->,shorten >= 2pt, shorten <=2pt,  very thick},
seqtrace/.style={->, line width=2},
aux_seqtrace/.style={->, line width=1, draw=gray},
und/.style={very thick, draw=gray},
event/.style={rectangle, minimum height=3.5mm, draw=black, fill=white, minimum width=8mm,   line width=1pt, inner sep=2, font={\footnotesize}},
aux_event/.style={event, draw=gray},
virt/.style={circle,draw=black!50,fill=black!20, opacity=0},
bad/.style={preaction={fill, white}, pattern color=red!40, pattern=north east lines},
good/.style={preaction={fill, white}, pattern color=green!60, pattern=north west lines},
neutral/.style={preaction={fill, white}, pattern color=yellow, pattern=crosshatch},
isLabel/.style={rectangle, fill opacity=0.75, fill=white, text opacity=1}
]

\pgfdeclarelayer{bg}
\pgfsetlayers{bg,main}

\newcommand{\xs}{1.7}%1.7
\newcommand{\ys}{0.57}%0.6
\newcommand{\yaux}{0.2}

\begin{pgfonlayer}{bg}
\fill[gray!10]
(-4.5*\xs, 0.5*\ys) to
(-4.5*\xs, -3.5*\ys) to
(3.4*\xs, -3.5*\ys) to
(3.4*\xs, 0.5*\ys) to cycle;

\fill[gray!20]
(-4.5*\xs, 0.5*\ys) to
(-4.5*\xs, -2.5*\ys) to
(3.4*\xs, -2.5*\ys) to
(3.4*\xs, 0.5*\ys) to cycle;
\end{pgfonlayer}

\node[] (t1m_0)   at (0*\xs, 0*\ys) {\small$\LocalTrace_3$};
\node[] (t1m_end) at (0*\xs, -8*\ys) {};

\node[] (t1ax_0)   at (1*\xs, 0*\ys) {\small$\LocalTrace_3'(x)$};
\node[] (t1ax_end) at (1*\xs, -8*\ys) {};

\node[] (t1ay_0)   at (2*\xs, 0*\ys) {\small$\LocalTrace_3'(y)$};
\node[] (t1ay_end) at (2*\xs, -8*\ys) {};

\node[] (t2m_0)   at (3*\xs, 0*\ys) {\small$\LocalTrace_4$};
\node[] (t2m_end) at (3*\xs, -8*\ys) {};

\node[] (t3m_0)   at (-4*\xs, 0*\ys) {\small$\LocalTrace_1$};
\node[] (t3m_end) at (-4*\xs, -8*\ys) {};

\node[] (t3ax_0)   at (-3*\xs, 0*\ys) {\small$\LocalTrace_1'(x)$};
\node[] (t3ax_end) at (-3*\xs, -8*\ys) {};

\node[] (t3ay_0)   at (-2*\xs, 0*\ys) {\small$\LocalTrace_1'(y)$};
\node[] (t3ay_end) at (-2*\xs, -8*\ys) {};

\node[] (t4m_0)   at (-1*\xs, 0*\ys) {\small$\LocalTrace_2$};
\node[] (t4m_end) at (-1*\xs, -8*\ys) {};

\begin{pgfonlayer}{bg}
\draw[seqtrace] (t1m_0) to (t1m_end);
\draw[seqtrace] (t2m_0) to (t2m_end);
\draw[seqtrace] (t3m_0) to (t3m_end);
\draw[seqtrace] (t4m_0) to (t4m_end);

\draw[aux_seqtrace] (t1ax_0) to (t1ax_end);
\draw[aux_seqtrace] (t1ay_0) to (t1ay_end);
\draw[aux_seqtrace] (t3ax_0) to (t3ax_end);
\draw[aux_seqtrace] (t3ay_0) to (t3ay_end);
\end{pgfonlayer}

\node[event] (t1m_2) at (0*\xs, -1*\ys) {$\WriteB_3(y)$};
\node[event] (t1m_3) at (0*\xs, -2*\ys) {$\WriteB_4(x)$};
\node[event, neutral] (t1m_4) at (0*\xs, -7*\ys) {$\Fence_1$};

\node[aux_event, neutral]
(t1ax_1) at (1*\xs, -5.75*\ys) {$\WriteM_4(x)$};
\node[aux_event, neutral]
(t1ay_2) at (2*\xs, -5*\ys) {$\WriteM_3(y)$};

\node[event, neutral] (t2m_1) at (3*\xs, -7*\ys) {$\Read_3(y)$};

\node[event] (t3m_2) at (-4*\xs, -1*\ys) {$\WriteB_1(y)$};
\node[event] (t3m_3) at (-4*\xs, -2*\ys) {$\WriteB_2(x)$};
\node[aux_event, good] (t3ax_1) at (-3*\xs, -3*\ys) {$\WriteM_2(x)$};
\node[aux_event, good] (t3ay_2) at (-2*\xs, -3*\ys) {$\WriteM_1(y)$};

\node[event, good]
          (t4m_3) at (-1*\xs, -4*\ys) {$\Read_1(y)$};
\node[event, bad]
          (t4m_4) at (-1*\xs, -5*\ys) {$\Read_2(x)$};

\begin{pgfonlayer}{bg}
\draw[post, gray] (t1m_2) to (t1ay_2);
\draw[post, gray] (t1m_3) to (t1ax_1);
\draw[post, gray] (t3m_2) to (t3ay_2);
\draw[post, gray] (t3m_3) to (t3ax_1);

\draw[post, \darkred, dashed] (t1ay_2) to (t2m_1);
\draw[post, \darkred, dashed] (t3ax_1) to (t4m_4);
\draw[post, \darkred, dashed] (t3ay_2) to (t4m_3);

\draw[post, gray, densely dotted] (t1ax_1) to (t1m_4);
\draw[post, gray, densely dotted] (t1ay_2) to[out=-90, in=0] (t1m_4);

\draw[post, color=black!20!blue, dotted] (t4m_3) to (t1ay_2);
\draw[post, color=black!20!blue, dotted] (t4m_4) to (t1ax_1);
\end{pgfonlayer}

\draw[post, color=black!20!blue, dash dot]
(t4m_4) -- (t1m_4) node [midway, isLabel, rotate=-35, xshift=-15pt, yshift=-20pt] {\small $\FenceMap_{\Trace'}(\Process_3, \Process_2)$};

\end{tikzpicture}
%\end{minipage}
%\vspace{-2mm}
%\begin{minipage}[c]{0.38\textwidth}
\vspace{-1mm}
\caption{
$\AlgoPSO$ completeness idea. Consider
the witness prefix $\Trace'$ (lighter gray) and the corresponding trace
$\Trace$ constructed by the algorithm (darker gray).
The fence $\Fence_1$ is $\PSO$-executable in $\Trace$ but not in $\Trace'$,
since in the latter, $\Proc{\Fence_1}$ has non-empty buffers,
but the variables $x$ and $y$ are held by $\WriteM_1$ and $\WriteM_2$, respectively.
This is equivalent to waiting until after $\Read_1$ and $\Read_2$ have been executed.
Since executing $\Read_2$ implies having executed $\Read_1$,
the fence map $\FenceMap_{\Trace'}(\Process_3, \Process_2)$ compresses this information by only pointing to $\Read_2$.
}
\label{fig:PSO_algo}
%\end{minipage}
%\vspace{-2mm}
\end{figure*}
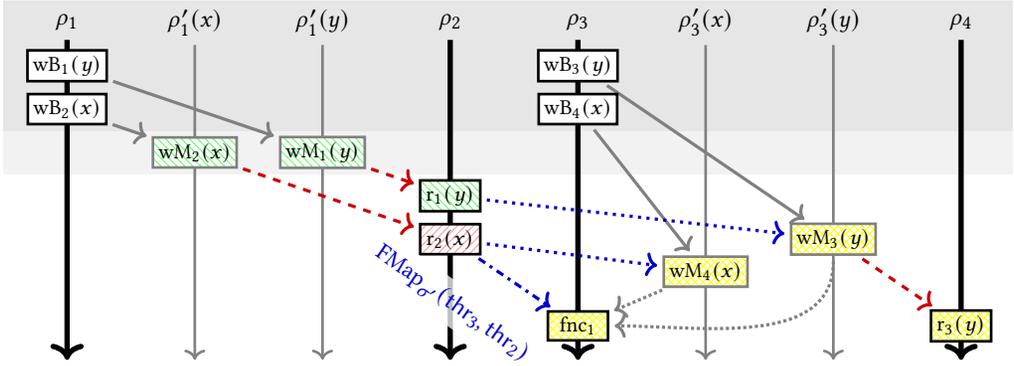

We now prove the argument in detail for the above $\Trace$,
$\Trace'$ and thread event $\Event$.
Assume that $\Trace' \Concat \Sequence \Concat \Event$ is a witness
prefix as well, for a sequence of memory-writes $\Sequence$. Consider the following cases.
\begin{enumerate}[noitemsep,topsep=0pt,partopsep=0px]%,wide,labelwidth=!,labelindent=0pt]
\item If $\Event$ is a read event, let $\Wpair{\Write}=(\WriteB, \WriteM)=\Observation(\Event)$.
If it is a local write (i.e., $\Proc{\Wpair{\Write}} = \Proc{\Event}$),
necessarily $\WriteB \in \Trace' \Concat \Sequence$, and since the traces agree
on thread events, we have $\WriteB \in \Trace$; thus $\Event$ is executable in $\Trace$.
Otherwise, $\Wpair{\Write}$ is a remote write (i.e., $\Proc{\Wpair{\Write}} \neq \Proc{\Event}$).
Assume towards contradiction that
$\Event$ is not executable in $\Trace$; this can happen in two cases.

In the first case, the variable $x=\Location{\Event}$ is held by another
(non-spurious) memory-write $\WriteM'$ in $\Trace$.
Since $\WritesM{\Trace}\setminus\SpuriousWritesM{\Trace}\subseteq \WritesM{\Trace'}$,
and $\LocalEvents{\Trace}= \LocalEvents{\Trace'}$, the variable $x$ is
also held by $\WriteM'$ in $\Trace' \Concat \Sequence$. But then, both
$\WriteM$ and $\WriteM'$ hold $x$ in $\Trace' \Concat \Sequence$, a contradiction.

In the second case, there is a write $(\WriteB', \WriteM')$ % \Wpair{\Write'}=
% that is local to $\Event$ (i.e., $\Proc{\Wpair{\Write'}} = \Proc{\Event}$),
% and such that $\WriteM'\not \in \Trace$.
with $\Location{\WriteM'} = \Location{\Event}$ and
$\WriteB'<_{\TO}\Event$ and
$\WriteM'\not \in \Trace$.
If $\WriteM' \not\in \Trace' \Concat \Sequence$, then $\Event$ would read-from $\WriteB'$ from the buffer
in $\Trace' \Concat \Sequence \Concat \Event$, contradicting $\Observation(\Event) = (\_, \WriteM)$.
%Therefore,
Thus $\WriteM' \in \Trace' \Concat \Sequence$,
and further $\WriteM \in \Trace' \Concat \Sequence$ with $\WriteM' <_{\Trace' \Concat \Sequence} \WriteM$.
Since $\Trace' \Concat \Sequence \Concat \Event$ is a witness prefix
and $\WriteB'<_{\TO}\Event$, %precedes $\Event$,
we have $\WriteB' \in \Trace'$.
From this and $\LocalEvents{\Trace} = \LocalEvents{\Trace'}$ we have that
$\WriteB' \in \Trace$ and %, thus
$\WriteM'$ is pending in $\Trace$.
This together gives us that $\WriteM'$ is spurious in $\Trace$.
Consider the earliest memory-write pending in $\Trace$ on
the same buffer (i.e., $\Proc{\WriteM'}$ and $\Location{\WriteM'}$), denote it $\WriteM''$.
% The below is true, just commenting it anyway to save some space.
%Since it is pending,
% none of the prior reads (in $\LocalEvents{\Trace} = \LocalEvents{\Trace'}$) have read-from it;
% since $\WriteM'' \leq_{\TO} \WriteM'$ and $\WriteM' <_{\Trace' \Concat \Sequence} \WriteM$,
% none of the future reads will read-from it. Thus it is a spurious memory write.
We have that $\WriteM'' \leq_{\TO} \WriteM'$ and $\WriteM''$ is spurious in $\Trace$.
Further, $\WriteM''$ is executable in $\Trace$.
% We save space here by arguing why $\WriteM''$ is executable, its just a simple case-based argument.
But then it would
have been added to $\Trace$ in the while loop of \cref{algo:pso_while_spurious}, a contradiction.

\item Assume that $\Event$ is a fence event, and let
$\WriteM_1,\dots,\WriteM_j$ be the pending memory-writes of $\Proc{\Event}$ in $\Trace$.
Suppose towards contradiction that $\Event$ is not executable. Then one of the $\WriteM_i$ %, 1\leq i\leq j$
is not executable,
let $x=\Location{\WriteM_i}$. % be the variable of $\WriteM_i$.
Similarly to the %previous case,
above, there can be two cases where this might happen.

The first case is when $\WriteM_i$ must be read-from by some read event $\Read\not \in \Trace$,
but $\Read$ is preceded by a local write $(\WriteB, \WriteM)$
(i.e., $\WriteB<_{\TO}\Read$)
on the same variable $x$ while $\WriteM\not \in \Trace$.
A similar analysis to the previous case shows that the earliest pending write
on $\Proc{\WriteM}$ for variable $x$ is spurious, and thus already
added to $\Trace$ due to the while loop in \cref{algo:pso_while_spurious}, a contradiction.

The second case is when the variable $x$ is held in $\Trace$.
Since $\FenceMap_{\Trace}\leq \FenceMap_{\Trace'}$, the variable $x$ is also held in $\Trace'$, and thus $\WriteM_i$ is not executable in $\Trace'$ either.
But then $\Trace' \Concat \Sequence \Concat \Event$ cannot be a witness prefix, a contradiction.
\end{enumerate}
The following lemma states the correctness of $\AlgoPSO$,
which together with the complexity argument establishes \cref{them:vpso}.
%We refer to \cref{subsec:app_verifyingpso} for the proof.

%\smallskip
\begin{restatable}{lemma}{lemverifyingpsocorrectness}\label{lem:verifyingpso_correctness}
$(X, \Observation)$ is realizable under $\PSO$ iff $\AlgoPSO$ returns a trace $\Trace\neq \epsilon$.
\end{restatable}
%\vspace{-2mm}

We conclude this section with some insights on the
relationship between $\VTSOrm$ and $\VPSOrm$.
%We end this section outlining the differences of the
%reads-from consistency verification of $\TSO$ and $\PSO$ executions.

\Paragraph{Relation between $\TSO$ and $\PSO$ verification.}
In high level, $\TSO$ might be perceived as a special case of $\PSO$,
where every thread is equipped with one buffer ($\TSO$)
as opposed to one buffer per global variable ($\PSO$).
However, the communication patterns between $\TSO$ and $\PSO$ are drastically different.
As a result, our algorithm $\AlgoPSO$ is not applicable to $\TSO$,
and we do not see an extension of $\AlgoTSO$ for handling $\PSO$ efficiently.
In particular, the minimal strategy of $\AlgoPSO$ on memory-writes is
based on the following observation: for a read $\Read$ observing
a remote memory-write $\WriteM$, it always suffices to execute
$\WriteM$ exactly before executing $\Read$ (unless $\WriteM$ has already been executed).
This holds because the corresponding buffer contains memory-writes
\emph{only} on the same variable, and thus all such memory-writes that
precede $\WriteM$ cannot be read-from by any subsequent read.
This property does not hold for $\TSO$: as there is a single buffer,
$\WriteM$ might be executed as a result of flushing the buffer of
thread $\Proc{\WriteM}$ to make another memory-write $\WriteM'$ visible,
on a \emph{different} variable than $\Location{\WriteM}$, and thus
$\WriteM'$ might be observable by a subsequent read.
Hence the minimal strategy of $\AlgoPSO$ on memory-writes does not apply to $\TSO$.
On the other hand, the maximal strategy of $\AlgoTSO$ is not effective
for $\PSO$, as it requires enumerating all lower sets of
$(\WritesM{X}, \Observation)$, which are $n^{k\cdot \NumVariables}$ many in $\PSO$
(where $\NumVariables$ is the number of variables),
and thus this leads to worse bounds than the ones we achieve in \cref{them:vpso}.

\subsection{Closure for $\AlgoTSO$ and $\AlgoPSO$}\label{subsec:heuristics}
% Heuristic

In this section we introduce \emph{closure}, a practical heuristic
to efficiently detect whether a given instance $(X,\Observation)$ of the
verification problem $\VTSOrm$ resp. $\VPSOrm$ is unrealizable.
%for early detection of instances $(X,\Observation)$ that are not realizable. %with no solution.
Closure is sound, meaning that a realizable instance $(X,\Observation)$
is never declared unrealizable by closure. Further, closure is not complete,
which means there exist unrealizable instances $(X,\Observation)$
not detected as such by closure. Finally, closure can be computed in
time polynomial with respect to the number of events (i.e., size of $X$),
irrespective of the underlying number of threads and variables.

Given an instance $(X,\Observation)$, any solution of
$\VTSOrm$/$\VPSOrm(X,\Observation)$ respects $\TO\Project X$, i.e.,
the program order upon $X$. Closure constructs the weakest partial
order $P(X)$ that refines the program order
(i.e., $P\Refines \TO\Project X$) and further satisfies % the following conditions
for each read $\Read \in \Reads{X}$ with $\Observation(\Read) = (\WriteB,\WriteM)$:
%the following conditions for each $\Read\in \Reads{X}$ with its desired
%reads-from $\Observation(\Read) = (\WriteB,\WriteM)$:
\begin{enumerate}[noitemsep,topsep=0pt,partopsep=0px]
\item\label{item:closure1} If $\Proc{\Read} \neq \Proc{\Observation(\Read)}$, then (i) $\WriteM <_P \Read$ and
(ii) $\bad{\WriteM} <_P \WriteM$ for any
$(\bad{\WriteB},\bad{\WriteM}) \in \Writes{X_{\Proc{\Read}}}$ such that $\Confl{\bad{\WriteM}}{\Read}$ and $\bad{\WriteB} <_{\TO} \Read$.
\item\label{item:closure2} For any $\bad{\WriteM} \in \WritesM{X_{\neq\Proc{\Read}}}$ such that $\Confl{\bad{\WriteM}}{\Read}$ and $\bad{\WriteM} \neq \WriteM$, $\bad{\WriteM} <_P \Read$ implies $\bad{\WriteM} <_P \WriteM$.
%\vspace{1mm}\\\vspace{1mm}if $\bad{\WriteM} <_P \Read$ then also $\bad{\WriteM} <_P \WriteM$.
\item\label{item:closure3} For any $\bad{\WriteM} \in \WritesM{X_{\neq\Proc{\Read}}}$ such that $\Confl{\bad{\WriteM}}{\Read}$ and $\bad{\WriteM} \neq \WriteM$, $\WriteM <_P \bad{\WriteM}$ implies $\Read <_P \bad{\WriteM}$.
%\vspace{1mm}\\if $\WriteM <_P \bad{\WriteM}$ then also $\Read <_P \bad{\WriteM}$.
\end{enumerate}

If no above $P$ exists, %$P$ with the above conditions exists,
the instance $\VTSOrm$/$\VPSOrm(X,\Observation)$ provably has no solution.
In case $P$ exists,
% $\VTSOrm$/$\VPSOrm(X,\Observation)$ may or may not have a solution, but
each solution $\Trace$ of $\VTSOrm$/$\VPSOrm(X,\Observation)$
provably respects $P$ (formally, $\Trace\Refines P$). % formally,

\begin{figure}[h]
% \vspace{2mm}

\begin{subfigure}[b]{0.31\textwidth}
\centering
\begin{tikzpicture}[thick,
pre/.style={<-,shorten >= 2pt, shorten <=2pt, very thick},
post/.style={->,shorten >= 2pt, shorten <=2pt,  very thick},
seqtrace/.style={->, line width=2},
aux_seqtrace/.style={->, line width=1, draw=gray},
und/.style={very thick, draw=gray},
event/.style={rectangle, minimum height=4.3mm, draw=black, fill=white, minimum width=8mm, line width=1pt, inner sep=2, font={\footnotesize}},
aux_event/.style={event, draw=gray},
virt/.style={circle,draw=black!50,fill=black!20, opacity=0},
bad/.style={preaction={fill, white}, pattern color=red!40, pattern=north east lines},
good/.style={preaction={fill, white}, pattern color=green!60, pattern=north west lines},
]

\newcommand{\xs}{0.95}
\newcommand{\ys}{0.65}

\node[] (t1m_0)   at (0*\xs, 0*\ys) {\small$\LocalTrace_1$};
\node[] (t1m_end) at (0*\xs, -5.2*\ys) {};

\node[] (t1a_0)   at (0.8*\xs, 0*\ys) {\small$\LocalTrace_1'$};
\node[] (t1a_end) at (0.8*\xs, -2.7*\ys) {};

\node[] (t2a_0)   at (2.2*\xs, 0*\ys) {\small$\LocalTrace_2'$};
\node[] (t2a_end) at (2.2*\xs, -3.3*\ys) {};

\node[] (t2m_0)   at (3*\xs, 0*\ys) {\small$\LocalTrace_2$};
\node[] (t2m_end) at (3*\xs, -5.2*\ys) {};

\draw[seqtrace] (t1m_0) to (t1m_end);
\draw[aux_seqtrace] (t1a_0) to (t1a_end);
\draw[aux_seqtrace] (t2a_0) to (t2a_end);
\draw[seqtrace] (t2m_0) to (t2m_end);

\node[event] (t1m_1) at (0*\xs, -1*\ys) {$\bad{\WriteB}$};
\node[event] (t1m_2) at (0*\xs, -4*\ys) {$\Read$};
\node[aux_event] (t1a_1) at (0.8*\xs, -2.4*\ys) {$\bad{\WriteM}$};

\node[event] (t2m_1) at (3*\xs, -1.6*\ys) {$\WriteB$};
\node[aux_event] (t2a_1) at (2.2*\xs, -3.0*\ys) {$\WriteM$};

\draw[post, draw=gray] (t1m_1) to (t1a_1);
\draw[post, draw=gray] (t2m_1) to (t2a_1);
\draw[post, \darkred, dashed] (t1a_1) to (t2a_1);
\draw[post, \darkred, dashed] (t2a_1) to (t1m_2);

\end{tikzpicture}
\vspace{-2mm}
\caption{
Rule \cref{item:closure1}.
Both new orderings are necessary, as a reversal of either
of them would
``hide'' $\WriteM$ from $\Read$,
making it impossible for $\Read$ to read-from
$(\WriteB,\WriteM)$.
}
\label{subfig:closure_rule1}
\end{subfigure}
%%%%%%%%%%%%%%%%%%%%%%%%%%%%%%%%%%%%%%%%%%%%%%%%%%%%%%%%%%%%%%%%%%%%%%%%%%%%%%%%%%%%%%%%%%%%%%%%%%%%%%%%
%\qquad%\quad
\quad
\begin{subfigure}[b]{0.31\textwidth}
\centering
\begin{tikzpicture}[thick,
pre/.style={<-,shorten >= 2pt, shorten <=2pt, very thick},
post/.style={->,shorten >= 2pt, shorten <=2pt,  very thick},
seqtrace/.style={->, line width=2},
aux_seqtrace/.style={->, line width=1, draw=gray},
und/.style={very thick, draw=gray},
event/.style={rectangle, minimum height=4.3mm, draw=black, fill=white, minimum width=8mm, line width=1pt, inner sep=2, font={\footnotesize}},
aux_event/.style={event, draw=gray},
virt/.style={circle,draw=black!50,fill=black!20, opacity=0},
bad/.style={preaction={fill, white}, pattern color=red!40, pattern=north east lines},
good/.style={preaction={fill, white}, pattern color=green!60, pattern=north west lines},
]

\newcommand{\xs}{0.95}
\newcommand{\ys}{0.65}

\node[] (t1m_0)   at (0*\xs, 0*\ys) {\small$\LocalTrace_1$};
\node[] (t1m_end) at (0*\xs, -5.2*\ys) {};

\node[] (t1a_0)   at (0.8*\xs, 0*\ys) {\small$\LocalTrace_1'$};
\node[] (t1a_end) at (0.8*\xs, -2.7*\ys) {};

\node[] (t2a_0)   at (2.2*\xs, 0*\ys) {\small$\LocalTrace_2'$};
\node[] (t2a_end) at (2.2*\xs, -3.3*\ys) {};

\node[] (t2m_0)   at (3*\xs, 0*\ys) {\small$\LocalTrace_2$};
\node[] (t2m_end) at (3*\xs, -5.2*\ys) {};

\draw[seqtrace] (t1m_0) to (t1m_end);
\draw[aux_seqtrace] (t1a_0) to (t1a_end);
\draw[aux_seqtrace] (t2a_0) to (t2a_end);
\draw[seqtrace] (t2m_0) to (t2m_end);

\node[event] (t1m_1) at (0*\xs, -1*\ys) {$\WriteB$};
\node[event] (t1m_2) at (0*\xs, -4*\ys) {$\Read$};
\node[aux_event] (t1a_1) at (0.8*\xs, -2.4*\ys) {$\WriteM$};

\node[event] (t2m_1) at (3*\xs, -1.6*\ys) {$\bad{\WriteB}$};
\node[aux_event] (t2a_1) at (2.2*\xs, -3.0*\ys) {$\bad{\WriteM}$};

\draw[post, draw=gray] (t1m_1) to (t1a_1);
\draw[post, draw=gray] (t2m_1) to (t2a_1);
\draw[post, draw=gray] (t2a_1) to (t1m_2);

\draw[post, \darkred, dashed] (t2a_1) to (t1a_1);

\end{tikzpicture}
\vspace{-2mm}
\caption{
Rule \cref{item:closure2}.
The new ordering is necessary; its reversal would make
$\bad{\WriteM}$ appear between $(\WriteB,\WriteM)$ and $\Read$,
making it impossible for $\Read$ to read-from $(\WriteB,\WriteM)$.
}
\label{subfig:closure_rule2}
\end{subfigure}
%%%%%%%%%%%%%%%%%%%%%%%%%%%%%%%%%%%%%%%%%%%%%%%%%%%%%%%%%%%%%%%%%%%%%%%%%%%%%%%%%%%%%%%%%%%%%%%%%%%%%%%%
%\qquad%\quad
\quad
\begin{subfigure}[b]{0.31\textwidth}
\centering
\begin{tikzpicture}[thick,
pre/.style={<-,shorten >= 2pt, shorten <=2pt, very thick},
post/.style={->,shorten >= 2pt, shorten <=2pt,  very thick},
seqtrace/.style={->, line width=2},
aux_seqtrace/.style={->, line width=1, draw=gray},
und/.style={very thick, draw=gray},
event/.style={rectangle, minimum height=4.3mm, draw=black, fill=white, minimum width=8mm, line width=1pt, inner sep=2, font={\footnotesize}},
aux_event/.style={event, draw=gray},
virt/.style={circle,draw=black!50,fill=black!20, opacity=0},
bad/.style={preaction={fill, white}, pattern color=red!40, pattern=north east lines},
good/.style={preaction={fill, white}, pattern color=green!60, pattern=north west lines},
]

\newcommand{\xs}{0.95}
\newcommand{\ys}{0.65}

\node[] (t1m_0)   at (0*\xs, 0*\ys) {\small$\LocalTrace_1$};
\node[] (t1m_end) at (0*\xs, -5.2*\ys) {};

\node[] (t1a_0)   at (0.8*\xs, 0*\ys) {\small$\LocalTrace_1'$};
\node[] (t1a_end) at (0.8*\xs, -2.7*\ys) {};

\node[] (t2a_0)   at (2.2*\xs, 0*\ys) {\small$\LocalTrace_2'$};
\node[] (t2a_end) at (2.2*\xs, -3.3*\ys) {};

\node[] (t2m_0)   at (3*\xs, 0*\ys) {\small$\LocalTrace_2$};
\node[] (t2m_end) at (3*\xs, -5.2*\ys) {};

\draw[seqtrace] (t1m_0) to (t1m_end);
\draw[aux_seqtrace] (t1a_0) to (t1a_end);
\draw[aux_seqtrace] (t2a_0) to (t2a_end);
\draw[seqtrace] (t2m_0) to (t2m_end);

\node[event] (t1m_1) at (0*\xs, -1*\ys) {$\WriteB$};
\node[event] (t1m_2) at (0*\xs, -4*\ys) {$\Read$};
\node[aux_event] (t1a_1) at (0.8*\xs, -2.4*\ys) {$\WriteM$};

\node[event] (t2m_1) at (3*\xs, -1.6*\ys) {$\bad{\WriteB}$};
\node[aux_event] (t2a_1) at (2.2*\xs, -3.0*\ys) {$\bad{\WriteM}$};

\draw[post, draw=gray] (t1m_1) to (t1a_1);
\draw[post, draw=gray] (t2m_1) to (t2a_1);
\draw[post, draw=gray] (t1a_1) to (t2a_1);
\draw[post, \darkred, dashed] (t1m_2) to (t2a_1);

\end{tikzpicture}
\vspace{-2mm}
\caption{
Rule \cref{item:closure3}.
The new ordering is necessary; its reversal would make
$\bad{\WriteM}$ appear between $(\WriteB,\WriteM)$ and $\Read$,
making it impossible for $\Read$ to read-from $(\WriteB,\WriteM)$.
}
\label{subfig:closure_rule3}
\end{subfigure}

\caption{
Illustration of the three closure rules. In each example,
the read $\Read$ has to read-from the write $(\WriteB, \WriteM)$, i.e.,
$\Observation(\Read) = (\WriteB,\WriteM)$.
All depicted events are on the same variable
(which is omitted for clarity).
The gray solid edges illustrate orderings already present in the
partial order, and the red dashed edges illustrate
the resulting new orderings
enforced by the specific rule.
}
\label{fig:closure}
\end{figure}

The intuition behind closure is as follows.
The construction starts with the program order $\TO\Project X$,
and then, utilizing the above rules
\cref{item:closure1}, \cref{item:closure2} and \cref{item:closure3},
it iteratively adds further event orderings such that every witness
execution provably has to follow the orderings.
Consequently, if the added orderings induce a cycle, this serves as
a proof that there exists no witness of the input instance
$(X,\Observation)$.
The rules
\cref{item:closure1}, \cref{item:closure2} and \cref{item:closure3}
can intuitively be though of as simple reasoning arguments why specific
orderings have to be present in each witness of $(X,\Observation)$,
and \cref{fig:closure} provides an illustration of the rules.

We leverage the guarantees of closure by computing it before
executing $\AlgoTSO$ resp. $\AlgoPSO$.
If no closure $P$ of $(X,\Observation)$ exists,
the algorithm $\AlgoTSO$ resp. $\AlgoPSO$ does not need to be executed
at all, as we already know that $(X,\Observation)$ is unrealizable.
%Once the closure $P$ of $(X,\Observation)$ is constructed, since
Otherwise we obtain the closure $P$, we execute $\AlgoTSO$/$\AlgoPSO$
to search for a witness of $(X,\Observation)$,
and we restrict $\AlgoTSO$/$\AlgoPSO$ to only consider prefixes
$\Trace'$ respecting $P$ (formally, $\Trace'\Refines P\Project \Events{\Trace'}$),
since we know that
each solution of $\VTSOrm$/$\VPSOrm(X,\Observation)$ has to respect $P$.

The notion of closure, its beneficial properties, as well as construction
algorithms are well-studied for the $\SC$ memory model~\cite{Chalupa17,Abdulla19,Pavlogiannis20}.
Our conditions above extend this notion to $\TSO$ and $\PSO$.
Moreover, the closure we introduce here is \emph{complete} for concurrent programs with two threads,
i.e., if $P$ exists then there is a valid trace realizing $(X,\Observation)$ under the respective memory model.
% linearizing

% a low-cost filter
% This is important as $\AlgoTSO$ (resp. $\AlgoPSO$) first enumerates
% all lower sets it considers before being able to declare an
% instance unsolvable.
% In our SMC experiments, all instances
% $(X,\Observation)$ with no solution were declared as such by the closure,
% and hence only realizable instances were passed on to $\AlgoTSO$/$\AlgoPSO$.
% \Viktor{We unfortunately did not track this, though I am quite certain it is true. Do we delete the claim just to be safe?}
% To preserve completeness of SMC, the closure
% produces no false alarms, meaning that each solvable
% instance passes the closure.

% Although the worst-case complexity bounds of $\AlgoTSO$ and
% $\AlgoPSO$ closely match the lower bounds, %established in~\cref{subsec:hardness}, of~\cref{them:vtspso_w1hard},
% there is significant potential for practical improvement, as the
% $\VTSOrm$/$\VPSOrm$ instances generated during SMC rarely come close
% to the cases difficult for $\AlgoTSO$/$\AlgoPSO$.
% \Paragraph{Closure.}

% \Paragraph{Search strategy.}
% This naturally preserves completeness, and significantly reduces
% This preserves completeness, while significantly reducing
% the search space to consider for $\AlgoTSO$/$\AlgoPSO$.
% Further, the search for a witness to $\VTSOrm$/$\VPSOrm(X,\Observation)$ % trace
% is guided by an \emph{auxiliary} trace, and the process is described
% in detail in~\cref{subsec:app_search}.

%\input{w1hardness}

\subsection{Verifying Executions with Atomic Primitives}\label{subsec:verifying_casrmw}

For clarity of presentation of the core algorithmic concepts,
we have thus far neglected more involved atomic operations,
namely atomic \emph{read-modify-write} (RMW) and atomic
\emph{compare-and-swap} (CAS).
We show how our approach handles verification
of $\TSO$ and $\PSO$ executions that also include
RMW and CAS operations here in a separate section.
Importantly, our treatment retains the complexity bounds established
in \cref{them:vtso} and \cref{them:vpso}.

\Paragraph{Atomic instructions.}
We consider the concurrent program under the $\TSO$ resp. $\PSO$
memory model, which can further atomically execute the following
types of instructions.
\begin{enumerate}[noitemsep,topsep=0pt,partopsep=0px]
\item A \emph{read-modify-write} instruction $\RMW$ executes atomically
the following sequence. It (i) reads, with respect to the $\TSO$ resp.
$\PSO$ semantics, the value $v$ of a global variable $x\in \Globals$,
then (ii) uses $v$ to compute a new value $v'$,
and finally (iii) writes the new value $v'$ to
the global variable $x$.
An example of a typical $\RMW$ computation is fetch-and-add
(resp. fetch-and-sub),
where $v' = v + c$ for some positive (resp. negative) constant $c$.
\item A \emph{compare-and-swap} instruction $\CAS$ executes atomically
the following sequence. It (i) reads, with respect to the $\TSO$ resp.
$\PSO$ semantics, the value $v$ of a global variable $x\in \Globals$,
(ii) compares it with a value $c$, and (iii) if $v = c$ then it
writes a new value $v'$ to the global variable $x$.
\end{enumerate}
Each instruction of the above two types blocks (i.e., it cannot get
executed) until the buffer of its
thread is empty (resp. all buffers of its thread are empty in $\PSO$).
Finally, the instruction specifies the
nature of its final write. This write is either enqueued into its
respective buffer (to be dequeued into shared memory at a later point),
or it gets immediately flushed into the shared memory.

\Paragraph{Atomic instructions modeling.}
In our approach we handle atomic RMW and CAS instructions
without introducing them as new event types. Instead, we model these
instructions as sequences of already considered events, i.e.,
reads, buffer-writes, memory-writes, and fences.
We annotate some events of an atomic instruction to constitute an
\emph{atomic block}, which intuitively
indicates that the event sequence of the atomic block cannot be
interleaved with other events, thus respecting the semantics of the
instruction.
\begin{enumerate}[noitemsep,topsep=0pt,partopsep=0px]
\item A \emph{read-modify-write} instruction $\RMW$ on a variable
$x$ is modeled as a sequence of four events: (i) a fence event,
(ii) a read of $x$, (iii) a buffer-write of $x$, and
(iv) a memory-write of $x$. The read and buffer-write events (ii)+(iii)
are annotated as constituting an atomic block; in case the write of
$\RMW$ is specified to proceed immediately to the shared memory,
the memory-write event (iv) is also part of the atomic block.
\item For a \emph{compare-and-swap} instruction $\CAS$ we consider
separately the following two cases. A \emph{successful} $\CAS$ (i.e.,
the write proceeds) is modeled the same way as a read-modify-write.
A \emph{failed} $\CAS$ (i.e., the write does not proceed) is modeled
simply as a fence followed by a read, with no atomic block.
\end{enumerate}

\Paragraph{Executable atomic blocks.}
Here we describe the $\TSO$- and $\PSO$-executability conditions
for an atomic block. No further additions for executability are
required, since no new event types are introduced to handle
RMW and CAS instructions.

Consider an instance $(X,\Observation)$ of $\VTSOrm$,
and a trace $\Trace$ with $\Events{\Trace}\subseteq X$.
An atomic block containing a sequence of events $\Event_1,...,\Event_j$
is $\TSO$-executable in $\Trace$ if:
\begin{enumerate}[noitemsep,topsep=0pt,partopsep=0px]
\item for each $1\leq i \leq j$ we have that
$\Event_i \in X\setminus \Events{\Trace}$, and
\item for each $1\leq i \leq j$ we have that
$\Event_i$ is $\TSO$-executable in $\Trace \Concat \Event_1 ... \Event_{i-1}$.
\end{enumerate}
Intuitively, an atomic block is $\TSO$-executable if it can be executed
as a sequence at once (i.e., without other events interleaved),
and the $\TSO$-executable conditions of each event
(i.e., a read or a buffer-write or a memory-write or a fence)
within the block are respected.

The $\PSO$-executable conditions are analogous.
Given an instance $(X,\Observation)$ of $\VPSOrm$ and a trace
$\Trace$ with $\Events{\Trace}\subseteq X$,
an atomic block of events $\Event_1,...,\Event_j$
is $\PSO$-executable in $\Trace$ if:
\begin{enumerate}[noitemsep,topsep=0pt,partopsep=0px]
\item for each $1\leq i \leq j$ we have that
$\Event_i \in X\setminus \Events{\Trace}$, and
\item for each $1\leq i \leq j$ we have that
$\Event_i$ is $\PSO$-executable in $\Trace \Concat \Event_1 ... \Event_{i-1}$.
\end{enumerate}

\Paragraph{Execution verification.}
Given the above executable conditions, the execution verification
algorithms $\AlgoTSO$ and $\AlgoPSO$ only require minor technical
modifications to verify executions including RMW and CAS instructions.
% listed here to keep our presentation self-contained.

The core idea of the $\AlgoTSO$ resp. $\AlgoPSO$ modifications is to not
extend prefixes with single events that are part of some atomic block,
and instead extend the atomic blocks fully. This way, a lower set of
$(X, \TO)$ is considered only if for each atomic block, the block
is either fully present or fully not present in the lower set.

In $\AlgoTSO$ (\cref{algo:verifytso}), in
\cref{algo:tso_local_extend_loop} we further consider each
$\TSO$-executable atomic block $\Event_1,...,\Event_j$ not containing
any memory-write event, and then in \cref{algo:tso_local_extend}
we extend the prefix with the entire atomic block, i.e.,
$\Trace\gets \Trace \Concat \Event_1,...,\Event_j$.
Further, in \cref{algo:tso_mem_execute} we further consider each
$\TSO$-executable atomic block $\Event_1,...,\Event_j$ containing
a memory-write event, and in \cref{algo:tso_mem_extend}
we then extend the prefix with the whole atomic block, i.e.,
$\Trace\gets \Trace \Concat \Event_1,...,\Event_j$.

In $\AlgoPSO$ (\cref{algo:verifypso}), in the loop of
\cref{algo:pso_ifextend} we further consider each
$\PSO$-executable atomic block. Consider a fixed iteration of this
loop with an atomic block $\Event_1,...,\Event_j$.
The first event of the atomic block $\Event_1$ is a read,
thus the condition in \cref{algo:pso_is_read} is evaluated true with
$\Event_1$ and the control flow moves to
\cref{algo:pso_read_needs_its_mw}. Later, the condition in
\cref{algo:pso_is_fence} is evaluated false (since $\Event_1$ is a
read). Finally, in \cref{algo:pso_execute_event} the prefix is extended
with the whole atomic block, i.e.,
$\Trace_{\Event}\gets\Trace_{\Event} \Concat \Event_1,...,\Event_j$.

For $\AlgoTSO$ the argument of maintaining maximality in the set of
thread events applies also in the presence of RMW and CAS, and thus the
bound of \cref{them:vtso} is retained.
Similarly, for $\AlgoPSO$ the enumeration of fence maps and the
maximality in the spurious writes is preserved also with RMW and CAS,
and hence the bound of \cref{them:vpso} holds.

\Paragraph{Closure.}
When verifying executions with RMW and CAS instructions,
while the closure retains its guarantees as is, it can more effectively
detect unrealizable instances with additional rules.
Specifically, the closure $P$ of $(X,\Observation)$ satisfies
the rules \ref{item:closure1}--\ref{item:closure3} described in
\cref{subsec:heuristics}, and additionally,
given an event $\Event$ and an atomic block $\Event_1,...,\Event_j$,
$P$ satisfies the following.
\begin{enumerate}[noitemsep,topsep=0pt,partopsep=0px]
\setcounter{enumi}{3}
\item If $\Event_i <_P \Event$ for any $1\leq i \leq j$,
then $\Event_j <_P \Event$ (i.e., if some part of the block is
before $\Event$ then the entire block is before $\Event$).
\item If $\Event <_P \Event_i$ for any $1\leq i \leq j$,
then $\Event <_P \Event_1$ (i.e., if $\Event$ is before some part of
the block then $\Event$ is before the entire block).
\end{enumerate}

% Mention consistency? (X,RF) induces a value function for each
% read/write. (X,RF) is consistent if the induced value function
% agrees with the present events, i.e., the events in X are exactly
% those that would appear under the program's semantics when the
% reads/writes have the value determined by the value function.
% This includes also determining for each CAS, whether it is
% successful or failed.

%\input{equivalence}
%\input{closure}

%Having developed efficient algorithms for $\VTSOrm$ (\cref{subsec:verifyingtso})
%and $\VPSOrm$ (\cref{subsec:verifyingpso}),
%We are now enabled to

%In this section we
%develop a reads-from SMC algorithm for $\TSO$ and $\PSO$.
%The algorithm will use as subroutines $\AlgoTSO$ (resp. $\AlgoPSO$)
%to decide whether any given class of the RF partitioning
%is consistent under the $\TSO$ (resp. $\PSO$) semantics.

%In summary:%, in this section:
%\begin{enumerate}
%\item (\cref{subsec:maximal}) We present an exploration-optimal RF-SMC
%algorithm for $\TSO$ and $\PSO$.
%\item (\cref{subsec:efficient}) We present a space-optimal RF-SMC
%algorithm for $\TSO$ and $\PSO$.
%\item (\cref{subsec:heuristics}) We introduce %simple (but potent) heuristics methods
%heuristics %that significantly improve practical performance of
%for $\AlgoTSO$ and $\AlgoPSO$ when used in SMC.
%\end{enumerate}
%An extension of RF-SMC algorithms to incorporate locks is
%described in~\cref{subsec:app_locks}.

% \subsection{Exploration-Optimal SMC Algorithm}\label{subsec:maximal}

\section{READS-FROM SMC FOR TSO AND PSO}\label{sec:dpor}
% {Reads-From SMC for TSO and PSO}

In this section we present $\DCTSOPSOM$,
an %\emph{exploration-optimal}
exploration-optimal reads-from SMC algorithm
for $\TSO$ and $\PSO$. The algorithm $\DCTSOPSOM$ is
based on the reads-from algorithm for %Sequential Consistency
$\SC$~\cite{Abdulla19}, and adapted in this work to handle
the relaxed memory models $\TSO$ and $\PSO$.
The algorithm uses as subroutines $\AlgoTSO$ (resp. $\AlgoPSO$)
to decide whether any given class of the RF partitioning
is consistent under the $\TSO$ (resp. $\PSO$) semantics.

% WE NOW HAVE THIS AS A REMARK IN THIS SECTION
% An extension of $\DCTSOPSOM$ to incorporate locks is
% described in~\cref{subsec:app_locks}.

$\DCTSOPSOM$ is a recursive algorithm, each call of
$\DCTSOPSOM$ is argumented by a tuple $(\Seq, \Observation, \Trace, \NegativeMarked)$
where the following points hold:
\begin{itemize}[noitemsep,topsep=0pt,partopsep=0px]
\item $\Seq$ is a sequence of thread events. %(i.e., non-memory-write events).
Let $X$ denote the %proper
set of events of $\Seq$ together with their memory-write counterparts,
formally $X = \Events{\Seq} \cup \{ \WriteM : \exists (\WriteB,\WriteM) \in \SysWrites \textrm{ such that }\WriteB \in \WritesB{\Seq}\}$.
\item $\Observation\colon\Reads{X}\to \Writes{X}$ is a
desired reads-from function.
\item $\Trace$ is a concrete valid trace that is a witness of
$(X, \Observation)$, i.e., $\Events{\Trace} = X$ and $\RF{\Trace} = \Observation$.
\item $\NegativeMarked \subseteq \Reads{\Seq}$ is a set of reads
that are \emph{marked} to be committed to the source they read-from in $\Trace$.
\end{itemize}
Further, a globally accessible set of schedule sets called $\schedules$
is maintained throughout the recursion.
The $\schedules$ set is initialized empty ($\schedules = \emptyset$) and
the initial call of the algorithm is argumented with empty sequences and
sets --- $\DCTSOPSOM(\epsilon, \emptyset, \epsilon, \emptyset)$.

%\smallskip
\begin{algorithm}%[H]
\small
\SetInd{0.4em}{0.4em}
\DontPrintSemicolon
%\setstretch{1.05}
\caption{$\DCTSOPSOM(\Seq, \Observation, \Trace, \NegativeMarked)$}\label{algo:dctsopsom}
\KwIn{Sequence $\Seq$,
desired reads-from $\Observation$,
valid trace $\Trace$ such that $\RF{\Trace} = \Observation$,
marked reads $\NegativeMarked$.
}
\BlankLine
$\widetilde{\Trace} \gets \Trace \Concat \widehat{\Trace}$
where $\widehat{\Trace}$ is an arbitrary maximal extension of $\Trace$
\tcp*[f]{Maximally extend trace $\Trace$}\label{line:dctsopsom_extendtrace}\\
$\widetilde{\Seq} \gets \Seq \Concat \,\widehat{\Trace}\Project\LocalEvents{\widehat{\Trace}}$\tcp*[f]{Extend $\Seq$ with the thread-events subsequence of the extension $\widehat{\Trace}$}\label{line:dctsopsom_extendsequence}\\
\ForEach(\tcp*[f]{Reads of the extension $\widehat{\Trace}$})
{$\Read\in \Reads{\widehat{\Trace}}$}{\label{line:dctsopsom_forinitschedule}%\lForEach
  $\schedules(\pre{\widetilde{\Seq}}{\Read}) \gets \emptyset$\tcp*[f]{Initialize new schedule set}\label{line:dctsopsom_initschedule}\\
}\label{line:dctsopsom_forinitscheduleend}
\ForEach(\tcp*[f]{Unmarked reads})
{$\Read\in \Reads{\widetilde{\Seq}} \setminus \NegativeMarked$}{\label{line:dctsopsom_forreadstomutate}
  $P \gets \TO \Project \Events{\widetilde{\Trace}}$\tcp*[f]{Program order on all the events of $\widetilde{\Trace}$}\label{line:dctsopsom_to}\\
  \ForEach(\tcp*[f]{Different-thread-$\RF{\widetilde{\Trace}}$ reads except $\Read$})
  {$\Read'\in \Reads{\widetilde{\Seq}} \setminus \{\Read\}$ with $\Process(\Read') \neq \Process(\RF{\widetilde{\Trace}}(\Read'))$}{\label{line:dctsopsom_forrf}
    insert $\WriteM \to \Read'$ into $P$ where
    $\RF{\widetilde{\Trace}}(\Read') = \WriteM$%(\WriteB,\WriteM)$
    \tcp*[f]{Add the reads-from ordering into $P$}\label{line:dctsopsom_rf}\\
  }
  %$\mutations \gets \VisibleWrites_{P}(\Read) \setminus \{\RF{\widetilde{\Trace}}(\Read)\}$
  $\mutations \gets \{ (\WriteB, \WriteM) \in \Writes{\widetilde{\Trace}} \;|\; \Confl{\Read}{\WriteM}\} \setminus \{\RF{\widetilde{\Trace}}(\Read)\}$
  \tcp*[f]{All different writes $\Read$ may read-from}\label{line:dctsopsom_getmut}\\
  \If(\tcp*[f]{If $\Read$ is not part of the extension then})
  %{$\Read \in \Reads{\Seq}$}{\label{line:dctsopsom_ifext}
  {$\Read \not\in \Reads{\widehat{\Trace}}$}{\label{line:dctsopsom_ifext}
    $\mutations \gets \mutations \cap \Writes{\widehat{\Trace}}$
    \tcp*[f]{Only consider writes of the extension}\label{line:dctsopsom_extmut}\\
  }
  %$\Seq_{\Read} = \pre{\widetilde{\Seq}}{\Read}$\label{line:dctsopsom_pre}\\
  \ForEach(\tcp*[f]{Considered mutations})
  {$(\WriteB, \WriteM)\in \mutations$}{ \label{line:dctsopsom_mutwrite}
    $\mathsf{causesafter} \gets \{\Event \in \Events{\widetilde{\Seq}} \;|\; \Read <_{\widetilde{\Seq}} \Event \textrm{ and } \Event \leq_P \WriteB\}$
    \tcp*[f]{Causal past of $\WriteB$ after $\Read$ in $\widetilde{\Seq}$}\label{line:dctsopsom_causes}\\
    $\Seq' \gets \pre{\widetilde{\Seq}}{\Read} \Concat \,\widetilde{\Seq}\Project \mathsf{causesafter}$
    \tcp*[f]{$\Read$-prefix followed by $\mathsf{causesafter}$}\label{line:dctsopsom_seqp}\\
    $X' \gets \Events{\Seq'} \cup \{ \WriteM' : (\WriteB',\WriteM') \in \Writes{\widetilde{\Trace}} \textrm{ and }\WriteB' \in \WritesB{\Seq'}\}$
    \tcp*[f]{Event set for this mutation}\label{line:dctsopsom_eventsp}\\
    $\Observation' \gets \{ (\Read', \RF{\widetilde{\Trace}}(\Read')) :
    \Read' \in \Reads{\Seq'} \textrm{ and }\Read' \neq \Read \} \cup \{(\Read, (\WriteB, \WriteM))\}$
    \tcp*[f]{Reads-from for this mutation}\label{line:dctsopsom_rfp}\\
    \If(\tcp*[f]{If this is a new schedule})
    {$(\Seq', \Observation', \_, \_) \not\in \schedules(\pre{\widetilde{\Seq}}{\Read})$}{\label{line:dctsopsom_ifnotsch}
      $\Trace' \gets \linearize(X', \Observation')$
      \tcp*[f]{$\AlgoTSO$ (\cref{algo:verifytso}) or $\AlgoPSO$ (\cref{algo:verifypso})}\label{line:dctsopsom_lin}\\
      \If(\tcp*[f]{If the mutation is realizable})
      {$\Trace' \neq \bot$}{\label{line:dctsopsom_iffeasible}
        %{\bf \textcolor{\myred}{this is wrong:}}
        %$\NegativeMarked' \gets (\NegativeMarked \cap \Events{\Seq'}) \cup
        %\{\Read' \in \Reads{\widetilde{\Seq}} \;|\; \Read' \leq_P \Read \textrm{ or } \Read' <_P \WriteB\}$\\%\tcp*[f]{Causal past reads become committed}\label{line:dctsopsom_commit}\\
        %{\bf \textcolor{\myred}{this is correct:}}
        $\NegativeMarked' \gets (\NegativeMarked \cap \Reads{\Seq'}) \cup
        \Reads{\mathsf{causesafter}}$
        \tcp*[f]{Reads in $\mathsf{causesafter}$ get newly marked}\label{line:dctsopsom_commit}\\
        add $(\Seq', \Observation', \Trace', \NegativeMarked')$ to $\schedules(\pre{\widetilde{\Seq}}{\Read})$
        \tcp*[f]{Add the successful new schedule}\label{line:dctsopsom_add}\\
      }
    }
  }
}\label{line:dctsopsom_forreadstomutateend}
\ForEach(\tcp*[f]{Extension reads starting from the end})
{$\widehat{\Read} \in \Reads{\widehat{\Trace}}$ in the reverse order of $<_{\widehat{\Trace}}$}{\label{line:dctsopsom_forrecurread}
  \ForEach(\tcp*[f]{Collected schedules mutating $\widehat{\Read}$})
  {$(\Seq', \Observation', \Trace', \NegativeMarked') \in \schedules(\pre{\widetilde{\Seq}}{\widehat{\Read}})$}{\label{line:dctsopsom_forrecursch}
    $\DCTSOPSOM(\Seq', \Observation', \Trace', \NegativeMarked')$
    \tcp*[f]{Recursive call on the schedule}\label{line:dctsopsom_recur}\\
  }
  delete $\schedules(\pre{\widetilde{\Seq}}{\widehat{\Read}})$
  \tcp*[f]{This schedule set has been fully explored, hence it can be deleted}\label{line:dctsopsom_delsch}\\
}\label{line:dctsopsom_forrecurreadend}

\end{algorithm}

\cref{algo:dctsopsom} presents the pseudocode of $\DCTSOPSOM$.
In each call of $\DCTSOPSOM$, a number of possible changes (or \emph{mutations})
of the desired reads-from function $\Observation$ is proposed
in iterations of the loop in \cref{line:dctsopsom_forreadstomutate}.
Consider the read $\Read$ of a fixed iteration of the \cref{line:dctsopsom_forreadstomutate}
loop.
First, in Lines~\ref{line:dctsopsom_to}--\ref{line:dctsopsom_rf}
a partial order $P$ is constructed to capture the causal past of
write events. % (that are not themselves in causal future of $\Read$). Then,
In Lines~\ref{line:dctsopsom_getmut}--\ref{line:dctsopsom_extmut}
the set of mutations for $\Read$ is computed. Then in each iteration of
the \cref{line:dctsopsom_mutwrite} loop a mutation is constructed
(Lines~\ref{line:dctsopsom_causes}--\ref{line:dctsopsom_rfp}). Here the
partial order $P$ is utilized in \cref{line:dctsopsom_causes}
to help determine the event set of the mutation.
The constructed mutation, if deemed novel
(checked in \cref{line:dctsopsom_ifnotsch}),
is probed whether it is realizable (in \cref{line:dctsopsom_lin}).
In case it is realizable, it gets
added into $\schedules$ in \cref{line:dctsopsom_add}.
After all the mutations are proposed, then in
Lines~\ref{line:dctsopsom_forrecurread}--\ref{line:dctsopsom_delsch}
a number of recursive calls of $\DCTSOPSOM$ is performed,
and the recursive $\DCTSOPSOM$ calls are argumented
by the specific $\schedules$ retrieved.

\cref{fig:ExpOpt} illustrates the run of $\DCTSOPSOM$ on
a simple concurrent program (the run is identical under both
$\TSO$ and $\PSO$). % memory models).
An initial trace (A) is obtained
where $\textcolor{\myred}{\Read_1(y)}$ reads-from the initial event and
$\textcolor{\myblue}{\Read_2(x)}$ reads-from $\textcolor{\myblue}{\Write_1(x)}$.
Here two mutations are probed and both are realizable.
In the first mutation (B), $\textcolor{\myred}{\Read_1(y)}$
is mutated to read-from $\textcolor{\myred}{\Write_2(y)}$ and $\textcolor{\myblue}{\Read_2(x)}$
is not retained (since it appears after $\textcolor{\myred}{\Read_1(y)}$ and it is not
in the causal past of $\textcolor{\myred}{\Write_2(y)}$). In the second
mutation (C), $\textcolor{\myblue}{\Read_2(x)}$ is mutated to read-from the initial event
and $\textcolor{\myred}{\Read_1(y)}$ is retained
(since it appears before $\textcolor{\myblue}{\Read_2(x)}$)
with initial event as its reads-from.
After both mutations are added to $\schedules$, recursive calls are performed
in the reverse order of reads appearing in the trace, thus starting with
(C). Here no mutations are probed since there are no events in the extension,
the algorithm backtracks to (A) and a recursive call to (B)
is performed. Here one mutation (D) is added, where
$\textcolor{\myblue}{\Read_2(x)}$ is mutated to read-from the initial event
and $\textcolor{\myred}{\Read_1(y)}$
is retained (it appears before $\textcolor{\myblue}{\Read_2(x)}$)
with $\textcolor{\myred}{\Write_2(y)}$ as its reads-from.
%Subsequently,
The call to (D) is performed and here no mutations
are probed (there are no events in the extension).
The algorithm backtracks %to (B) and (A) and then
and concludes,
exploring four RF partitioning classes in total.

\begin{figure}%[!h]
% \vspace{-3mm}
  \begin{minipage}{0.7\textwidth}
  %  \centering
    \begin{tikzpicture}[thick,
      pre/.style={<-,shorten >= 2pt, shorten <=2pt, very thick},
      post/.style={->,shorten >= 2pt, shorten <=2pt,  very thick},
      seqtrace/.style={->, line width=2},
      aux_seqtrace/.style={->, line width=1, draw=gray},
      und/.style={very thick, draw=gray},
      event/.style={rectangle, minimum height=3.5mm, draw=black, fill=white, minimum width=8mm,   line width=1pt, inner sep=2, font={\footnotesize}},
      aux_event/.style={event, draw=gray},
      virt/.style={circle,draw=black!50,fill=black!20, opacity=0, text opacity=1},
      nd_bg/.style={fill=gray!80, fill opacity=0.5, draw=none},
      read/.style={draw=black, fill=gray!20, font={\footnotesize}},
      default/.style={font={\footnotesize}}
      ]

      \pgfdeclarelayer{bg}
      \pgfsetlayers{bg,main}

      \newcommand{\sep}{||}

      \newcommand{\xs}{0.9}
      \newcommand{\ys}{0.8}%8

      \newcommand{\yA}{0.*\ys}
      \newcommand{\yB}{-1.7*\ys}
      \newcommand{\yC}{-3.0*\ys}
      \newcommand{\yD}{-4.4*\ys}

      \newcommand{\rootAx}{0*\xs}
      \newcommand{\rootBx}{2.5*\xs}
      \newcommand{\rootCx}{0*\xs}
      \newcommand{\rootDx}{2.5*\xs}

      \newcommand{\rAx}{\rootAx + 4*\xs}
      \newcommand{\rBx}{\rootAx + 1*\xs}

      \node[virt] (rootX) at (\rootAx + 7*\xs, \yA) {(A)};
      \node[default] (root0) at (\rootAx + 0*\xs, \yA) {$\mathrm{init}$ \sep};
      \node[default] (root1) at (\rootAx + 1*\xs, \yA) {$\WriteB_1$};
      \node[default] (root2) at (\rootAx + 2*\xs, \yA) {$\WriteM_1$};
      \node[read]
              (root3) at (\rootAx + 3*\xs, \yA) {$\Read_1$};
      \node[default] (root4) at (\rootAx + 4*\xs, \yA) {$\WriteB_2$};
      \node[default] (root5) at (\rootAx + 5*\xs, \yA) {$\WriteM_2$};
      \node[read]
              (root6) at (\rootAx + 6*\xs, \yA) {$\Read_2$};

      \node[default] (r2_init0) at (\rootBx + 0*\xs, \yB) {$\mathrm{init}$};
      \node[default] (r2_init1) at (\rootBx + 1*\xs, \yB) {$\WriteB_2$};
      \node[read]
              (r2_init2) at (\rootBx + 2*\xs, \yB) {$\Read_2$};
      \node[default] (r2_init3) at (\rootBx + 3*\xs, \yB) {$\WriteB_1$};
      \node[read]
              (r2_init4) at (\rootBx + 4*\xs, \yB) {$\Read_1$};
      \node[default] (r2_init5) at (\rootBx + 5*\xs, \yB) {$\WriteM_1$};
      \node[default] (r2_init6) at (\rootBx + 6*\xs, \yB) {$\WriteM_2$ \sep};
      \node[virt] (r2_init_X) at (\rootBx + 7*\xs, \yB) {(C)};

      \node[default] (r1_w2_0) at (\rootCx + 0*\xs, \yC) {$\mathrm{init}$};
      \node[default] (r1_w2_1) at (\rootCx + 1*\xs, \yC) {$\WriteB_2$};
      \node[default] (r1_w2_2) at (\rootCx + 2*\xs, \yC) {$\WriteM_2$};
      \node[default] (r1_w2_3) at (\rootCx + 3*\xs, \yC) {$\WriteB_1$};
      \node[read]
              (r1_w2_4) at (\rootCx + 4*\xs, \yC) {$\Read_1$};
      \node[default] (r1_w2_5) at (\rootCx + 5*\xs, \yC) {$\WriteM_1$ \sep};
      \node[read]
              (r1_w2_6) at (\rootCx + 6*\xs, \yC) {$\Read_2$};
      \node[virt] (r1_w2_X) at (\rootCx + 7*\xs, \yC) {(B)};

      \node[default] (r1_w2_r2_init0) at (\rootDx + 0*\xs, \yD) {$\mathrm{init}$};
      \node[default] (r1_w2_r2_init1) at (\rootDx + 1*\xs, \yD) {$\WriteB_2$};
      \node[read]
              (r1_w2_r2_init2) at (\rootDx + 2*\xs, \yD) {$\Read_2$};
      \node[default] (r1_w2_r2_init3) at (\rootDx + 3*\xs, \yD) {$\WriteM_2$};
      \node[default] (r1_w2_r2_init4) at (\rootDx + 4*\xs, \yD) {$\WriteB_1$};
      \node[read]
              (r1_w2_r2_init5) at (\rootDx + 5*\xs, \yD) {$\Read_1$};
      \node[default] (r1_w2_r2_init6) at (\rootDx + 6*\xs, \yD) {$\WriteM_1$ \sep};
      \node[virt] (r1_w2_r2_init_X) at (\rootDx + 7*\xs, \yD) {(D)};

      \begin{pgfonlayer}{bg}
        \draw[post] (root3) -- (r1_w2_0) node [midway, left] {$\Write_2$};
        \draw[post] (root6) -- (r2_init0) node [midway, right, xshift=9pt, yshift=3pt] {$\mathrm{init}$};
        \draw[post] (r1_w2_6) -- (r1_w2_r2_init0) node [midway, right, xshift=9pt, yshift=3pt] {$\mathrm{init}$};

        \draw[post, \darkred, dashed] (root0) to[out=45, in=145] (root3);
        \draw[post, \darkred, dashed] (root2) to[out=35, in=155] (root6);
        \draw[post, \darkred, dashed] (r2_init0) to[out=50, in=140] (r2_init2);
        \draw[post, \darkred, dashed] (r2_init0) to[out=50, in=150] (r2_init4);
        \draw[post, \darkred, dashed] (r1_w2_2) to[out=60, in=120] (r1_w2_4);
        \draw[post, \darkred, dashed] (r1_w2_5) to[out=60, in=120] (r1_w2_6);
        \draw[post, \darkred, dashed] (r1_w2_r2_init0) to[out=60, in=120] (r1_w2_r2_init2);
        \draw[post, \darkred, dashed] (r1_w2_r2_init3) to[out=60, in=120] (r1_w2_r2_init5);

        \draw[nd_bg] ($(root0) + (-0.5*\xs, 0.3*\ys)$) rectangle ($(root6) + (0.5*\xs, -0.3*\ys)$);
        \draw[nd_bg] ($(r2_init0) + (-0.5*\xs, 0.3*\ys)$) rectangle ($(r2_init6) + (0.5*\xs, -0.3*\ys)$);
        \draw[nd_bg] ($(r1_w2_0) + (-0.5*\xs, 0.3*\ys)$) rectangle ($(r1_w2_6) + (0.5*\xs, -0.3*\ys)$);
        \draw[nd_bg] ($(r1_w2_r2_init0) + (-0.5*\xs, 0.3*\ys)$) rectangle ($(r1_w2_r2_init6) + (0.5*\xs, -0.3*\ys)$);
      \end{pgfonlayer}

    \end{tikzpicture}
  \end{minipage}
  \begin{minipage}{0.11\textwidth}
  \small
  \vspace{-2mm}
    \begin{align*}
      \text{Th}&\text{read}~\Process_{1}\\
      \hline\\[-1em]
      1.~& \textcolor{\myblue}{\Write_1(x)}\\
      \\[-1.9em]
      2.~& \textcolor{\myred}{\Read_1(y)}
    \end{align*}
  \end{minipage}
  \quad
  \begin{minipage}{0.11\textwidth}
  \small
  \vspace{-2mm}
    \begin{align*}
      \text{Th}&\text{read}~\Process_{2}\\
      \hline\\[-1em]
      1.~& \textcolor{\myred}{\Write_2(y)}\\
      \\[-1.9em]
      2.~& \textcolor{\myblue}{\Read_2(x)}
    \end{align*}
  \end{minipage}
  %\end{subfigure}
  %%%%%%%%%%%%%%%
  \caption{
    %Visualization of
    $\DCTSOPSOM$ (\cref{algo:dctsopsom}). %the exploration-optimal SMC algorithm $\DCTSOPSOM$. %under $\TSO$.
    The gray boxes
    represent individual calls to $\DCTSOPSOM$. The sequence of events % individual calls   the calls
    inside a gray box is the trace $\widetilde{\Trace}$; the part left of the $||$-separator
    is $\Trace$ (before extending), and to the right is $\widehat{\Trace}$
    (the extension). The red dashed arrows represent the reads-from function $\RF{\widetilde{\Trace}}$.
    Each black solid arrow represents a recursive call, where the arrow's
    outgoing tail and label describes the corresponding mutation.
  }
  \label{fig:ExpOpt}
\end{figure}

%Utilization of $\Seq$ to store and retrieve $\schedules$, together
%with marking of events in $\NegativeMarked$, ensure that
$\DCTSOPSOM$ is sound, complete and exploration-optimal, and
we formally state this in~\cref{them:dctsopso_maximal}.
%The work of~\citet{Abdulla19} introduced this algorithm for the Sequentially
%Consistent ($\SC$) memory model and proved its completeness and
%exploration-optimality. The completeness and optimality
%arguments naturally transfer to the settings of $\TSO$ and $\PSO$ used in $\DCTSOPSOM$.
%\input{figures/OptimalSMC_expmemory.tex}
%$\DCTSOPSOM$, as also its $\SC$ counterpart of~\citet{Abdulla19}, uses
%space exponential wrt the length of program's
%longest trace in the worst case
%. \cref{fig:expmemory} illustrates this worst-case behaviour.
%(detailed on in~\cref{subsec:app_expspace}).
%Although this exponential-space case happens for adversarially
%created programs, $\DCTSOPSOM$ does not seem to have memory issues
%in practice, as none were reported by~\citet{Abdulla19} and we
%have also not encountered any during our experiments.

\Paragraph{Extension from $\SC$ to $\TSO$ and $\PSO$.}
The fundamental challenge in extending the $\SC$ algorithm
of~\citet{Abdulla19} to $\TSO$ and $\PSO$ is verifying execution consistency
for $\TSO$ and $\PSO$, which we address in \cref{sec:verifyingtsopso}
(\cref{line:dctsopsom_lin}~of~\cref{algo:dctsopsom} calls our algorithms
$\AlgoTSO$ and $\AlgoPSO$).
The main remaining challenge is then to ensure
that the exploration optimality is preserved. To that end,
we have to exclude certain events (in particular, memory-write events)
from subsequences and event subsets that guide the exploration
of \cref{algo:dctsopsom}. Specifically, the sequences
$\Seq$, $\Seq'$, and $\widetilde{\Seq}$ invariantly contain only
the thread events, which is ensured in
\cref{line:dctsopsom_extendsequence},
\cref{line:dctsopsom_causes} and \cref{line:dctsopsom_seqp},
and then in \cref{line:dctsopsom_eventsp} the absent memory-writes are
reintroduced. No such distinction is required under $\SC$.

\begin{remark}[Handling locks and atomic primitives]\label{rmk:lockscasrmw}
{\normalfont
For clarity of presentation, so far we have neglected locks in our
model. However, lock events can be naturally handled by our approach
as follows. We consider each lock-release event $\Release$
as an atomic write event (i.e., its effects are not deferred by a buffer
but instead are instantly visible to each thread). Then, each
lock-acquire event $\Aquire$ is considered as a read event that accesses
the unique memory location.

In SMC, we enumerate the reads-from functions that also consider locks,
thus having constraints of the form $\Observation(\Aquire) = \Release$.
% During closure (\cref{subsec:heuristics}),
% Then for each constraint $\Observation(\Aquire) = \Release$,
% we require any witness $\Trace$ to satisfy $\Release <_{\Trace} \Aquire$.
% consider the following condition:
% $\Proc{\Aquire} \neq \Proc{\Release}$ implies $\Release <_P \Aquire$.
% Thus $P$ totally orders the critical sections of each lock,
This treatment totally orders the critical sections of each lock,
which naturally solves all reads-from constraints of locks,
% , as respecting $P$ trivially solves all reads-from constraints of locks
and further ensures that no thread acquires
an already acquired (and so-far unreleased) lock.
Therefore $\AlgoTSO$/$\AlgoPSO$ need not take additional care
for locks.
%No modifications to the $\DCTSOPSOP$ algorithm are needed to
%incorporate locks. The changes required in $\DCTSOPSOM$ to handle locks
%are described by~\citet{Abdulla19}.
The approach to handle locks by~\citet{Abdulla19} directly carries over % for the $\SC$ algorithm,
to our exploration algorithm $\DCTSOPSOM$.
% are applied.

%%%% RMW CAS below

The atomic operations read-modify-write (RMW) and compare-and-swap
(CAS) are modeled as in \cref{subsec:verifying_casrmw}, except for the
fact that the atomic blocks are not necessary for SMC.
Then $\DCTSOPSOM$ can handle programs with such
operations as described by~\citet{Abdulla19}.
In particular, the modification of $\DCTSOPSOM$ (\cref{algo:dctsopsom})
to handle RMW and CAS operations is as follows.

% We describe the modification of $\DCTSOPSOM$ (\cref{algo:dctsopsom})
% to handle RMW and CAS operations, as also employed by~\citet{Abdulla19},
% is as follows. in detail.
Consider an iteration of the loop in
\cref{line:dctsopsom_forreadstomutate} where $\Read$ is the read-part
of either a RMW or a successful CAS,
denoted $\Event$, and let
$(\WriteB'', \WriteM'') = \RF{\widetilde{\Trace}}(\Read)$.
Then, in \cref{line:dctsopsom_getmut}
we additionally consider as an extra mutation
each atomic instruction $\Event'$ satisfying:
% the following:
%
\begin{enumerate}[noitemsep,topsep=0pt,partopsep=0px]
\item The read-part $\Read'$ of $\Event'$
%% ?? is unmarked and it ??
reads-from the write-part $(\WriteB, \WriteM)$ of $\Event$ % in $\widetilde{\Trace}$
(i.e.,
$\RF{\widetilde{\Trace}}(\Read')  = (\WriteB, \WriteM)$), and
\item $\Event'$ is either a RMW, or it will be a successful
CAS when it reads-from $(\WriteB'', \WriteM'')$. In this case,
let $(\WriteB', \WriteM')$ denote the write-part of $\Event'$.
\end{enumerate}
% Note that at most once instruction satisfying the above properties can exist.
% ABOVE SENTENCE IS NOT TRUE! You can have one rmw in one thread, and
% then many failed cas in other threads, all of those cas could be
% successful if reading before the rmw, hence they would all be satisfying
% the above two things (as they also should; they should all be considered)
When considering the above mutation in \cref{line:dctsopsom_mutwrite},
we set
$\Observation'(\Read') = (\WriteB'', \WriteM'')$
and $\Observation'(\Read) = (\WriteB', \WriteM')$
in \cref{line:dctsopsom_rfp}, which intuitively
aims to ``reverse'' $\Event$ and $\Event'$ in the trace.
}
\end{remark}
% Keep in mind relaxed variants of these primitives. Eg CAS atomically reads
% but swaps only into the buffer / and swap directly into main memory.
% I know that in the mpi benchmarks there are such C11-variants of CAS,
% where you specify SC/Acquire/Release behaviour for the compare-part
% and swap-part. I think writing into the main memory directly has to
% be preceded by flushing the buffer in any case, even if that direct
% write is part of a CAS/RMW. At least that's what the Nidhugg
% implementation seems to be doing.

%\input{dpor_efficient}

%\input{dpor_heuristics}

\section{EXPERIMENTS}\label{sec:exp}
% {Experiments}

In this section we report on an experimental evaluation
of the consistency verification algorithms
$\AlgoTSO$ and $\AlgoPSO$, as well as the reads-from SMC
algorithm $\DCTSOPSOM$.
We have implemented our algorithms as an extension in Nidhugg~\cite{Abdulla2015},
%which is
a state-of-the-art stateless model checker for multithreaded
C/C++ programs with {\sf pthreads} library,
operating on LLVM IR.
%Nidhugg operates on input program's LLVM intermediate representation.
%\footnote{\url{https://github.com/ViToSVK/nidhugg/tree/tsopso_DC} for $\DCTSOPSOM$, \url{https://github.com/ViToSVK/nidhugg/tree/tsopso_DCpartial} for $\DCTSOPSOP$.}

\Paragraph{Benchmarks.}
For our experimental evaluation of both the consistency verification
and SMC, we consider 109 benchmarks coming
from four different categories, namely:
(i)~SV-COMP benchmarks,
(ii)~benchmarks from related papers and works~\cite{Abdulla19,Abdulla2015,Huang16,Chatterjee19},
(iii)~mutual-exclusion algorithms, and
(iv)~dynamic-programming benchmarks of~\citet{Chatterjee19}.
%, described in~\cref{subsec:app_expfull}.
Although the consistency and SMC algorithms can be extended to support
atomic compare-and-swap and read-modify-write primitives
(cf. \cref{rmk:lockscasrmw}),
our current implementation does not support these primitives.
Therefore, we used all benchmarks without such
primitives that we could obtain (e.g., we include every benchmark
of the relevant $\SC$ reads-from work~\cite{Abdulla19} except the one
benchmark with compare-and-swap).
Each benchmark comes with a scaling parameter, called the
\emph{unroll} bound, which controls the bound
on the number of iterations in all loops of the benchmark
(and in some cases it further controls the number of threads).

% Decide whether to describe dynamic benchmarks once again in appendix and refer to it.
%, namely
%Dekker~\cite{Knuth66},
%Lamport~\cite{Lamport87},
%Burns~\cite{Burns80},
%Peterson~\cite{Peterson81},
%%Dijkstra~\cite{Dijkstra83},
%Peterson-Fischer~\cite{Peterson77},
%Kessels~\cite{Kessels82},
%Szymanski~\cite{Szymanski88},
%Tsay~\cite{Tsay98},
%as well as novel solutions of Correia-Ramalhete~\cite{Correia16}.

\subsection{Experiments on Execution Verification for TSO and PSO}\label{subsec:exp_verify}
% Reads-From

In this section we perform an experimental evaluation of our execution verification
algorithms $\AlgoTSO$ and $\AlgoPSO$. For the purpose of comparison,
we have also implemented within Nidhugg the naive lower-set enumeration algorithm
of~\citet{Abdulla19,BiswasE19}, extended to $\TSO$ and $\PSO$.
Intuitively, this approach enumerates all lower sets of the program
order restricted to the input event set, which yields a better
complexity bound than enumerating write-coherence orders
(even with just one location).
The extensions to $\TSO$ and $\PSO$ are called
$\NaiveAlgoTSO$ and $\NaiveAlgoPSO$, respectively, and their worst-case
complexity is $n^{2\cdot k}$ and $n^{k\cdot (\NumVariables+1)}$,
respectively (as discussed in~\cref{sec:summary}).
Further, for each of the above verification algorithms, we consider
two variants, namely,
with and without the closure heuristic of~\cref{subsec:heuristics}.

\Paragraph{Setup.}
We evaluate the verification algorithms on execution consistency
instances induced during SMC of the benchmarks. For $\TSO$ we have
collected 9400 instances, 1600 of which are not realizable.
For $\PSO$ we have collected
% 7850 realizable, 880 unrealizable small, 520 unrealizable big
9250 instances, 1400 of which are not realizable.
The collection process is described in detail in~\cref{subsec:app_verifysetup}.
For each instance, we run the verification algorithms subject to
a timeout of one minute, and we report the average time achieved
over 5 runs.

Below we present the results using logarithmically scaled % log-scaled
plots, where the opaque and semi-transparent red lines represent
identity and an order-of-magnitude difference, respectively.

\Paragraph{Results -- algorithms with closure.}
Here we evaluate the verification algorithms that execute the
closure as the preceding step. The plots in
\cref{fig:withcl_all} present the results for
% $\TSO$ (left) and $\PSO$ (right).
$\TSO$ and $\PSO$.

\begin{figure}[h]
\centering
%%%%%%%%%%%%%%%
\begin{subfigure}[b]{0.44\textwidth}
\centering
\includegraphics[height=5.4cm]{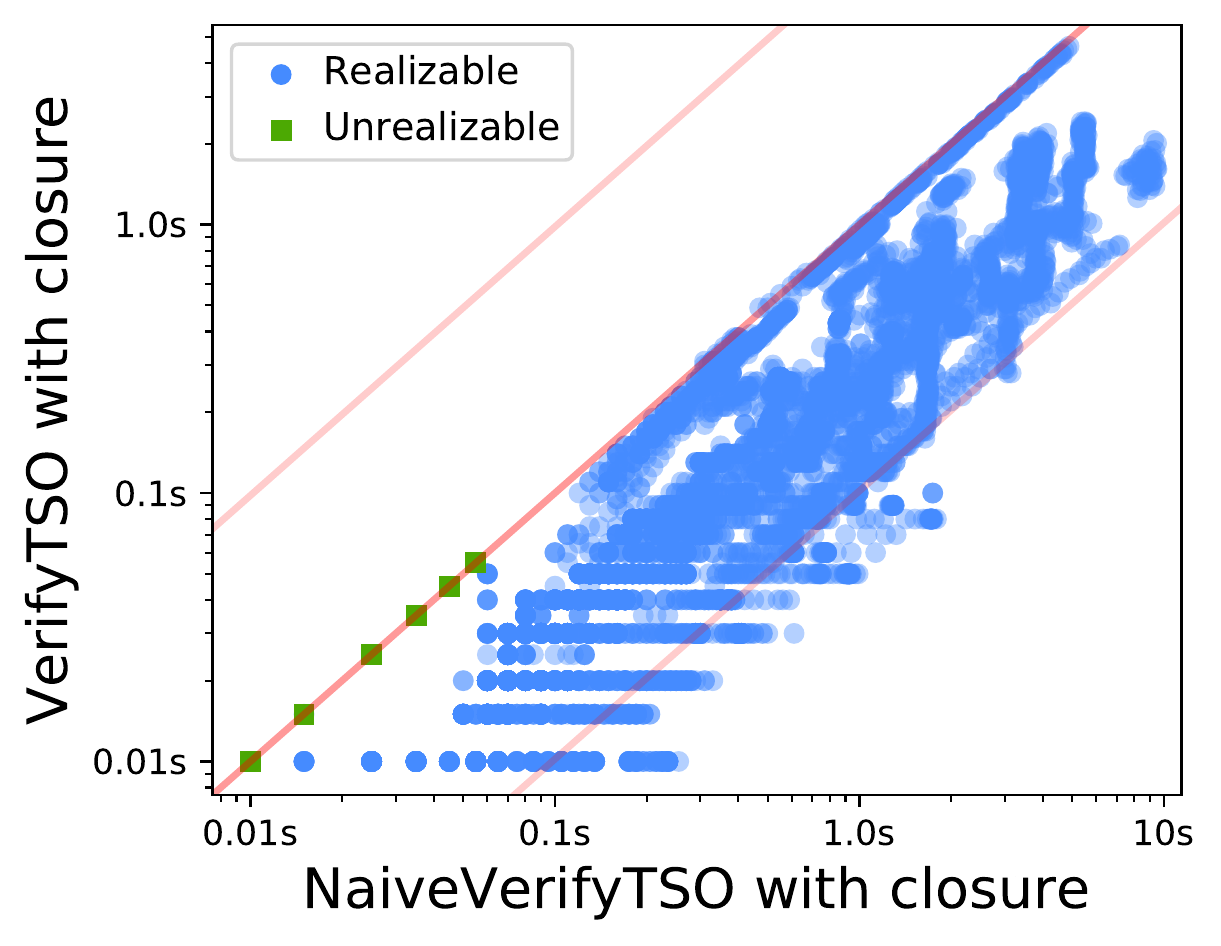}
%\caption{}
\label{subfig:tso_withcl_all}
\end{subfigure}
%%%%%%%%%%%%%%%
\qquad%\qquad
%%%%%%%%%%%%%%%
\begin{subfigure}[b]{0.44\textwidth}
\centering
\includegraphics[height=5.4cm]{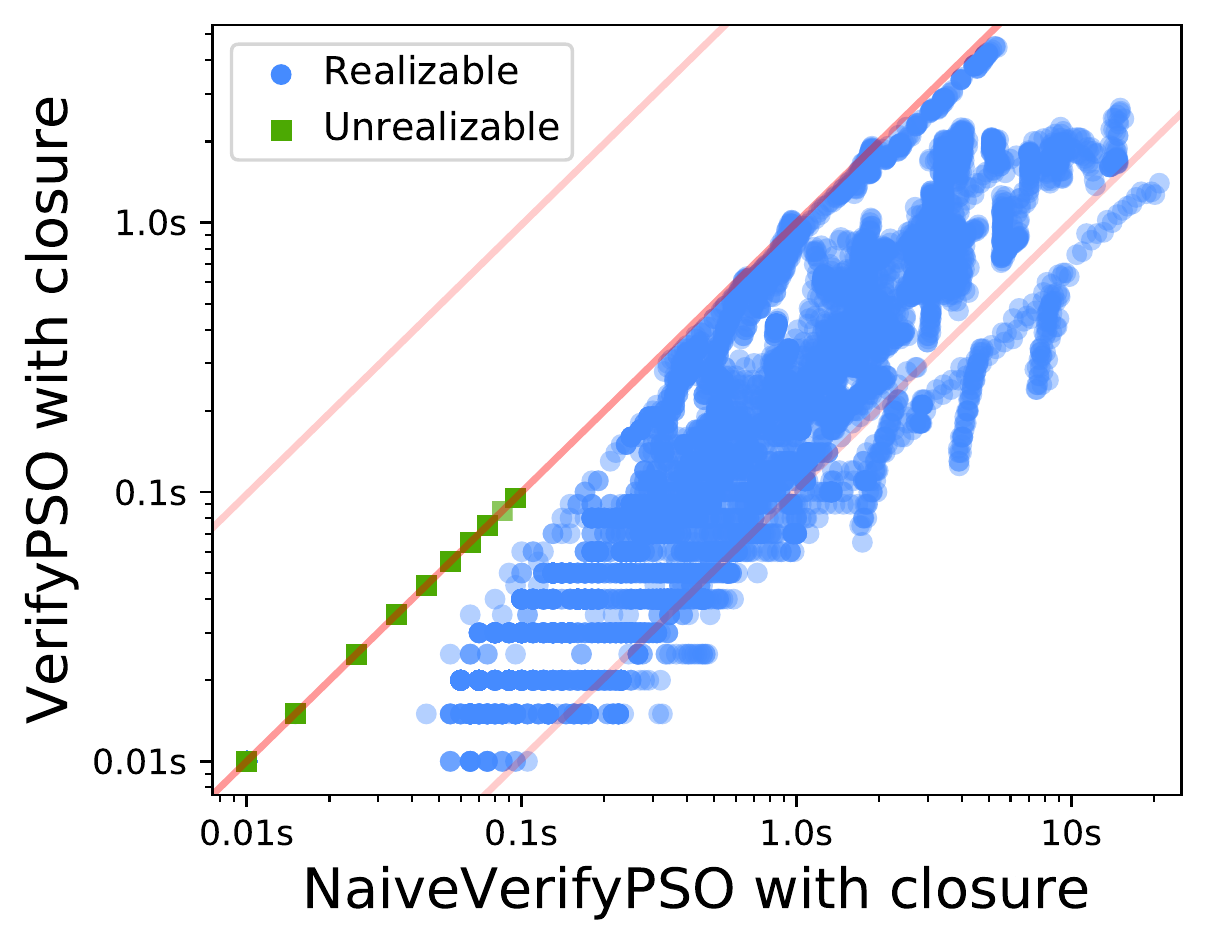}
%\caption{}
\label{subfig:pso_withcl_all}
\end{subfigure}
\qquad%\qquad
%%%%%%%%%%%%%%%
\vspace{-5mm}
\caption{
Consistency verification comparison on $\TSO$ (left) and $\PSO$ (right) when using closure.
}
\label{fig:withcl_all}
\end{figure}

In $\TSO$, our algorithm $\AlgoTSO$ is similar to or faster
than $\NaiveAlgoTSO$ on the realizable instances (blue dots),
and the improvement is mostly
within an order of magnitude. All unrealizable instances (green dots)
were detected as such by closure, and hence the closure-using
$\AlgoTSO$ and $\NaiveAlgoTSO$ coincide on these instances.

We make similar observations in $\PSO$, where $\AlgoPSO$ is similar or
superior to $\NaiveAlgoPSO$ for the realizable instances, and the
algorithms are indentical on the unrealizable instances, since these are
all detected as unrealizable by closure.

\begin{figure}[h]
\centering
%%%%%%%%%%%%%%%
\begin{subfigure}[b]{0.44\textwidth}
\centering
\includegraphics[height=5.4cm]{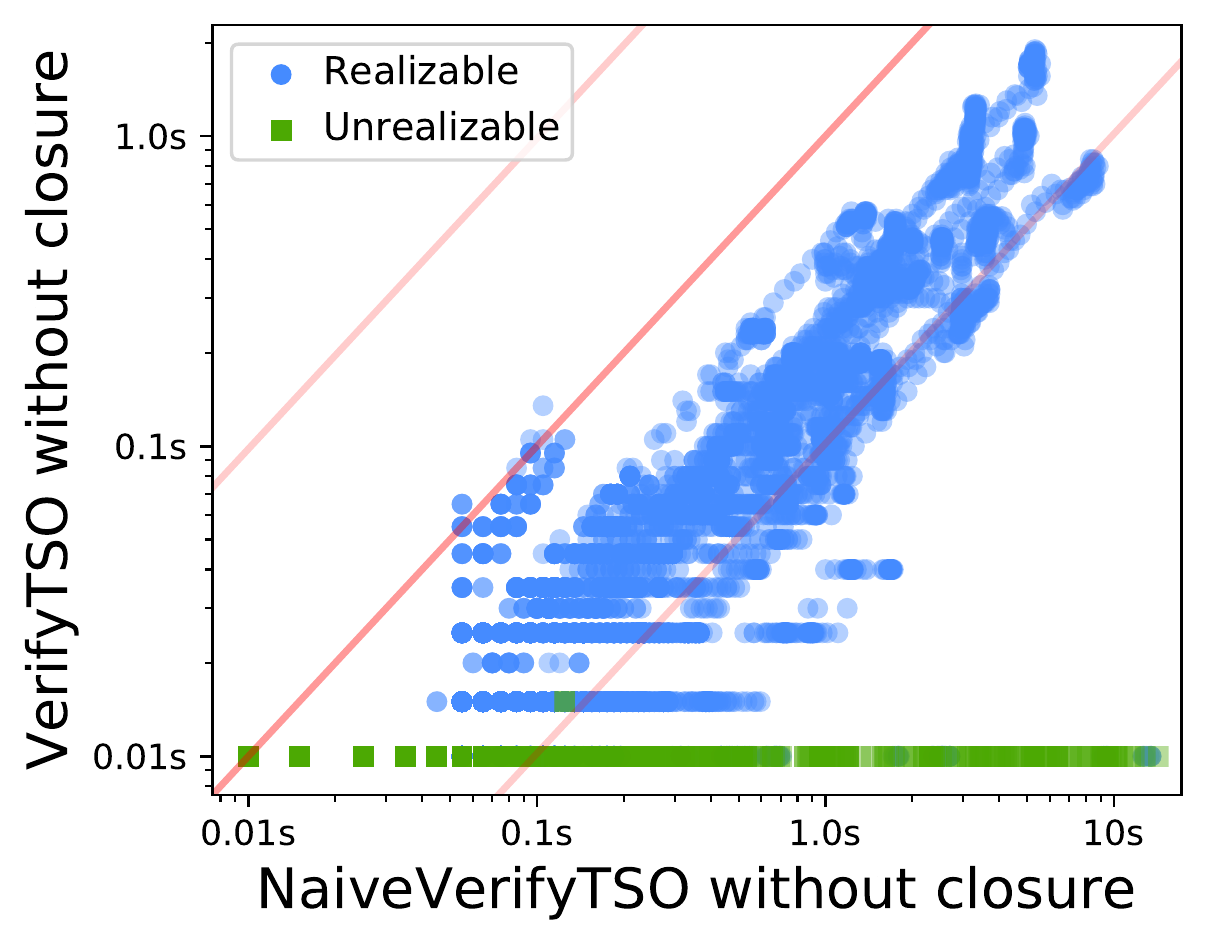}
%\caption{}
\label{subfig:tso_nocl_all}
\end{subfigure}
%%%%%%%%%%%%%%%
\qquad%\qquad
%%%%%%%%%%%%%%%
\begin{subfigure}[b]{0.44\textwidth}
\centering
\includegraphics[height=5.4cm]{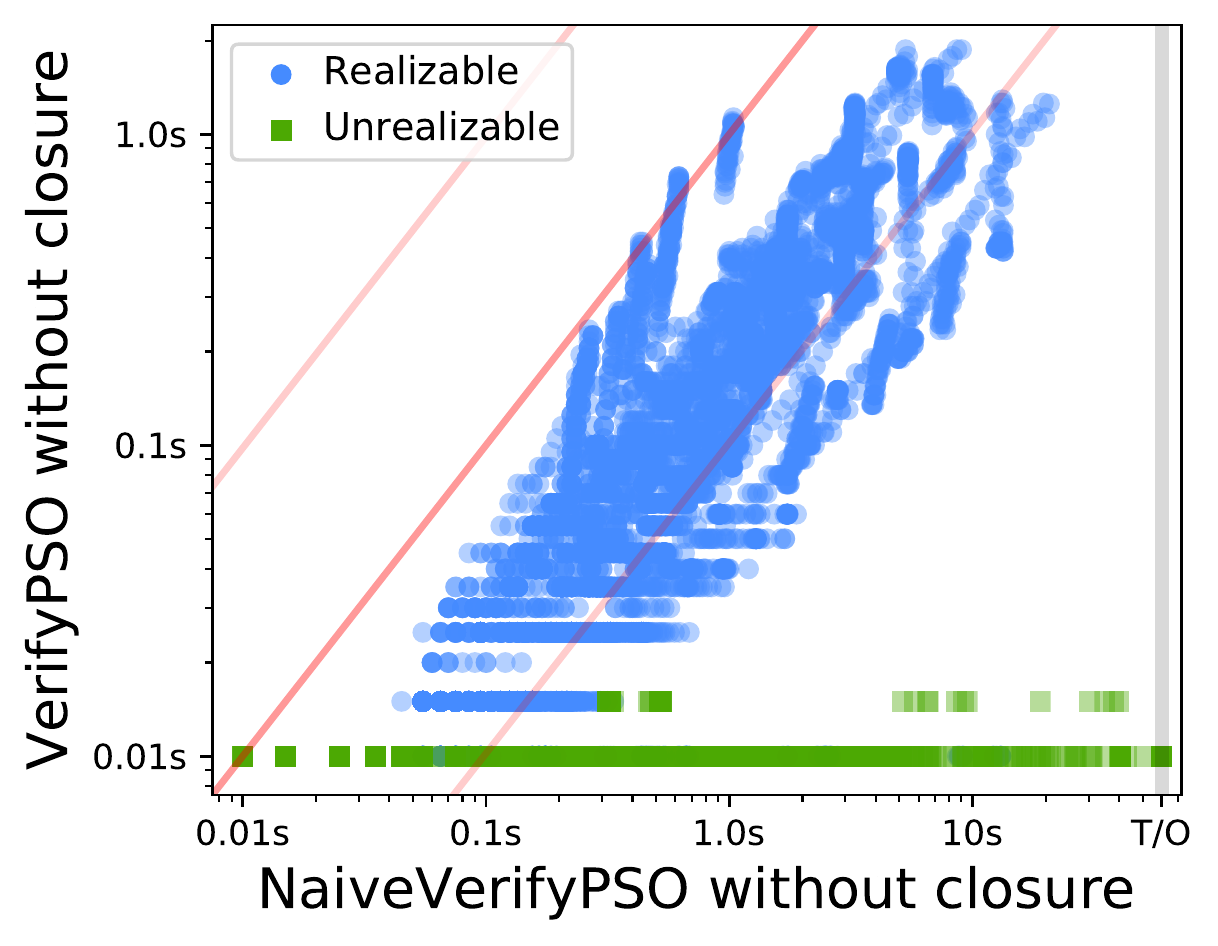}
%\caption{}
\label{subfig:pso_nocl_all}
\end{subfigure}
\qquad%\qquad
%%%%%%%%%%%%%%%
\vspace{-5mm}
\caption{
Consistency verification comparison on $\TSO$ (left) and $\PSO$ (right) without the closure.
}
\label{fig:nocl_all}
\end{figure}

\Paragraph{Results -- algorithms without closure.}
Here we evaluate the verification algorithms without the closure.
The plots in \cref{fig:nocl_all} present the results for
% $\TSO$ (left) and $\PSO$ (right).
$\TSO$ and $\PSO$.

In $\TSO$, the algorithm $\AlgoTSO$ outperforms $\NaiveAlgoTSO$
on most of the realizable instances (blue dots). Further,
$\AlgoTSO$ significantly outperforms $\NaiveAlgoTSO$
on the unrealizable instances (green dots).
This is because without closure, a verification algorithm can declare
an instance unrealizable only after an exhaustive exploration of its
respective lower-set space. $\AlgoTSO$ explores a significantly smaller
space compared to $\NaiveAlgoTSO$, as outlined in~\cref{sec:summary}.

Similar observations as above hold in $\PSO$ for
the algorithms $\AlgoPSO$ and $\NaiveAlgoPSO$ without closure,
both for the realizable and the unrealizable instances.

\Paragraph{Results -- effect of closure.}
Here we comment on the effect of closure for the verification algorithms,
in~\cref{subsec:app_verifyfurther} we present the detailed analysis.
Recall that closure constructs a partial order that each witness has to
satisfy, and declares an instance unrealizable when it detects that the
partial order cannot be constructed for this instance
(we refer to~\cref{subsec:heuristics} for details).
% We refer to \cref{subsec:heuristics} for the description of the closure.

For each verification algorithm, its version without closure is
faster on most instances that are realizable (i.e., a witness exists).
This means that the overhead of computing the closure % fully computing
typically outweighs the consecutive benefit of the verification being guided by the partial
order.

On the other hand, for each verification algorithm, its version with
closure is significantly faster on the unrealizable instances (i.e.,
no witness exists). This is because a verification algorithm has to
enumerate all its lower sets before declaring an instance unrealizable,
and this is much slower than the polynomial closure computation.

\Paragraph{Results -- verification with atomic operations.}
Here we present additional experiments to evaluate $\TSO$ verification
algorithms $\AlgoTSO$ and $\NaiveAlgoTSO$ on executions containing
atomic operations read-modify-write (RMW) and compare-and-swap (CAS).
To that end, we consider 1088 verification instances
(779 realizable and 309 not realizable) that arise during stateless
model checking of benchmarks containing RMW and CAS, namely:
\begin{itemize}[noitemsep,topsep=0pt,partopsep=0px]
\item synthetic benchmarks
$\sf{casrot}$~\cite{Abdulla19} and
$\sf{cinc}$~\cite{Kokologiannakis19},
\item data structure benchmarks
$\sf{barrier}$, $\sf{chase}$-$\sf{lev}$,
$\sf{ms}$-$\sf{queue}$ and $\sf{linuxrwlocks}$
\cite{NorrisD13,Kokologiannakis19}, and
\item Linux kernel benchmarks
$\sf{mcs\_spinlock}$ and $\sf{qspinlock}$
\cite{Kokologiannakis19}.
\end{itemize}

\begin{figure}[h]
\centering
%%%%%%%%%%%%%%%
\begin{subfigure}[b]{0.44\textwidth}
\centering
\includegraphics[height=5.4cm]{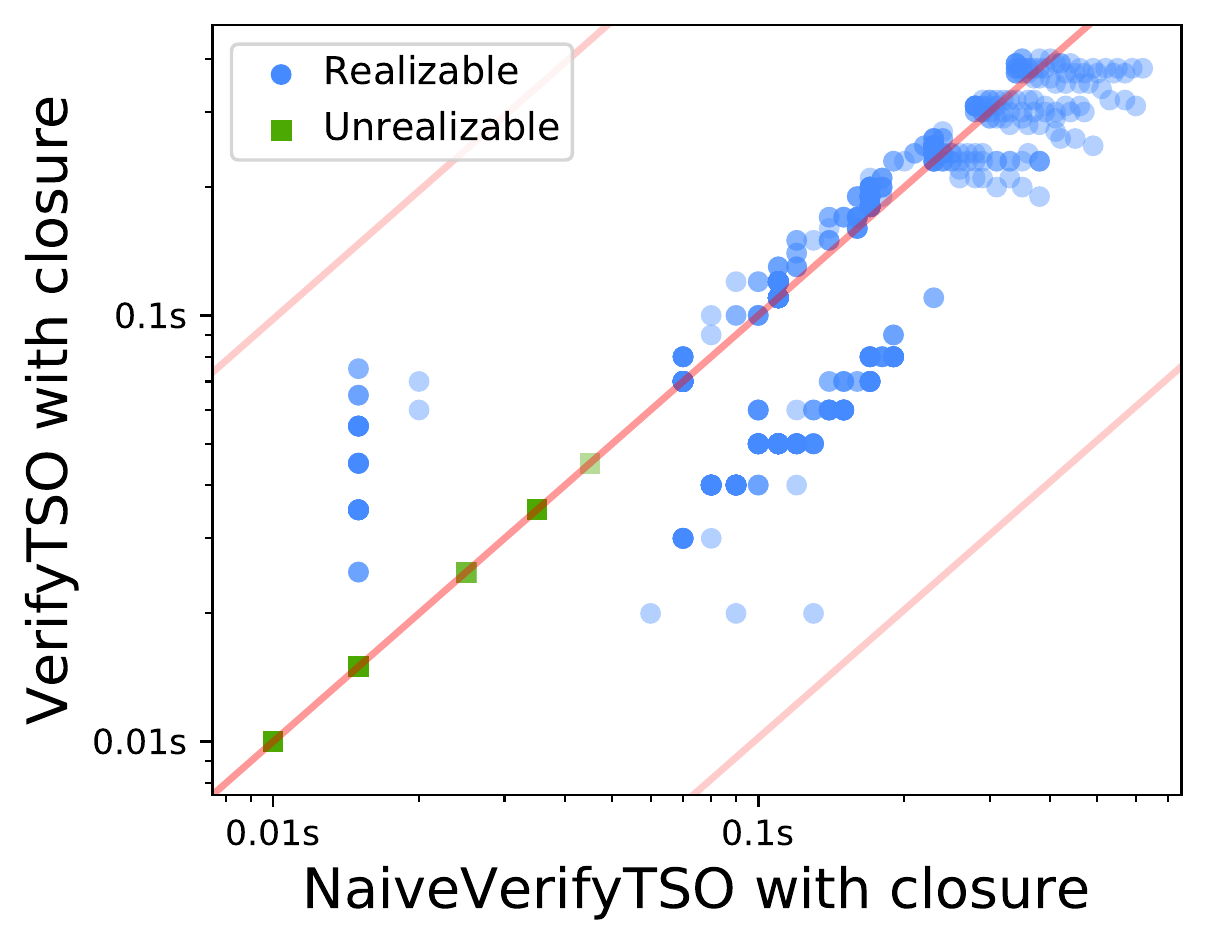}
%\caption{}
\label{subfig:tso_cas_withcl}
\end{subfigure}
%%%%%%%%%%%%%%%
\qquad%\qquad
%%%%%%%%%%%%%%%
\begin{subfigure}[b]{0.44\textwidth}
\centering
\includegraphics[height=5.4cm]{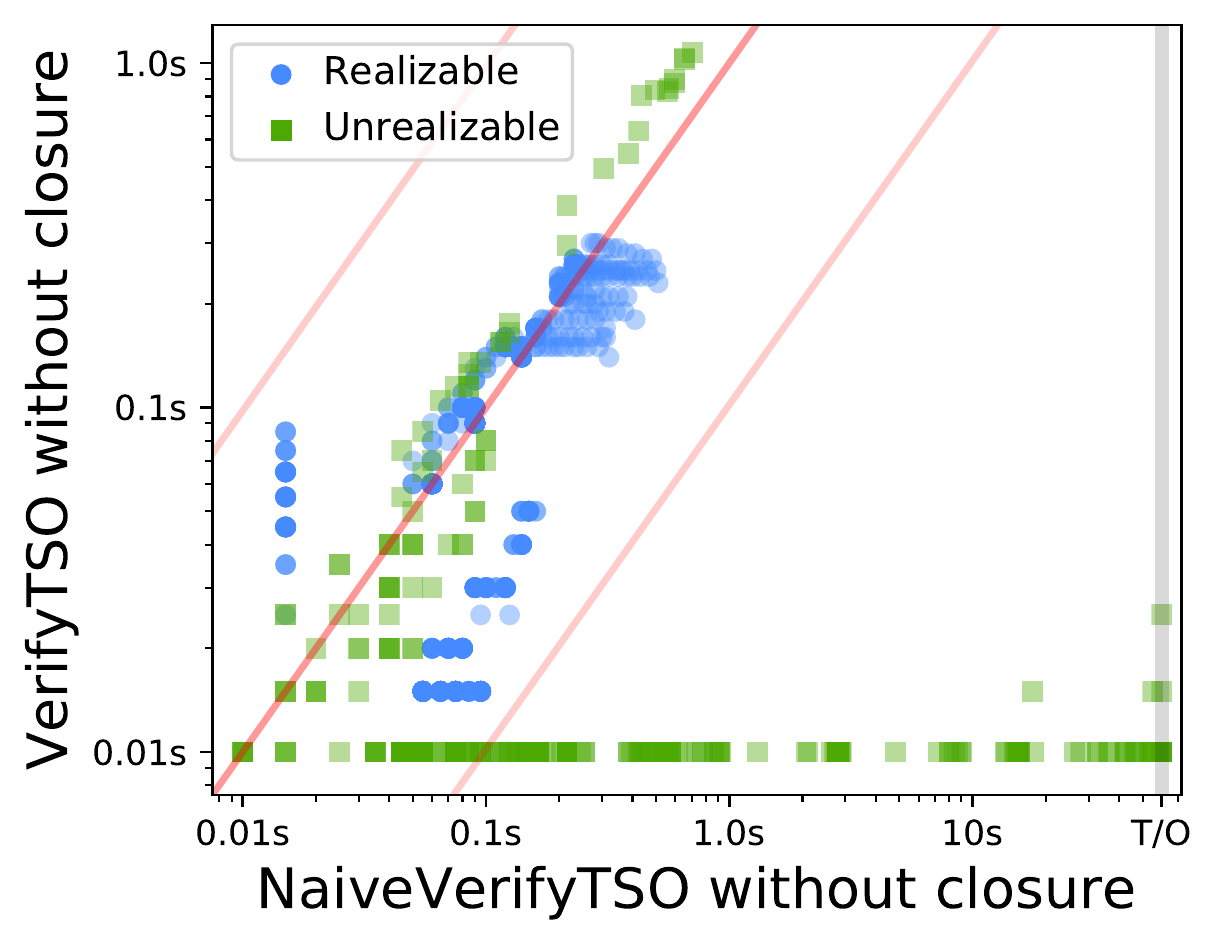}
%\caption{}
\label{subfig:tso_cas_nocl}
\end{subfigure}
\qquad%\qquad
%%%%%%%%%%%%%%%
\vspace{-5mm}
\caption{
Consistency verification comparison of $\AlgoTSO$ and $\NaiveAlgoTSO$
with closure (left) and without closure (right)
on verification instances that contain RMW and CAS instructions.
}
\label{fig:tso_cas_us_them}
\end{figure}

The results are presented in \cref{fig:tso_cas_us_them}.
The left plot depicts the results for $\AlgoTSO$ and $\NaiveAlgoTSO$ when
closure is used as a preceding step. Here the results are all within
an order-of-magnitude difference, and they are identical
for unrealizable instances, since all of them were detected as
unrealizable already by the closure.
The right plot depicts the results for $\AlgoTSO$ and $\NaiveAlgoTSO$
without using the closure. Here the difference for realizable instances
is also within an order of magnitude, but for some unrealizable
instances the algorithm $\AlgoTSO$ is significantly faster.
Generally, the observed improvement of our $\AlgoTSO$ as compared to
$\NaiveAlgoTSO$ is somewhat smaller in \cref{fig:tso_cas_us_them},
which could be due to the fact that executions with RMW and CAS
instructions typically have fewer concurrent writes
(indeed, in an execution where each write event is a part of a RMW/CAS
instruction, each conflicting pair of writes is inherently
ordered by the reads-from orderings together with $\TO$).
Finally, in~\cref{subsec:app_verifyfurther} the effect of closure
is evaluated for both verification algorithms $\AlgoTSO$ and
$\NaiveAlgoTSO$ on instances with RMW and CAS.

\subsection{Experiments on SMC for TSO and PSO}\label{subsec:exp_dpor}
% Reads-From

% In our experiments, we found that the heuristics for $\AlgoTSO$ and $\AlgoPSO$ (\cref{subsec:heuristics}) practically never fail,
% making the consistency checking very efficient in practice.
% A similar observation was made by~\citet{Abdulla19} for RF SMC under $\SC$.
% Our nontrivial generalization of the heuristics of~\citet{Abdulla19} to $\TSO$ and $\PSO$
% extends the practical efficiency into $\TSO$ and $\PSO$.
% In light of these findings, we focus our experimental evaluation
% on assessing the advantages of utilizing the RF equivalence for SMC in $\TSO$ and $\PSO$.

In this section we focus on assessing the advantages of utilizing the reads-from %RF
equivalence for SMC in $\TSO$ and $\PSO$.
We have used $\DCTSOPSOM$ for stateless model checking
of 109 benchmarks
%We have analyzed each benchmark
under each memory model $\MemoryModel\in\{\SC, \TSO, \PSO \}$,
where $\SC$ is handled in our implementation as $\TSO$ with
a fence after each thread event. % inserted
\cref{subsec:app_expdetails} provides further details on our SMC setup. % experiments.

\Paragraph{Comparison.}
As a baseline for comparison,
we have also executed $\Source$-$\DPOR$~\cite{Abdulla14},
which is implemented in Nidhugg and explores the trace space using
the partitioning %$\TraceSpaceMax_{\MemoryModel}/\MazE$
based on the Shasha--Snir equivalence.
%Mazurkiewicz equivalence $\MazE$
%(called the Shasha--Snir equivalence for $\TSO$ and $\PSO$).
In $\SC$, we have further executed $\ReadsFrom$, the
Nidhugg implementation
of the reads-from SMC algorithm for $\SC$ by~\citet{Abdulla19},
and the full comparison that includes $\ReadsFrom$ for $\SC$
is in~\cref{subsec:app_expfull}.
Both $\ReadsFrom$ and $\Source$ are well-optimized, and recently started using advanced data-structures for SMC~\cite{LangS20}.
The works of~\citet{Kokologiannakis19,Kokologiannakis20} provide a general interface for reads-from SMC
in relaxed memory models. However, they handle a given memory model assuming
that an auxiliary consistency verification algorithm for that memory
model is provided.
No such consistency algorithm for $\TSO$ or $\PSO$ is
presented by~\citet{Kokologiannakis19,Kokologiannakis20}, and, to our
knowledge, the tool implementations of~\citet{Kokologiannakis19,Kokologiannakis20}
also lack a consistency algorithm for both $\TSO$ and $\PSO$.
Thus these tools are not included in the evaluation.\footnote{
Another related work is MCR~\cite{Huang16},
however, the corresponding tool operates on Java programs and uses heavyweight SMT solvers that
require fine tuning, and thus is beyond the experimental scope of this work.}
%Andreas: I don't see a good way to do this. Let's keep it for the rebuttale, if needed
% Perhaps somehow mention and cite the bounded model
%checker DARTAGNAN, MPI people ran experiments against it. Paper is CAV tool paper 2019
%'BMC for Weak Memory Models: Relation Analysis for Compact SMT Encodings'.
%For C++ programs. Handles TSO but you have to input TSO semantics as a formula; not sure about PSO.
%Does not do any SMC/DPOR, instead uses SMT solvers and encodes entire program and memory model as an SMT formula.

%To the best of our knowledge, the only other SMC tool that handles the memory models $\TSO$ and $\PSO$ is presented by~\citet{Huang16}.
%uses NP-oracles to determine realizability under $\TSO$ and $\PSO$, thus sidestepping
%the central task of developing efficient realizability algorithms for $\TSO$ and $\PSO$. Further,
%the implementation of~\citet{Huang16} unfortunately operates only on Java
%programs, and thus it is not included in our experimental evaluation.

\Paragraph{Evaluation objective.}
Our objective for the SMC evaluation is three-fold. % goal
First, we want to %are interested to
quantify how each memory model
$\MemoryModel\in\{\SC, \TSO, \PSO \}$ impacts the size of the RF partitioning.
%$\TraceSpaceMax_{\MemoryModel}/\ObsE$.
%We expect that as we move towards more relaxed semantics, the size of
%the partitioning increases.
Second, we are interested to see whether, as compared to the
baseline Shasha--Snir equivalence, the RF equivalence %Shasha--Snir
leads to coarser partitionings for $\TSO$ and $\PSO$,
as it does for $\SC$~\cite{Abdulla19}. %Chalupa17 uses RF only in acyclic-communication programs
Finally, we want to determine whether a coarser RF partitioning
%$\TraceSpaceMax_{\MemoryModel}/\ObsE$
leads to faster exploration.
%, and (ii) in examples with RF partitioning
%$\TraceSpaceMax_{\MemoryModel}/\ObsE$
%as coarse as Shasha--Snir partitioning, %$\TraceSpaceMax_{\MemoryModel}/\MazE$,
%whether the overhead of operating on the reads-from equivalence hurts the exploration speed.
\cref{them:dctsopso_maximal} %(resp. \cref{them:dctsopso_partial})
states that $\DCTSOPSOM$ %(resp. $\DCTSOPSOP$)
spends polynomial time per
partitioning class, %on average,
and we aim to see whether
this is a small polynomial in practice.

\begin{figure}[h]
\raggedright
%%%%%%%%%%%%%%%
\begin{subfigure}[b]{0.443\textwidth}
\raggedright
\includegraphics[height=5.38cm]{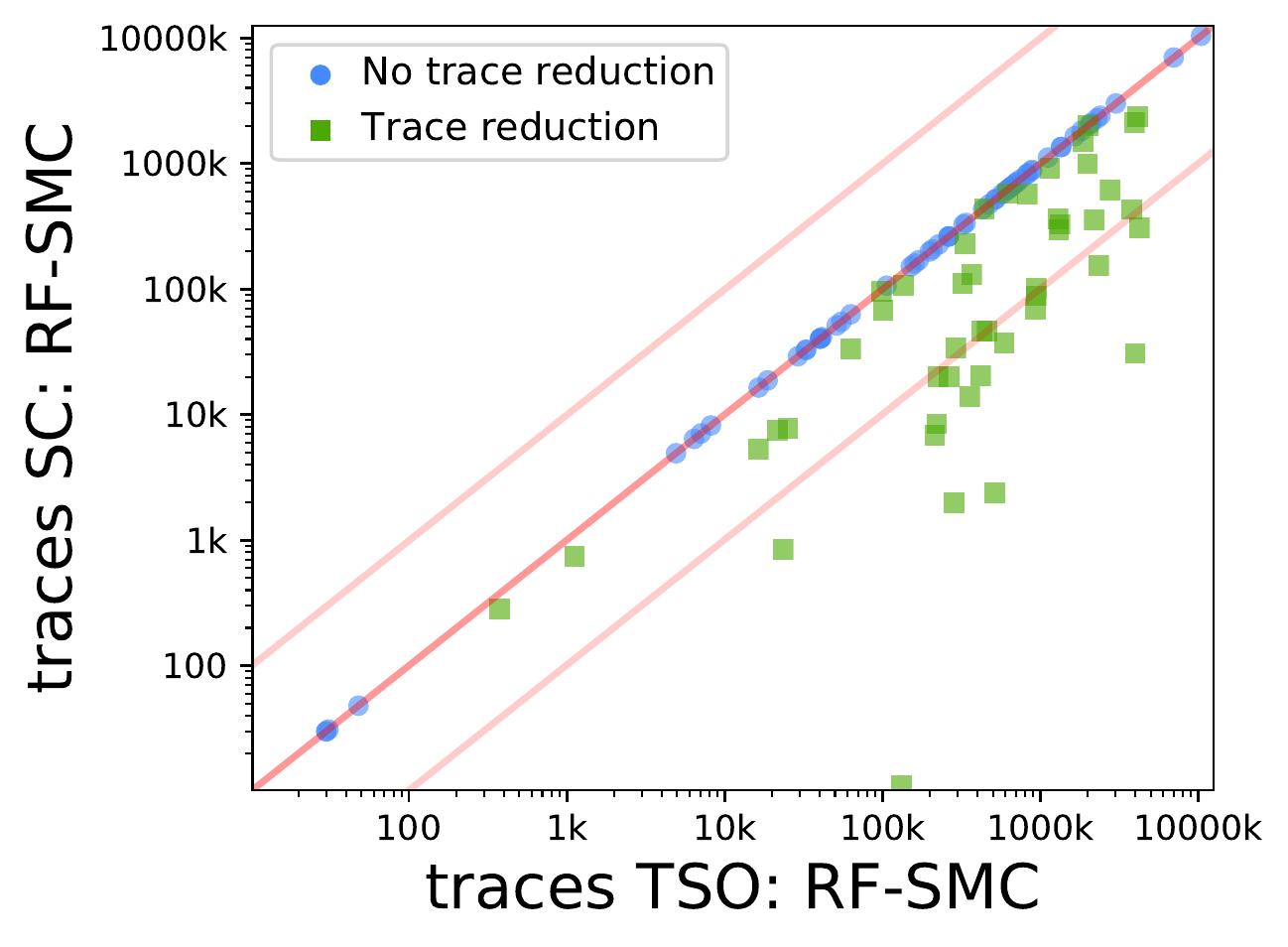}
%\caption{}
\label{subfig:p_sctso_tr}
\end{subfigure}
%%%%%%%%%%%%%%%
\qquad%\qquad
%%%%%%%%%%%%%%%
\begin{subfigure}[b]{0.44\textwidth}
\raggedright
\includegraphics[height=5.38cm]{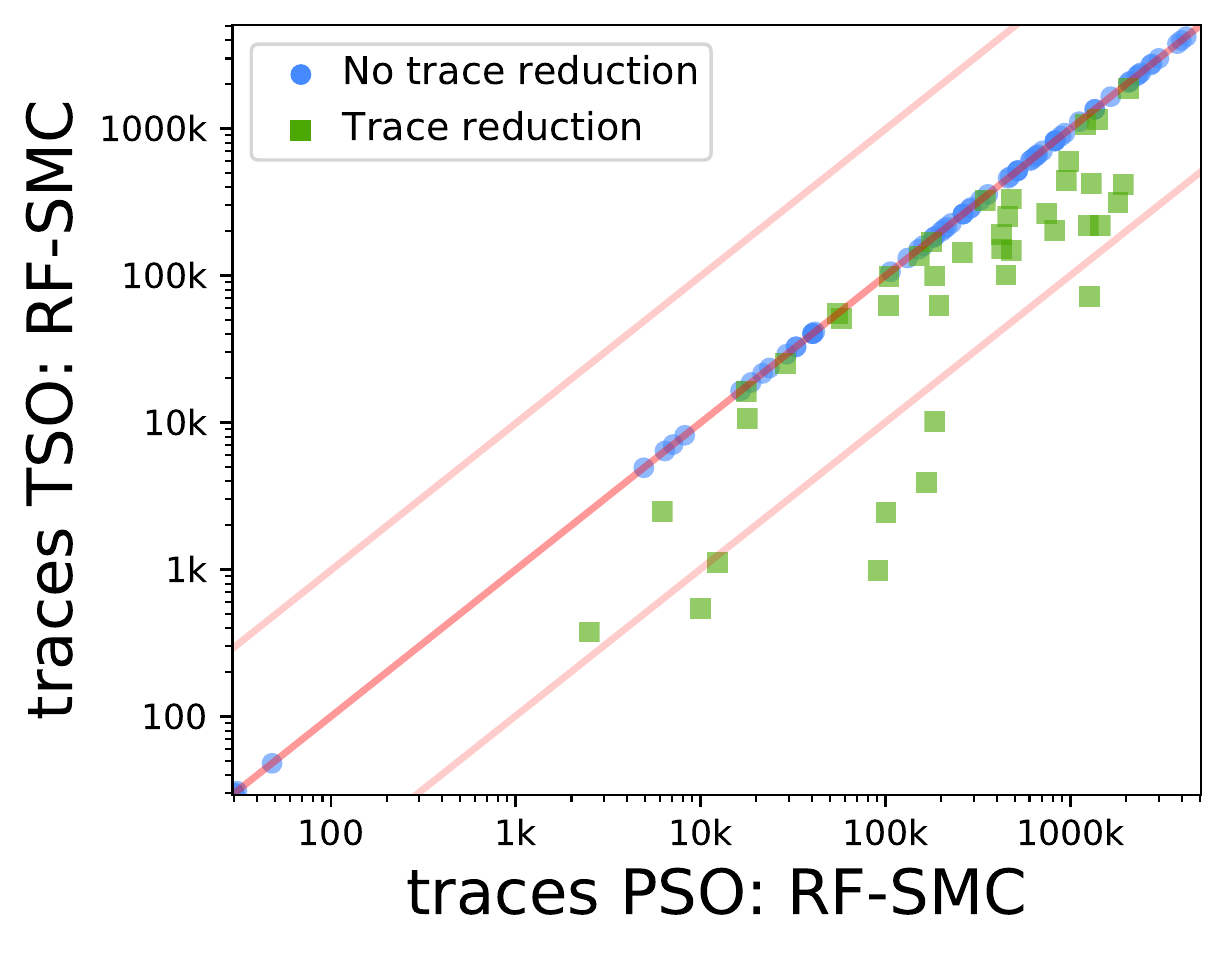}
%\caption{}
\label{subfig:p_tsopso_tr}
\end{subfigure}
\qquad%\qquad
%%%%%%%%%%%%%%%
\vspace{-5mm}
\caption{
Traces comparison as $\DCTSOPSOM$ moves from $\SC$ to $\TSO$ (left) and from $\TSO$ to $\PSO$ (right).
}
\label{fig:models_traces}
\vspace{2mm}
\end{figure}

\Paragraph{Results.}
We illustrate the obtained results with several scatter plots.
Each plot compares two algorithms executing under specified memory
models.
Then for each benchmark,
we consider the highest attempted unroll bound
where both the compared algorithms
%$\DCTSOPSOM$ and $\Source$
finish
%under the specified memory model
before the one-hour timeout.
Green dots indicate that a trace
reduction was achieved on the underlying benchmark by
the algorithm on the y-axis as compared to the algorithm on the x-axis.
Benchmarks with no trace reduction are represented by the blue dots.
All scatter plots are in log scale, the opaque and semi-transparent
red lines represent identity and an order-of-magnitude difference,
respectively.

The plots in \cref{fig:models_traces} illustrate how the size of the
RF partitioning explored by $\DCTSOPSOM$ changes
as we move to more relaxed memory models
($\SC$ to $\TSO$ to $\PSO$).
The plots in \cref{fig:exo_source_tsopso_traces} capture how the
size of the RF partitioning explored by $\DCTSOPSOM$ relates to
the size of the Shasha--Snir partitioning explored by $\Source$.
Finally, the plots in \cref{fig:exo_source_tsopso_times} demonstrate
the time comparison of $\DCTSOPSOM$ and $\Source$ when there is
some (green dots) or no (blue dots) RF-induced trace reduction.

Below we discuss the observations on the obtained results.
\cref{tab:smc_selected} captures detailed results on several
benchmarks that we refer to as examples in the discussion.

\begin{figure}[h]
\vspace{3mm}
\raggedright
%%%%%%%%%%%%%%%
\begin{subfigure}[b]{0.443\textwidth}
\raggedright
\includegraphics[height=5.38cm]{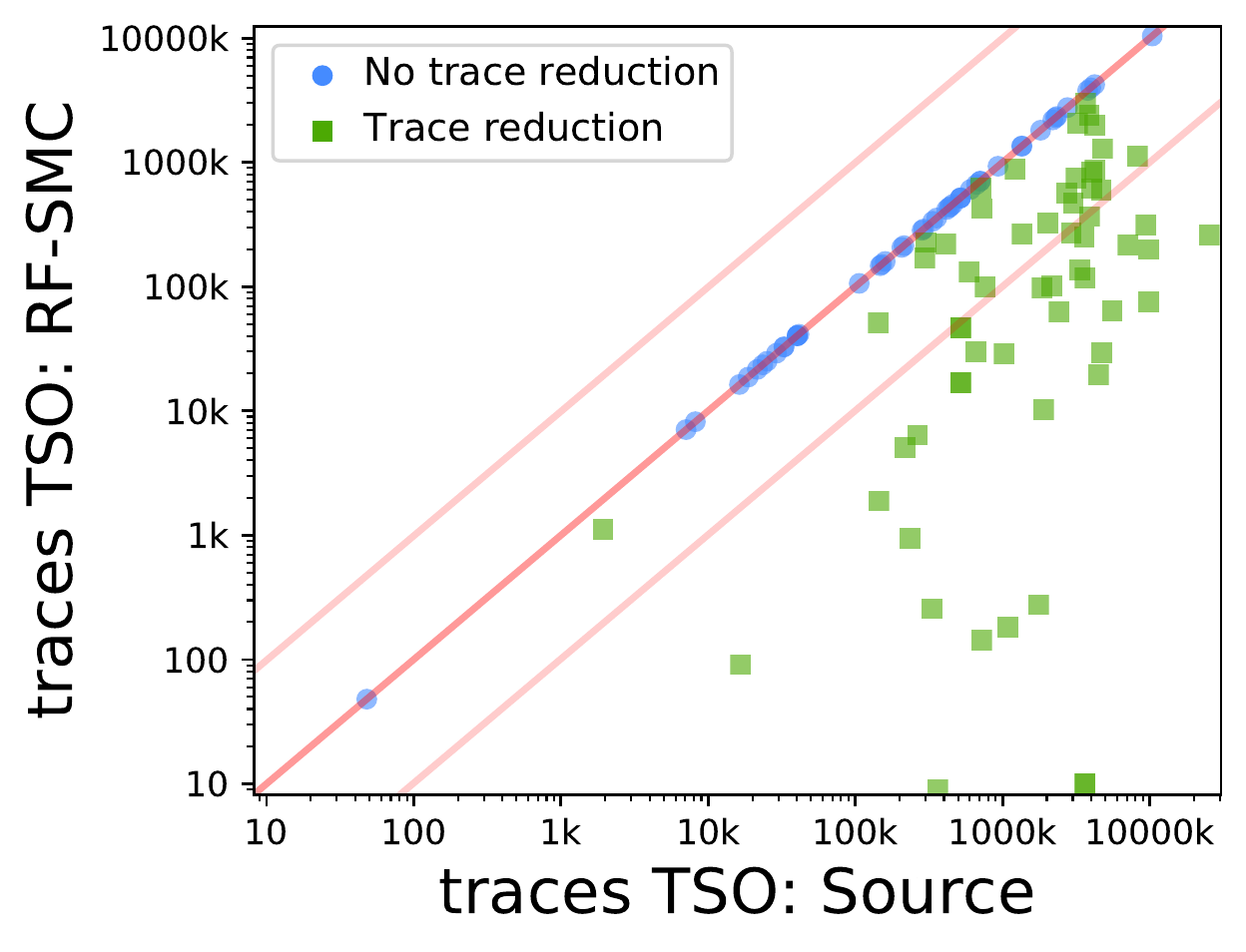}
%\caption{}
\label{subfig:p_tso_exo_source_tr}
\end{subfigure}
%%%%%%%%%%%%%%%
\qquad%\qquad
%%%%%%%%%%%%%%%
\begin{subfigure}[b]{0.44\textwidth}
\raggedright
\includegraphics[height=5.38cm]{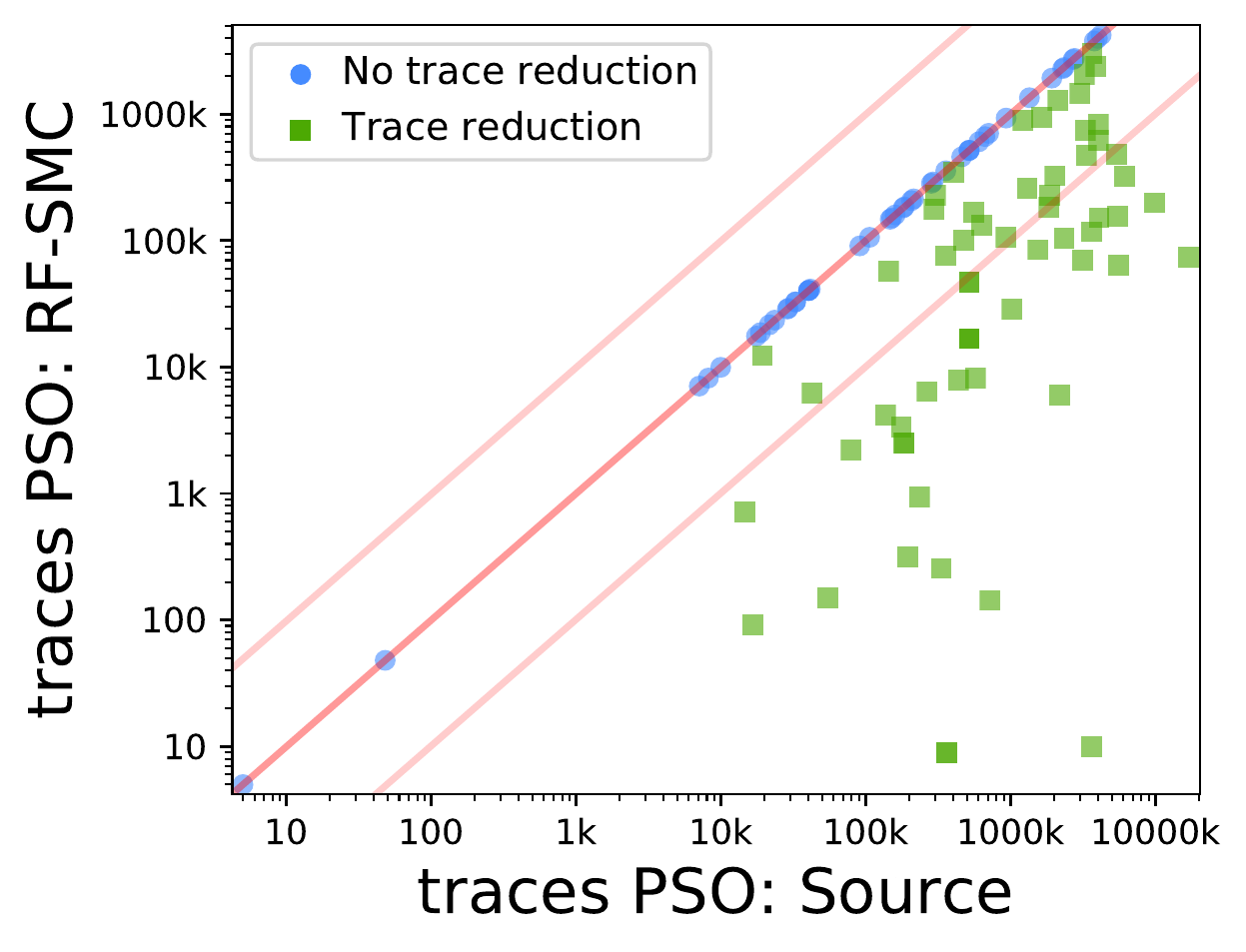}
%\caption{}
\label{subfig:p_pso_exo_source_tr}
\end{subfigure}
\qquad%\qquad
%%%%%%%%%%%%%%%
\vspace{-5mm}
\caption{
Traces comparison for $\DCTSOPSOM$ and $\Source$ on the $\TSO$ (left) and $\PSO$ (right) memory model.
}
\label{fig:exo_source_tsopso_traces}
\end{figure}

\begin{figure}[h]
\raggedright
%%%%%%%%%%%%%%%
\begin{subfigure}[b]{0.44\textwidth}
\raggedright
\includegraphics[height=5.4cm]{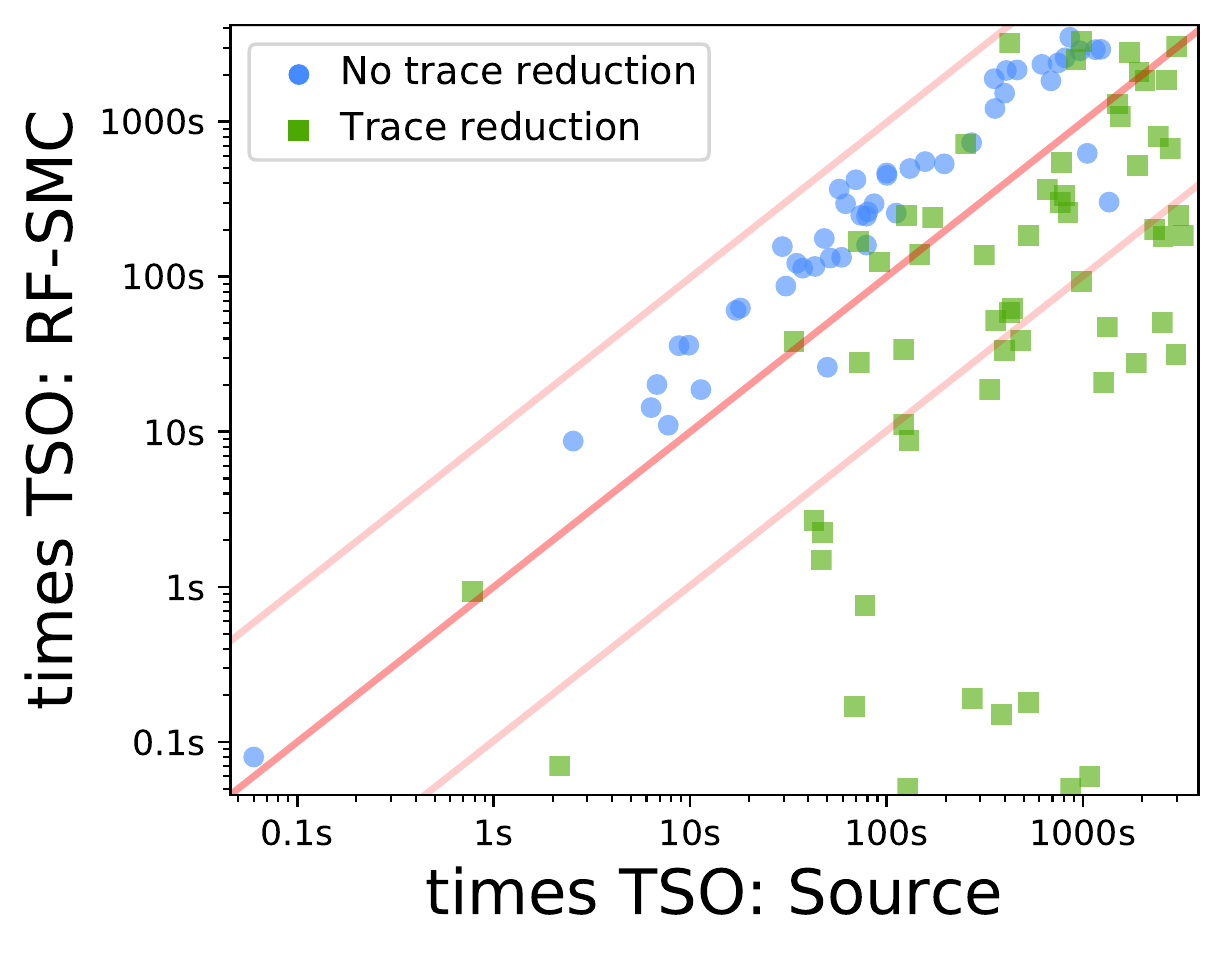}
%\caption{}
\label{subfig:p_tso_exo_source_ti}
\end{subfigure}
%%%%%%%%%%%%%%%
\qquad%\qquad
%%%%%%%%%%%%%%%
\begin{subfigure}[b]{0.44\textwidth}
\raggedright
\includegraphics[height=5.4cm]{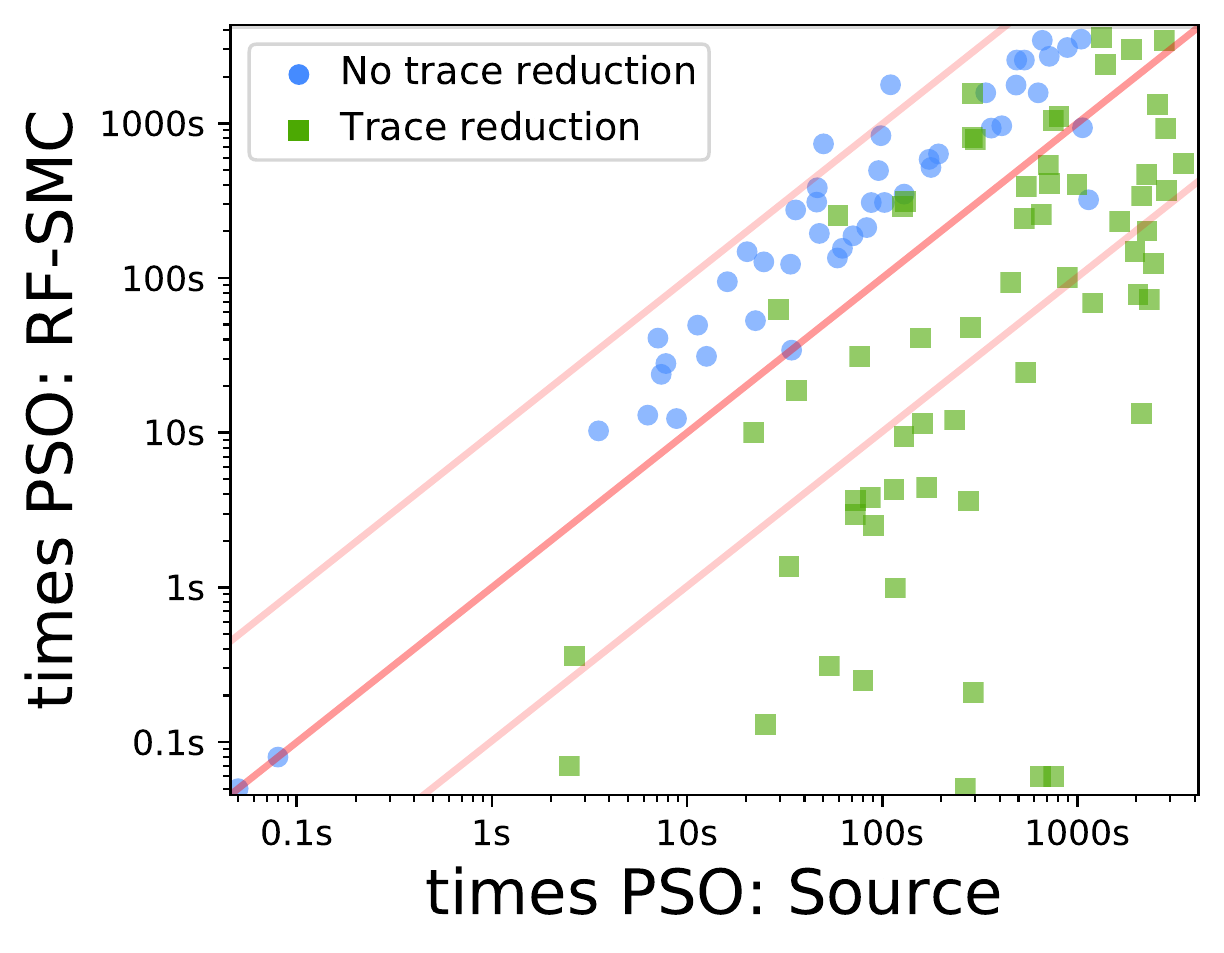}
%\caption{}
\label{subfig:p_pso_exo_source_ti}
\end{subfigure}
\qquad%\qquad
%%%%%%%%%%%%%%%
\vspace{-5mm}
\caption{
Times comparison for $\DCTSOPSOM$ and $\Source$ on the $\TSO$ (left) and $\PSO$ (right) memory model.
}
\label{fig:exo_source_tsopso_times}
\end{figure}

\Paragraph{Discussion.}
We notice that the analysed programs can often exhibit additional
behavior in relaxed memory settings. This causes an increase
in the size of the partitionings explored by SMC algorithms
(see {\tt 27\_Boop4} in \cref{tab:smc_selected} as an example).
\cref{fig:models_traces} illustrates the overall phenomenon for $\DCTSOPSOM$,
where the increase of the RF partitioning size (and hence the number
of traces explored) is sometimes beyond an order of magnitude
when moving from $\SC$ to $\TSO$, or from $\TSO$ to $\PSO$.

We observe that across all memory models,
the reads-from equivalence can offer significant reduction
in the trace partitioning as compared to Shasha--Snir equivalence.
This leads to fewer traces that need to be explored, see the
plots of \cref{fig:exo_source_tsopso_traces}.
As we move
towards more relaxed memory ($\SC$ to $\TSO$ to $\PSO$), the
reduction of RF partitioning often becomes more prominent
(see {\tt 27\_Boop4} in \cref{tab:smc_selected}). Interestingly,
in some cases
the size of the Shasha--Snir partitioning explored by $\Source$ increases
as we move to more relaxed settings, while the RF partitioning remains
unchanged (cf.~{\tt fillarray\_false} in \cref{tab:smc_selected}). % fillarray fillarr
All these observations signify advantages of RF for analysis of
the more complex program behavior that arises due to relaxed memory.

We now discuss how trace partitioning coarseness
affects execution time, observing the plots of
\cref{fig:exo_source_tsopso_times}.
We see that in cases where RF
partitioning is coarser (green dots), our RF algorithm
$\DCTSOPSOM$ often becomes significantly faster than the
Shasha--Snir-based $\Source$, allowing us to analyse programs scaled
several levels further (see {\tt eratosthenes} in \cref{tab:smc_selected}).
%as one of many specific examples).
% Naturally, larger RF-based partitioning
% reduction typically leads to more drastic speedup.
In cases where the sizes of the RF partitioning and the Shasha--Snir
partitioning coincide (blue dots), the well-engineered $\Source$ outperforms
our $\DCTSOPSOM$ implementation.
% The scatter plots illustrate that
The time differences in these
cases are reasonably moderate, suggesting that the
polynomial overhead incurred to operate on the RF partitioning
is small in practice.

\cref{subsec:app_expfull} contains the complete results on
all 109 benchmarks, as well as further scatter plots, illustrating
(i) comparison of $\DCTSOPSOM$ with $\Source$ and $\ReadsFrom$ in $\SC$,
(ii) time comparison of $\DCTSOPSOM$ across memory models, and
(iii) the effect of using closure in the constency checking during SMC.

\begin{table}[t]
\footnotesize
\newcolumntype{?}{!{\vrule width 1.5pt}}
\setlength{\extrarowheight}{.08em}
\begin{tabular}{? l | c | c ? r r ? r r ? r r ? }
\specialrule{.15em}{0em}{0em}
\multicolumn{2}{?c|}{\multirow{2}{*}{\textbf{Benchmark}}} & \multirow{2}{*}{U} & \multicolumn{2}{c?}{\textbf{Seq. Consistency}} & \multicolumn{2}{c?}{\textbf{Total Store Order}} & \multicolumn{2}{c?}{\textbf{Partial Store Order}}\\
\cline{4-9}
\multicolumn{2}{?c|}{} & & $\DCTSOPSOM$ & $\Source$ & $\DCTSOPSOM$ & $\Source$ & $\DCTSOPSOM$ & $\Source$\\
\specialrule{.1em}{0em}{0em}
% 27_Boop_simple_vf_false-unreach-call_4thr.c
\multirow{4}{*}{\begin{tabular}{l}
\textbf{27\_Boop4}\\
% lines of code: 74\\
% variables: 4\\
threads: 4
\end{tabular}}
&
\multirow{2}{*}{Traces}
                                                                                                      &                                                1 &                                    \textbf{2902} &                                            21948 &                                    \textbf{3682} &                                            36588 &                                    \textbf{8233} &                                           572436 \\
%                                                   &                                                  &                                                2 &                                   \textbf{18040} &                                           178236 &                                   \textbf{24766} &                                           355716 &                                   \textbf{80680} &                                                - \\
%                                                   &                                                  &                                                3 &                                   \textbf{68083} &                                           966834 &                                  \textbf{100897} &                                          2157426 &                                  \textbf{447739} &                                                - \\
                                                   &                                                  &                                                4 &                                  \textbf{197260} &                                          3873348 &                                  \textbf{313336} &                                          9412428 &                                 \textbf{1807408} &                                                - \\
\cline{2-9}
&
\multirow{2}{*}{Times}
                                                                                                      &                                                1 &                                   \textbf{1.22s} &                                            1.74s &                                   \textbf{1.46s} &                                            6.18s &                                   \textbf{4.40s} &                                             169s \\
%                                                   &                                                  &                                                2 &                                   \textbf{8.59s} &                                              17s &                                     \textbf{11s} &                                              57s &                                     \textbf{51s} &                                                - \\
%                                                   &                                                  &                                                3 &                                     \textbf{38s} &                                             186s &                                     \textbf{52s} &                                             340s &                                    \textbf{338s} &                                                - \\
                                                   &                                                  &                                                4 &                                    \textbf{124s} &                                             550s &                                    \textbf{182s} &                                            2556s &                                   \textbf{1593s} &                                                - \\
\specialrule{.1em}{0em}{0em}
% eratosthenes.c
\multirow{4}{*}{\begin{tabular}{l}
\textbf{eratosthenes}\\
% lines of code: 25\\
% variables: U\\
threads: 2
\end{tabular}}
&
\multirow{2}{*}{Traces}
%                                                                                                      &                                               13 &                                     \textbf{535} &                                             3358 &                                    \textbf{1467} &                                            32800 &                                    \textbf{3375} &                                           175616 \\
%                                                   &                                                  &                                               15 &                                    \textbf{1465} &                                            15460 &                                    \textbf{4791} &                                           240160 &                                   \textbf{16875} &                                                - \\
                                                                                                     &                                               17 &                                    \textbf{4667} &                                           100664 &                                   \textbf{29217} &                                          4719488 &                                  \textbf{253125} &                                                - \\
%                                                   &                                                  &                                               19 &                                    \textbf{8495} &                                           359598 &                                   \textbf{71607} &                                                - &                                 \textbf{1265625} &                                                - \\
                                                   &                                                  &                                               21 &                                   \textbf{19991} &                                          1527736 &                                  \textbf{223929} &                                                - &                                                - &                                                - \\
\cline{2-9}
&
\multirow{2}{*}{Times}
%                                                                                                      &                                               13 &                                   \textbf{0.49s} &                                            0.69s &                                   \textbf{1.05s} &                                              19s &                                   \textbf{3.62s} &                                             277s \\
%                                                   &                                                  &                                               15 &                                   \textbf{1.49s} &                                            6.51s &                                   \textbf{4.02s} &                                             195s &                                     \textbf{22s} &                                                - \\
                                                                                                     &                                               17 &                                   \textbf{6.70s} &                                              46s &                                     \textbf{32s} &                                            2978s &                                    \textbf{475s} &                                                - \\
%                                                   &                                                  &                                               19 &                                     \textbf{15s} &                                             108s &                                     \textbf{93s} &                                                - &                                   \textbf{3264s} &                                                - \\
                                                   &                                                  &                                               21 &                                     \textbf{41s} &                                             736s &                                    \textbf{342s} &                                                - &                                                - &                                                - \\
\specialrule{.1em}{0em}{0em}
% fillarray_false-valid-deref.c
\multirow{4}{*}{\begin{tabular}{l}
\textbf{fillarray\_false}\\
% lines of code: 42\\
% variables: 14\\
threads: 2
\end{tabular}}
&
\multirow{2}{*}{Traces}
%                                                                                                      &                                                1 &                                      \textbf{17} &                                               26 &                                      \textbf{17} &                                               28 &                                      \textbf{17} &                                               28 \\
%                                                   &                                                  &                                                2 &                                     \textbf{475} &                                             1058 &                                     \textbf{475} &                                             1226 &                                     \textbf{475} &                                             1258 \\
                                                                                                     &                                                3 &                                   \textbf{14625} &                                            47892 &                                   \textbf{14625} &                                            59404 &                                   \textbf{14625} &                                            63088 \\
                                                   &                                                  &                                                4 &                                  \textbf{471821} &                                          2278732 &                                  \textbf{471821} &                                          3023380 &                                  \textbf{471821} &                                          3329934 \\
\cline{2-9}
&
\multirow{2}{*}{Times}
%                                                                                                      &                                                1 &                                   \textbf{0.05s} &                                   \textbf{0.05s} &                                   \textbf{0.05s} &                                            0.07s &                                            0.06s &                                   \textbf{0.05s} \\
%                                                   &                                                  &                                                2 &                                            0.30s &                                   \textbf{0.15s} &                                            0.30s &                                   \textbf{0.27s} &                                            0.42s &                                   \textbf{0.38s} \\
                                                                                                     &                                                3 &                                              12s &                                   \textbf{6.18s} &                                     \textbf{12s} &                                     \textbf{12s} &                                     \textbf{18s} &                                              39s \\
                                                   &                                                  &                                                4 &                                             553s &                                    \textbf{331s} &                                    \textbf{547s} &                                             778s &                                    \textbf{930s} &                                            2844s \\
\specialrule{.1em}{0em}{0em}
\end{tabular}
% \addtocounter{table}{-1}
% % \vspace{-2mm}
\caption{
SMC results on several benchmarks. \textbf{U} denotes the unroll bound.
% Each benchmark has the number of its threads
% (sometimes parameterized by \textbf{U})
% written underneath its name in the table.
The timeout of one hour is indicated by ``-''.
% from the SMC experiments.  with additional relaxed-memory behaviour.
Bold-font entries indicate the smallest numbers for the respective memory model.
}
\label{tab:smc_selected}
\end{table}

\section{CONCLUSIONS}\label{sec:conclusions}
% {Conclusions}

In this work we have solved the consistency verification problem under a reads-from map for the $\TSO$ and $\PSO$ relaxed memory models.
Our algorithms scale as $O(k\cdot n^{k+1})$ for $\TSO$, and as \mbox{$O(k\cdot n^{k+1}\cdot \min(n^{k\cdot (k-1)}, 2^{k\cdot \NumVariables}))$} for $\PSO$, for $n$ events, $k$ threads and $\NumVariables$ variables.
Thus, they both become polynomial-time for a bounded number of threads, similar to the case for $\SC$ that was established recently~\cite{Abdulla19,BiswasE19}.
In practice, our algorithms perform much better than the standard baseline methods, offering significant scalability improvements.
Encouraged by these scalability improvements, we have used these algorithms to develop, for the first time, SMC under $\TSO$ and $\PSO$ using the reads-from equivalence, as opposed to the standard Shasha--Snir equivalence.
Our experiments show that the underlying reads-from partitioning is often much coarser than the Shasha--Snir partitioning, which yields a significant speedup in the model checking task.

We remark that our consistency-verification algorithms have direct applications beyond SMC.
In particular, most predictive dynamic analyses solve a consistency-verification problem in order to infer whether an erroneous execution can be generated by a concurrent system (see, e.g., \citet{Smaragdakis12,Kini17,Mathur20}).
Hence, the results of this work allow to extend predictive analyses to $\TSO/\PSO$ in a scalable way that does not sacrifice precision.
We will pursue this direction in our future work.

%In this work, we have generalized the reads-from partitioning of~\citet{Abdulla19,Chalupa17}
%from sequential consistency to the $\TSO$ and $\PSO$ memory models.
%For the execution-consistency verification problem,
%we have presented
%algorithms $\AlgoTSO$ for $\TSO$ and $\AlgoPSO$ for $\PSO$.
%Further, we have presented an SMC algorithm $\DCTSOPSOM$ for
%$\TSO$ and $\PSO$
%operating on the reads-from partitioning, which spends polynomial
%time per partitioning class for bounded number of threads,
%and is exploration-optimal.
%This is the first provable coarsening of
%the standard Shasha--Snir abstraction for $\TSO$ and $\PSO$
%with tractable (polynomial-time) complexity per partitioning class for bounded number of threads.
%We have implemented $\DCTSOPSOM$ in a tool that handles
%concurrent C/C++ programs in LLVM IR.
%Our experimental results show that the reads-from partitioning is often
%coarser than the Shasha--Snir partitioning across both
%$\TSO$ and $\PSO$ in practice, which also leads to significant
%speedup in the model checking of the respective benchmarks.
%Development of verification techniques based on reads-from equivalence
%for further memory models is an interesting option for future work.
%Additionally, design of algorithms based on equivalences even coarser
%than the reads-from is an attractive avenue of future work.

%\clearpage
\renewcommand{\refname}{REFERENCES}
\bibliography{bibliography}


\begin{thebibliography}{58}

%%% the Web address.

\ifx \showCODEN    \undefined \def \showCODEN     #1{\unskip}     \fi
\ifx \showDOI      \undefined \def \showDOI       #1{#1}\fi
\ifx \showISBNx    \undefined \def \showISBNx     #1{\unskip}     \fi
\ifx \showISBNxiii \undefined \def \showISBNxiii  #1{\unskip}     \fi
\ifx \showISSN     \undefined \def \showISSN      #1{\unskip}     \fi
\ifx \showLCCN     \undefined \def \showLCCN      #1{\unskip}     \fi
\ifx \shownote     \undefined \def \shownote      #1{#1}          \fi
\ifx \showarticletitle \undefined \def \showarticletitle #1{#1}   \fi
\ifx \showURL      \undefined \def \showURL       {\relax}        \fi
% invisible to TeX
\providecommand\bibfield[2]{#2}
\providecommand\bibinfo[2]{#2}
\providecommand\natexlab[1]{#1}
\providecommand\showeprint[2][]{arXiv:#2}

\bibitem[\protect\citeauthoryear{Abdulla, Aronis, Jonsson, and Sagonas}{Abdulla
  et~al\mbox{.}}{2014}]%
        {Abdulla14}
\bibfield{author}{\bibinfo{person}{Parosh Abdulla}, \bibinfo{person}{Stavros
  Aronis}, \bibinfo{person}{Bengt Jonsson}, {and} \bibinfo{person}{Konstantinos
  Sagonas}.} \bibinfo{year}{2014}\natexlab{}.
\newblock \showarticletitle{Optimal Dynamic Partial Order Reduction}
  \emph{(\bibinfo{series}{POPL})}.
\newblock


\bibitem[\protect\citeauthoryear{Abdulla, Aronis, Atig, Jonsson, Leonardsson,
  and Sagonas}{Abdulla et~al\mbox{.}}{2015}]%
        {Abdulla2015}
\bibfield{author}{\bibinfo{person}{Parosh~Aziz Abdulla},
  \bibinfo{person}{Stavros Aronis}, \bibinfo{person}{Mohamed~Faouzi Atig},
  \bibinfo{person}{Bengt Jonsson}, \bibinfo{person}{Carl Leonardsson}, {and}
  \bibinfo{person}{Konstantinos Sagonas}.} \bibinfo{year}{2015}\natexlab{}.
\newblock \showarticletitle{Stateless Model Checking for TSO and PSO}. In
  \bibinfo{booktitle}{\emph{TACAS}}.
\newblock


\bibitem[\protect\citeauthoryear{Abdulla, Atig, Jonsson, L\r{a}ng, Ngo, and
  Sagonas}{Abdulla et~al\mbox{.}}{2019}]%
        {Abdulla19}
\bibfield{author}{\bibinfo{person}{Parosh~Aziz Abdulla},
  \bibinfo{person}{Mohamed~Faouzi Atig}, \bibinfo{person}{Bengt Jonsson},
  \bibinfo{person}{Magnus L\r{a}ng}, \bibinfo{person}{Tuan~Phong Ngo}, {and}
  \bibinfo{person}{Konstantinos Sagonas}.} \bibinfo{year}{2019}\natexlab{}.
\newblock \showarticletitle{Optimal Stateless Model Checking for Reads-from
  Equivalence under Sequential Consistency}.
\newblock \bibinfo{journal}{\emph{Proc. ACM Program. Lang.}}
  \bibinfo{volume}{3}, \bibinfo{number}{OOPSLA}, Article
  \bibinfo{articleno}{150} (\bibinfo{date}{Oct.} \bibinfo{year}{2019}),
  \bibinfo{numpages}{29}~pages.
\newblock
\urldef\tempurl%
\url{https://doi.org/10.1145/3360576}
\showDOI{\tempurl}


\bibitem[\protect\citeauthoryear{Abdulla, Atig, Jonsson, and Ngo}{Abdulla
  et~al\mbox{.}}{2018}]%
        {AbdullaAJN18}
\bibfield{author}{\bibinfo{person}{Parosh~Aziz Abdulla},
  \bibinfo{person}{Mohamed~Faouzi Atig}, \bibinfo{person}{Bengt Jonsson}, {and}
  \bibinfo{person}{Tuan~Phong Ngo}.} \bibinfo{year}{2018}\natexlab{}.
\newblock \showarticletitle{Optimal stateless model checking under the
  release-acquire semantics}.
\newblock \bibinfo{journal}{\emph{Proc. {ACM} Program. Lang.}}
  \bibinfo{volume}{2}, \bibinfo{number}{{OOPSLA}} (\bibinfo{year}{2018}),
  \bibinfo{pages}{135:1--135:29}.
\newblock
\urldef\tempurl%
\url{https://doi.org/10.1145/3276505}
\showDOI{\tempurl}


\bibitem[\protect\citeauthoryear{{Adve} and {Gharachorloo}}{{Adve} and
  {Gharachorloo}}{1996}]%
        {Adve96}
\bibfield{author}{\bibinfo{person}{S.~V. {Adve}} {and} \bibinfo{person}{K.
  {Gharachorloo}}.} \bibinfo{year}{1996}\natexlab{}.
\newblock \showarticletitle{Shared memory consistency models: a tutorial}.
\newblock \bibinfo{journal}{\emph{Computer}} \bibinfo{volume}{29},
  \bibinfo{number}{12} (\bibinfo{date}{Dec} \bibinfo{year}{1996}),
  \bibinfo{pages}{66--76}.
\newblock
\urldef\tempurl%
\url{https://doi.org/10.1109/2.546611}
\showDOI{\tempurl}


\bibitem[\protect\citeauthoryear{Albert, Arenas, de~la Banda,
  G{\'o}mez-Zamalloa, and Stuckey}{Albert et~al\mbox{.}}{2017}]%
        {Elvira17}
\bibfield{author}{\bibinfo{person}{Elvira Albert}, \bibinfo{person}{Puri
  Arenas}, \bibinfo{person}{Mar{\'i}a~Garc{\'i}a de~la Banda},
  \bibinfo{person}{Miguel G{\'o}mez-Zamalloa}, {and} \bibinfo{person}{Peter~J.
  Stuckey}.} \bibinfo{year}{2017}\natexlab{}.
\newblock \showarticletitle{Context-Sensitive Dynamic Partial Order Reduction}.
  In \bibinfo{booktitle}{\emph{Computer Aided Verification}},
  \bibfield{editor}{\bibinfo{person}{Rupak Majumdar} {and}
  \bibinfo{person}{Viktor Kun{\v{c}}ak}} (Eds.). \bibinfo{publisher}{Springer
  International Publishing}, \bibinfo{address}{Cham},
  \bibinfo{pages}{526--543}.
\newblock
\showISBNx{978-3-319-63387-9}


\bibitem[\protect\citeauthoryear{Albert, G{\'o}mez-Zamalloa, Isabel, and
  Rubio}{Albert et~al\mbox{.}}{2018}]%
        {Elvira18}
\bibfield{author}{\bibinfo{person}{Elvira Albert}, \bibinfo{person}{Miguel
  G{\'o}mez-Zamalloa}, \bibinfo{person}{Miguel Isabel}, {and}
  \bibinfo{person}{Albert Rubio}.} \bibinfo{year}{2018}\natexlab{}.
\newblock \showarticletitle{Constrained Dynamic Partial Order Reduction}. In
  \bibinfo{booktitle}{\emph{Computer Aided Verification}},
  \bibfield{editor}{\bibinfo{person}{Hana Chockler} {and}
  \bibinfo{person}{Georg Weissenbacher}} (Eds.). \bibinfo{publisher}{Springer
  International Publishing}, \bibinfo{address}{Cham},
  \bibinfo{pages}{392--410}.
\newblock
\showISBNx{978-3-319-96142-2}


\bibitem[\protect\citeauthoryear{Alglave}{Alglave}{2010}]%
        {Alglave10}
\bibfield{author}{\bibinfo{person}{Jade Alglave}.}
  \bibinfo{year}{2010}\natexlab{}.
\newblock \emph{\bibinfo{title}{A Shared Memory Poetics}}.
\newblock \bibinfo{thesistype}{Ph.D. Dissertation}. \bibinfo{school}{Paris
  Diderot University}.
\newblock


\bibitem[\protect\citeauthoryear{Alglave, Cousot, and Urban}{Alglave
  et~al\mbox{.}}{2017}]%
        {Alglave17}
\bibfield{author}{\bibinfo{person}{Jade Alglave}, \bibinfo{person}{Patrick
  Cousot}, {and} \bibinfo{person}{Caterina Urban}.}
  \bibinfo{year}{2017}\natexlab{}.
\newblock \showarticletitle{{Concurrency with Weak Memory Models (Dagstuhl
  Seminar 16471)}}.
\newblock \bibinfo{journal}{\emph{Dagstuhl Reports}} \bibinfo{volume}{6},
  \bibinfo{number}{11} (\bibinfo{year}{2017}), \bibinfo{pages}{108--128}.
\newblock
\showISSN{2192-5283}
\urldef\tempurl%
\url{https://doi.org/10.4230/DagRep.6.11.108}
\showDOI{\tempurl}


\bibitem[\protect\citeauthoryear{Aronis, Jonsson, L{\aa}ng, and Sagonas}{Aronis
  et~al\mbox{.}}{2018}]%
        {Aronis18}
\bibfield{author}{\bibinfo{person}{Stavros Aronis}, \bibinfo{person}{Bengt
  Jonsson}, \bibinfo{person}{Magnus L{\aa}ng}, {and}
  \bibinfo{person}{Konstantinos Sagonas}.} \bibinfo{year}{2018}\natexlab{}.
\newblock \showarticletitle{Optimal Dynamic Partial Order Reduction with
  Observers}. In \bibinfo{booktitle}{\emph{Tools and Algorithms for the
  Construction and Analysis of Systems}},
  \bibfield{editor}{\bibinfo{person}{Dirk Beyer} {and} \bibinfo{person}{Marieke
  Huisman}} (Eds.). \bibinfo{publisher}{Springer International Publishing},
  \bibinfo{address}{Cham}, \bibinfo{pages}{229--248}.
\newblock
\showISBNx{978-3-319-89963-3}


\bibitem[\protect\citeauthoryear{Biswas and Enea}{Biswas and Enea}{2019}]%
        {BiswasE19}
\bibfield{author}{\bibinfo{person}{Ranadeep Biswas} {and}
  \bibinfo{person}{Constantin Enea}.} \bibinfo{year}{2019}\natexlab{}.
\newblock \showarticletitle{On the complexity of checking transactional
  consistency}.
\newblock \bibinfo{journal}{\emph{Proc. {ACM} Program. Lang.}}
  \bibinfo{volume}{3}, \bibinfo{number}{{OOPSLA}} (\bibinfo{year}{2019}),
  \bibinfo{pages}{165:1--165:28}.
\newblock
\urldef\tempurl%
\url{https://doi.org/10.1145/3360591}
\showDOI{\tempurl}


\bibitem[\protect\citeauthoryear{Bouajjani, Derevenetc, and Meyer}{Bouajjani
  et~al\mbox{.}}{2013}]%
        {Bouajjani13}
\bibfield{author}{\bibinfo{person}{Ahmed Bouajjani}, \bibinfo{person}{Egor
  Derevenetc}, {and} \bibinfo{person}{Roland Meyer}.}
  \bibinfo{year}{2013}\natexlab{}.
\newblock \showarticletitle{Checking and Enforcing Robustness against TSO}. In
  \bibinfo{booktitle}{\emph{Programming Languages and Systems}},
  \bibfield{editor}{\bibinfo{person}{Matthias Felleisen} {and}
  \bibinfo{person}{Philippa Gardner}} (Eds.). \bibinfo{publisher}{Springer
  Berlin Heidelberg}, \bibinfo{address}{Berlin, Heidelberg},
  \bibinfo{pages}{533--553}.
\newblock
\showISBNx{978-3-642-37036-6}


\bibitem[\protect\citeauthoryear{Bouajjani, Meyer, and M{\"o}hlmann}{Bouajjani
  et~al\mbox{.}}{2011}]%
        {Bouajjani11}
\bibfield{author}{\bibinfo{person}{Ahmed Bouajjani}, \bibinfo{person}{Roland
  Meyer}, {and} \bibinfo{person}{Eike M{\"o}hlmann}.}
  \bibinfo{year}{2011}\natexlab{}.
\newblock \showarticletitle{Deciding Robustness against Total Store Ordering}.
  In \bibinfo{booktitle}{\emph{Automata, Languages and Programming}},
  \bibfield{editor}{\bibinfo{person}{Luca Aceto}, \bibinfo{person}{Monika
  Henzinger}, {and} \bibinfo{person}{Ji{\v{r}}{\'i} Sgall}} (Eds.).
  \bibinfo{publisher}{Springer Berlin Heidelberg}, \bibinfo{address}{Berlin,
  Heidelberg}, \bibinfo{pages}{428--440}.
\newblock
\showISBNx{978-3-642-22012-8}


\bibitem[\protect\citeauthoryear{Cain and Lipasti}{Cain and Lipasti}{2002}]%
        {Cain02}
\bibfield{author}{\bibinfo{person}{Harold~W. Cain} {and}
  \bibinfo{person}{Mikko~H. Lipasti}.} \bibinfo{year}{2002}\natexlab{}.
\newblock \showarticletitle{Verifying Sequential Consistency Using Vector
  Clocks}. In \bibinfo{booktitle}{\emph{Proceedings of the Fourteenth Annual
  ACM Symposium on Parallel Algorithms and Architectures}} (Winnipeg, Manitoba,
  Canada) \emph{(\bibinfo{series}{SPAA ’02})}.
  \bibinfo{publisher}{Association for Computing Machinery},
  \bibinfo{address}{New York, NY, USA}, \bibinfo{pages}{153–154}.
\newblock
\showISBNx{1581135297}
\urldef\tempurl%
\url{https://doi.org/10.1145/564870.564897}
\showDOI{\tempurl}


\bibitem[\protect\citeauthoryear{Chalupa, Chatterjee, Pavlogiannis, Sinha, and
  Vaidya}{Chalupa et~al\mbox{.}}{2017}]%
        {Chalupa17}
\bibfield{author}{\bibinfo{person}{Marek Chalupa}, \bibinfo{person}{Krishnendu
  Chatterjee}, \bibinfo{person}{Andreas Pavlogiannis}, \bibinfo{person}{Nishant
  Sinha}, {and} \bibinfo{person}{Kapil Vaidya}.}
  \bibinfo{year}{2017}\natexlab{}.
\newblock \showarticletitle{Data-centric Dynamic Partial Order Reduction}.
\newblock \bibinfo{journal}{\emph{Proc. ACM Program. Lang.}}
  \bibinfo{volume}{2}, \bibinfo{number}{POPL}, Article \bibinfo{articleno}{31}
  (\bibinfo{date}{Dec.} \bibinfo{year}{2017}), \bibinfo{numpages}{30}~pages.
\newblock
\showISSN{2475-1421}
\urldef\tempurl%
\url{https://doi.org/10.1145/3158119}
\showDOI{\tempurl}


\bibitem[\protect\citeauthoryear{Chatterjee, Pavlogiannis, and
  Toman}{Chatterjee et~al\mbox{.}}{2019}]%
        {Chatterjee19}
\bibfield{author}{\bibinfo{person}{Krishnendu Chatterjee},
  \bibinfo{person}{Andreas Pavlogiannis}, {and} \bibinfo{person}{Viktor
  Toman}.} \bibinfo{year}{2019}\natexlab{}.
\newblock \showarticletitle{Value-Centric Dynamic Partial Order Reduction}.
\newblock \bibinfo{journal}{\emph{Proc. ACM Program. Lang.}}
  \bibinfo{volume}{3}, \bibinfo{number}{OOPSLA}, Article
  \bibinfo{articleno}{124} (\bibinfo{date}{Oct.} \bibinfo{year}{2019}),
  \bibinfo{numpages}{29}~pages.
\newblock
\urldef\tempurl%
\url{https://doi.org/10.1145/3360550}
\showDOI{\tempurl}


\bibitem[\protect\citeauthoryear{{Chen}, {Yi Lv}, {Hu}, {Chen}, {Haihua Shen},
  {Pengyu Wang}, and {Hong Pan}}{{Chen} et~al\mbox{.}}{2009}]%
        {Chen09}
\bibfield{author}{\bibinfo{person}{Y. {Chen}}, \bibinfo{person}{{Yi Lv}},
  \bibinfo{person}{W. {Hu}}, \bibinfo{person}{T. {Chen}},
  \bibinfo{person}{{Haihua Shen}}, \bibinfo{person}{{Pengyu Wang}}, {and}
  \bibinfo{person}{{Hong Pan}}.} \bibinfo{year}{2009}\natexlab{}.
\newblock \showarticletitle{Fast complete memory consistency verification}. In
  \bibinfo{booktitle}{\emph{2009 IEEE 15th International Symposium on High
  Performance Computer Architecture}}. \bibinfo{pages}{381--392}.
\newblock


\bibitem[\protect\citeauthoryear{Chini and Saivasan}{Chini and
  Saivasan}{2020}]%
        {ChiniS20}
\bibfield{author}{\bibinfo{person}{Peter Chini} {and} \bibinfo{person}{Prakash
  Saivasan}.} \bibinfo{year}{2020}\natexlab{}.
\newblock \showarticletitle{A Framework for Consistency Algorithms}. In
  \bibinfo{booktitle}{\emph{40th {IARCS} Annual Conference on Foundations of
  Software Technology and Theoretical Computer Science, {FSTTCS} 2020, December
  14-18, 2020, {BITS} Pilani, {K} {K} Birla Goa Campus, Goa, India (Virtual
  Conference)}} \emph{(\bibinfo{series}{LIPIcs}, Vol.~\bibinfo{volume}{182})},
  \bibfield{editor}{\bibinfo{person}{Nitin Saxena} {and} \bibinfo{person}{Sunil
  Simon}} (Eds.). \bibinfo{publisher}{Schloss Dagstuhl - Leibniz-Zentrum
  f{\"{u}}r Informatik}, \bibinfo{pages}{42:1--42:17}.
\newblock
\urldef\tempurl%
\url{https://doi.org/10.4230/LIPIcs.FSTTCS.2020.42}
\showDOI{\tempurl}


\bibitem[\protect\citeauthoryear{Clarke, Grumberg, Minea, and Peled}{Clarke
  et~al\mbox{.}}{1999}]%
        {Clarke99}
\bibfield{author}{\bibinfo{person}{E.M. Clarke}, \bibinfo{person}{O. Grumberg},
  \bibinfo{person}{M. Minea}, {and} \bibinfo{person}{D. Peled}.}
  \bibinfo{year}{1999}\natexlab{}.
\newblock \showarticletitle{State space reduction using partial order
  techniques}.
\newblock \bibinfo{journal}{\emph{STTT}} \bibinfo{volume}{2},
  \bibinfo{number}{3} (\bibinfo{year}{1999}), \bibinfo{pages}{279--287}.
\newblock


\bibitem[\protect\citeauthoryear{Correia and Ramalhete}{Correia and
  Ramalhete}{2016}]%
        {Correia16}
\bibfield{author}{\bibinfo{person}{Andreia Correia} {and}
  \bibinfo{person}{Pedro Ramalhete}.} \bibinfo{year}{2016}\natexlab{}.
\newblock \bibinfo{title}{2-thread software solutions for the mutual exclusion
  problem}.
\newblock
  \bibinfo{howpublished}{\url{https://github.com/pramalhe/ConcurrencyFreaks/blob/master/papers/cr2t-2016.pdf}}.
\newblock


\bibitem[\protect\citeauthoryear{Demsky and Lam}{Demsky and Lam}{2015}]%
        {Demsky15}
\bibfield{author}{\bibinfo{person}{Brian Demsky} {and} \bibinfo{person}{Patrick
  Lam}.} \bibinfo{year}{2015}\natexlab{}.
\newblock \showarticletitle{SATCheck: SAT-directed Stateless Model Checking for
  SC and TSO} \emph{(\bibinfo{series}{OOPSLA})}. \bibinfo{publisher}{ACM},
  \bibinfo{address}{New York, NY, USA}, \bibinfo{pages}{20--36}.
\newblock
\showISBNx{978-1-4503-3689-5}
\urldef\tempurl%
\url{https://doi.org/10.1145/2814270.2814297}
\showDOI{\tempurl}


\bibitem[\protect\citeauthoryear{Flanagan and Godefroid}{Flanagan and
  Godefroid}{2005}]%
        {Flanagan05}
\bibfield{author}{\bibinfo{person}{Cormac Flanagan} {and}
  \bibinfo{person}{Patrice Godefroid}.} \bibinfo{year}{2005}\natexlab{}.
\newblock \showarticletitle{Dynamic Partial-order Reduction for Model Checking
  Software}. In \bibinfo{booktitle}{\emph{POPL}}.
\newblock


\bibitem[\protect\citeauthoryear{Furbach, Meyer, Schneider, and
  Senftleben}{Furbach et~al\mbox{.}}{2015}]%
        {Furbach15}
\bibfield{author}{\bibinfo{person}{Florian Furbach}, \bibinfo{person}{Roland
  Meyer}, \bibinfo{person}{Klaus Schneider}, {and} \bibinfo{person}{Maximilian
  Senftleben}.} \bibinfo{year}{2015}\natexlab{}.
\newblock \showarticletitle{Memory-Model-Aware Testing: A Unified Complexity
  Analysis}.
\newblock \bibinfo{journal}{\emph{ACM Trans. Embed. Comput. Syst.}}
  \bibinfo{volume}{14}, \bibinfo{number}{4}, Article \bibinfo{articleno}{63}
  (\bibinfo{date}{Sept.} \bibinfo{year}{2015}), \bibinfo{numpages}{25}~pages.
\newblock
\showISSN{1539-9087}
\urldef\tempurl%
\url{https://doi.org/10.1145/2753761}
\showDOI{\tempurl}


\bibitem[\protect\citeauthoryear{Gibbons and Korach}{Gibbons and
  Korach}{1997}]%
        {Gibbons97}
\bibfield{author}{\bibinfo{person}{Phillip~B. Gibbons} {and}
  \bibinfo{person}{Ephraim Korach}.} \bibinfo{year}{1997}\natexlab{}.
\newblock \showarticletitle{Testing Shared Memories}.
\newblock \bibinfo{journal}{\emph{SIAM J. Comput.}} \bibinfo{volume}{26},
  \bibinfo{number}{4} (\bibinfo{date}{Aug.} \bibinfo{year}{1997}),
  \bibinfo{pages}{1208--1244}.
\newblock
\showISSN{0097-5397}
\urldef\tempurl%
\url{https://doi.org/10.1137/S0097539794279614}
\showDOI{\tempurl}


\bibitem[\protect\citeauthoryear{Godefroid}{Godefroid}{1996}]%
        {G96}
\bibfield{author}{\bibinfo{person}{P. Godefroid}.}
  \bibinfo{year}{1996}\natexlab{}.
\newblock \bibinfo{booktitle}{\emph{Partial-Order Methods for the Verification
  of Concurrent Systems: An Approach to the State-Explosion Problem}}.
\newblock \bibinfo{publisher}{Springer-Verlag}, \bibinfo{address}{Secaucus, NJ,
  USA}.
\newblock


\bibitem[\protect\citeauthoryear{Godefroid}{Godefroid}{1997}]%
        {Godefroid97}
\bibfield{author}{\bibinfo{person}{Patrice Godefroid}.}
  \bibinfo{year}{1997}\natexlab{}.
\newblock \showarticletitle{Model Checking for Programming Languages Using
  VeriSoft}. In \bibinfo{booktitle}{\emph{POPL}}.
\newblock


\bibitem[\protect\citeauthoryear{Godefroid}{Godefroid}{2005}]%
        {Godefroid05}
\bibfield{author}{\bibinfo{person}{Patrice Godefroid}.}
  \bibinfo{year}{2005}\natexlab{}.
\newblock \showarticletitle{Software Model Checking: The VeriSoft Approach}.
\newblock \bibinfo{journal}{\emph{FMSD}} \bibinfo{volume}{26},
  \bibinfo{number}{2} (\bibinfo{year}{2005}), \bibinfo{pages}{77--101}.
\newblock


\bibitem[\protect\citeauthoryear{Herlihy and Wing}{Herlihy and Wing}{1990}]%
        {Herlihy90}
\bibfield{author}{\bibinfo{person}{Maurice~P. Herlihy} {and}
  \bibinfo{person}{Jeannette~M. Wing}.} \bibinfo{year}{1990}\natexlab{}.
\newblock \showarticletitle{Linearizability: A Correctness Condition for
  Concurrent Objects}.
\newblock \bibinfo{journal}{\emph{ACM Trans. Program. Lang. Syst.}}
  \bibinfo{volume}{12}, \bibinfo{number}{3} (\bibinfo{date}{July}
  \bibinfo{year}{1990}), \bibinfo{pages}{463–492}.
\newblock
\showISSN{0164-0925}
\urldef\tempurl%
\url{https://doi.org/10.1145/78969.78972}
\showDOI{\tempurl}


\bibitem[\protect\citeauthoryear{{Hu}, {Chen}, {Chen}, {Qian}, and {Li}}{{Hu}
  et~al\mbox{.}}{2012}]%
        {Hu12}
\bibfield{author}{\bibinfo{person}{W. {Hu}}, \bibinfo{person}{Y. {Chen}},
  \bibinfo{person}{T. {Chen}}, \bibinfo{person}{C. {Qian}}, {and}
  \bibinfo{person}{L. {Li}}.} \bibinfo{year}{2012}\natexlab{}.
\newblock \showarticletitle{Linear Time Memory Consistency Verification}.
\newblock \bibinfo{journal}{\emph{IEEE Trans. Comput.}} \bibinfo{volume}{61},
  \bibinfo{number}{4} (\bibinfo{year}{2012}), \bibinfo{pages}{502--516}.
\newblock


\bibitem[\protect\citeauthoryear{Huang}{Huang}{2015}]%
        {HUANG15}
\bibfield{author}{\bibinfo{person}{Jeff Huang}.}
  \bibinfo{year}{2015}\natexlab{}.
\newblock \showarticletitle{Stateless Model Checking Concurrent Programs with
  Maximal Causality Reduction}. In \bibinfo{booktitle}{\emph{PLDI}}.
\newblock


\bibitem[\protect\citeauthoryear{Huang and Huang}{Huang and Huang}{2016}]%
        {Huang16}
\bibfield{author}{\bibinfo{person}{Shiyou Huang} {and} \bibinfo{person}{Jeff
  Huang}.} \bibinfo{year}{2016}\natexlab{}.
\newblock \showarticletitle{Maximal Causality Reduction for TSO and PSO}.
\newblock \bibinfo{journal}{\emph{SIGPLAN Not.}} \bibinfo{volume}{51},
  \bibinfo{number}{10} (\bibinfo{date}{Oct.} \bibinfo{year}{2016}),
  \bibinfo{pages}{447--461}.
\newblock
\showISSN{0362-1340}
\urldef\tempurl%
\url{https://doi.org/10.1145/3022671.2984025}
\showDOI{\tempurl}


\bibitem[\protect\citeauthoryear{Huang and Huang}{Huang and Huang}{2017}]%
        {Huang017}
\bibfield{author}{\bibinfo{person}{Shiyou Huang} {and} \bibinfo{person}{Jeff
  Huang}.} \bibinfo{year}{2017}\natexlab{}.
\newblock \showarticletitle{Speeding Up Maximal Causality Reduction with Static
  Dependency Analysis}. In \bibinfo{booktitle}{\emph{31st European Conference
  on Object-Oriented Programming, {ECOOP} 2017, June 19-23, 2017, Barcelona,
  Spain}}. \bibinfo{pages}{16:1--16:22}.
\newblock
\urldef\tempurl%
\url{https://doi.org/10.4230/LIPIcs.ECOOP.2017.16}
\showDOI{\tempurl}


\bibitem[\protect\citeauthoryear{Kahlon, Wang, and Gupta}{Kahlon
  et~al\mbox{.}}{2009}]%
        {Kahlon09}
\bibfield{author}{\bibinfo{person}{Vineet Kahlon}, \bibinfo{person}{Chao Wang},
  {and} \bibinfo{person}{Aarti Gupta}.} \bibinfo{year}{2009}\natexlab{}.
\newblock \showarticletitle{Monotonic Partial Order Reduction: An Optimal
  Symbolic Partial Order Reduction Technique}. In
  \bibinfo{booktitle}{\emph{Proceedings of the 21st International Conference on
  Computer Aided Verification}} (Grenoble, France) \emph{(\bibinfo{series}{CAV
  '09})}. \bibinfo{publisher}{Springer-Verlag}, \bibinfo{address}{Berlin,
  Heidelberg}, \bibinfo{pages}{398--413}.
\newblock
\showISBNx{978-3-642-02657-7}
\urldef\tempurl%
\url{https://doi.org/10.1007/978-3-642-02658-4_31}
\showDOI{\tempurl}


\bibitem[\protect\citeauthoryear{Kini, Mathur, and Viswanathan}{Kini
  et~al\mbox{.}}{2017}]%
        {Kini17}
\bibfield{author}{\bibinfo{person}{Dileep Kini}, \bibinfo{person}{Umang
  Mathur}, {and} \bibinfo{person}{Mahesh Viswanathan}.}
  \bibinfo{year}{2017}\natexlab{}.
\newblock \showarticletitle{Dynamic Race Prediction in Linear Time}. In
  \bibinfo{booktitle}{\emph{Proceedings of the 38th ACM SIGPLAN Conference on
  Programming Language Design and Implementation}} (Barcelona, Spain)
  \emph{(\bibinfo{series}{PLDI 2017})}. \bibinfo{publisher}{ACM},
  \bibinfo{address}{New York, NY, USA}, \bibinfo{pages}{157--170}.
\newblock
\showISBNx{978-1-4503-4988-8}
\urldef\tempurl%
\url{https://doi.org/10.1145/3062341.3062374}
\showDOI{\tempurl}


\bibitem[\protect\citeauthoryear{Kokologiannakis, Lahav, Sagonas, and
  Vafeiadis}{Kokologiannakis et~al\mbox{.}}{2017}]%
        {Kokologiannakis17}
\bibfield{author}{\bibinfo{person}{Michalis Kokologiannakis},
  \bibinfo{person}{Ori Lahav}, \bibinfo{person}{Konstantinos Sagonas}, {and}
  \bibinfo{person}{Viktor Vafeiadis}.} \bibinfo{year}{2017}\natexlab{}.
\newblock \showarticletitle{Effective Stateless Model Checking for C/C++
  Concurrency}.
\newblock \bibinfo{journal}{\emph{Proc. ACM Program. Lang.}}
  \bibinfo{volume}{2}, \bibinfo{number}{POPL}, Article \bibinfo{articleno}{17}
  (\bibinfo{date}{Dec.} \bibinfo{year}{2017}), \bibinfo{numpages}{32}~pages.
\newblock
\showISSN{2475-1421}
\urldef\tempurl%
\url{https://doi.org/10.1145/3158105}
\showDOI{\tempurl}


\bibitem[\protect\citeauthoryear{Kokologiannakis, Raad, and
  Vafeiadis}{Kokologiannakis et~al\mbox{.}}{2019a}]%
        {Kokologiannakis19b}
\bibfield{author}{\bibinfo{person}{Michalis Kokologiannakis},
  \bibinfo{person}{Azalea Raad}, {and} \bibinfo{person}{Viktor Vafeiadis}.}
  \bibinfo{year}{2019}\natexlab{a}.
\newblock \showarticletitle{Effective Lock Handling in Stateless Model
  Checking}.
\newblock \bibinfo{journal}{\emph{Proc. ACM Program. Lang.}}
  \bibinfo{volume}{3}, \bibinfo{number}{OOPSLA}, Article
  \bibinfo{articleno}{173} (\bibinfo{date}{Oct.} \bibinfo{year}{2019}),
  \bibinfo{numpages}{26}~pages.
\newblock
\urldef\tempurl%
\url{https://doi.org/10.1145/3360599}
\showDOI{\tempurl}


\bibitem[\protect\citeauthoryear{Kokologiannakis, Raad, and
  Vafeiadis}{Kokologiannakis et~al\mbox{.}}{2019b}]%
        {Kokologiannakis19}
\bibfield{author}{\bibinfo{person}{Michalis Kokologiannakis},
  \bibinfo{person}{Azalea Raad}, {and} \bibinfo{person}{Viktor Vafeiadis}.}
  \bibinfo{year}{2019}\natexlab{b}.
\newblock \showarticletitle{Model Checking for Weakly Consistent Libraries}. In
  \bibinfo{booktitle}{\emph{Proceedings of the 40th ACM SIGPLAN Conference on
  Programming Language Design and Implementation}} (Phoenix, AZ, USA)
  \emph{(\bibinfo{series}{PLDI 2019})}. \bibinfo{publisher}{ACM},
  \bibinfo{address}{New York, NY, USA}, \bibinfo{pages}{96--110}.
\newblock
\showISBNx{978-1-4503-6712-7}
\urldef\tempurl%
\url{https://doi.org/10.1145/3314221.3314609}
\showDOI{\tempurl}


\bibitem[\protect\citeauthoryear{Kokologiannakis and Vafeiadis}{Kokologiannakis
  and Vafeiadis}{2020}]%
        {Kokologiannakis20}
\bibfield{author}{\bibinfo{person}{Michalis Kokologiannakis} {and}
  \bibinfo{person}{Viktor Vafeiadis}.} \bibinfo{year}{2020}\natexlab{}.
\newblock \showarticletitle{{HMC:} Model Checking for Hardware Memory Models}.
  In \bibinfo{booktitle}{\emph{{ASPLOS} '20: Architectural Support for
  Programming Languages and Operating Systems, Lausanne, Switzerland, March
  16-20, 2020}}, \bibfield{editor}{\bibinfo{person}{James~R. Larus},
  \bibinfo{person}{Luis Ceze}, {and} \bibinfo{person}{Karin Strauss}} (Eds.).
  \bibinfo{publisher}{{ACM}}, \bibinfo{pages}{1157--1171}.
\newblock
\urldef\tempurl%
\url{https://doi.org/10.1145/3373376.3378480}
\showDOI{\tempurl}


\bibitem[\protect\citeauthoryear{Lahav, Vafeiadis, Kang, Hur, and Dreyer}{Lahav
  et~al\mbox{.}}{2017}]%
        {LahavVKHD17}
\bibfield{author}{\bibinfo{person}{Ori Lahav}, \bibinfo{person}{Viktor
  Vafeiadis}, \bibinfo{person}{Jeehoon Kang}, \bibinfo{person}{Chung{-}Kil
  Hur}, {and} \bibinfo{person}{Derek Dreyer}.} \bibinfo{year}{2017}\natexlab{}.
\newblock \showarticletitle{Repairing sequential consistency in {C/C++11}}. In
  \bibinfo{booktitle}{\emph{Proceedings of the 38th {ACM} {SIGPLAN} Conference
  on Programming Language Design and Implementation, {PLDI} 2017, Barcelona,
  Spain, June 18-23, 2017}}, \bibfield{editor}{\bibinfo{person}{Albert Cohen}
  {and} \bibinfo{person}{Martin~T. Vechev}} (Eds.). \bibinfo{publisher}{{ACM}},
  \bibinfo{pages}{618--632}.
\newblock
\urldef\tempurl%
\url{https://doi.org/10.1145/3062341.3062352}
\showDOI{\tempurl}


\bibitem[\protect\citeauthoryear{Lamport}{Lamport}{1979}]%
        {Lamport79}
\bibfield{author}{\bibinfo{person}{L. Lamport}.}
  \bibinfo{year}{1979}\natexlab{}.
\newblock \showarticletitle{How to Make a Multiprocessor Computer That
  Correctly Executes Multiprocess Programs}.
\newblock \bibinfo{journal}{\emph{IEEE Trans. Comput.}} \bibinfo{volume}{28},
  \bibinfo{number}{9} (\bibinfo{year}{1979}), \bibinfo{pages}{690--691}.
\newblock


\bibitem[\protect\citeauthoryear{L{\aa}ng and Sagonas}{L{\aa}ng and
  Sagonas}{2020}]%
        {LangS20}
\bibfield{author}{\bibinfo{person}{Magnus L{\aa}ng} {and}
  \bibinfo{person}{Konstantinos Sagonas}.} \bibinfo{year}{2020}\natexlab{}.
\newblock \showarticletitle{Parallel Graph-Based Stateless Model Checking}. In
  \bibinfo{booktitle}{\emph{Automated Technology for Verification and Analysis
  - 18th International Symposium, {ATVA} 2020, Hanoi, Vietnam, October 19-23,
  2020, Proceedings}} \emph{(\bibinfo{series}{Lecture Notes in Computer
  Science}, Vol.~\bibinfo{volume}{12302})},
  \bibfield{editor}{\bibinfo{person}{Dang~Van Hung} {and} \bibinfo{person}{Oleg
  Sokolsky}} (Eds.). \bibinfo{publisher}{Springer}, \bibinfo{pages}{377--393}.
\newblock
\urldef\tempurl%
\url{https://doi.org/10.1007/978-3-030-59152-6\_21}
\showDOI{\tempurl}


\bibitem[\protect\citeauthoryear{Madan~Musuvathi}{Madan~Musuvathi}{2007}]%
        {Musuvathi07b}
\bibfield{author}{\bibinfo{person}{Tom~Ball Madan~Musuvathi, Shaz~Qadeer}.}
  \bibinfo{year}{2007}\natexlab{}.
\newblock \bibinfo{booktitle}{\emph{CHESS: A systematic testing tool for
  concurrent software}}.
\newblock \bibinfo{type}{{T}echnical {R}eport}.
\newblock


\bibitem[\protect\citeauthoryear{{Manovit} and {Hangal}}{{Manovit} and
  {Hangal}}{2006}]%
        {Manovit06}
\bibfield{author}{\bibinfo{person}{C. {Manovit}} {and} \bibinfo{person}{S.
  {Hangal}}.} \bibinfo{year}{2006}\natexlab{}.
\newblock \showarticletitle{Completely verifying memory consistency of test
  program executions}. In \bibinfo{booktitle}{\emph{The Twelfth International
  Symposium on High-Performance Computer Architecture, 2006.}}
  \bibinfo{pages}{166--175}.
\newblock


\bibitem[\protect\citeauthoryear{Mathur, Pavlogiannis, and Viswanathan}{Mathur
  et~al\mbox{.}}{2020}]%
        {Mathur20}
\bibfield{author}{\bibinfo{person}{Umang Mathur}, \bibinfo{person}{Andreas
  Pavlogiannis}, {and} \bibinfo{person}{Mahesh Viswanathan}.}
  \bibinfo{year}{2020}\natexlab{}.
\newblock \showarticletitle{The Complexity of Dynamic Data Race Prediction}. In
  \bibinfo{booktitle}{\emph{Proceedings of the 35th Annual ACM/IEEE Symposium
  on Logic in Computer Science}} (Saarbr\"{u}cken, Germany)
  \emph{(\bibinfo{series}{LICS ’20})}. \bibinfo{publisher}{Association for
  Computing Machinery}, \bibinfo{address}{New York, NY, USA},
  \bibinfo{pages}{713–727}.
\newblock
\showISBNx{9781450371049}
\urldef\tempurl%
\url{https://doi.org/10.1145/3373718.3394783}
\showDOI{\tempurl}


\bibitem[\protect\citeauthoryear{Mathur, Pavlogiannis, and Viswanathan}{Mathur
  et~al\mbox{.}}{2021}]%
        {Mathur21}
\bibfield{author}{\bibinfo{person}{Umang Mathur}, \bibinfo{person}{Andreas
  Pavlogiannis}, {and} \bibinfo{person}{Mahesh Viswanathan}.}
  \bibinfo{year}{2021}\natexlab{}.
\newblock \showarticletitle{Optimal Prediction of Synchronization-Preserving
  Races} \emph{(\bibinfo{series}{POPL '21})}.
\newblock
\newblock
\shownote{To Appear.}


\bibitem[\protect\citeauthoryear{Norris and Demsky}{Norris and Demsky}{2013}]%
        {NorrisD13}
\bibfield{author}{\bibinfo{person}{Brian Norris} {and} \bibinfo{person}{Brian
  Demsky}.} \bibinfo{year}{2013}\natexlab{}.
\newblock \showarticletitle{CDSchecker: checking concurrent data structures
  written with {C/C++} atomics}. In \bibinfo{booktitle}{\emph{OOPSLA}},
  \bibfield{editor}{\bibinfo{person}{Antony~L. Hosking},
  \bibinfo{person}{Patrick~Th. Eugster}, {and} \bibinfo{person}{Cristina~V.
  Lopes}} (Eds.). \bibinfo{publisher}{{ACM}}, \bibinfo{pages}{131--150}.
\newblock
\urldef\tempurl%
\url{https://doi.org/10.1145/2509136.2509514}
\showDOI{\tempurl}


\bibitem[\protect\citeauthoryear{Owens, Sarkar, and Sewell}{Owens
  et~al\mbox{.}}{2009}]%
        {Sewell09}
\bibfield{author}{\bibinfo{person}{Scott Owens}, \bibinfo{person}{Susmit
  Sarkar}, {and} \bibinfo{person}{Peter Sewell}.}
  \bibinfo{year}{2009}\natexlab{}.
\newblock \showarticletitle{A Better x86 Memory Model: x86-TSO}. In
  \bibinfo{booktitle}{\emph{Theorem Proving in Higher Order Logics}},
  \bibfield{editor}{\bibinfo{person}{Stefan Berghofer}, \bibinfo{person}{Tobias
  Nipkow}, \bibinfo{person}{Christian Urban}, {and} \bibinfo{person}{Makarius
  Wenzel}} (Eds.). \bibinfo{publisher}{Springer Berlin Heidelberg},
  \bibinfo{address}{Berlin, Heidelberg}, \bibinfo{pages}{391--407}.
\newblock
\showISBNx{978-3-642-03359-9}


\bibitem[\protect\citeauthoryear{Pavlogiannis}{Pavlogiannis}{2019}]%
        {Pavlogiannis20}
\bibfield{author}{\bibinfo{person}{Andreas Pavlogiannis}.}
  \bibinfo{year}{2019}\natexlab{}.
\newblock \showarticletitle{Fast, Sound, and Effectively Complete Dynamic Race
  Prediction}.
\newblock \bibinfo{journal}{\emph{Proc. ACM Program. Lang.}}
  \bibinfo{volume}{4}, \bibinfo{number}{POPL}, Article \bibinfo{articleno}{17}
  (\bibinfo{date}{Dec.} \bibinfo{year}{2019}), \bibinfo{numpages}{29}~pages.
\newblock
\urldef\tempurl%
\url{https://doi.org/10.1145/3371085}
\showDOI{\tempurl}


\bibitem[\protect\citeauthoryear{Peled}{Peled}{1993}]%
        {Peled93}
\bibfield{author}{\bibinfo{person}{Doron Peled}.}
  \bibinfo{year}{1993}\natexlab{}.
\newblock \showarticletitle{All from One, One for All: On Model Checking Using
  Representatives}. In \bibinfo{booktitle}{\emph{CAV}}.
\newblock


\bibitem[\protect\citeauthoryear{Podkopaev, Lahav, and Vafeiadis}{Podkopaev
  et~al\mbox{.}}{2019}]%
        {PodkopaevLV19}
\bibfield{author}{\bibinfo{person}{Anton Podkopaev}, \bibinfo{person}{Ori
  Lahav}, {and} \bibinfo{person}{Viktor Vafeiadis}.}
  \bibinfo{year}{2019}\natexlab{}.
\newblock \showarticletitle{Bridging the gap between programming languages and
  hardware weak memory models}.
\newblock \bibinfo{journal}{\emph{Proc. {ACM} Program. Lang.}}
  \bibinfo{volume}{3}, \bibinfo{number}{{POPL}} (\bibinfo{year}{2019}),
  \bibinfo{pages}{69:1--69:31}.
\newblock
\urldef\tempurl%
\url{https://doi.org/10.1145/3290382}
\showDOI{\tempurl}


\bibitem[\protect\citeauthoryear{Rodr{\'{\i}}guez, Sousa, Sharma, and
  Kroening}{Rodr{\'{\i}}guez et~al\mbox{.}}{2015}]%
        {Sousa15}
\bibfield{author}{\bibinfo{person}{C{\'{e}}sar Rodr{\'{\i}}guez},
  \bibinfo{person}{Marcelo Sousa}, \bibinfo{person}{Subodh Sharma}, {and}
  \bibinfo{person}{Daniel Kroening}.} \bibinfo{year}{2015}\natexlab{}.
\newblock \showarticletitle{Unfolding-based Partial Order Reduction}. In
  \bibinfo{booktitle}{\emph{CONCUR}}.
\newblock


\bibitem[\protect\citeauthoryear{Roemer, Gen\c{c}, and Bond}{Roemer
  et~al\mbox{.}}{2020}]%
        {Roemer20}
\bibfield{author}{\bibinfo{person}{Jake Roemer}, \bibinfo{person}{Kaan
  Gen\c{c}}, {and} \bibinfo{person}{Michael~D. Bond}.}
  \bibinfo{year}{2020}\natexlab{}.
\newblock \showarticletitle{SmartTrack: Efficient Predictive Race Detection}.
  In \bibinfo{booktitle}{\emph{Proceedings of the 41st ACM SIGPLAN Conference
  on Programming Language Design and Implementation}} (London, UK)
  \emph{(\bibinfo{series}{PLDI 2020})}. \bibinfo{publisher}{Association for
  Computing Machinery}, \bibinfo{address}{New York, NY, USA},
  \bibinfo{pages}{747–762}.
\newblock
\showISBNx{9781450376136}
\urldef\tempurl%
\url{https://doi.org/10.1145/3385412.3385993}
\showDOI{\tempurl}


\bibitem[\protect\citeauthoryear{Sewell, Sarkar, Owens, Nardelli, and
  Myreen}{Sewell et~al\mbox{.}}{2010}]%
        {Sewell10}
\bibfield{author}{\bibinfo{person}{Peter Sewell}, \bibinfo{person}{Susmit
  Sarkar}, \bibinfo{person}{Scott Owens}, \bibinfo{person}{Francesco~Zappa
  Nardelli}, {and} \bibinfo{person}{Magnus~O. Myreen}.}
  \bibinfo{year}{2010}\natexlab{}.
\newblock \showarticletitle{X86-TSO: A Rigorous and Usable Programmer's Model
  for x86 Multiprocessors}.
\newblock \bibinfo{journal}{\emph{Commun. ACM}} \bibinfo{volume}{53},
  \bibinfo{number}{7} (\bibinfo{date}{July} \bibinfo{year}{2010}),
  \bibinfo{pages}{89--97}.
\newblock
\showISSN{0001-0782}
\urldef\tempurl%
\url{https://doi.org/10.1145/1785414.1785443}
\showDOI{\tempurl}


\bibitem[\protect\citeauthoryear{Shasha and Snir}{Shasha and Snir}{1988}]%
        {Shasha88}
\bibfield{author}{\bibinfo{person}{Dennis Shasha} {and} \bibinfo{person}{Marc
  Snir}.} \bibinfo{year}{1988}\natexlab{}.
\newblock \showarticletitle{Efficient and Correct Execution of Parallel
  Programs That Share Memory}.
\newblock \bibinfo{journal}{\emph{ACM Trans. Program. Lang. Syst.}}
  \bibinfo{volume}{10}, \bibinfo{number}{2} (\bibinfo{date}{April}
  \bibinfo{year}{1988}), \bibinfo{pages}{282--312}.
\newblock
\showISSN{0164-0925}
\urldef\tempurl%
\url{https://doi.org/10.1145/42190.42277}
\showDOI{\tempurl}


\bibitem[\protect\citeauthoryear{Smaragdakis, Evans, Sadowski, Yi, and
  Flanagan}{Smaragdakis et~al\mbox{.}}{2012}]%
        {Smaragdakis12}
\bibfield{author}{\bibinfo{person}{Yannis Smaragdakis}, \bibinfo{person}{Jacob
  Evans}, \bibinfo{person}{Caitlin Sadowski}, \bibinfo{person}{Jaeheon Yi},
  {and} \bibinfo{person}{Cormac Flanagan}.} \bibinfo{year}{2012}\natexlab{}.
\newblock \showarticletitle{Sound Predictive Race Detection in Polynomial
  Time}. In \bibinfo{booktitle}{\emph{Proceedings of the 39th Annual ACM
  SIGPLAN-SIGACT Symposium on Principles of Programming Languages}}
  (Philadelphia, PA, USA) \emph{(\bibinfo{series}{POPL '12})}.
  \bibinfo{publisher}{ACM}, \bibinfo{address}{New York, NY, USA},
  \bibinfo{pages}{387--400}.
\newblock
\showISBNx{978-1-4503-1083-3}
\urldef\tempurl%
\url{https://doi.org/10.1145/2103656.2103702}
\showDOI{\tempurl}


\bibitem[\protect\citeauthoryear{SPARC~International}{SPARC~International}{1994}]%
        {SPARCInternational94}
\bibfield{author}{\bibinfo{person}{CORPORATE SPARC~International, Inc.}}
  \bibinfo{year}{1994}\natexlab{}.
\newblock \bibinfo{booktitle}{\emph{The SPARC Architecture Manual (Version
  9)}}.
\newblock \bibinfo{publisher}{Prentice-Hall, Inc.}, \bibinfo{address}{Upper
  Saddle River, NJ, USA}.
\newblock
\showISBNx{0-13-099227-5}


\bibitem[\protect\citeauthoryear{Zennou, Bouajjani, Enea, and Erradi}{Zennou
  et~al\mbox{.}}{2019}]%
        {ZennouBEE19}
\bibfield{author}{\bibinfo{person}{Rachid Zennou}, \bibinfo{person}{Ahmed
  Bouajjani}, \bibinfo{person}{Constantin Enea}, {and}
  \bibinfo{person}{Mohammed Erradi}.} \bibinfo{year}{2019}\natexlab{}.
\newblock \showarticletitle{Gradual Consistency Checking}. In
  \bibinfo{booktitle}{\emph{Computer Aided Verification - 31st International
  Conference, {CAV} 2019, New York City, NY, USA, July 15-18, 2019,
  Proceedings, Part {II}}} \emph{(\bibinfo{series}{Lecture Notes in Computer
  Science}, Vol.~\bibinfo{volume}{11562})},
  \bibfield{editor}{\bibinfo{person}{Isil Dillig} {and} \bibinfo{person}{Serdar
  Tasiran}} (Eds.). \bibinfo{publisher}{Springer}, \bibinfo{pages}{267--285}.
\newblock
\urldef\tempurl%
\url{https://doi.org/10.1007/978-3-030-25543-5\_16}
\showDOI{\tempurl}


\bibitem[\protect\citeauthoryear{Zhang, Kusano, and Wang}{Zhang
  et~al\mbox{.}}{2015}]%
        {Zhang15}
\bibfield{author}{\bibinfo{person}{Naling Zhang}, \bibinfo{person}{Markus
  Kusano}, {and} \bibinfo{person}{Chao Wang}.} \bibinfo{year}{2015}\natexlab{}.
\newblock \showarticletitle{Dynamic Partial Order Reduction for Relaxed Memory
  Models}. In \bibinfo{booktitle}{\emph{PLDI}}.
\newblock


\end{thebibliography}

%% Appendix
\newpage
\appendix
\section{Details of {\cref{sec:verifyingtsopso}}}\label{sec:app_verifyingtsopso}

% Here we provide details to reads-from verification of $\TSO$ and $\PSO$ executions.
% We conclude this section with some insights on the relationship between $\VTSOrm$ and $\VPSOrm$.
% We start with some insights on the relationship between $\VTSOrm$ and $\VPSOrm$,
Here we proceed with proofs of our theorems and lemmas
regarding algorithms $\AlgoTSO$ and $\AlgoPSO$.
Then, we describe an extension of $\AlgoPSO$ to handle store-store fences.

\subsection{Proofs of {\cref{subsec:verifyingtso}}}\label{subsec:app_verifyingtso}

\smallskip
\lemverifyingtsocorrectness*
\begin{proof}
We argue separately about soundness and completeness.

\noindent{\em Soundness.}
We prove by induction that every trace $\Trace$ extracted from $\Worklist$ in \cref{algo:tso_extract_worklist} is a trace that realizes $(X\Project\Events{\Trace}, \Observation\Project\Events{\Trace})$ under $\TSO$.
The claim clearly holds for $\Trace=\epsilon$.
Now consider a trace $\Trace$ such that $\Trace\neq \emptyset$, hence $\Trace$ was inserted in $\Worklist$ in \cref{algo:tso_insert_worklist} while executing a previous iteration of the while-loop in \cref{algo:tso_main_while}.
Let $\Trace'$ be the trace that was extracted from $\Worklist$ in that iteration.
Observe that $\Trace'$ is extended with $\TSO$-executable events in \cref{algo:tso_local_extend} and \cref{algo:tso_mem_extend}, hence it is well-formed.
It remains to argue that for every new read  $\Read$ executed in  \cref{algo:tso_local_extend}, we have  $\Observation_{\Trace'}(\Read)=\Observation(\Read)$.
Assume towards contradiction otherwise, and let $\Read$ be the first read for which this equality fails.
For the remaining of the proof, we let $\Trace'$ be the trace in the iteration of \cref{algo:tso_local_extend} that executed $\Read$, i.e., $\Trace'$ ends in $\Read$.
Let $\Observation(\Read)=(\WriteB, \WriteM)$ and  $\Observation_{\Trace'}(\Read)=(\WriteB', \WriteM')$.
We distinguish the following cases.
\begin{enumerate}
\item If $\Read$ reads-from $\WriteB'$ in $\Trace'$,  then $\Proc{\Read}\neq \Proc{\WriteB}$, while also $\WriteM'\not \in \Events{\Trace}$.
Since $\Read$ became $\TSO$-executable, we have $\WriteM\in \Events{\Trace'}$, hence $\WriteM$ has already become $\TSO$-executable.
This violates \cref{item:tso_execw2} of the definition of $\TSO$-executable memory-writes for $\WriteM$, a contradiction.
\item If $\Read$ reads-from $\WriteM'$ in $\Trace'$, then $\WriteM\in \Events{\Trace'}$ and $\WriteM'$ was executed after $\WriteM$ was executed in $\Trace'$.
This violates \cref{item:tso_execw1} of the definition of $\TSO$-executable memory-writes for $\WriteM'$, a contradiction.
\end{enumerate}
It follows that $\Observation_{\Trace'}(\Read)=\Observation(\Read)$ for all reads $\Read\in \Reads{\Trace'}$, and hence $\Trace'$ realizes $(X\Project\Events{\Trace'}, \Observation\Project\Events{\Trace'})$ under $\TSO$.
The above soundness argument carries over to executions
containing RMW and CAS instructions, since
(i) such instructions are
modeled by events of already considered types
(c.f.~\cref{subsec:verifying_casrmw}), while respecting the
$\TSO$-executability requirements of these events
(as were defined in~\cref{subsec:verifyingtso}), and
(ii) in \cref{algo:tso_local_extend_loop} resp.
\cref{algo:tso_mem_execute} we only consider
$\TSO$-executable atomic blocks (described in detail
in \cref{subsec:verifying_casrmw}).

\noindent{\em Completeness.}
Consider any trace $\Trace^*$ that realizes $(X, \Observation)$.
We show by induction that for every prefix $\ov{\Trace}$ of $\Trace^*$,
the algorithm examines a trace $\Trace$ in \cref{algo:tso_extract_worklist} such that
(i)~$\WritesM{\ov{\Trace}}=\WritesM{\Trace}$, and
(ii)~$\LocalEvents{\ov{\Trace}}\subseteq \LocalEvents{\Trace}$.
The proof is by induction on the number of memory-writes of $\ov{\Trace}$.
Let $\ov{\Trace}=\ov{\Trace}'\Concat\Sequence \Concat\WriteM$, where $\Sequence$ is a sequence of thread events.
Assume by the induction hypothesis that the algorithm extracts a trace $\Trace'$ in \cref{algo:tso_extract_worklist} such that
(i)~$\WritesM{\ov{\Trace}'}=\WritesM{\Trace'}$, and
(ii)~$\LocalEvents{\ov{\Trace'}}\subseteq \LocalEvents{\Trace'}$.
(note that the statement clearly holds for the base case where $\ov{\Trace}'=\epsilon$).
By a straightforward induction, all the events of $\Sequence$
not already present in $\Trace'$ become eventually $\TSO$-executable in $\Trace'$, and thus appended in $\Trace'$,
as the algorithm executes the while-loop in \cref{algo:tso_local_extend_loop}.
Hence, at the end of this while-loop, we have
(i)~$\WritesM{\ov{\Trace}'}=\WritesM{\Trace'}$, and
(ii)~$\LocalEvents{\ov{\Trace}'}\cup \Events{\Sequence}\subseteq \LocalEvents{\Trace'}$.

It remains to argue that $\WriteM$ is $\TSO$-executable in $\Trace'$ at this point
(i.e., in \cref{algo:tso_mem_execute}).
Assume towards contradiction otherwise, hence one of the following hold.
\begin{enumerate}
\item There is a read $\Read\in \Reads{X}$ with $\Observation(\Read)=(\WriteB', \WriteM')$ and such that
(i)~$\Confl{\Read}{\WriteM}$,
(ii)~$\WriteM\neq \WriteM'$,
(iii)~$\WriteM'\in \Trace'$, and
(iv)~$\Read\not \in \Trace'$.
By the induction hypothesis, we have $\WritesM{\Trace'}=\WritesM{\ov{\Trace}'}$ and thus $\WriteM'\in \ov{\Trace}'$.
Moreover, we have $\Events{\ov{\Trace}'\Concat\Sequence}\subseteq \Events{\Trace'}$, and thus $\Read\not \in \ov{\Trace}'\Concat\Sequence$.
This violates the fact that $\ov{\Trace}$ is a witness prefix for $(X,\Observation)$.
\item There is a read $\Read\in \Reads{X}$ with $\Observation(\Read)=(\WriteB, \WriteM)$ and such that
there exists a two-phase write $(\WriteB', \WriteM')$  with
(i)~$\Confl{\Read}{\WriteB'}$,
(ii)~$\WriteB'<_{\TO}\Read$,
(iii)~$\WriteM'\not \in \Trace'$.
By the induction hypothesis, we have $\WritesM{\Trace'}=\WritesM{\ov{\Trace}'}$ and thus  $\WriteM'\not \in \ov{\Trace}'$.
Moreover, we have $\Events{\ov{\Trace}'\Concat\Sequence}\subseteq \Events{\Trace'}$, and thus $\Read\not \in \ov{\Trace}'\Concat\Sequence$.
This violates the fact that $\ov{\Trace}$ is a witness prefix for $(X,\Observation)$.
\end{enumerate}
Hence $\WriteM$ is $\TSO$-executable in $\Trace'$ in
\cref{algo:tso_mem_execute}, and thus the algorithm will construct the
trace $\Trace'_{\WriteM}=\Trace'\Concat \WriteM$ in \cref{algo:tso_mem_extend}.
If $\WritesM{\Trace'_{\WriteM}}\not \in \DoneSet$, the test in
\cref{algo:tso_if_new} succeeds, and the statement holds for $\Trace$
being $\Trace'_{\WriteM}$ extracted from $\Worklist$ in a later iteration.
Otherwise, the algorithm previously constructed a trace $\Trace''$
with $\WritesM{\Trace''}=\WritesM{\Trace'_{\WriteM}}$, and the statement
holds for $\Trace$ being $\Trace''$ extracted from $\Worklist$ in a later iteration.

When arguing about completeness in the presence of
RMW and CAS instructions, additional care needs to be taken, as follows.
The above induction argument applies, but it needs to
additionally consider a case with
$\ov{\Trace}=\ov{\Trace}'\Concat\Sequence \Concat
\Read \Concat \WriteB \Concat \WriteM$ and
$\Fence \in \Events{\ov{\Trace}'\Concat\Sequence}$,
where $\Fence$ together with $\Read \Concat \WriteB \Concat \WriteM$
represent an atomic RMW resp. CAS instruction with the write-part
designated to be immediately propagated to the shared memory. Let us
consider this case in what follows.

As above, we start with the induction hypothesis that
in \cref{algo:tso_extract_worklist} we have $\Trace'$ with
(i)~$\WritesM{\ov{\Trace}'}=\WritesM{\Trace'}$, and
(ii)~$\LocalEvents{\ov{\Trace'}}\subseteq \LocalEvents{\Trace'}$.
Further, by an argument similar to the above, we reach
\cref{algo:tso_mem_execute} where $\Trace'$ now satisfies
(i)~$\WritesM{\ov{\Trace}'}=\WritesM{\Trace'}$, and
(ii)~$\LocalEvents{\ov{\Trace}'}\cup \Events{\Sequence}\subseteq \LocalEvents{\Trace'}$.
At this point, we have $\Fence \in \Events{\ov{\Trace}'\Concat\Sequence}$
and $\Fence \in \LocalEvents{\Trace'}$.
Further, since in our approach we emplace $\Read$, $\WriteB$ and
$\WriteM$ in an atomic block, and we never allow execution of a
singular event that is part of some atomic block
(described in detail in~\cref{subsec:verifying_casrmw}),
we have that $\Events{\Trace'} \cap \{ \Read, \WriteB, \WriteM \} = \emptyset$.
As a result, since
there are no events between $\Fence$ and $\Read$ in the thread
of the atomic instruction, we have that the buffer of the thread
of the atomic instruction is empty in both
$\ov{\Trace}'\Concat\Sequence$ and $\Trace'$.
What remains to argue is that the atomic block
$\Read \Concat \WriteB \Concat \WriteM$
is $\TSO$-executable in $\Trace'$.
For this, we refer to the $\TSO$-executable conditions of atomic blocks
defined in~\cref{subsec:verifying_casrmw}. In turn, utilizing the
$\TSO$-executable conditions of (i) reads, (ii) buffer-writes, and
(iii) memory-writes, defined in~\cref{subsec:verifyingtso},
we show that
(i) $\Read$ is $\TSO$-executable in $\Trace'$,
(ii) $\WriteB$ is $\TSO$-executable in $\Trace' \Concat \Read$, and
(iii) $\WriteM$ is $\TSO$-executable in $\Trace' \Concat \Read \Concat \WriteB$.
This together with $\Events{\Trace'} \cap \{ \Read, \WriteB, \WriteM \} = \emptyset$
gives us that the atomic block
$\Read \Concat \WriteB \Concat \WriteM$
is $\TSO$-executable in $\Trace'$,
and thus in \cref{algo:tso_mem_extend} the algorithm will construct the
trace $\Trace''=\Trace'\Concat \Read \Concat \WriteB \Concat \WriteM$.

The desired completeness result follows.
\end{proof}

We conclude the section with the proof of \cref{them:vtso}.

\themvtso*

\smallskip
\begin{proof}%[Proof of \cref{them:vtso}.]
\cref{lem:verifyingtso_correctness} establishes the correctness, so here we focus on the complexity,
and the following argument applies also to executions containing
RMW and CAS instructions.

Since there are $k$ threads, there exist at most $n^k$ distinct traces $\Trace_1, \Trace_2$ with $\WritesM{\Trace_1}\neq \WritesM{\Trace_2}$.
Hence, the main loop in \cref{algo:tso_main_while} is executed at  most $ n^k$ times.
For each of the $\leq n^k$ traces $\Trace$ inserted in $\Worklist$ in \cref{algo:tso_insert_worklist}, there exist at most $k-1$ traces that are not inserted in $\Worklist$ because $\WritesM{\Trace}=\WritesM{\Trace'}$ (hence the test  in \cref{algo:tso_if_new} fails).
Hence, the algorithm handles $O(k\cdot n^{k})$ traces in total, while each trace is constructed in $O(n)$ time.
Thus, the complexity of $\AlgoTSO$ is $O(k\cdot n^{k+1})$.
The desired result follows.
\end{proof}

\subsection{Proofs of {\cref{subsec:verifyingpso}}}\label{subsec:app_verifyingpso}

\smallskip
\lemfencemapmonotonicityone*
\begin{proof}
Since $\WriteM$ is executable in $\Trace_1$, the variable $\Location{\Trace_1}$ is not held in $\Trace_1$.
It follows directly from the definition of fence maps that the read sets $A_{\Process, \Process'}$ can only increase in $\FenceMap_{\Trace_2}$ compared to $\FenceMap_{\Trace_1}$.
Hence,  $\FenceMap_{\Trace_1}(\Process_1, \Process_2)\leq \FenceMap_{\Trace_2}(\Process_1, \Process_2)$ for all $\Process_1, \Process_2$.
Moreover, if $\WriteM$ is spurious then the sets $A_{\Process, \Process'}$ are identical, thus $\FenceMap_{\Trace_1}(\Process_1, \Process_2)=\FenceMap_{\Trace_2}(\Process_1, \Process_2)$ for all $\Process_1, \Process_2$.

The desired result follows.
\end{proof}

\smallskip
\lemfencemapmonotonicitytwo*
\begin{proof}
We distinguish cases based on the type of $\Event$.
\begin{enumerate}
\item \emph{If $\Event$ is a fence $\Fence$}, the fence maps do not chance, hence the claim holds directly from the fact that $\FenceMap_{\Trace_1}\leq \FenceMap_{\Trace_2}$.

\item \emph{If $\Event$ is a read $\Read$}, observe that $\FenceMap_{\Trace'_i}\leq \FenceMap_{\Trace_i}$ for each $i\in [2]$.
Hence we must have $\FenceMap_{\Trace'_2}(\Process_1, \Process_2)<\FenceMap_{\Trace_2}(\Process_1, \Process_2)$,
for some thread $\Process_1\in \Threads$ and $\Process_2=\Proc{\Read}$.
Note that in fact $\FenceMap_{\Trace'_2}(\Process_1, \Process_2)=0$, which occurs because $\FenceMap_{\Trace_2}(\Process_1, \Process_2)$
is the index of $\Read$ in $\Process_2$.
Since $\FenceMap_{\Trace_1}\leq \FenceMap_{\Trace_2}$, we have either $\FenceMap_{\Trace_1}(\Process_1,\Process_2)=0$ or $\FenceMap_{\Trace_1}(\Process_1,\Process_2)=\FenceMap_{\Trace_2}(\Process_1, \Process_2)$.
In either case, we have $\FenceMap_{\Trace'_1}\leq \FenceMap_{\Trace_2}=0$, a contradiction.

\item \emph{If $\Event$ is a buffer-write $\WriteB$}, observe that $\FenceMap_{\Trace_i}\leq \FenceMap_{\Trace'_i}$ for each $i\in [2]$.
Hence we must have $\FenceMap_{\Trace_1}(\Process_1, \Process_2)<\FenceMap_{\Trace'_1}(\Process_1, \Process_2)$, where $\Process_1=\Proc{\WriteB}$ and $\Process_2$ is some other thread.
It follows that $v=\Location{\WriteB}$ is held in $\Trace_1$ by an active memory-write $\WriteM'$ (thus $\WriteM'$ is not spurious in $\Trace_1$),
and $\FenceMap_{\Trace'_1}(\Process_1, \Process_2)$ is the index of $\Process_2$ that contains a read $\Read$ with
$\Observation(\Read)= (\_, \WriteM')$.
Since $\WritesM{\Trace_1}\setminus\SpuriousWritesM{\Trace_1}\subseteq \WritesM{\Trace_2}$, we have $\WritesM'\in \Trace_2$
Since $\LocalEvents{\Trace_1}=\LocalEvents{\Trace_2}$, we have that $\WriteM'$ is an active memory-write in $\Trace_2$.
Hence $\FenceMap_{\Trace'_2}(\Process_1, \Process_2)\geq \FenceMap_{\Trace'_1}(\Process_1, \Process_2)$, a contradiction.
\end{enumerate}

The desired result follows.
\end{proof}

\smallskip
\lemfencemap*
\begin{proof}
Given a trace $\Trace$, we define the \emph{non-empty-buffer map} $\NEBMap_{\Trace}\colon\Threads\times\Globals\to \{\True, \False \}$,
such that $\NEBMap_{\Trace}(\Process,v)=\True$ iff
(i)~$\Process$ does not hold variable $v$, and
(ii)~the buffer of thread $\Process$ on variable $v$ is non-empty.
Clearly there exist at most $2^{k\cdot d}$ different non-empty-buffer maps.
We argue that for every two traces $\Trace_1, \Trace_2$, if  $\LocalEvents{\Trace_1}=\LocalEvents{\Trace_2}$ and $\NEBMap_{\Trace_1}=\NEBMap_{\Trace_2}$
then $\FenceMap_{\Trace_1}= \FenceMap_{\Trace_2}$, from which the $2^{k\cdot d}$ bound of the lemma follows.

Assume towards contradiction that $\FenceMap_{\Trace_1}\neq \FenceMap_{\Trace_2}$.
Hence, wlog, there exist two threads $\Process_1, \Process_2$ such that $\FenceMap_{\Trace_2}(\Process_1, \Process_2) > \FenceMap_{\Trace_1}(\Process_1, \Process_2)$.
Let $\FenceMap_{\Trace_2}(\Process_1, \Process_2)=m$, and consider the read $\Read$ of $\Process_2$ at index $m$.
Let $v=\Location{\Read}$ and  $\Observation(\Read)=(\WriteB, \WriteM)$ and $\Process_3=\Proc{\WriteB}$.
By the definition of fence maps, we have that $\Process_3$ holds variable $v$ in $\Trace_2$.
By the definition of non-empty-buffer maps, we have that $\NEBMap_{\Trace_2}(\Process_3,v)=\False$,
and since $\NEBMap_{\Trace_1}=\NEBMap_{\Trace_2}$, we also have $\NEBMap_{\Trace_1}(\Process_3,v)=\False$.
Since $\LocalEvents{\Trace_1}=\LocalEvents{\Trace_2}$, we have that $\WriteB\in \Trace_1$.
Moreover, we have $\WriteM\not \in\Trace_1$, as otherwise, since $\NEBMap_{\Trace_1}(\Process_1)=\NEBMap_{\Trace_1}(\Process_2)$,
we would have $\FenceMap_{\Trace_1}(\Process_1, \Process_2)\geq m$.
Hence, the buffer of thread $\Process_3$ on variable $v$ is non-empty in $\Trace_1$.
Since $\NEBMap_{\Trace_1}(\Process_3,v)=\False$, we have that $\Process_3$ holds $v$ in $\Trace_1$.
Thus, there is a read $\Read'\not \in \Trace_1$ with $\Observation(\Read')=(\WriteB', \WriteM')$, where $\WriteM'<_{\TO}\WriteM$.
Since $\LocalEvents{\Trace_1}=\LocalEvents{\Trace_2}$, we have that $\Read'\not \in \Trace_2$, which violates the observation of $\Read'$ in any extension of $\Trace_2$.

The desired result follows.
\end{proof}

\smallskip
\lemverifyingpsocorrectness*
\begin{proof}
We argue separately about soundness and completeness.

\noindent{\em Soundness.}
We prove by induction that every trace $\Trace$ extracted from $\Worklist$ in \cref{algo:pso_extract_worklist} is a trace that realizes $(X\Project\Events{\Trace}, \Observation\Project\Events{\Trace})$ under $\PSO$.
The claim clearly holds for $\Trace=\epsilon$.
Now consider a trace $\Trace$ such that $\Trace\neq \emptyset$, hence $\Trace$ was inserted in $\Worklist$ in \cref{algo:pso_insert_worklist} while executing a previous iteration of the while-loop in \cref{algo:pso_main_while}.
Let $\Trace'$ be the trace that was extracted from $\Worklist$ in that iteration, and consider the trace $\Trace_{\Event}$ constructed in \cref{algo:pso_execute_event}.
Since $\Trace_{\Event}$ is obtained by extending $\Trace'$ with $\PSO$-executable events, it follows that $\Trace_{\Event}$ is well-formed.
It remains to argue that $\Observation_{\Trace}\subseteq \Observation$.
If $\Event$ is not a read, then the claim holds by the induction hypothesis as $\Reads{\Trace_{\Event}}=\Reads{\Trace'}$.
Now assume that $\Event$ is a read with $\Observation_{\Trace_{\Event}}(\Event)=(\WriteB', \WriteM')$.
Let $\Observation(\Event)=(\WriteB, \WriteM)$, and assume towards contradiction that $\WriteB\neq \WriteB$.
We distinguish the following cases.
\begin{enumerate}
\item If $\Event$ reads-from $\WriteB'$ in $\Trace_{\Event}$, we have that $\Proc{\Read}\neq \Proc{\WriteB}$.
But then $\WriteM\in \Trace_{\Event}$, hence $\WriteM$ has become $\PSO$-executable, and thus $\WriteM'\in \Trace_{\Event}$.
Since $\Event$ is the last event of $\Trace_{\Event}$ this violates the fact that $\Event$ reads-from $\WriteB'$ in $\Trace_{\Event}$.
\item If $\Event$ reads-from $\WriteM'$ in $\Trace_{\Event}$, then $\WriteM\in \Trace'$, and $\WriteM'$ was executed after $\WriteM$ in $\Trace'$.
By the definition of $\PSO$-executable events, $\WriteM'$ could not have been $\PSO$-executable at that point, a contradiction.
\end{enumerate}
It follows that $\Observation_{\Trace_{\Event}}(\Event)=\Observation(\Event)$,
and this with the induction hypothesis gives us that
$\Observation_{\Trace_{\Event}}(\Read)=\Observation(\Read)$ for all reads $\Read\in \Reads{\Trace_{\Event}}$.
As a result, $\Trace_{\Event}$ realizes $(X\Project\Events{\Trace_{\Event}}, \Observation\Project\Events{\Trace_{\Event}})$ under $\PSO$.
The soundness argument carries over directly to executions
containing RMW and CAS instructions. Indeed,
since in \cref{algo:pso_ifextend} we only consider
atomic blocks that are $\PSO$-executable
(described in~\cref{subsec:verifying_casrmw}), consequently,
the $\PSO$-executable conditions of fences, reads, buffer-writes and
memory-writes (as defined in~\cref{subsec:verifyingpso}) used to model
RMW and CAS are preserved, which by the above argument implies
soundness.

\noindent{\em Completeness.}
Consider any trace $\Trace^*$ that realizes $(X, \Observation)$.
We show by induction that for every prefix $\ov{\Trace}$ of $\Trace^*$,
the algorithm examines a trace $\Trace$ in \cref{algo:pso_extract_worklist} such that
(i)~$\LocalEvents{\Trace}= \LocalEvents{\ov{\Trace}}$,
(ii)~$\WritesM{\Trace}\setminus\SpuriousWritesM{\Trace}\subseteq \WritesM{\ov{\Trace}}$, and
(iii)~$\FenceMap_{\Trace}\leq \FenceMap_{\ov{\Trace}}$.

The proof is by induction on the number of thread events of $\ov{\Trace}$.
The statement clearly holds when $\ov{\Trace}=\epsilon$ due to the initialization of $\Worklist$.
For the inductive step, let $\ov{\Trace}=\ov{\Trace}'\Concat \Sequence \Concat\Event$,
where $\Sequence$ is a sequence of memory-writes and $\Event$ is a thread event.
By the induction hypothesis, the algorithm extracts a trace $\Trace'$ in \cref{algo:pso_extract_worklist} such that
(i)~$\LocalEvents{\Trace'}= \LocalEvents{\ov{\Trace}'}$,
(ii)~$\WritesM{\Trace'}\setminus\SpuriousWritesM{\Trace'}\subseteq \WritesM{\ov{\Trace}'}$,
(iii)~$\FenceMap_{\Trace'}\leq \FenceMap_{\ov{\Trace}'} $.
Let $\ov{\Trace}_1=\ov{\Trace}'\Concat\Sequence$, and
$\Trace_1$  be the trace $\Trace'$ after the algorithm has extended $\Trace'$ with all  events in the while-loop of \cref{algo:pso_while_spurious}.
By \cref{lem:fencemap_monotonicity1}, we have $\FenceMap_{\ov{\Trace}'}\leq \FenceMap_{\ov{\Trace}_1}$.
Since all events appended to $\Trace'$ are spurious memory-writes in $\Trace'$, by \cref{lem:fencemap_monotonicity1}, we have $\FenceMap_{\Trace_1}= \FenceMap_{\Trace'}$ and thus $\FenceMap_{\Trace_1}\leq  \FenceMap_{\ov{\Trace}_1}$.
Moreover, since the while-loop only appends spurious memory-writes to $\Trace'$, we have
$\WritesM{\Trace_1}\setminus\SpuriousWritesM{\Trace_1}\subseteq \WritesM{\ov{\Trace}_1}$.
Finally, we trivially have $\LocalEvents{\Trace_1}=\LocalEvents{\ov{\Trace}_1}$.

We now argue that $\Event$ is $\PSO$-executable in $\Trace_1$ in \cref{algo:pso_ifextend}, and the statement holds for the new trace $\Trace_{\Event}$ constructed in \cref{algo:pso_execute_event}.
We distinguish cases based on the type of $\Event$.
\begin{enumerate}
\item\label{item:pso_compl_bw} \emph{If $\Event$ is a buffer-write}, then $\Events{\Trace_1}\cup \{\Event \}$ is a lower set of $(X,\Observation)$, hence
$\Event$ is $\PSO$-executable in $\Trace_1$.
Thus, we have $\LocalEvents{\ov{\Trace}}=\LocalEvents{\Trace_{\Event}}$.
Moreover, note that $\Trace_{\Event}=\Trace_1\Concat \Event$ and $\ov{\Trace}=\ov{\Trace}_1\Concat \Event$.
By \cref{lem:fencemap_monotonicity2} on $\Trace_1$ and $\ov{\Trace}_1$, we have $\FenceMap_{\Trace_{\Event}}\leq \FenceMap_{\ov{\Trace}}$.
Finally, we have $\WritesM{\Trace_{\Event}} = \WritesM{\Trace_1}$
and thus $\WritesM{\Trace_{\Event}} \setminus \SpuriousWritesM{\Trace_{\Event}} \subseteq \WritesM{\ov{\Trace}}$.
\item\label{item:pso_compl_read} \emph{If $\Event$ is a read}, let $\Observation(\Event)=(\WriteB, \WriteM)$ and $v=\Location{\Event}$.
We have $\WriteB\in \ov{\Trace}_1$ and thus $\WriteB\in \Trace_1$.
If $\Proc{\WriteB}=\Proc{\Event}$, then $\Event$ is $\PSO$-executable in $\Trace_1$.
Now consider that $\Proc{\WriteB}\neq \Proc{\Event}$, and
assume towards contradiction that $\Event$ is
not $\PSO$-executable in $\Trace_1$.
There are two cases where this can happen.

The first case is when the variable $v$ is held by another memory write in $\Trace_1$.
Since $\LocalEvents{\Trace_1}=\LocalEvents{\ov{\Trace}_1}$ and $\WritesM{\Trace_1}\setminus\SpuriousWritesM{\Trace_1}\subseteq \WritesM{\ov{\Trace}_1}$,
the variable $v$ is also held by another memory write in $\ov{\Trace}_1$,
and thus $\WriteM$ is neither in $\ov{\Trace}_1$ nor $\PSO$-executable in $\ov{\Trace}_1$.
Thus $\Event$ is not $\PSO$-executable in $\ov{\Trace}_1$ either, a contradiction.

The second case is when there exists a read $\Read\not \in \Trace_1$ such that $\Observation(\Read)= (\_, \WriteM)$, and
there exists a local write event $\Write'=(\WriteB', \WriteM')$ with $\Proc{\WriteB'}=\Proc{\Read}$ but $\WriteM'\not \in \Trace_1$.
Since $\ov{\Trace}_1$ is a witness prefix, we have $\WriteM'\in \ov{\Trace}_1$, hence $\WriteB'\in \ov{\Trace}_1$, and since $\LocalEvents{\Trace_1}=\LocalEvents{\ov{\Trace}_1}$, we also have $\WriteB'\in \Trace_1$.
Thus $\WriteM'$ is a pending memory write for the thread $\Process'=\Proc{\WriteM'}$.
Let $\WriteM''$ be the earliest (wrt $\TO$) pending memory-write of $\Process'$ for the variable $v$.
Thus $\WriteM''<_{\TO}\WriteM'$, and hence $\WriteM''\in \ov{\Trace}_1$.
Note that $\WriteM''$ is not read-from by any read not in $\Trace_1$, and hence $\WriteM''$ is spurious in $\Trace_1$.
But then, the while loop in \cref{algo:pso_while_spurious} must have added $\WriteM''$ in $\Trace_1$, a contradiction.

It follows that $\Event$ is $\PSO$-executable in $\Trace_1$, and thus $\LocalEvents{\ov{\Trace}}=\LocalEvents{\Trace_{\Event}}$.
Let $\Trace_2=\Trace_1$ if $\WriteM\in \Trace_1$, else $\Trace_2=\Trace_1\Concat \WriteM$.
Observe that if $\WriteM$ is $\PSO$-executable in $\Trace_1$,
all pending memory-writes $\WriteM'$ on variable $v$ of threads
other than $\Proc{\WriteB}$ are spurious in $\Trace'$,
and thus all such buffers are empty in $\Trace_1$.
It follows that $\FenceMap_{\Trace_2}\leq \FenceMap_{\Trace_1}$ and thus $\FenceMap_{\Trace_2}\leq \FenceMap_{\ov{\Trace}_1}$.
Moreover, trivially $\WritesM{\Trace_2}\setminus\SpuriousWritesM{\Trace_2}\subseteq \ov{\Trace_1}$.
Finally, executing $\Event$ in $\Trace_2$ and $\ov{\Trace}_1$, we obtain respectively $\Trace_{\Event}$ and $\ov{\Trace}$,
and by \cref{lem:fencemap_monotonicity2}, we have $\FenceMap_{\Trace_{\Event}}\leq \FenceMap_{\ov{\Trace}}$.
Moreover, clearly $\WritesM{\Trace_{\Event}}=\WritesM{\Trace_2}$ and thus $\WritesM{\Trace_{\Event}}\setminus\SpuriousWritesM{\Trace_{\Event}}\subseteq \ov{\Trace}$.

\item\label{item:pso_compl_fence} \emph{If $\Event$ is a fence}, let $\SequenceVar=\WriteM_1, \dots, \WriteM_j$ be the sequence of pending memory-writes constructed in \cref{algo:pso_linearize_pending_writes}.
By a similar analysis to the case where $\Event$ is a a read event, we have that $\SequenceVar$ contains at most one memory-write per variable, as all preceding ones (wrt $\TO$) must be spurious.
Assume towards contradiction that some pending memory write
$\WriteM_i$ is not $\PSO$-executable in $\Trace$, and let $v=\Location{\WriteM_i}$.
There are two cases to consider.
\begin{enumerate}
\item $\WriteM_i$ is not $\PSO$-executable because $v$ is held in $\Trace_1$.
Let $\WriteM$ be the memory-write that holds $v$ in $\Trace_1$, and $\Read$ be the corresponding read with
$\Observation(\Read)= (\_, \WriteM)$ and $\Read\not \in \Trace_1$.
Since $\LocalEvents{\Trace_1}=\LocalEvents{\ov{\Trace}_1}$, we have $\Read\not \in \ov{\Trace}_1$.
Let $\Process_1=\Proc{\Event}$, $\Process_2=\Proc{\Read}$, and $m$ be the index of $\Read$ in $\Process_2$.
We have $\FenceMap_{\Trace_1}(\Process_1, \Process_2) \geq m$, and since $\FenceMap_{\Trace_1}\leq \FenceMap_{\ov{\Trace}'}$,
we also have $\FenceMap_{\ov{\Trace}'}(\Process_1, \Process_2) \geq m$.
But then there is a pending memory-write $\WriteM'\in \ov{\Trace}'$ with $\Proc{\WriteM'}=\Process_1$ and $\WriteM'\not \in \ov{\Trace}_1$.
Hence $\Event$ is not $\PSO$-executable in $\ov{\Trace}_1$, a contradiction.
\item $\WriteM_i$ is not $\PSO$-executable because there exists a read $\Read\not \in \Trace_1$ such that $\Observation(\Read)= (\_, \WriteM_i)$, and
there exists a local write event $\Write=(\WriteB, \WriteM)$ with $\Proc{\WriteB}=\Proc{\Read}$ but $\WriteM\not \in \Trace_1$.
The analysis is similar to the case of $\Event$ being a read, which leads to a contradiction.
\end{enumerate}
Thus, we have that the fence $\Event$ is $\PSO$-executable in $\Trace_1$.
It is straightforward to see that $\WritesM{\Trace_{\Event}}\setminus\SpuriousWritesM{\Trace_{\Event}}\subseteq \ov{\Trace}$,
and thus it remains to argue that $\FenceMap_{\Trace_{\Event}}\leq \FenceMap_{\ov{\Trace}}$.
Let $\Trace_1^{j}=\Trace_1\Concat\WriteM_1,\dots, \WriteM_{j-1}$.
It suffices to argue that $\FenceMap_{\Trace_1^{j+1}}\leq \FenceMap_{\ov{\Trace}_1}$, as
$\Trace_{\Event}=\Trace_1^{j+1}\Concat \Event$ and $\ov{\Trace}=\ov{\Trace}_1\Concat \Event$, and the claim holds by \cref{lem:fencemap_monotonicity2} on   $\Trace_1^{j+1}$ and $\ov{\Trace}_1$.
The proof is by induction on $\Trace_1^{i}$.
The claim clearly holds for $i=1$, as then $\Trace_1^{1}=\Trace_1$ and we have $\FenceMap_{\Trace_1}\leq \FenceMap_{\ov{\Trace}_1}$.
Now consider that for some $i>1$, there exist two threads $\Process_1, \Process_2,\in \Threads$ such that $\FenceMap_{\Trace_1^{i}}>\FenceMap_{\Trace_1^{i-1}}(\Process_1, \Process_2)$.
Hence, variable $v=\Location{\WriteM^{i}}$ is held in $\Trace_1^{i}$ and $\WriteM^{i}$ is the respective active-memory-write, and
thread $\Process_2$ has a read $\Read$ in index $m=\FenceMap_{\Trace_1^{i}}(\Process_1, \Process_2)$ with $\Observation(\Read)= (\_, \WriteM^{i})$.
In addition, there exists a buffer-write $\WriteB\in \Trace_1$ such that $\Proc{\WriteB}=\Process_1$ and $\Location{\WriteB}=v$.
Since $\LocalEvents{\Trace_1^{i}}=\LocalEvents{\ov{\Trace}}$, we have that $\WriteB\in \ov{\Trace}$ and $\Read\not \in \ov{\Trace}$.
Moreover, since $\WriteM^i<_{\TO}\Event$, we have $\WriteM^{i}\in \ov{\Trace}$.
Hence $\WriteM^{i}$ is an active-memory-write in $\ov{\Trace}$ as well, and thus $\FenceMap_{\ov{\Trace}}(\Process_1, \Process_2)\geq \FenceMap_{\Trace_1^i}(\Process_1, \Process_2)$.
At the end of the induction, we have $\FenceMap_{\Trace_1^{j+1}}\leq \FenceMap_{\ov{\Trace}}$, as desired.
\end{enumerate}
This concludes the completeness argument for executions without
RMW and CAS instructions.

%% RMW CAS
When executions contain RMW and CAS instructions, additional argument
has to be made for completeness, as follows.
We proceed with the same induction argument as above, but additionally
consider the inductive case where
$\ov{\Trace}=\ov{\Trace}'\Concat \Sequence \Concat\Event$ such
that $\Event$ is an atomic block corresponding to a RMW or a CAS
instruction. In this case, $\Event$ is a sequence of
(i) a read $\Read$, (ii) a buffer-write $\WriteB$, and
optionally (in case the write-part of $\Event$ is designated to
proceed directly into the shared memory) (iii) a memory-write $\WriteM$.
Finally, $\Event$ is preceded in its thread by a fence $\Fence$.

First, since $\ov{\Trace}$ is a witness prefix we have
$\Fence \in \LocalEvents{\ov{\Trace}'}$, and
from the induction hypothesis regarding $\Trace'$ such that
$\LocalEvents{\Trace'} = \LocalEvents{\ov{\Trace}'}$
we also have $\Fence \in \LocalEvents{\Trace'}$. Thus all buffers
of the thread of $\Event$ are empty in both $\ov{\Trace}'$
and $\Trace'$. Then the argument is followed identically to above
until \cref{algo:pso_ifextend}, where we have to show that the atomic
block $\Event$ is $\PSO$-executable in $\Trace_1$ in
\cref{algo:pso_ifextend}, and that consequently the induction statement
holds for the new trace
$\Trace_{\Event}$ constructed in \cref{algo:pso_execute_event}.

The crucial observation is that no event from the second event onward
in the atomic block $\Event$ is a read or a fence. This is important
as reads and fences may need additional events executed right before
them (see Lines
\ref{algo:pso_is_read}--\ref{algo:pso_execute_pending_writes}),
which would invalidate the atomicity of the atomic block.
Given this observation, we simply utilize the $\PSO$-executable
requirements to prove the following. First, using the argument of
\cref{item:pso_compl_read} above we show that $\Read$ is
$\PSO$-executable in $\Trace_1$, let $\Trace_{\Read}$ denote the trace
resulting after executing $\Read$. Second, using
\cref{item:pso_compl_bw} above we show that $\WriteB$ is $\PSO$-executable
in $\Trace_{\Read}$, and further that
(i)~$\LocalEvents{\Trace_{\Read} \Concat \WriteB}= \LocalEvents{\ov{\Trace}'\Concat \Sequence \Concat \Read \Concat \WriteB}$,
(ii)~$\WritesM{\Trace_{\Read} \Concat \WriteB}\setminus\SpuriousWritesM{\Trace_{\Read} \Concat \WriteB}
\subseteq \WritesM{\ov{\Trace}'\Concat \Sequence \Concat \Read \Concat \WriteB}$, and
(iii)~$\FenceMap_{\Trace_{\Read} \Concat \WriteB}
\leq \FenceMap_{\ov{\Trace}'\Concat \Sequence \Concat \Read \Concat \WriteB}$.
Finally, in the case where $\WriteM$ is part of the atomic block
$\Event$,
% from the (i)(ii)(iii) above and the spurious-mwrites-maximally-executed
% and the fact that mwrite psoexecutable in \ov{\Trace}'\Concat \Sequence \Concat \Read \Concat \WriteB
we have that $\WriteM$ is $\PSO$-executable in
$\Trace_{\Read} \Concat \WriteB$, resulting in the trace $\Trace_{\Event}$.
Further, since the induction statement held already for
$\Trace_{\Read} \Concat \WriteB$ with respect to
$\ov{\Trace}'\Concat \Sequence \Concat \Read \Concat \WriteB$
(see (i),(ii),(iii) above), we have
that the induction statement holds also for $\Trace_{\Event}$
with respect to $\ov{\Trace}$, which
concludes the argument.

The desired result follows.
\end{proof}

We can now proceed with the proof of \cref{them:vpso}.

\themvpso*

\begin{proof}%[Proof of \cref{them:vpso}.]
\cref{lem:verifyingpso_correctness} establishes the correctness, so here we focus on the complexity,
and the following argument applies also for executions with RMW and CAS instructions.
Since there are $k$ threads, there exist at most $n^k$ distinct traces $\Trace_1, \Trace_2$ with $\LocalEvents{\Trace_1}\neq \LocalEvents{\Trace_2}$.
Because of the test in \cref{algo:pso_if_new}, for any two traces $\Trace_1,\Trace_2$ inserted in the worklist with $\LocalEvents{\Trace_1}=\LocalEvents{\Trace_2}$, we have $\FenceMap_{\Trace_1}\neq \FenceMap_{\Trace_2}$.
If there are no fences, there is only one possible fence map, hence there are $n^k$ traces inserted in $\Worklist$.
If there are fences, the number of different fence maps with $\FenceMap_{\Trace_1}\neq \FenceMap_{\Trace_2}$ when $\LocalEvents{\Trace_1}=\LocalEvents{\Trace_2}$ is bounded by $2^{k\cdot \NumVariables}$ (by \cref{lem:fencemap_size}) and also by $n^{k\cdot (k-1) }$ (since there are at most that many difference fence maps).
Hence the number of traces inserted in the worklist is bounded by $n^{k}\cdot \min(n^{k \cdot (k-1)}, 2^{k\cdot \NumVariables})$.
Since there are $k$ threads, for every trace $\Trace_1$ inserted in the worklist, the algorithm examines at most $k-1$ other traces $\Trace_2$ that are not inserted in the worklist because
$\LocalEvents{\Trace_1}=\LocalEvents{\Trace_2}$ and $\FenceMap_{\Trace_1}=\FenceMap_{\Trace_2}$.
Hence the algorithm examines at most $k\cdot  n^{k}\cdot \min(n^{k\cdot (k-1)}, 2^{k\cdot \NumVariables})$ traces in total, while each such trace is handled in $O(n)$ time.
Hence the total running time is $O(k\cdot n^{k+1}\cdot \min(n^{k\cdot (k-1) }, 2^{k\cdot \NumVariables}))$.

Finally, note that if there are no fences present, we can completely drop the fence maps from the algorithm, which results in complexity $O(k\cdot n^{k+1})$.
The desired result follows.
\end{proof}

\subsection{Verifying $\PSO$ Executions with Store-store Fences}\label{subsec:app_psostorestore}

Here we describe our extension to handle $\VPSOrm$ in
the presence of store-store-fences.

A \emph{store-store fence} event $\SFence$ happening on a thread $\Process$
% enforces for every write of $\Process$ with its buffer-part $\WriteB$
% before $\SFence$, that its memory-part $\WriteM$ must happen before
% any $\WriteM'$ of $\Process$ with its buffer-part $\WriteB'$ after $\SFence$.
% Formally, $\SFence$ of a thread $\Process$
introduces further orderings into the program order $\TO$, namely
$\WriteM <_{\TO} \WriteM'$ for each
$(\WriteB, \WriteM), (\WriteB', \WriteM') \in \SysWritesM_{\Process}$
with $\WriteB <_{\TO} \SFence <_{\TO} \WriteB'$.

Store-store fences are considered only for the $\PSO$ memory model, as they
would have no effect in $\TSO$, since in $\TSO$ all memory-writes
within the same thread are already ordered. In fact, the $\TSO$ model can
be seen as $\PSO$ with a store-store fence inserted after every
buffer-write event.

We extend our notion of $\PSO$-executability to accommodate store-store-fences.
Given $(X,\TO)$ and $\Trace$ with $\Events{\Trace} \subseteq X$:
\begin{enumerate}
\item A store-store fence $\SFence \in X \setminus \Events{\Trace}$ is
$\PSO$-executable if $\Events{\Trace} \cup \{\SFence\}$ is a lower set
of $(X,\TO)$.
\item An additional condition for a memory-write
$\WriteM \in X \setminus \Events{\Trace}$ to be $\PSO$-executable,
is that every memory-write $\WriteM' \in X \setminus \Events{\Trace}$
with $\WriteM' <_{\TO} \WriteM$ is $\PSO$-executable.
\end{enumerate}

We consider a notion very similar to the fence maps introduced in
\cref{subsec:verifyingpso}, to efficiently represent the
$\PSO$-executability requirements introduced by store-store fences,
namely \emph{store-store fence maps}
$\SFenceMap_{\Trace}\colon \Threads\times \Threads \to [n]$.
While $\FenceMap_{\Trace}(\Process)$ efficiently captures
the requirements for executing a fence event of $\Process$,
$\SFenceMap_{\Trace}(\Process)$ captures efficiently, in the same
manner as $\FenceMap_{\Trace}(\Process)$ does, the following.
Consider the latest $\SFence \in \Events{\Trace}$ of thread $\Process$,
and consider that no memory-write of $\Process$ has been executed
in $\Trace$ after $\SFence$ yet.
Then, $\SFenceMap_{\Trace}(\Process)$ captures
the requirements for executing a memory-write of $\Process$.

We utilize the store-store fence maps to refine our identification of
duplicate witness-prefixes.
% Namely, in \cref{algo:pso_if_new}~of~\cref{algo:verifypso} there
% has to be no prefix in $\DoneSet$ such that futher its store-store
% fence map is identical to that of our current prefix.
This then gives us a time-complexity bound of
$O(k\cdot n^{k+1}\cdot \min(n^{2\cdot k\cdot (k-1)}, 2^{k\cdot \NumVariables}))$.

\section{Details of~\cref{sec:dpor}}\label{sec:app_algoproperties}

In this section we provide the proof of \cref{them:dctsopso_maximal}
regarding $\DCTSOPSOM$.
% then we detail on the search heuristic
% for $\AlgoTSO$/$\AlgoPSO$ when used in $\DCTSOPSOM$, and finally
% we describe an extension of $\DCTSOPSOM$ to handle lock events.

% \subsection{Properties of $\DCTSOPSOM$}\label{subsec:app_maximal}

\smallskip
\themdctsopsomaximal*
\begin{proof}
Let $\MemoryModel$ be the memory model from $\{\TSO, \PSO \}$.
We sketch the correctness (i.e., soundness and completeness),
exploration-optimality, and time complexity of $\DCTSOPSOM$.

\noindent{\em Soudness.}
The soundness trivially follows from soundness of $\AlgoTSO$
used in $\TSO$ and of $\AlgoPSO$ used in $\PSO$,
which are used as subroutines for verifying execution consistency.

\noindent{\em Completeness.}
The completeness of $\DCTSOPSOM$ rests upon the completeness
of its variant for $\SC$ introduced by~\citet{Abdulla19}. We now
argue that the modifications to accomodate $\TSO$ and $\PSO$
have no effect on completeness. First, consider in each recursive call
the sequences $\Seq$ (argument of the call) and
$\widetilde{\Seq}$ (\cref{line:dctsopsom_extendsequence}~of~\cref{algo:dctsopsom}).
The sequence $\Seq$ (resp. $\widetilde{\Seq}$) in each call
contains exactly the thread
events of the trace $\Trace$ (resp. $\widehat{\Trace}$) in that call.
Thus $\Seq$ (resp. $\widetilde{\Seq}$) contains exactly the events
of local traces of each thread in $\Trace$ (resp. $\widehat{\Trace}$).
This gives that the usage of $\widetilde{\Seq}$ to manipulate
$\schedules$ is equivalent to the $\SC$ case where there are
only thread events. Second, the proper event set formed
in~\cref{line:dctsopsom_eventsp}~of~\cref{algo:dctsopsom}
is uniquely determined, and mirrors the set of events $\Events{\Seq'}$
of the sequence $\Seq'$ created in \cref{line:dctsopsom_seqp}~of~\cref{algo:dctsopsom}.
The set of events $\Events{\Seq'}$ would be considered for the mutation
in the $\SC$ case, given that we consider buffer-writes of $\Events{\Seq'}$
as simply atomic write events that $\SC$ models.
Finally, the witness subroutine is handled by
$\AlgoTSO$ for $\TSO$ and $\AlgoPSO$ for $\PSO$, whose completeness
is established in \cref{lem:verifyingtso_correctness}
and \cref{lem:verifyingpso_correctness}.
Thus the completeness of $\DCTSOPSOM$ follows.

\noindent{\em Exploration-optimality.}
The exploration-optimality argument mirrors the one made
by~\citet{Abdulla19}, and can be simply established by considering
the sequence $\widetilde{\Seq}$
(\cref{line:dctsopsom_extendsequence}~of~\cref{algo:dctsopsom})
of each recursive call. The sequences
$\widetilde{\Seq}$ of all calls, coalesced together
with equal events merged, form a rooted tree.
Each node in the tree with multiple children is some read $\Read$.
Let us label each child branch by the source $\Read$ reads-from, in the
trace of the same call that owns the sequence introducing the child
branch. The source for $\Read$ is different in each branch,
and thus the same trace can never appear when following two
different branches of $\Read$. The exploration-optimality follows.

\noindent{\em Time complexity.}
From exploration-optimality we have that a run of $\DCTSOPSOM$
performs exactly $\left|\TraceSpaceMax_{\MemoryModel}/\sim_{\Observation}\right|$
calls. It remains to argue that each class of $\TraceSpaceMax_{\MemoryModel}/\sim_{\Observation}$ spends time
$O(\alpha)$ where
\begin{enumerate}
\item $\alpha=n^{k + O(1)}$ under $\MemoryModel=\TSO$, and
\item $\alpha=n^{k + O(1)} \cdot \min(n^{k\cdot (k-1) }, 2^{k\cdot \NumVariables})$
under $\MemoryModel=\PSO$.
\end{enumerate}
We split this argument to three parts.
\begin{enumerate}
\item Lines~\ref{line:dctsopsom_extendtrace}-\ref{line:dctsopsom_forinitscheduleend} spend $O(n)$ time per call.
\item One call of $\AlgoTSO$ resp. $\AlgoPSO$ spends $O(\alpha)$ time by~\cref{them:vtso} resp.~\cref{them:vpso}.
Thus Lines~\ref{line:dctsopsom_forreadstomutate}-\ref{line:dctsopsom_forreadstomutateend}
spend $O(n^2 \cdot \alpha)$ time per call.
\item The total number of mutations added into $\schedules$~(on~\cref{line:dctsopsom_add})
equals $\left|\TraceSpaceMax_{\MemoryModel}/\sim_{\Observation}\right| - 1$,
i.e., it equals the total number of calls minus the initial call.
However, we note that (i) each call adds only polynomialy many
new schedules, and (ii) a call to a new schedule is considered
work spent on the \emph{class corresponding to} the new schedule.
Thus Lines~\ref{line:dctsopsom_forrecurread}-\ref{line:dctsopsom_forrecurreadend} spend $O(1)$ amortized time per
recursive call, and $O(1)$ time is spent in this location per partitioning class.
\end{enumerate}
The complexity result follows.

\end{proof}

\section{Details of~\cref{sec:exp}}\label{sec:app_exp}

In this section we provide further details on our
consistency verification and SMC experiments.

\smallskip\noindent{\bf Technical details.}
For all our experiments we have used a Linux machine with Intel(R) Xeon(R)
CPU E5-1650 v3 @ 3.50GHz (12 CPUs) and 128GB of RAM.
We have run the Nidhugg version of 26. November 2020,
with Clang and LLVM version 8.

\subsection{Consistency Verification -- Experimental Setup Details}\label{subsec:app_verifysetup}

Here we describe in detail the collection of instances
% experimental setup
for evaluation of the consistency verification algorithms.
We generate and collect the $\VTSOrm$ and $\VPSOrm$ instances
that appear during SMC of our benchmarks using the
reads-from SMC algorithm $\DCTSOPSOM$.

We supply the unroll bounds to the benchnmarks % appropriately
so that % most of
the created $\VTSOrm$/$\VPSOrm$ instances
are solvable in a time reasonable for experiments
(i.e., more than a tiny fraction of a second,
and within a minute). We run each benchmark with several such unroll
bounds. Further, as a filter of too small instances,
we only consider realizable instances where at least one verification
algorithm without closure took at least 0.05 seconds.

For each SMC run, to collect a diverse set of instances, we collect
every fifth realizable instance we encounter, and every fifth
unrealizable instance we encounter. In this way
we collect 50 realizable instances and 20
unrealizable instances.
For each collected instance, we run all verification algorithms
and closure/no-closure variants % thereof
5 times, and average the results.
We run all verification algorithms subject to a timeout of
one minute.

%
% Histogram of number of events and number of threads in cases ?
%

\subsection{Consistency Verification -- Further Results}\label{subsec:app_verifyfurther}

Here we provide further analysis of the results obtained
for evaluation of the consistency verification algorithms
$\AlgoTSO$, $\AlgoPSO$, $\NaiveAlgoTSO$ and $\NaiveAlgoPSO$, as well as
the closure heuristic of~\cref{subsec:heuristics}.

\Paragraph{Detailed results -- effect of closure.}
Here we compare each veficiation algorithm against itself, where
one version uses closure and the other one does not.

\begin{figure}[h]
\centering
%%%%%%%%%%%%%%%
\begin{subfigure}[b]{0.44\textwidth}
\centering
\includegraphics[height=5.4cm]{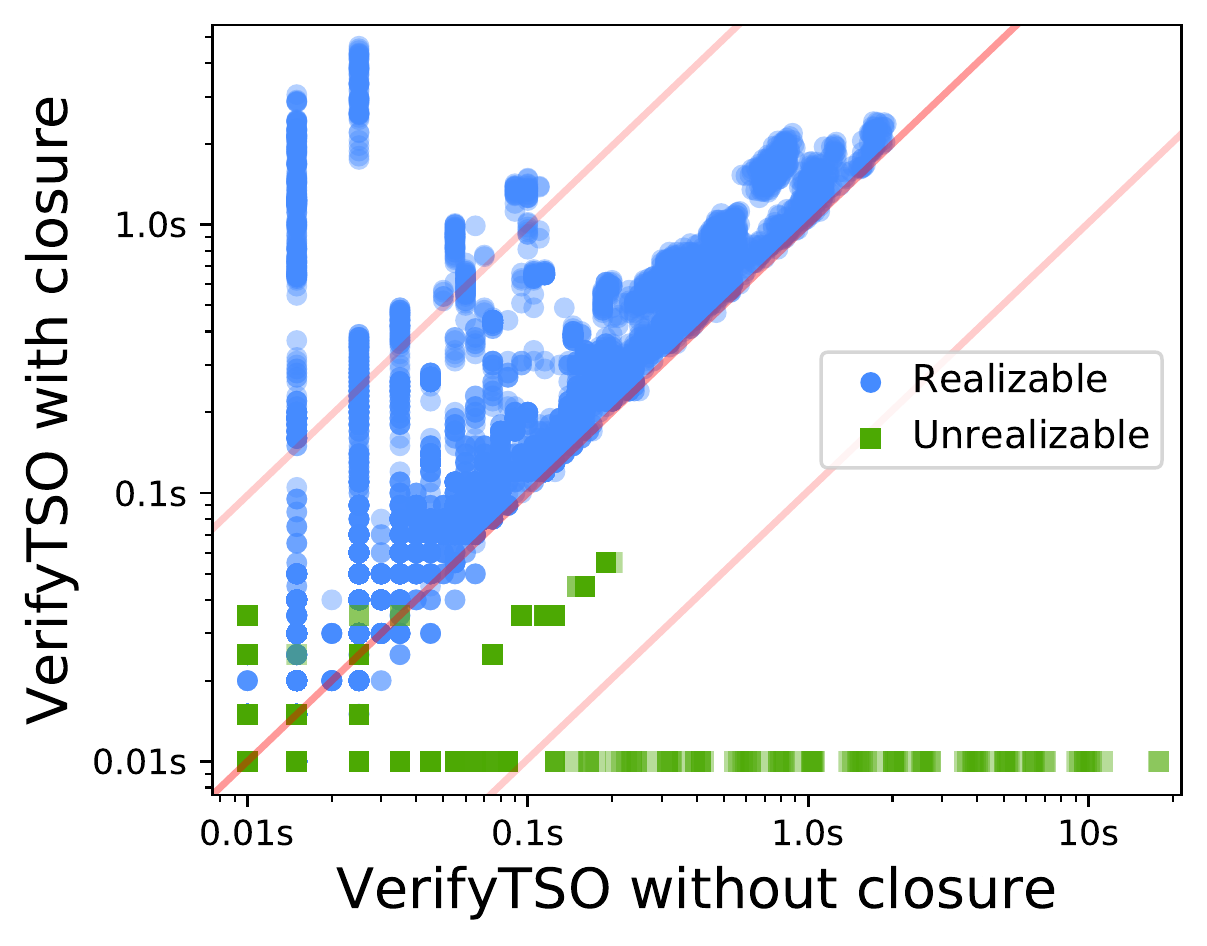}
%\caption{}
\label{subfig:tso_us_clnocl_all}
\end{subfigure}
%%%%%%%%%%%%%%%
\qquad%\qquad
%%%%%%%%%%%%%%%
\begin{subfigure}[b]{0.44\textwidth}
\centering
\includegraphics[height=5.4cm]{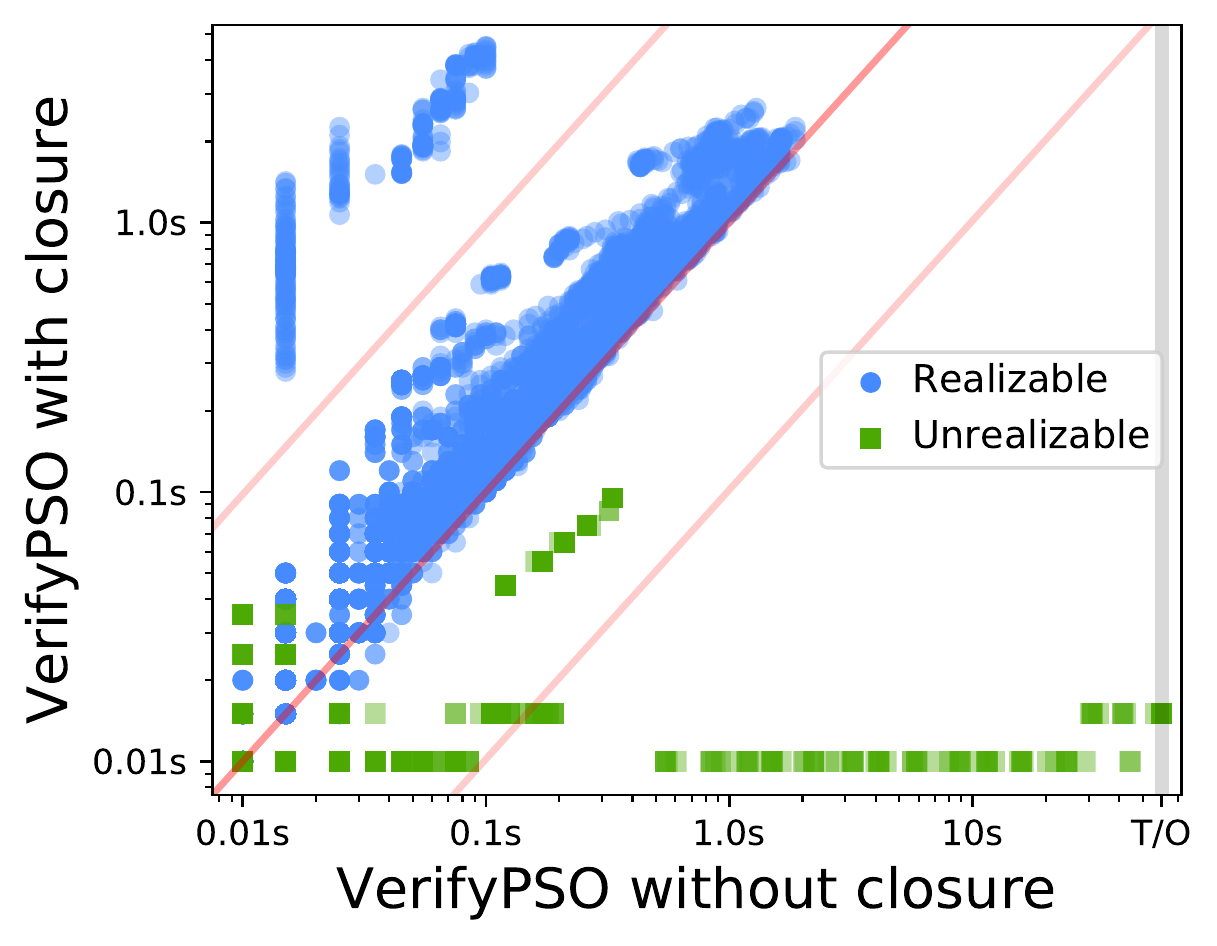}
%\caption{}
\label{subfig:pso_us_clnocl_all}
\end{subfigure}
\qquad%\qquad
%%%%%%%%%%%%%%%
\vspace{-5mm}
\caption{
Comparison of $\AlgoTSO$ (left) and $\AlgoPSO$ (right) with and without the closure.
}
\label{fig:us_clnocl_all}
\end{figure}

\cref{fig:us_clnocl_all} presents the comparison of
$\AlgoTSO$ (used for $\VTSOrm$) and $\AlgoPSO$ (used for $\VPSOrm$)
with and without closure. For both memory models, we see that for
instances that are realizable (blue dots), the version without closure
is superior, sometimes even beyond an order of magnitude. This suggests
that computing the closure-partial-order takes more time than is
subsequently saved by utilizing it during the witness search.
On the other hand, we observe that for the instances that are not
realizable (green dots), the version with closure is orders-of-magnitude
faster. This signifies that closure detects unrealizable instances
much faster than complete exploration of a consistency verification
algorithm.

\begin{figure}[h]
\centering
%%%%%%%%%%%%%%%
\begin{subfigure}[b]{0.44\textwidth}
\centering
\includegraphics[height=5.4cm]{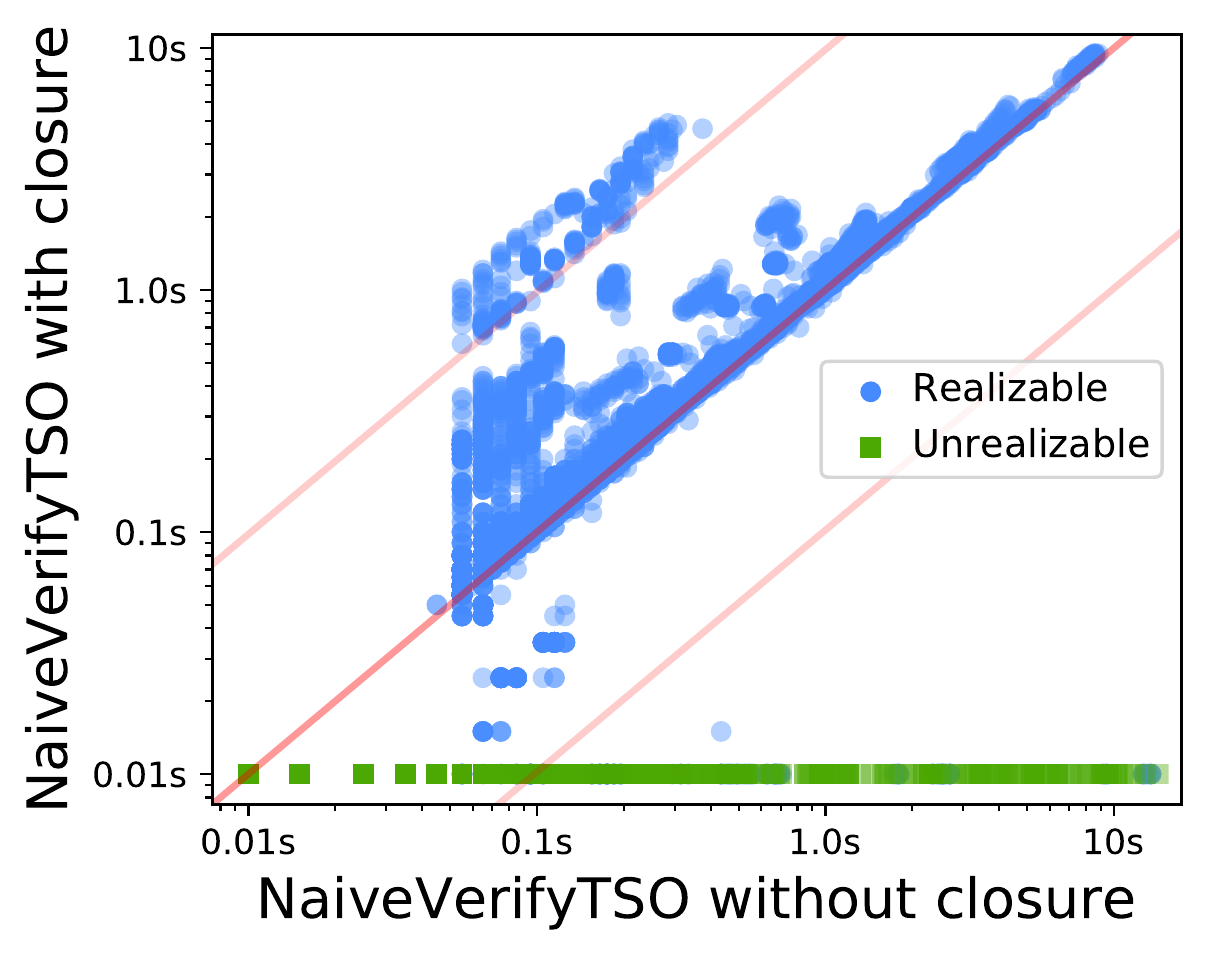}
%\caption{}
\label{subfig:tso_them_clnocl_all}
\end{subfigure}
%%%%%%%%%%%%%%%
\qquad%\qquad
%%%%%%%%%%%%%%%
\begin{subfigure}[b]{0.44\textwidth}
\centering
\includegraphics[height=5.4cm]{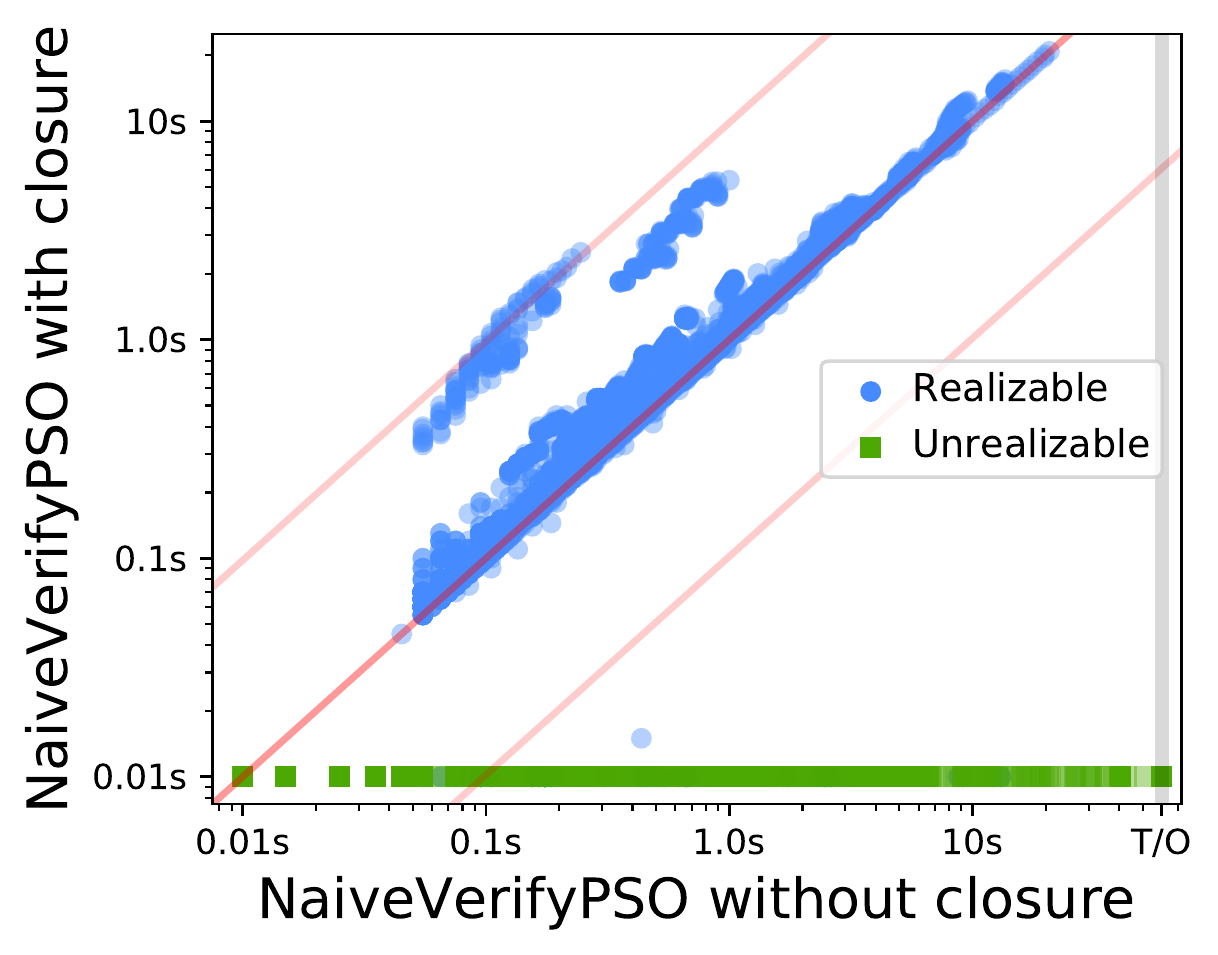}
%\caption{}
\label{subfig:pso_them_clnocl_all}
\end{subfigure}
\qquad%\qquad
%%%%%%%%%%%%%%%
\vspace{-5mm}
\caption{
Comparison of $\NaiveAlgoTSO$ (left) and $\NaiveAlgoPSO$ (right) with and without the closure.
}
\label{fig:them_clnocl_all}
\end{figure}

\cref{fig:them_clnocl_all} presents the comparison of
$\NaiveAlgoTSO$ (used for $\VTSOrm$) and $\NaiveAlgoPSO$ (used for $\VPSOrm$)
with and without closure. We observe trends similar to the paragraph
above. Specifically, both $\NaiveAlgoTSO$ and $\NaiveAlgoPSO$ are
mostly faster without closure on realizable instances, while they are
significantly faster with closure on unrealizable instances.

\begin{figure}[h]
\centering
%%%%%%%%%%%%%%%
\begin{subfigure}[b]{0.44\textwidth}
\centering
\includegraphics[height=5.4cm]{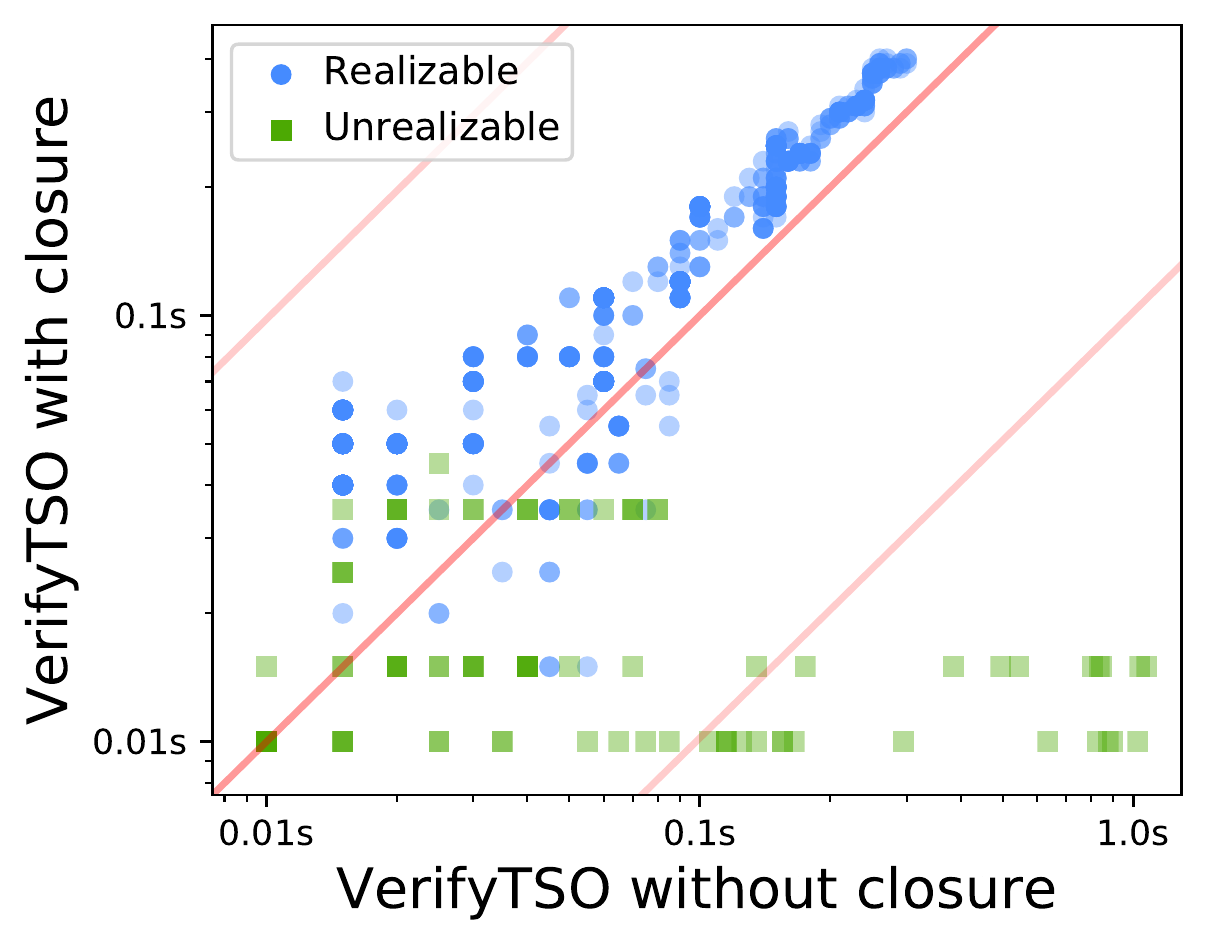}
%\caption{}
\label{subfig:tso_cas_us_cleffect}
\end{subfigure}
%%%%%%%%%%%%%%%
\qquad%\qquad
%%%%%%%%%%%%%%%
\begin{subfigure}[b]{0.44\textwidth}
\centering
\includegraphics[height=5.4cm]{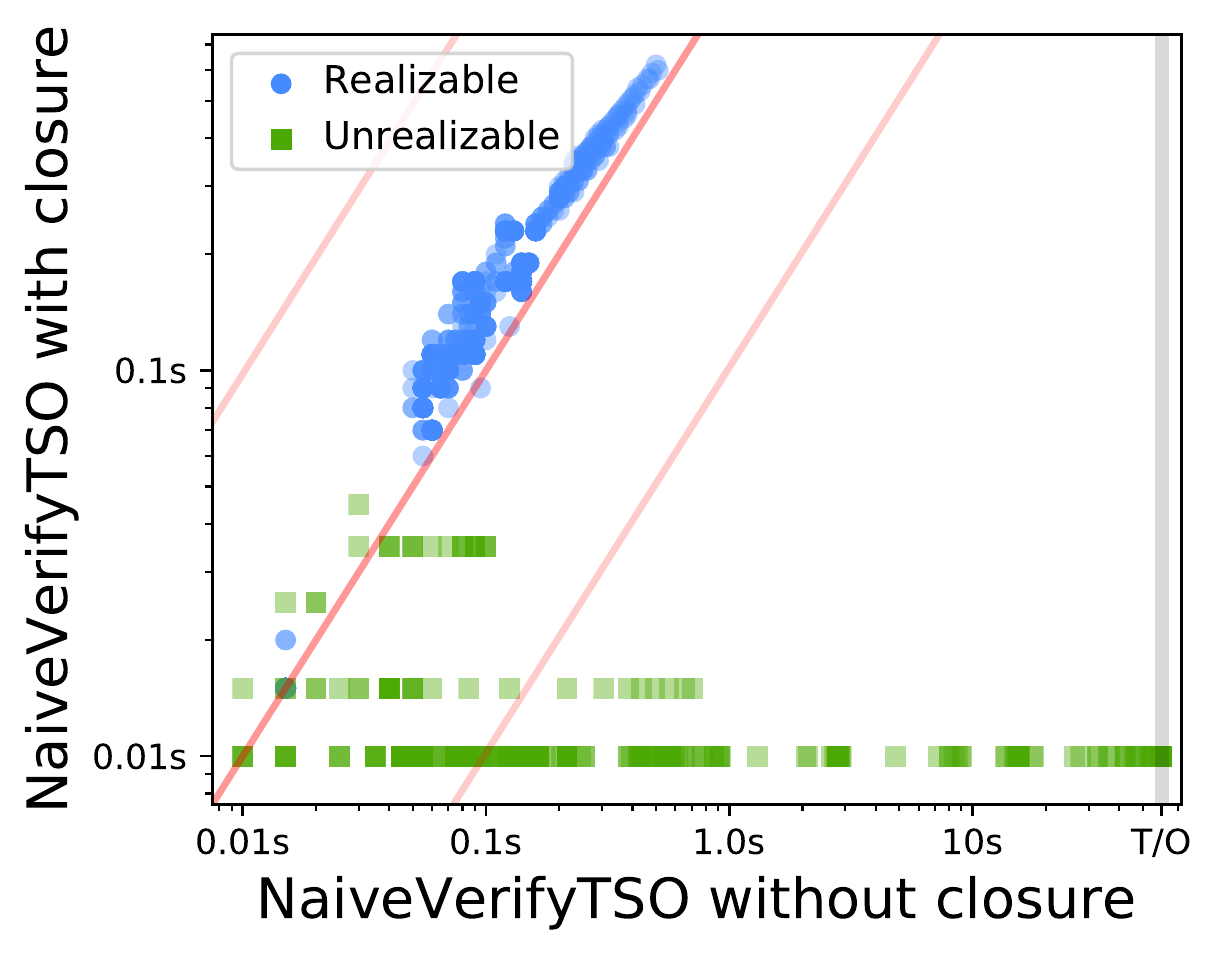}
%\caption{}
\label{subfig:tso_cas_them_cleffect}
\end{subfigure}
\qquad%\qquad
%%%%%%%%%%%%%%%
\vspace{-5mm}
\caption{
Comparison of $\AlgoTSO$ (left) and $\NaiveAlgoTSO$ (right) with and without closure
on verification instances with RMW and CAS instructions.
}
\label{fig:tso_cas_cleffect}
\end{figure}

Finally, \cref{fig:tso_cas_cleffect} presents the effect of closure
for $\AlgoTSO$ and $\NaiveAlgoTSO$
on verification instances that contain RMW and CAS instructions.
Similarly to the verification without RMW and CAS instructions,
both verification algorithms are somewhat slower when using closure
on the realizable instances, and they are significantly faster
when using closure on the unrealizable instances.

\subsection{SMC -- Experimental Setup Details}\label{subsec:app_expdetails}

Here we present further details regarding the setup for SMC experiments.

\smallskip\noindent{\bf Handling assertion violations.}
We note that not all benchmarks behave as intended under all memory models,
e.g., a benchmark might be correct under $\SC$, but contain bugs under $\TSO$.
However, this is not an issue, as our goal is to characterize the size
of the underlying partitionings, rather than detecting assertion violations.
We have disabled all assertions, in order to not have the measured
parameters be affected by how fast a violation is discovered, as the
latter is arbitrary.
As a sanity check, we have confirmed that for each memory model,
all algorithms considered for that model discover the same bugs when assertions are enabled.

\smallskip\noindent{\bf Identifying events.}
Our implementation extends the Nidhugg model checker and we rely on the
interpreter built inside Nidhugg to identify events. An event $\Event$ is
defined by a triple ($a_{\Event}, b_{\Event}, c_{\Event}$), where $a_{\Event}$
is the thread-id of $\Event$, $b_{\Event}$ is the id of either the buffer of
$a_{\Event}$ or the main-thread of $a_{\Event}$ that $\Event$ is a part of, and
$c_{\Event}$ is the sequential number of the last LLVM instruction
(of the corresponding thread/buffer) that is part of $\Event$. It can happen
that there exist two traces $\Trace_1$ and $\Trace_2$, and two different events
$\Event_1 \in \Trace_1$, $\Event_2 \in \Trace_2$, such that their identifiers are
equal, i.e., $a_{\Event_1}=a_{\Event_2}$, $b_{\Event_2}=b_{\Event_2}$, and
$c_{\Event_1}=c_{\Event_2}$. However, this means that the control-flow leading
to each event is different. In this case, $\Trace_1$ and $\Trace_2$ differ in
the reads-from of a common event that is ordered by the program
order $\TO$ both before $\Event_1$ in
$\Trace_1$ and before $\Event_2$ in $\Trace_2$, and hence $\Event_1$ and
$\Event_2$ are treated as inequivalent.

\smallskip\noindent{\bf Dynamic thread creation.}
For simplicity of presentation of our approach, we have neglected dynamic thread creation
and assumed a static set of threads for a given concurrent program. In practice, all our
benchmarks spawn threads dynamically. This situation is handled straightforwardly, by
including in the %closure~(\cref{subsec:heuristics})
program order $\TO$
the orderings naturally induced by spawn and join events.

\smallskip\noindent{\bf Benchmark adaptations.}
We have made small changes to some of the SVCOMP benchmarks so they can be
processed by our prototype implementation in Nidhugg:
\begin{itemize}
\item Verifier calls to perform acquire and release are handled by a \texttt{pthread\_mutex}.
\item In order to eliminate intra-thread nondeterminism,
%which is typically not handled by SMC techniques (including this work),
verifier calls to nondeterministically produce an arbitrary integer are replaced by a constant value.
\end{itemize}
Further, we have made steps to obtain scalable versions of benchmarks:
\begin{itemize}
\item In mutual exclusion benchmarks, the thread routines are put in a loop
with scalable size, so threads can reenter a critical section multiple times.
\item We manually perform loop unrolling, i.e., we limit the amount of times
each loop is executed by a scalable bound, instead of relying
on the loop bounding technique provided by Nidhugg.
\end{itemize}

\subsection{SMC -- Full Experimental Results}\label{subsec:app_expfull}

Here we provide full results of the SMC experiments.
We first provide several further scatter plots to compactly
illustrate the full experimental results.
% We have analyzed 109 benchmarks under each memory model $\MemoryModel\in\{\SC, \TSO, \PSO \}$.
For each fixed plot comparing two algorithms, %under selected memory models,
%and for each memory model,
we plot the execution times and
the numbers of explored maximal traces as follows. For each benchmark,
we consider the highest attempted unroll bound % (from the ones we have attempted)
where both compared algorithms finish % under their respective memory models
before the one-hour timeout.
% \footnote{
% No longer true for any model::: For benchmark \texttt{coin\_all\_td4},
% $\Source$ times out for the lowest unroll bound, in all memory models.
% No longer true::: For benchmark \texttt{eratosthenes}, $\Source$ times out for the lowest unroll bound, in $\PSO$.
% Below: not it doesnt crash, but report fewer traces than should instead, on sc and tso
% For benchmark \texttt{exponential\_bug}, $\Source$ crashes in $\SC$ and $\TSO$.
% These cases were excluded from the plots.}
Then we plot the time and the number of traces obtained
by the two algorithms on the benchmark scaled with the
above unroll bound.

In each plot, the opaque
(resp. semi-transparent) red line represents identity
(resp. order-of-magnitude improvement).
Green dots indicate that a trace reduction was achieved on the
underlying benchmark by the algorithm on the y-axis, as compared
to the algorithm on the x-axis. Benchmarks with no trace reduction
are represented by blue dots.

\begin{figure}[h]
\raggedright
%%%%%%%%%%%%%%%
\begin{subfigure}[b]{0.44\textwidth}
\raggedright
\includegraphics[height=5.4cm]{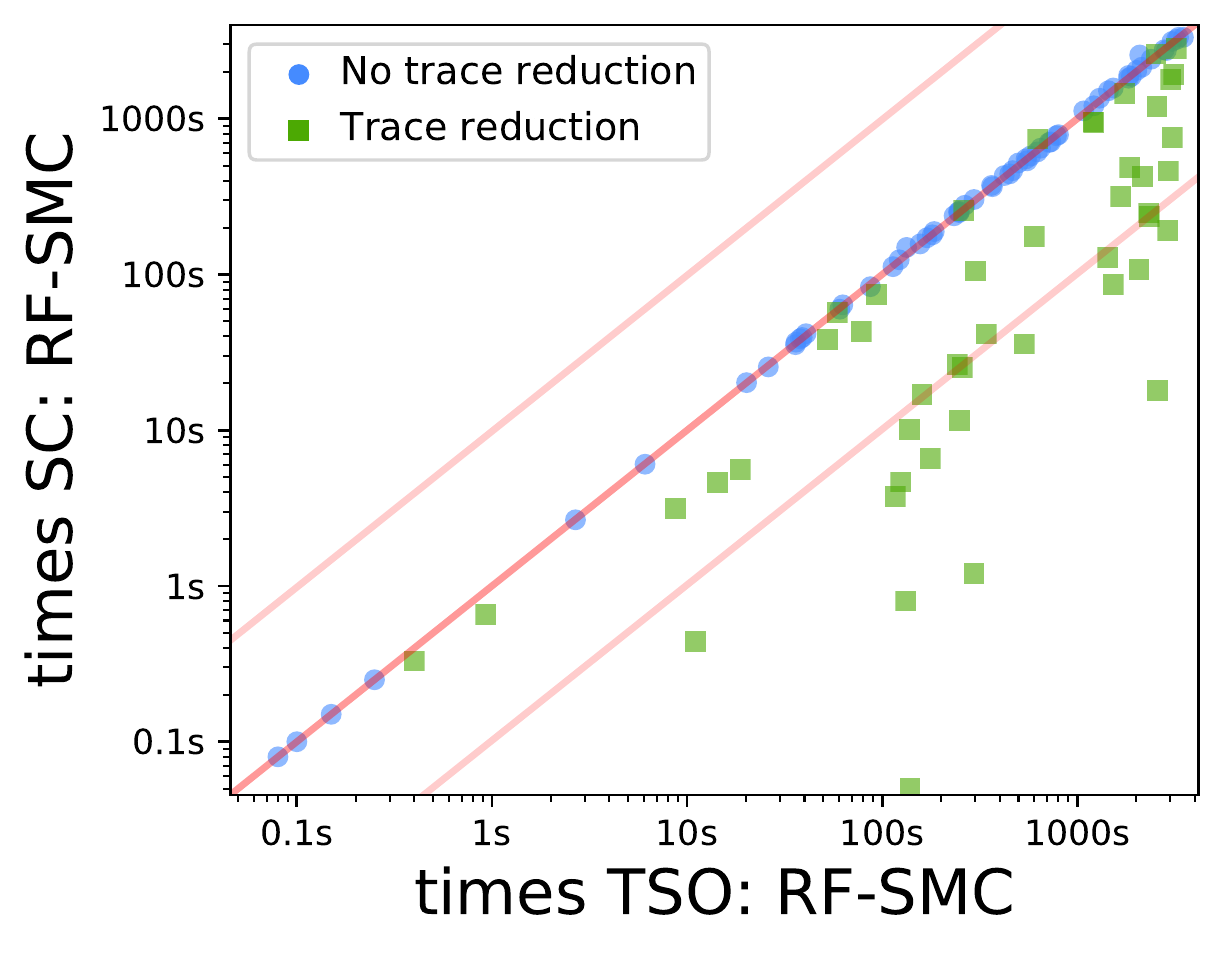}
%\caption{}
\label{subfig:p_sctso_ti}
\end{subfigure}
%%%%%%%%%%%%%%%
\qquad%\qquad
%%%%%%%%%%%%%%%
\begin{subfigure}[b]{0.44\textwidth}
\raggedright
\includegraphics[height=5.4cm]{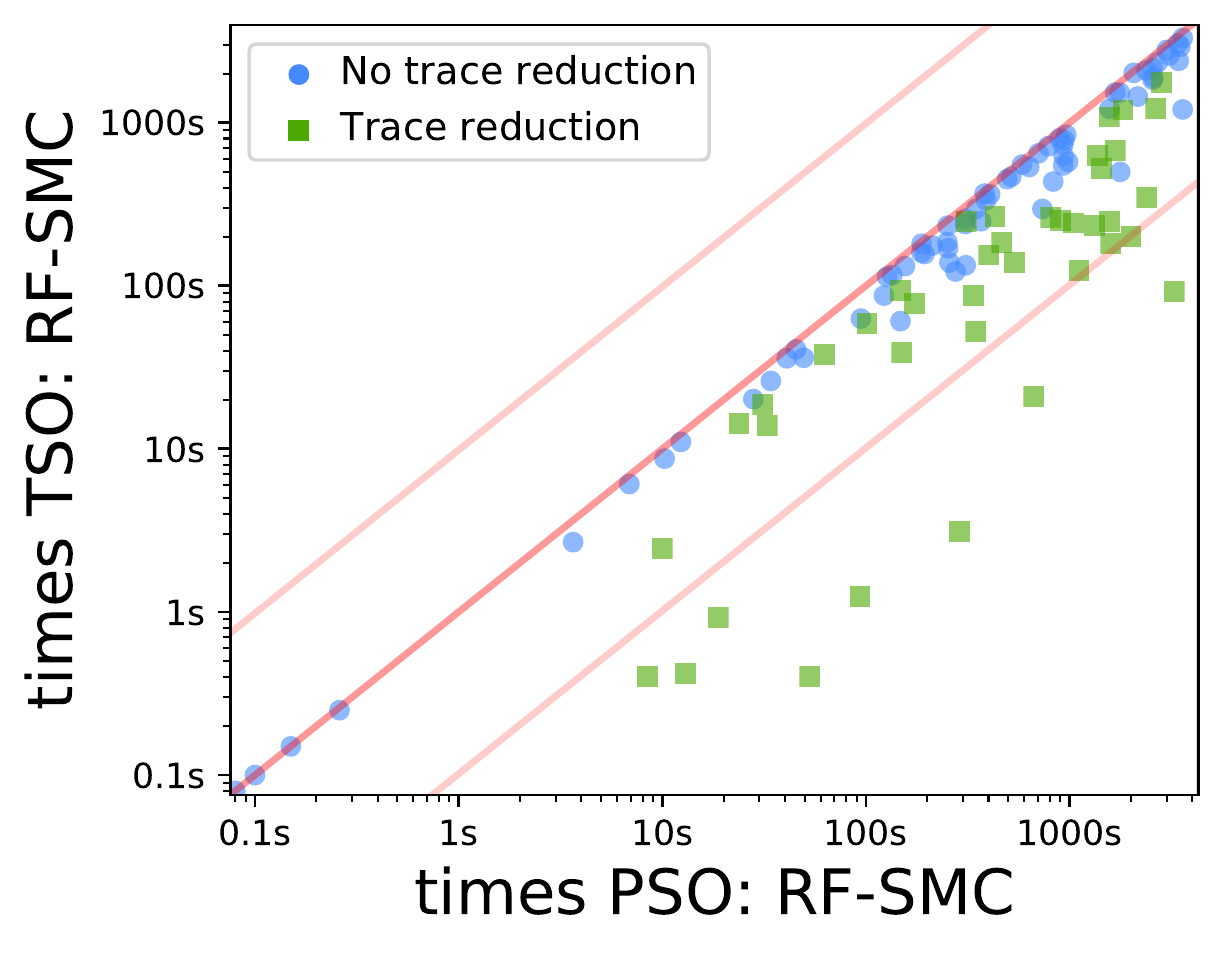}
%\caption{}
\label{subfig:p_tsopso_ti}
\end{subfigure}
\qquad%\qquad
%%%%%%%%%%%%%%%
\vspace{-5mm}
\caption{
Times comparison as $\DCTSOPSOM$ moves from $\SC$ to $\TSO$ (left) and from $\TSO$ to $\PSO$ (right).
}
\label{fig:models_times}
\end{figure}

\cref{fig:models_times} captures how analyzing a concurrent program
by $\DCTSOPSOM$
under more relaxed memory settings affects the execution time.
Unsurprisingly, when a program exhibits additional behavior under
a more relaxed model, more time is required to fully analyze it under
the more relaxed model. Green dots represent such programs.
On the other hand, for programs (represented by blue dots) where
the number of traces stays the same in the more relaxed model,
the time required for analysis is only minorly impacted.

\begin{figure}[h]
\raggedright
%%%%%%%%%%%%%%%
\begin{subfigure}[b]{0.44\textwidth}
\raggedright
\includegraphics[height=5.4cm]{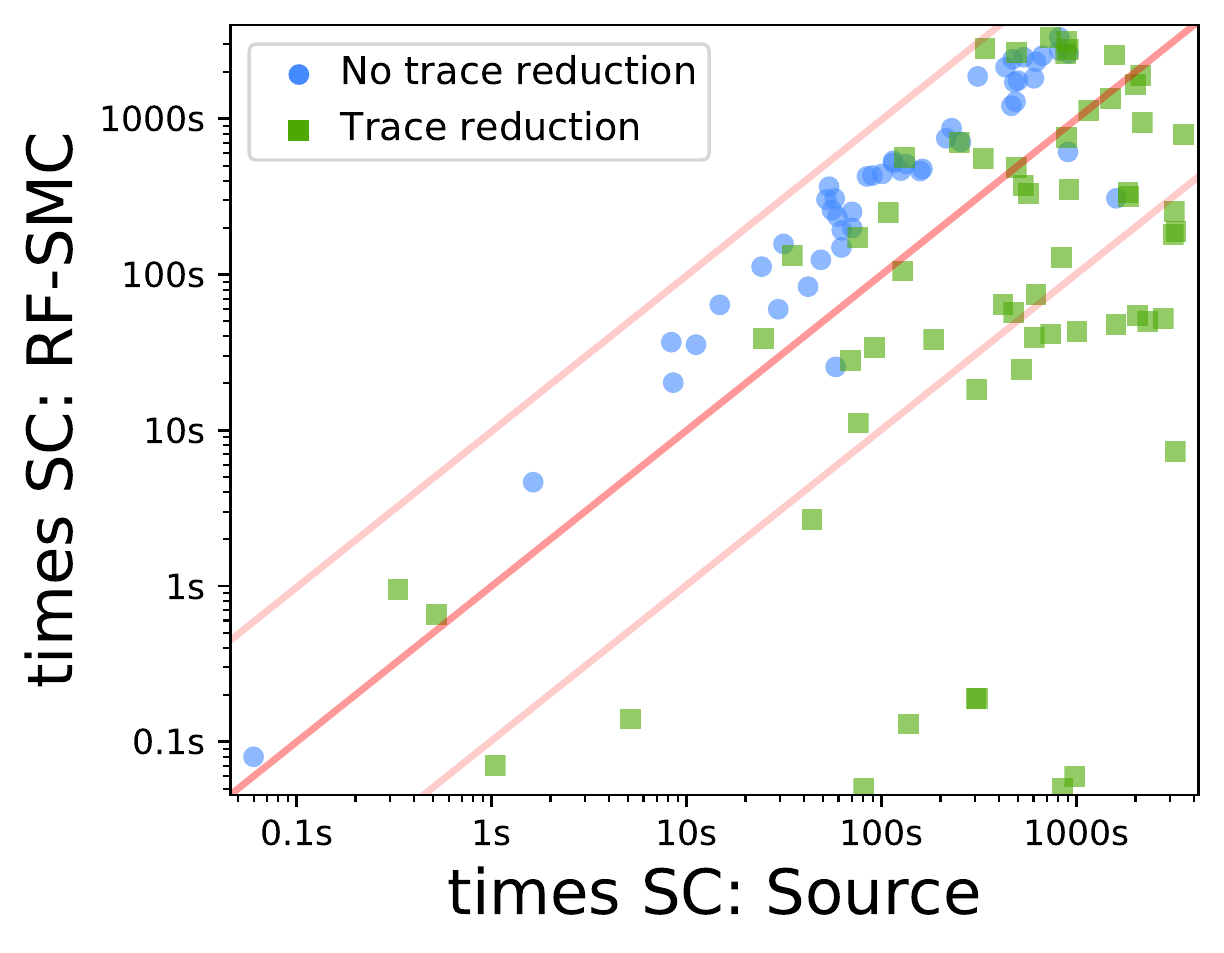}
%\caption{}
\label{subfig:p_sc_exo_source_ti}
\end{subfigure}
%%%%%%%%%%%%%%%
\qquad%\qquad
%%%%%%%%%%%%%%%
\begin{subfigure}[b]{0.44\textwidth}
\raggedright
\includegraphics[height=5.4cm]{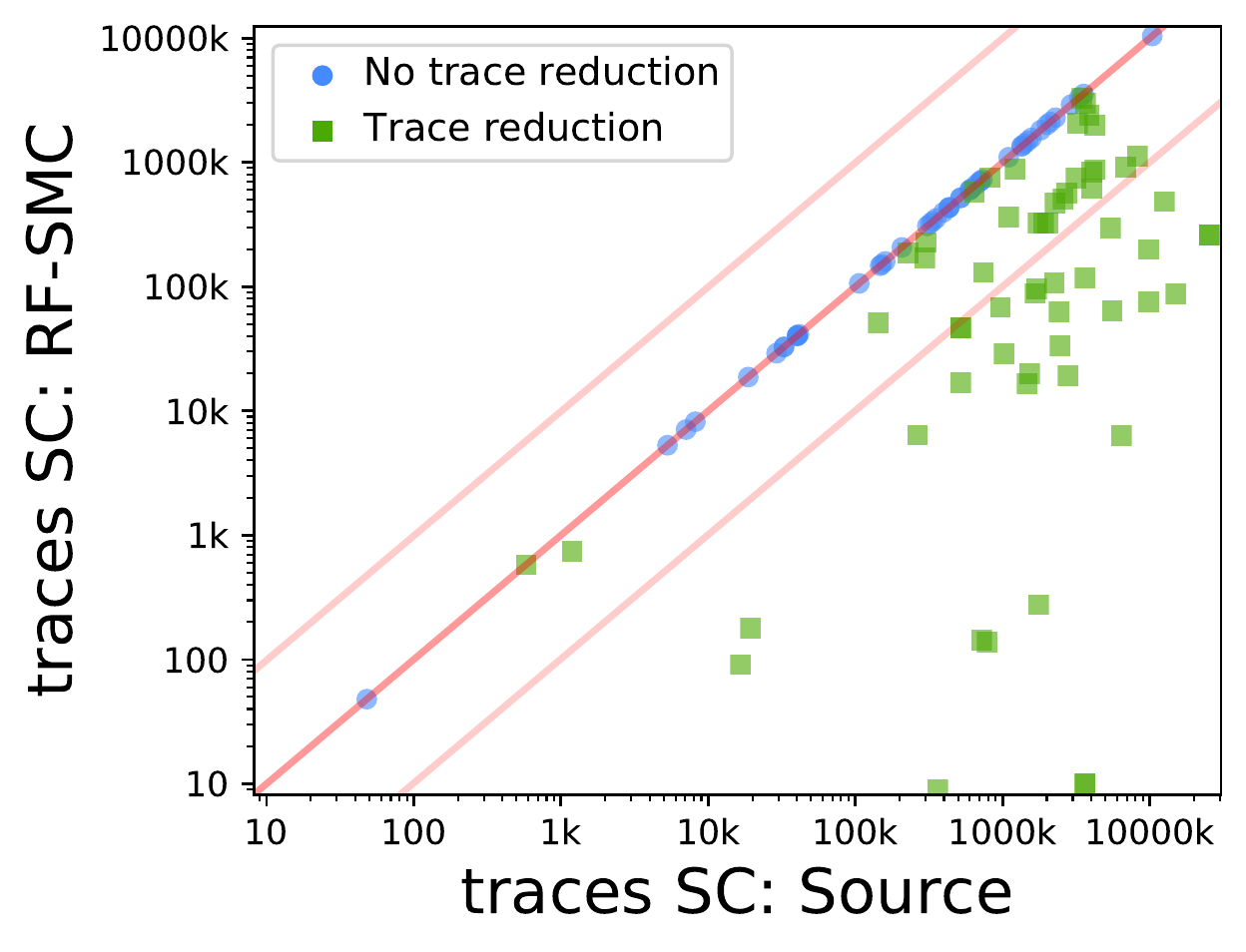}
%\caption{}
\label{subfig:p_sc_exo_source_tr}
\end{subfigure}
\qquad%\qquad
%%%%%%%%%%%%%%%
\vspace{-5mm}
\caption{
Times (left) and traces (right) comparison for $\DCTSOPSOM$ and $\Source$ on the $\SC$ memory model.
}
\label{fig:sc_exo_source}
\end{figure}

\begin{figure}[h]
\centering
%%%%%%%%%%%%%%%
\begin{subfigure}[b]{0.44\textwidth}
\centering
\includegraphics[height=5.4cm]{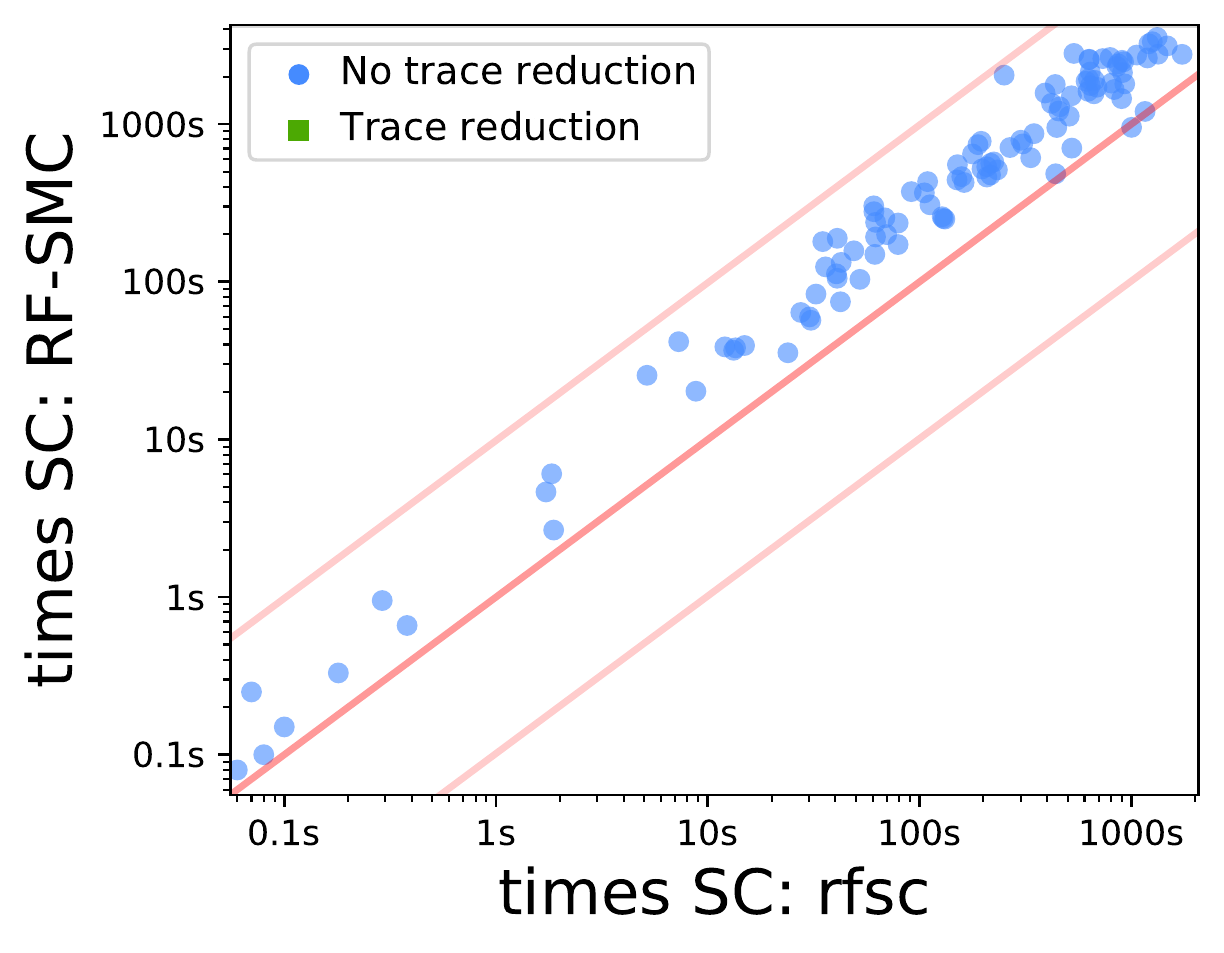}
%\caption{}
\label{subfig:p_sc_exo_rfsc_ti}
\end{subfigure}
\qquad\quad
%%%%%%%%%%%%%%%
\vspace{-5mm}
\caption{
Times comparison for $\DCTSOPSOM$ and $\ReadsFrom$ on the $\SC$ memory model (traces coincide).
}
\label{fig:sc_exo_rfsc}
\end{figure}

\cref{fig:sc_exo_source} compares in $\SC$ the algorithm $\Source$
with our algorithm % reads-from SMC algorithm
$\DCTSOPSOM$ that handles $\SC$ as $\TSO$ where a fence event is
inserted after every buffer-write event. Similar trends are observed
as when $\Source$ and $\DCTSOPSOM$ are compared in $\TSO$ and $\PSO$.
Specifically, there are cases where the RF partitioning offers
reduction of the trace space size to be explored (green dots),
and this often leads to significant speedup of the exploration.
On the other hand, $\Source$ dominates in cases where no RF-based
partitioning is induced (blue dots).

Further,
\cref{fig:sc_exo_rfsc} compares in $\SC$
the algorithm $\DCTSOPSOM$ with $\ReadsFrom$,
the reads-from SMC algorithm for $\SC$ presented by~\citet{Abdulla19}.
These two are essentially identical algorithms, thus unsurprisingly,
the number of explored traces coincides in all cases. However, the
well-engineered implementation of $\ReadsFrom$ is faster than our
implementation of $\DCTSOPSOM$. This comparison provides a rough
illustration of the effect of the optimizations and data-structures
recently employed by $\ReadsFrom$ in the work of~\citet{LangS20}.

\begin{figure}[h]
\raggedright
%%%%%%%%%%%%%%%
\begin{subfigure}[b]{0.445\textwidth}
\raggedright
\includegraphics[height=5.26cm]{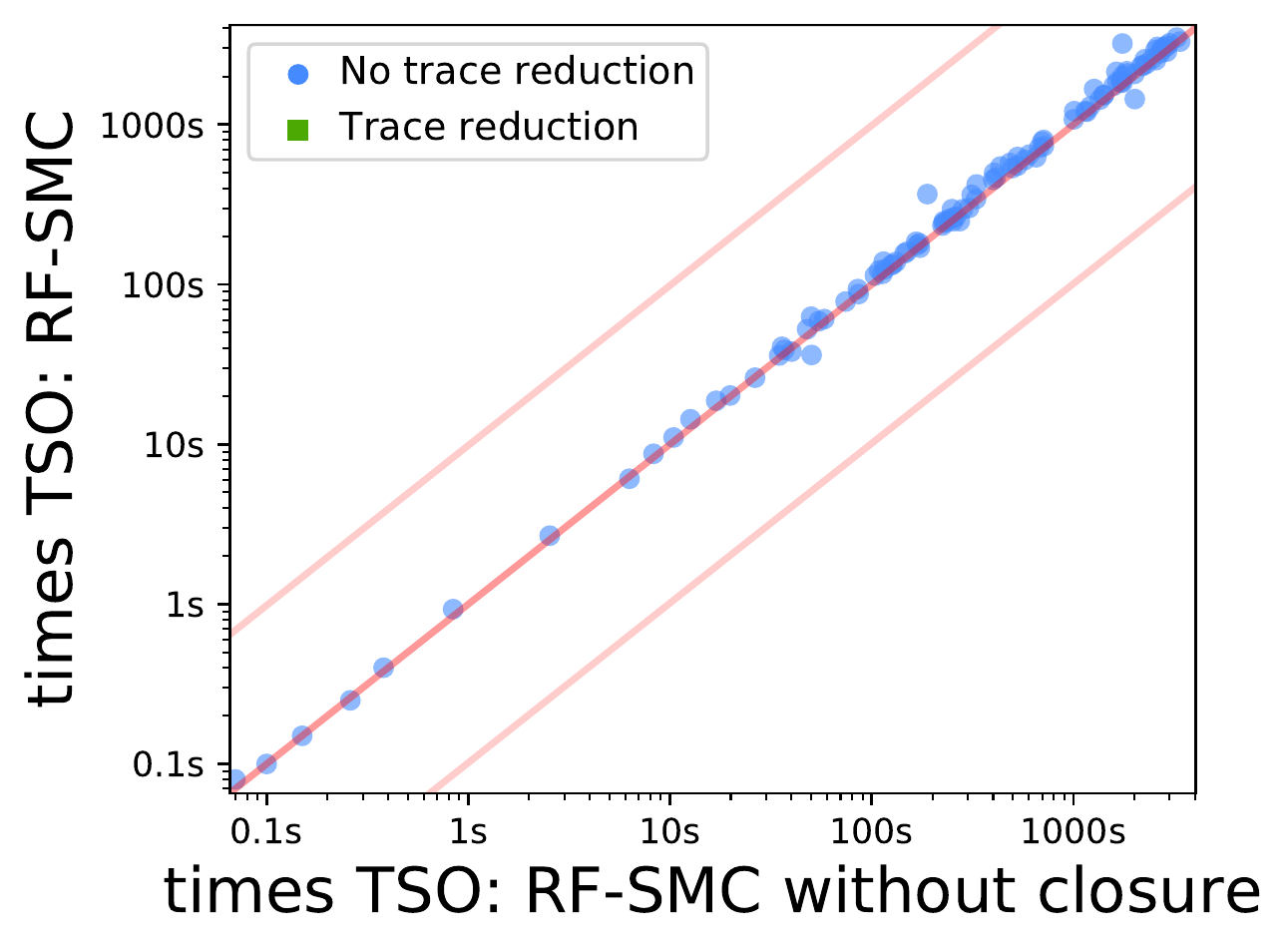}
%\caption{}
\label{subfig:p_tso_exoclosure_ti}
\end{subfigure}
%%%%%%%%%%%%%%%
\qquad%\qquad
%%%%%%%%%%%%%%%
\begin{subfigure}[b]{0.44\textwidth}
\raggedright
\includegraphics[height=5.26cm]{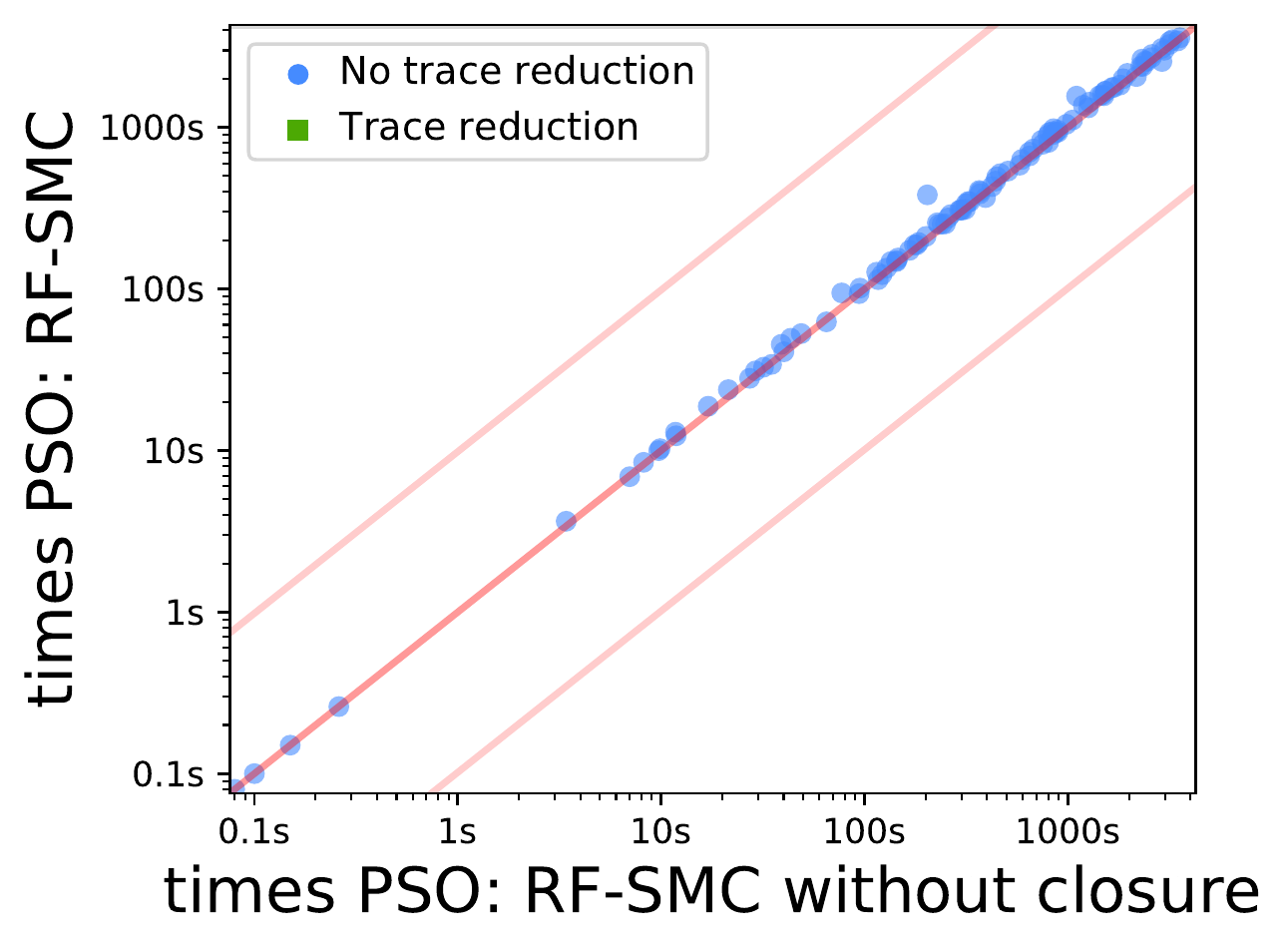}
%\caption{}
\label{subfig:p_pso_exoclosure_ti}
\end{subfigure}
\qquad%\qquad
%%%%%%%%%%%%%%%
\vspace{-5mm}
\caption{
Times comparison for $\DCTSOPSOM$ with and without closure on $\TSO$ (left) and $\PSO$ (right) memory models (traces coincide).
}
\label{fig:exo_closure}
\end{figure}

In \cref{subsec:app_verifyfurther} we have seen that utilizing closure
in consistency verification of realizable instances is mostly detrimental,
whereas in consistency verification of unrealizable cases it is
extremely helpful. This naturally begs a question whether it is
overall beneficial to use closure in SMC. The plots in
\cref{fig:exo_closure} present the results for such an experiment.
The plots demonstrate that the time differences are negligible.
The number of traces is, unsurprisingly, unaffected (it is also supposed
to be unaffected, since closure is sound and $\AlgoTSO$/$\AlgoPSO$ are sound and complete).

We have further considered an auxiliary-trace heuristic for guiding
$\AlgoTSO$ resp. $\AlgoPSO$, similar to the heuristic reported
by~\citet{Abdulla19}. Similar to the paragraph above, this heuristic
provided little-to-no time difference in the model checking task
in our experiments.

%
% Plot showing the percentage of time during SMC spent on AlgoTSO / AlgoPSO ?
%

%\iffalse

\begin{table}
\scriptsize
\newcolumntype{?}{!{\vrule width 1.5pt}}
\setlength{\extrarowheight}{.05em}
\begin{tabular}{? l | c | c ? r r r ? r r ? r r ? }
\specialrule{.15em}{0em}{0em}
\multicolumn{2}{?c|}{\multirow{2}{*}{\textbf{Benchmark}}} & \multirow{2}{*}{U} & \multicolumn{3}{c?}{\textbf{Sequential Consistency}} & \multicolumn{2}{c?}{\textbf{Total Store Order}} & \multicolumn{2}{c?}{\textbf{Partial Store Order}}\\
\cline{4-10}
\multicolumn{2}{?c|}{} & & $\ReadsFrom$ & $\DCTSOPSOM$ & $\Source$ & $\DCTSOPSOM$ & $\Source$ & $\DCTSOPSOM$ & $\Source$\\
\specialrule{.1em}{0em}{0em}
% approxds_append.c
\multirow{4}{*}{\begin{tabular}{l}
\textbf{approxds\_append}\\
% lines of code: 60\\
% variables: U-1\\
threads: U
\end{tabular}}
&
\multirow{2}{*}{Traces}
%                                                                                                      &                                                3 &                                      \textbf{36} &                                      \textbf{36} &                                               63 &                                      \textbf{36} &                                               63 &                                      \textbf{36} &                                               63 \\
%                                                   &                                                  &                                                4 &                                     \textbf{580} &                                     \textbf{580} &                                             2460 &                                     \textbf{580} &                                             2460 &                                     \textbf{580} &                                             2460 \\
                                                                                                     &                                                5 &                                    \textbf{9945} &                                    \textbf{9945} &                                           127740 &                                    \textbf{9945} &                                           127740 &                                    \textbf{9945} &                                           127740 \\
                                                   &                                                  &                                                6 &                                  \textbf{198936} &                                  \textbf{198936} &                                          9847080 &                                  \textbf{198936} &                                          9847080 &                                  \textbf{198936} &                                          9847080 \\
%                                                   &                                                  &                                                7 &                                 \textbf{4645207} &                                                - &                                                - &                                                - &                                                - &                                                - &                                                - \\
\cline{2-10}
&
\multirow{2}{*}{Times}
%                                                                                                      &                                                3 &                                   \textbf{0.07s} &                                            0.12s &                                            0.09s &                                   \textbf{0.08s} &                                            0.10s &                                   \textbf{0.08s} &                                            0.13s \\
%                                                   &                                                  &                                                4 &                                   \textbf{0.19s} &                                            0.60s &                                            0.87s &                                   \textbf{0.48s} &                                            0.56s &                                            0.61s &                                   \textbf{0.56s} \\
                                                                                                     &                                                5 &                                   \textbf{3.00s} &                                            9.64s &                                              32s &                                   \textbf{9.55s} &                                              29s &                                     \textbf{14s} &                                              25s \\
                                                   &                                                  &                                                6 &                                     \textbf{68s} &                                             254s &                                            3161s &                                    \textbf{249s} &                                            3068s &                                    \textbf{368s} &                                            2871s \\
%                                                   &                                                  &                                                7 &                                   \textbf{2447s} &                                                - &                                                - &                                                - &                                                - &                                                - &                                                - \\
\specialrule{.1em}{0em}{0em}
% check_bad_array.c
\multirow{4}{*}{\begin{tabular}{l}
\textbf{check\_bad\_array}\\
% lines of code: 33\\
% variables: 1\\
threads: U
\end{tabular}}
&
\multirow{2}{*}{Traces}
%                                                                                                      &                                                3 &                                      \textbf{34} &                                      \textbf{34} &                                               55 &                                      \textbf{34} &                                               55 &                                      \textbf{34} &                                               55 \\
%                                                   &                                                  &                                                4 &                                     \textbf{415} &                                     \textbf{415} &                                              908 &                                     \textbf{415} &                                              908 &                                     \textbf{415} &                                              908 \\
%                                                   &                                                  &                                                5 &                                    \textbf{3884} &                                    \textbf{3884} &                                            12838 &                                    \textbf{3884} &                                            12838 &                                    \textbf{3884} &                                            12838 \\
                                                                                                     &                                                6 &                                   \textbf{75921} &                                   \textbf{75921} &                                           357368 &                                   \textbf{75921} &                                           357368 &                                   \textbf{75921} &                                           357368 \\
                                                   &                                                  &                                                7 &                                 \textbf{1115240} &                                 \textbf{1115240} &                                          8245810 &                                 \textbf{1115240} &                                          8245810 &                                 \textbf{1115240} &                                                - \\
%                                                   &                                                  &                                                8 &                                 \textbf{9063529} &                                                - &                                                - &                                                - &                                                - &                                                - &                                                - \\
\cline{2-10}
&
\multirow{2}{*}{Times}
%                                                                                                      &                                                3 &                                   \textbf{0.05s} &                                            0.06s &                                   \textbf{0.05s} &                                   \textbf{0.05s} &                                   \textbf{0.05s} &                                            0.06s &                                   \textbf{0.05s} \\
%                                                   &                                                  &                                                4 &                                   \textbf{0.11s} &                                            0.19s &                                            0.17s &                                   \textbf{0.20s} &                                            0.27s &                                            0.21s &                                   \textbf{0.20s} \\
%                                                   &                                                  &                                                5 &                                   \textbf{1.06s} &                                            1.73s &                                            2.40s &                                   \textbf{1.72s} &                                            1.97s &                                   \textbf{1.91s} &                                            3.56s \\
                                                                                                     &                                                6 &                                     \textbf{20s} &                                              42s &                                             102s &                                     \textbf{43s} &                                              94s &                                     \textbf{48s} &                                             284s \\
                                                   &                                                  &                                                7 &                                    \textbf{300s} &                                             790s &                                            3524s &                                    \textbf{802s} &                                            2418s &                                    \textbf{893s} &                                                - \\
%                                                   &                                                  &                                                8 &                                   \textbf{2874s} &                                                - &                                                - &                                                - &                                                - &                                                - &                                                - \\
\specialrule{.1em}{0em}{0em}
% circular_buffer.c
\multirow{4}{*}{\begin{tabular}{l}
\textbf{circular\_buffer}\\
% lines of code: 100\\
% variables: 26\\
% locks: 1\\
threads: 2
\end{tabular}}
&
\multirow{2}{*}{Traces}
                                                                                                     &                                               10 &                                  \textbf{184756} &                                  \textbf{184756} &                                  \textbf{184756} &                                  \textbf{184756} &                                  \textbf{184756} &                                  \textbf{184756} &                                  \textbf{184756} \\
                                                   &                                                  &                                               11 &                                  \textbf{705432} &                                  \textbf{705432} &                                  \textbf{705432} &                                  \textbf{705432} &                                  \textbf{705432} &                                                - &                                  \textbf{705432} \\
%                                                   &                                                  &                                               12 &                                 \textbf{2704156} &                                                - &                                 \textbf{2704156} &                                                - &                                 \textbf{2704156} &                                                - &                                 \textbf{2704156} \\
\cline{2-10}
&
\multirow{2}{*}{Times}
                                                                                                     &                                               10 &                                             140s &                                             442s &                                     \textbf{83s} &                                             435s &                                     \textbf{91s} &                                             831s &                                     \textbf{99s} \\
                                                   &                                                  &                                               11 &                                             639s &                                            2140s &                                    \textbf{431s} &                                            2140s &                                    \textbf{408s} &                                                - &                                    \textbf{383s} \\
%                                                   &                                                  &                                               12 &                                            2997s &                                                - &                                   \textbf{1438s} &                                                - &                                   \textbf{1465s} &                                                - &                                   \textbf{1576s} \\
\specialrule{.1em}{0em}{0em}
% control_flow.c
\multirow{4}{*}{\begin{tabular}{l}
\textbf{control\_flow}\\
% lines of code: 55\\
% variables: 3\\
threads: 2U+2
\end{tabular}}
&
\multirow{2}{*}{Traces}
%                                                                                                      &                                                5 &                                      \textbf{43} &                                      \textbf{43} &                                             5160 &                                      \textbf{43} &                                             5160 &                                      \textbf{43} &                                             5160 \\
%                                                   &                                                  &                                                6 &                                      \textbf{77} &                                      \textbf{77} &                                            55440 &                                      \textbf{77} &                                            55440 &                                      \textbf{77} &                                            55440 \\
                                                                                                     &                                                7 &                                     \textbf{143} &                                     \textbf{143} &                                           720720 &                                     \textbf{143} &                                           720720 &                                     \textbf{143} &                                           720720 \\
                                                   &                                                  &                                               14 &                                   \textbf{16413} &                                   \textbf{16413} &                                                - &                                   \textbf{16413} &                                                - &                                   \textbf{16413} &                                                - \\
\cline{2-10}
&
\multirow{2}{*}{Times}
%                                                                                                      &                                                5 &                                   \textbf{0.06s} &                                            0.08s &                                            1.41s &                                   \textbf{0.08s} &                                            1.49s &                                   \textbf{0.08s} &                                            1.96s \\
%                                                   &                                                  &                                                6 &                                   \textbf{0.06s} &                                            0.11s &                                              14s &                                   \textbf{0.15s} &                                              17s &                                   \textbf{0.12s} &                                              30s \\
                                                                                                     &                                                7 &                                   \textbf{0.08s} &                                            0.19s &                                             310s &                                   \textbf{0.19s} &                                             273s &                                   \textbf{0.21s} &                                             293s \\
                                                   &                                                  &                                               14 &                                   \textbf{7.28s} &                                              42s &                                                - &                                     \textbf{41s} &                                                - &                                     \textbf{45s} &                                                - \\
\specialrule{.1em}{0em}{0em}
% dispatcher.cc
\multirow{4}{*}{\begin{tabular}{l}
\textbf{dispatcher}\\
% lines of code: 151\\
% variables: 2U+1\\
% locks: U+1\\
threads: U+1
\end{tabular}}
&
\multirow{2}{*}{Traces}
%                                                                                                      &                                                2 &                                      \textbf{32} &                                      \textbf{32} &                                      \textbf{32} &                                      \textbf{32} &                                      \textbf{32} &                                      \textbf{32} &                                      \textbf{32} \\
%                                                   &                                                  &                                                3 &                                     \textbf{398} &                                     \textbf{398} &                                     \textbf{398} &                                     \textbf{398} &                                     \textbf{398} &                                     \textbf{398} &                                     \textbf{398} \\
                                                                                                     &                                                4 &                                    \textbf{6854} &                                    \textbf{6854} &                                    \textbf{6854} &                                    \textbf{6854} &                                    \textbf{6854} &                                    \textbf{6854} &                                    \textbf{6854} \\
                                                   &                                                  &                                                5 &                                  \textbf{151032} &                                  \textbf{151032} &                                  \textbf{151032} &                                  \textbf{151032} &                                  \textbf{151032} &                                  \textbf{151032} &                                  \textbf{151032} \\
\cline{2-10}
&
\multirow{2}{*}{Times}
%                                                                                                      &                                                2 &                                   \textbf{0.09s} &                                            0.10s &                                   \textbf{0.09s} &                                            0.10s &                                   \textbf{0.09s} &                                   \textbf{0.11s} &                                   \textbf{0.11s} \\
%                                                   &                                                  &                                                3 &                                   \textbf{0.30s} &                                            0.64s &                                            0.52s &                                            0.66s &                                   \textbf{0.55s} &                                            0.84s &                                   \textbf{0.52s} \\
                                                                                                     &                                                4 &                                   \textbf{8.16s} &                                              16s &                                              22s &                                     \textbf{17s} &                                              23s &                                              23s &                                     \textbf{22s} \\
                                                   &                                                  &                                                5 &                                    \textbf{334s} &                                             612s &                                             901s &                                    \textbf{625s} &                                            1052s &                                    \textbf{935s} &                                            1064s \\
\specialrule{.1em}{0em}{0em}
% electron_microscope_mock.c
\multirow{4}{*}{\begin{tabular}{l}
\textbf{electron\_microscope}\\
% lines of code: 36\\
% variables: 3\\
threads: 2
\end{tabular}}
&
\multirow{2}{*}{Traces}
%                                                                                                      &                                                1 &                                       \textbf{3} &                                       \textbf{3} &                                       \textbf{3} &                                       \textbf{3} &                                       \textbf{3} &                                      \textbf{11} &                                      \textbf{11} \\
%                                                   &                                                  &                                                2 &                                      \textbf{50} &                                      \textbf{50} &                                      \textbf{50} &                                      \textbf{50} &                                      \textbf{50} &                                     \textbf{831} &                                     \textbf{831} \\
                                                                                                     &                                                3 &                                     \textbf{984} &                                     \textbf{984} &                                     \textbf{984} &                                     \textbf{984} &                                     \textbf{984} &                                   \textbf{90866} &                                   \textbf{90866} \\
                                                   &                                                  &                                                4 &                                   \textbf{20347} &                                   \textbf{20347} &                                   \textbf{20347} &                                   \textbf{20347} &                                   \textbf{20347} &                                                - &                                \textbf{11613535} \\
%                                                   &                                                  &                                                5 &                                  \textbf{432842} &                                  \textbf{432842} &                                  \textbf{432842} &                                  \textbf{432842} &                                  \textbf{432842} &                                                - &                                                - \\
%                                                   &                                                  &                                                6 &                                 \textbf{9379670} &                                                - &                                 \textbf{9379670} &                                                - &                                 \textbf{9379670} &                                                - &                                                - \\
%                                                   &                                                  &                                                7 &                                                - &                                                - &                                                - &                                                - &                                                - &                                                - &                                                - \\
\cline{2-10}
&
\multirow{2}{*}{Times}
%                                                                                                      &                                                1 &                                            0.07s &                                   \textbf{0.05s} &                                            0.06s &                                   \textbf{0.05s} &                                   \textbf{0.05s} &                                   \textbf{0.05s} &                                   \textbf{0.05s} \\
%                                                   &                                                  &                                                2 &                                            0.08s &                                   \textbf{0.06s} &                                   \textbf{0.06s} &                                   \textbf{0.06s} &                                            0.08s &                                            0.41s &                                   \textbf{0.21s} \\
                                                                                                     &                                                3 &                                            0.31s &                                            0.41s &                                   \textbf{0.12s} &                                            0.40s &                                   \textbf{0.14s} &                                              53s &                                     \textbf{22s} \\
                                                   &                                                  &                                                4 &                                            5.58s &                                            9.46s &                                   \textbf{1.82s} &                                            9.06s &                                   \textbf{3.29s} &                                                - &                                   \textbf{2479s} \\
%                                                   &                                                  &                                                5 &                                             129s &                                             252s &                                     \textbf{70s} &                                             246s &                                     \textbf{79s} &                                                - &                                                - \\
%                                                   &                                                  &                                                6 &                                            2425s &                                                - &                                   \textbf{1223s} &                                                - &                                   \textbf{1795s} &                                                - &                                                - \\
%                                                   &                                                  &                                                7 &                                                - &                                                - &                                                - &                                                - &                                                - &                                                - &                                                - \\
\specialrule{.1em}{0em}{0em}
% eratosthenes.c
\multirow{4}{*}{\begin{tabular}{l}
\textbf{eratosthenes}\\
% lines of code: 25\\
% variables: U\\
threads: 2
\end{tabular}}
&
\multirow{2}{*}{Traces}
                                                                                                      &                                               13 &                                     \textbf{535} &                                     \textbf{535} &                                             3358 &                                    \textbf{1467} &                                            32800 &                                    \textbf{3375} &                                           175616 \\
%                                                   &                                                  &                                               15 &                                    \textbf{1465} &                                    \textbf{1465} &                                            15460 &                                    \textbf{4791} &                                           240160 &                                   \textbf{16875} &                                                - \\
                                                   &                                                  &                                               17 &                                    \textbf{4667} &                                    \textbf{4667} &                                           100664 &                                   \textbf{29217} &                                          4719488 &                                  \textbf{253125} &                                                - \\
\cline{2-10}
&
\multirow{2}{*}{Times}
                                                                                                      &                                               13 &                                   \textbf{0.20s} &                                            0.49s &                                            0.69s &                                   \textbf{1.05s} &                                              19s &                                   \textbf{3.62s} &                                             277s \\
%                                                   &                                                  &                                               15 &                                   \textbf{0.65s} &                                            1.49s &                                            6.51s &                                   \textbf{4.02s} &                                             195s &                                     \textbf{22s} &                                                - \\
                                                   &                                                  &                                               17 &                                   \textbf{2.12s} &                                            6.70s &                                              46s &                                     \textbf{32s} &                                            2978s &                                    \textbf{475s} &                                                - \\
\specialrule{.1em}{0em}{0em}
% exponential_bug.c
\multirow{4}{*}{\begin{tabular}{l}
\textbf{exponential\_bug}\\
% lines of code: 85\\
% variables: 2\\
% locks: 1\\
threads: 3
\end{tabular}}
&
\multirow{2}{*}{Traces}
                                                                                                     &                                               20 &                                  \textbf{502771} &                                  \textbf{502771} &                                           538332 &                                  \textbf{503714} &                                           538332 &                                  \textbf{581482} &                                                - \\
                                                   &                                                  &                                               25 &                                 \textbf{1063451} &                                 \textbf{1063451} &                                          1644109 &                                 \textbf{1064879} &                                          1644109 &                                 \textbf{1204567} &                                                - \\
%                                                   &                                                  &                                               30 &                                 \textbf{1996531} &                                 \textbf{1996531} &                                          4227681 &                                 \textbf{1998544} &                                          4227681 &                                                - &                                                - \\
\cline{2-10}
&
\multirow{2}{*}{Times}
                                                                                                     &                                               20 &                                             172s &                                             505s &                                     \textbf{91s} &                                             490s &                                     \textbf{87s} &                                    \textbf{770s} &                                                - \\
                                                   &                                                  &                                               25 &                                             412s &                                            1236s &                                    \textbf{297s} &                                            1195s &                                    \textbf{500s} &                                   \textbf{1831s} &                                                - \\
%                                                   &                                                  &                                               30 &                                    \textbf{732s} &                                            2603s &                                             879s &                                            2532s &                                    \textbf{922s} &                                                - &                                                - \\
\specialrule{.1em}{0em}{0em}
% filesystem.c
\multirow{4}{*}{\begin{tabular}{l}
\textbf{filesystem}\\
% lines of code: 70\\
% variables: 58\\
% locks: 58\\
threads: U
\end{tabular}}
&
\multirow{2}{*}{Traces}
%                                                                                                      &                                               20 &                                     \textbf{128} &                                     \textbf{128} &                                     \textbf{128} &                                     \textbf{128} &                                     \textbf{128} &                                     \textbf{128} &                                     \textbf{128} \\
%                                                   &                                                  &                                               22 &                                     \textbf{512} &                                     \textbf{512} &                                     \textbf{512} &                                     \textbf{512} &                                     \textbf{512} &                                     \textbf{512} &                                     \textbf{512} \\
                                                                                                     &                                               24 &                                    \textbf{2048} &                                    \textbf{2048} &                                    \textbf{2048} &                                    \textbf{2048} &                                    \textbf{2048} &                                    \textbf{2048} &                                    \textbf{2048} \\
                                                   &                                                  &                                               26 &                                    \textbf{8192} &                                    \textbf{8192} &                                    \textbf{8192} &                                    \textbf{8192} &                                    \textbf{8192} &                                    \textbf{8192} &                                    \textbf{8192} \\
%                                                   &                                                  &                                               28 &                                                - &                                                - &                                                - &                                                - &                                                - &                                                - &                                                - \\
\cline{2-10}
&
\multirow{2}{*}{Times}
%                                                                                                      &                                               20 &                                            0.26s &                                            0.62s &                                   \textbf{0.24s} &                                            0.63s &                                   \textbf{0.38s} &                                            1.59s &                                   \textbf{0.37s} \\
%                                                   &                                                  &                                               22 &                                            1.58s &                                            2.77s &                                   \textbf{1.49s} &                                            2.79s &                                   \textbf{1.72s} &                                            7.07s &                                   \textbf{0.92s} \\
                                                                                                     &                                               24 &                                            6.99s &                                              13s &                                   \textbf{6.71s} &                                              13s &                                   \textbf{3.92s} &                                              32s &                                   \textbf{6.16s} \\
                                                   &                                                  &                                               26 &                                              30s &                                              60s &                                     \textbf{29s} &                                              61s &                                     \textbf{17s} &                                             148s &                                     \textbf{20s} \\
%                                                   &                                                  &                                               28 &                                                - &                                                - &                                                - &                                                - &                                                - &                                                - &                                                - \\
\specialrule{.1em}{0em}{0em}
% floating_read.c
\multirow{4}{*}{\begin{tabular}{l}
\textbf{floating\_read}\\
% lines of code: 25\\
% variables: 1\\
threads: U
\end{tabular}}
&
\multirow{2}{*}{Traces}
%                                                                                                      &                                                3 &                                       \textbf{4} &                                       \textbf{4} &                                               24 &                                       \textbf{4} &                                               24 &                                       \textbf{4} &                                               24 \\
                                                                                                     &                                                7 &                                       \textbf{8} &                                       \textbf{8} &                                            40320 &                                       \textbf{8} &                                            40320 &                                       \textbf{8} &                                            40320 \\
                                                   &                                                  &                                                8 &                                       \textbf{9} &                                       \textbf{9} &                                           362880 &                                       \textbf{9} &                                           362880 &                                       \textbf{9} &                                           362880 \\
%                                                   &                                                  &                                               10 &                                      \textbf{11} &                                      \textbf{11} &                                                - &                                      \textbf{11} &                                                - &                                      \textbf{11} &                                                - \\
%                                                   &                                                  &                                               20 &                                      \textbf{21} &                                      \textbf{21} &                                                - &                                      \textbf{21} &                                                - &                                      \textbf{21} &                                                - \\
%                                                   &                                                  &                                               30 &                                      \textbf{31} &                                      \textbf{31} &                                                - &                                      \textbf{31} &                                                - &                                      \textbf{31} &                                                - \\
\cline{2-10}
&
\multirow{2}{*}{Times}
%                                                                                                      &                                                3 &                                   \textbf{0.05s} &                                   \textbf{0.05s} &                                   \textbf{0.05s} &                                   \textbf{0.05s} &                                   \textbf{0.05s} &                                   \textbf{0.05s} &                                            0.07s \\
                                                                                                     &                                                7 &                                   \textbf{0.05s} &                                   \textbf{0.05s} &                                            8.49s &                                   \textbf{0.05s} &                                              12s &                                   \textbf{0.05s} &                                              30s \\
                                                   &                                                  &                                                8 &                                            0.07s &                                   \textbf{0.05s} &                                              81s &                                   \textbf{0.05s} &                                             128s &                                   \textbf{0.05s} &                                             267s \\
%                                                   &                                                  &                                               10 &                                   \textbf{0.05s} &                                   \textbf{0.05s} &                                                - &                                   \textbf{0.05s} &                                                - &                                   \textbf{0.05s} &                                                - \\
%                                                   &                                                  &                                               20 &                                            0.08s &                                   \textbf{0.07s} &                                                - &                                   \textbf{0.07s} &                                                - &                                   \textbf{0.07s} &                                                - \\
%                                                   &                                                  &                                               30 &                                   \textbf{0.08s} &                                            0.10s &                                                - &                                   \textbf{0.10s} &                                                - &                                   \textbf{0.10s} &                                                - \\
\specialrule{.1em}{0em}{0em}
% lastwrite.c
\multirow{4}{*}{\begin{tabular}{l}
\textbf{lastwrite}\\
% lines of code: 22\\
% variables: 1\\
threads: U
\end{tabular}}
&
\multirow{2}{*}{Traces}
%                                                                                                      &                                                5 &                                       \textbf{5} &                                       \textbf{5} &                                              120 &                                       \textbf{5} &                                              120 &                                       \textbf{5} &                                              120 \\
%                                                   &                                                  &                                                8 &                                       \textbf{8} &                                       \textbf{8} &                                            40320 &                                       \textbf{8} &                                            40320 &                                       \textbf{8} &                                            40320 \\
                                                                                                     &                                                9 &                                       \textbf{9} &                                       \textbf{9} &                                           362880 &                                       \textbf{9} &                                           362880 &                                       \textbf{9} &                                           362880 \\
                                                   &                                                  &                                               10 &                                      \textbf{10} &                                      \textbf{10} &                                          3628800 &                                      \textbf{10} &                                          3628800 &                                      \textbf{10} &                                          3628800 \\
\cline{2-10}
&
\multirow{2}{*}{Times}
%                                                                                                      &                                                5 &                                            0.06s &                                   \textbf{0.05s} &                                            0.06s &                                   \textbf{0.05s} &                                            0.06s &                                   \textbf{0.05s} &                                            0.08s \\
%                                                   &                                                  &                                                8 &                                            0.06s &                                   \textbf{0.05s} &                                            9.61s &                                   \textbf{0.05s} &                                            6.05s &                                   \textbf{0.05s} &                                            9.26s \\
                                                                                                     &                                                9 &                                   \textbf{0.05s} &                                   \textbf{0.05s} &                                              77s &                                   \textbf{0.05s} &                                              61s &                                   \textbf{0.05s} &                                              64s \\
                                                   &                                                  &                                               10 &                                   \textbf{0.05s} &                                   \textbf{0.05s} &                                             846s &                                   \textbf{0.05s} &                                             868s &                                   \textbf{0.06s} &                                             757s \\
\specialrule{.1em}{0em}{0em}
% lastzero.c
\multirow{4}{*}{\begin{tabular}{l}
\textbf{lastzero}\\
% lines of code: 35\\
% variables: U+1\\
threads: U+1
\end{tabular}}
&
\multirow{2}{*}{Traces}
                                                                                                     &                                               14 &                                   \textbf{69632} &                                   \textbf{69632} &                                   \textbf{69632} &                                   \textbf{69632} &                                   \textbf{69632} &                                   \textbf{69632} &                                   \textbf{69632} \\
                                                   &                                                  &                                               15 &                                  \textbf{147456} &                                  \textbf{147456} &                                  \textbf{147456} &                                  \textbf{147456} &                                  \textbf{147456} &                                  \textbf{147456} &                                  \textbf{147456} \\
%                                                   &                                                  &                                               16 &                                  \textbf{311296} &                                  \textbf{311296} &                                                - &                                  \textbf{311296} &                                                - &                                  \textbf{311296} &                                                - \\
%                                                   &                                                  &                                               17 &                                  \textbf{655360} &                                  \textbf{655360} &                                                - &                                  \textbf{655360} &                                                - &                                  \textbf{655360} &                                                - \\
\cline{2-10}
&
\multirow{2}{*}{Times}
                                                                                                     &                                               14 &                                     \textbf{20s} &                                             121s &                                             425s &                                    \textbf{119s} &                                             532s &                                    \textbf{127s} &                                             308s \\
                                                   &                                                  &                                               15 &                                     \textbf{53s} &                                             308s &                                            1597s &                                    \textbf{303s} &                                            1359s &                                    \textbf{320s} &                                            1143s \\
%                                                   &                                                  &                                               16 &                                    \textbf{111s} &                                             800s &                                                - &                                    \textbf{802s} &                                                - &                                    \textbf{851s} &                                                - \\
%                                                   &                                                  &                                               17 &                                    \textbf{250s} &                                            2045s &                                                - &                                   \textbf{2020s} &                                                - &                                   \textbf{2066s} &                                                - \\
\specialrule{.1em}{0em}{0em}
% mcr_account.c
\multirow{4}{*}{\begin{tabular}{l}
\textbf{account}\\
% lines of code: 55\\
% variables: 1\\
threads: U+1
\end{tabular}}
&
\multirow{2}{*}{Traces}
%                                                                                                      &                                                2 &                                      \textbf{27} &                                      \textbf{27} &                                               36 &                                      \textbf{27} &                                               36 &                                      \textbf{27} &                                               36 \\
%                                                   &                                                  &                                                3 &                                     \textbf{256} &                                     \textbf{256} &                                              576 &                                     \textbf{256} &                                              576 &                                     \textbf{256} &                                              576 \\
                                                                                                     &                                                4 &                                    \textbf{3125} &                                    \textbf{3125} &                                            14400 &                                    \textbf{3125} &                                            14400 &                                    \textbf{3125} &                                            14400 \\
                                                   &                                                  &                                                5 &                                   \textbf{46656} &                                   \textbf{46656} &                                           518400 &                                   \textbf{46656} &                                           518400 &                                   \textbf{46656} &                                           518400 \\
%                                                   &                                                  &                                                6 &                                  \textbf{823543} &                                  \textbf{823543} &                                                - &                                  \textbf{823543} &                                                - &                                  \textbf{823543} &                                                - \\
%                                                   &                                                  &                                                7 &                                                - &                                                - &                                                - &                                                - &                                                - &                                                - &                                                - \\
\cline{2-10}
&
\multirow{2}{*}{Times}
%                                                                                                      &                                                2 &                                   \textbf{0.05s} &                                            0.06s &                                   \textbf{0.05s} &                                            0.06s &                                   \textbf{0.05s} &                                            0.06s &                                   \textbf{0.05s} \\
%                                                   &                                                  &                                                3 &                                   \textbf{0.09s} &                                            0.15s &                                            0.10s &                                            0.15s &                                   \textbf{0.11s} &                                            0.16s &                                   \textbf{0.12s} \\
                                                                                                     &                                                4 &                                   \textbf{0.83s} &                                            1.57s &                                            2.17s &                                   \textbf{1.55s} &                                            3.33s &                                   \textbf{1.72s} &                                            2.92s \\
                                                   &                                                  &                                                5 &                                     \textbf{13s} &                                              28s &                                              69s &                                     \textbf{28s} &                                              73s &                                     \textbf{31s} &                                              77s \\
%                                                   &                                                  &                                                6 &                                    \textbf{178s} &                                             647s &                                                - &                                    \textbf{649s} &                                                - &                                    \textbf{704s} &                                                - \\
%                                                   &                                                  &                                                7 &                                                - &                                                - &                                                - &                                                - &                                                - &                                                - &                                                - \\
\specialrule{.1em}{0em}{0em}
% mcr_airline.c
\multirow{4}{*}{\begin{tabular}{l}
\textbf{airline}\\
% lines of code: 50\\
% variables: 3\\
threads: U
\end{tabular}}
&
\multirow{2}{*}{Traces}
%                                                                                                      &                                                2 &                                       \textbf{4} &                                       \textbf{4} &                                       \textbf{4} &                                       \textbf{4} &                                       \textbf{4} &                                       \textbf{4} &                                       \textbf{4} \\
%                                                   &                                                  &                                                3 &                                      \textbf{27} &                                      \textbf{27} &                                               36 &                                      \textbf{27} &                                               36 &                                      \textbf{27} &                                               36 \\
%                                                   &                                                  &                                                4 &                                     \textbf{256} &                                     \textbf{256} &                                              576 &                                     \textbf{256} &                                              576 &                                     \textbf{256} &                                              576 \\
                                                                                                     &                                                5 &                                    \textbf{3125} &                                    \textbf{3125} &                                            14400 &                                    \textbf{3125} &                                            14400 &                                    \textbf{3125} &                                            14400 \\
                                                   &                                                  &                                                6 &                                   \textbf{46656} &                                   \textbf{46656} &                                           518400 &                                   \textbf{46656} &                                           518400 &                                   \textbf{46656} &                                           518400 \\
%                                                   &                                                  &                                                7 &                                  \textbf{823543} &                                  \textbf{823543} &                                                - &                                  \textbf{823543} &                                                - &                                  \textbf{823543} &                                                - \\
%                                                   &                                                  &                                                8 &                                                - &                                                - &                                                - &                                                - &                                                - &                                                - &                                                - \\
\cline{2-10}
&
\multirow{2}{*}{Times}
%                                                                                                      &                                                2 &                                   \textbf{0.05s} &                                   \textbf{0.05s} &                                   \textbf{0.05s} &                                   \textbf{0.05s} &                                   \textbf{0.05s} &                                   \textbf{0.05s} &                                            0.07s \\
%                                                   &                                                  &                                                3 &                                            0.06s &                                            0.06s &                                   \textbf{0.05s} &                                            0.06s &                                   \textbf{0.05s} &                                   \textbf{0.06s} &                                            0.07s \\
%                                                   &                                                  &                                                4 &                                   \textbf{0.09s} &                                            0.17s &                                            0.11s &                                            0.16s &                                   \textbf{0.12s} &                                            0.19s &                                   \textbf{0.18s} \\
                                                                                                     &                                                5 &                                   \textbf{0.91s} &                                            1.87s &                                            3.69s &                                   \textbf{1.86s} &                                            4.05s &                                   \textbf{2.26s} &                                            3.48s \\
                                                   &                                                  &                                                6 &                                   \textbf{8.82s} &                                              34s &                                              92s &                                     \textbf{34s} &                                             122s &                                     \textbf{41s} &                                             157s \\
%                                                   &                                                  &                                                7 &                                    \textbf{195s} &                                             778s &                                                - &                                    \textbf{783s} &                                                - &                                    \textbf{942s} &                                                - \\
%                                                   &                                                  &                                                8 &                                                - &                                                - &                                                - &                                                - &                                                - &                                                - &                                                - \\
\specialrule{.1em}{0em}{0em}
% mcr_bubblesort.c
\multirow{4}{*}{\begin{tabular}{l}
\textbf{bubblesort}\\
% lines of code: 48\\
% variables: U\\
threads: U
\end{tabular}}
&
\multirow{2}{*}{Traces}
                                                                                                      &                                                2 &                                      \textbf{10} &                                      \textbf{10} &                                      \textbf{10} &                                      \textbf{10} &                                      \textbf{10} &                                      \textbf{10} &                                      \textbf{10} \\
                                                   &                                                  &                                                3 &                                   \textbf{51216} &                                   \textbf{51216} &                                           143448 &                                   \textbf{51216} &                                           143448 &                                   \textbf{57702} &                                           143448 \\
%                                                   &                                                  &                                                4 &                                                - &                                                - &                                                - &                                                - &                                                - &                                                - &                                                - \\
\cline{2-10}
&
\multirow{2}{*}{Times}
                                                                                                      &                                                2 &                                   \textbf{0.05s} &                                   \textbf{0.05s} &                                   \textbf{0.05s} &                                   \textbf{0.05s} &                                   \textbf{0.05s} &                                   \textbf{0.05s} &                                   \textbf{0.05s} \\
                                                   &                                                  &                                                3 &                                     \textbf{12s} &                                              39s &                                              25s &                                              38s &                                     \textbf{34s} &                                              63s &                                     \textbf{29s} \\
%                                                   &                                                  &                                                4 &                                                - &                                                - &                                                - &                                                - &                                                - &                                                - &                                                - \\
\specialrule{.1em}{0em}{0em}
% mcr_rv_example.c
\multirow{4}{*}{\begin{tabular}{l}
\textbf{rv\_example}\\
% lines of code: 66\\
% variables: 2\\
% locks: 1\\
threads: 3
\end{tabular}}
&
\multirow{2}{*}{Traces}
                                                                                                     &                                               25 &                                  \textbf{102028} &                                  \textbf{102028} &                                           180478 &                                  \textbf{102028} &                                           180478 &                                  \textbf{107303} &                                           180478 \\
                                                   &                                                  &                                               30 &                                  \textbf{169383} &                                  \textbf{169383} &                                           295848 &                                  \textbf{169383} &                                           295848 &                                  \textbf{176988} &                                           295848 \\
\cline{2-10}
&
\multirow{2}{*}{Times}
                                                                                                     &                                               25 &                                     \textbf{50s} &                                             124s &                                              61s &                                             125s &                                     \textbf{53s} &                                             157s &                                     \textbf{61s} \\
                                                   &                                                  &                                               30 &                                             131s &                                             250s &                                    \textbf{108s} &                                             248s &                                    \textbf{126s} &                                             311s &                                    \textbf{132s} \\
\specialrule{.1em}{0em}{0em}
\end{tabular}
% \addtocounter{table}{-1}
% \vspace{-2mm}
\caption{Part1: Related papers and works, and synthetic benchmarks.}
\label{tab:all_other_1}
\end{table}

\begin{table}
\scriptsize
\newcolumntype{?}{!{\vrule width 1.5pt}}
\setlength{\extrarowheight}{.08em}
\begin{tabular}{? l | c | c ? r r r ? r r ? r r ? }
\specialrule{.15em}{0em}{0em}
\multicolumn{2}{?c|}{\multirow{2}{*}{\textbf{Benchmark}}} & \multirow{2}{*}{U} & \multicolumn{3}{c?}{\textbf{Sequential Consistency}} & \multicolumn{2}{c?}{\textbf{Total Store Order}} & \multicolumn{2}{c?}{\textbf{Partial Store Order}}\\
\cline{4-10}
\multicolumn{2}{?c|}{} & & $\ReadsFrom$ & $\DCTSOPSOM$ & $\Source$ & $\DCTSOPSOM$ & $\Source$ & $\DCTSOPSOM$ & $\Source$\\
\specialrule{.1em}{0em}{0em}
% multiprodcon.cc
\multirow{4}{*}{\begin{tabular}{l}
\textbf{multiprodcon}\\
% lines of code: 108\\
% variables: 2U+1\\
% locks: U+1\\
threads: 2U
\end{tabular}}
&
\multirow{2}{*}{Traces}
%                                                                                                      &                                                2 &                                      \textbf{18} &                                      \textbf{18} &                                      \textbf{18} &                                      \textbf{18} &                                      \textbf{18} &                                      \textbf{18} &                                      \textbf{18} \\
%                                                   &                                                  &                                                3 &                                     \textbf{162} &                                     \textbf{162} &                                     \textbf{162} &                                     \textbf{162} &                                     \textbf{162} &                                     \textbf{162} &                                     \textbf{162} \\
                                                                                                     &                                                4 &                                    \textbf{1944} &                                    \textbf{1944} &                                    \textbf{1944} &                                    \textbf{1944} &                                    \textbf{1944} &                                    \textbf{1944} &                                    \textbf{1944} \\
                                                   &                                                  &                                                5 &                                   \textbf{29160} &                                   \textbf{29160} &                                   \textbf{29160} &                                   \textbf{29160} &                                   \textbf{29160} &                                   \textbf{29160} &                                   \textbf{29160} \\
\cline{2-10}
&
\multirow{2}{*}{Times}
%                                                                                                      &                                                2 &                                            0.11s &                                            0.11s &                                   \textbf{0.10s} &                                            0.11s &                                   \textbf{0.10s} &                                            0.12s &                                   \textbf{0.10s} \\
%                                                   &                                                  &                                                3 &                                   \textbf{0.19s} &                                            0.33s &                                            0.30s &                                            0.34s &                                   \textbf{0.18s} &                                            0.42s &                                   \textbf{0.24s} \\
                                                                                                     &                                                4 &                                   \textbf{2.13s} &                                            4.17s &                                            2.48s &                                            4.26s &                                   \textbf{1.43s} &                                            5.63s &                                   \textbf{2.28s} \\
                                                   &                                                  &                                                5 &                                     \textbf{32s} &                                              84s &                                              42s &                                              87s &                                     \textbf{31s} &                                             123s &                                     \textbf{34s} \\
\specialrule{.1em}{0em}{0em}
% optimistic_lock_2thr.c
\multirow{4}{*}{\begin{tabular}{l}
\textbf{opt\_lock2}\\
% lines of code: 31\\
% variables: 2\\
threads: 2
\end{tabular}}
&
\multirow{2}{*}{Traces}
%                                                                                                      &                                                4 &                                      \textbf{43} &                                      \textbf{43} &                                              346 &                                      \textbf{55} &                                              478 &                                      \textbf{91} &                                              998 \\
                                                                                                     &                                                8 &                                      \textbf{91} &                                      \textbf{91} &                                            16714 &                                     \textbf{119} &                                            23098 &                                     \textbf{315} &                                           194470 \\
                                                   &                                                  &                                               12 &                                     \textbf{139} &                                     \textbf{139} &                                           785662 &                                     \textbf{183} &                                          1085758 &                                     \textbf{667} &                                                - \\
%                                                   &                                                  &                                               16 &                                     \textbf{187} &                                     \textbf{187} &                                                - &                                     \textbf{247} &                                                - &                                    \textbf{1147} &                                                - \\
%                                                   &                                                  &                                               20 &                                     \textbf{235} &                                     \textbf{235} &                                                - &                                     \textbf{311} &                                                - &                                    \textbf{1755} &                                                - \\
%                                                   &                                                  &                                               24 &                                     \textbf{283} &                                     \textbf{283} &                                                - &                                     \textbf{375} &                                                - &                                    \textbf{2491} &                                                - \\
\cline{2-10}
&
\multirow{2}{*}{Times}
%                                                                                                      &                                                4 &                                   \textbf{0.07s} &                                   \textbf{0.07s} &                                            0.12s &                                   \textbf{0.07s} &                                            0.18s &                                   \textbf{0.09s} &                                            0.24s \\
                                                                                                     &                                                8 &                                            0.10s &                                   \textbf{0.09s} &                                            3.56s &                                   \textbf{0.10s} &                                            8.33s &                                   \textbf{0.31s} &                                              54s \\
                                                   &                                                  &                                               12 &                                   \textbf{0.08s} &                                            0.13s &                                             137s &                                   \textbf{0.15s} &                                             386s &                                   \textbf{0.92s} &                                                - \\
%                                                   &                                                  &                                               16 &                                   \textbf{0.10s} &                                            0.18s &                                                - &                                   \textbf{0.22s} &                                                - &                                   \textbf{2.20s} &                                                - \\
%                                                   &                                                  &                                               20 &                                   \textbf{0.17s} &                                            0.25s &                                                - &                                   \textbf{0.30s} &                                                - &                                   \textbf{4.57s} &                                                - \\
%                                                   &                                                  &                                               24 &                                   \textbf{0.18s} &                                            0.33s &                                                - &                                   \textbf{0.40s} &                                                - &                                   \textbf{8.45s} &                                                - \\
\specialrule{.1em}{0em}{0em}
% optimistic_lock_3thr.c
\multirow{4}{*}{\begin{tabular}{l}
\textbf{opt\_lock3}\\
% lines of code: 31\\
% variables: 2\\
threads: 3
\end{tabular}}
&
\multirow{2}{*}{Traces}
%                                                                                                      &                                                1 &                                      \textbf{79} &                                      \textbf{79} &                                              282 &                                      \textbf{91} &                                              756 &                                      \textbf{91} &                                              756 \\
                                                                                                     &                                                2 &                                    \textbf{3103} &                                    \textbf{3103} &                                            69252 &                                    \textbf{5080} &                                           217992 &                                    \textbf{7852} &                                           435798 \\
                                                   &                                                  &                                                3 &                                   \textbf{87451} &                                   \textbf{87451} &                                         15036174 &                                  \textbf{151570} &                                                - &                                  \textbf{425260} &                                                - \\
%                                                   &                                                  &                                                4 &                                 \textbf{2363899} &                                 \textbf{2363899} &                                                - &                                 \textbf{4123918} &                                                - &                                                - &                                                - \\
%                                                   &                                                  &                                                5 &                                                - &                                                - &                                                - &                                                - &                                                - &                                                - &                                                - \\
\cline{2-10}
&
\multirow{2}{*}{Times}
%                                                                                                      &                                                1 &                                   \textbf{0.06s} &                                            0.08s &                                            0.07s &                                   \textbf{0.08s} &                                            0.23s &                                   \textbf{0.08s} &                                            0.22s \\
                                                                                                     &                                                2 &                                   \textbf{0.55s} &                                            1.47s &                                            6.75s &                                   \textbf{2.23s} &                                              47s &                                   \textbf{4.27s} &                                             115s \\
                                                   &                                                  &                                                3 &                                     \textbf{28s} &                                              54s &                                            2049s &                                     \textbf{88s} &                                                - &                                    \textbf{337s} &                                                - \\
%                                                   &                                                  &                                                4 &                                    \textbf{627s} &                                            1934s &                                                - &                                   \textbf{3121s} &                                                - &                                                - &                                                - \\
%                                                   &                                                  &                                                5 &                                                - &                                                - &                                                - &                                                - &                                                - &                                                - &                                                - \\
\specialrule{.1em}{0em}{0em}
% overtake.c
\multirow{4}{*}{\begin{tabular}{l}
\textbf{overtake}\\
% lines of code: 36\\
% variables: 2\\
threads: U
\end{tabular}}
&
\multirow{2}{*}{Traces}
%                                                                                                      &                                                1 &                                       \textbf{2} &                                       \textbf{2} &                                       \textbf{2} &                                       \textbf{2} &                                       \textbf{2} &                                       \textbf{3} &                                       \textbf{3} \\
%                                                   &                                                  &                                                2 &                                      \textbf{23} &                                      \textbf{23} &                                               50 &                                      \textbf{23} &                                               50 &                                      \textbf{49} &                                              196 \\
                                                                                                     &                                                3 &                                     \textbf{826} &                                     \textbf{826} &                                             6168 &                                     \textbf{826} &                                             6168 &                                    \textbf{2197} &                                            79092 \\
                                                   &                                                  &                                                4 &                                   \textbf{62893} &                                   \textbf{62893} &                                          2418000 &                                   \textbf{62893} &                                          2418000 &                                  \textbf{194481} &                                                - \\
%                                                   &                                                  &                                                5 &                                 \textbf{8291996} &                                                - &                                                - &                                                - &                                                - &                                                - &                                                - \\
\cline{2-10}
&
\multirow{2}{*}{Times}
%                                                                                                      &                                                1 &                                   \textbf{0.05s} &                                   \textbf{0.05s} &                                   \textbf{0.05s} &                                   \textbf{0.05s} &                                   \textbf{0.05s} &                                   \textbf{0.05s} &                                            0.06s \\
%                                                   &                                                  &                                                2 &                                   \textbf{0.05s} &                                            0.06s &                                   \textbf{0.05s} &                                   \textbf{0.05s} &                                   \textbf{0.05s} &                                            0.08s &                                   \textbf{0.07s} \\
                                                                                                     &                                                3 &                                   \textbf{0.25s} &                                            0.40s &                                            0.82s &                                   \textbf{0.40s} &                                            0.71s &                                   \textbf{1.37s} &                                              33s \\
                                                   &                                                  &                                                4 &                                     \textbf{15s} &                                              39s &                                             605s &                                     \textbf{39s} &                                             482s &                                    \textbf{150s} &                                                - \\
%                                                   &                                                  &                                                5 &                                   \textbf{2205s} &                                                - &                                                - &                                                - &                                                - &                                                - &                                                - \\
\specialrule{.1em}{0em}{0em}
% parker.c
\multirow{4}{*}{\begin{tabular}{l}
\textbf{parker}\\
% lines of code: 134\\
% variables: 6\\
threads: 2
\end{tabular}}
&
\multirow{2}{*}{Traces}
%                                                                                                      &                                               10 &                                    \textbf{5657} &                                    \textbf{5657} &                                             6342 &                                    \textbf{6188} &                                             7083 &                                    \textbf{6556} &                                             7661 \\
%                                                   &                                                  &                                               20 &                                   \textbf{40087} &                                   \textbf{40087} &                                            44857 &                                   \textbf{43548} &                                            49838 &                                   \textbf{46381} &                                            54191 \\
%                                                   &                                                  &                                               30 &                                  \textbf{129317} &                                  \textbf{129317} &                                           144572 &                                  \textbf{140108} &                                           160293 &                                  \textbf{149506} &                                           174621 \\
                                                                                                     &                                               40 &                                  \textbf{299347} &                                  \textbf{299347} &                                           334487 &                                  \textbf{323868} &                                           370448 &                                  \textbf{345931} &                                           403951 \\
                                                   &                                                  &                                               50 &                                  \textbf{576177} &                                  \textbf{576177} &                                           643602 &                                  \textbf{622828} &                                           712303 &                                                - &                                  \textbf{777181} \\
%                                                   &                                                  &                                               60 &                                  \textbf{985807} &                                                - &                                          1100917 &                                                - &                                 \textbf{1217858} &                                                - &                                 \textbf{1329311} \\
%                                                   &                                                  &                                               70 &                                 \textbf{1554237} &                                                - &                                          1735432 &                                                - &                                 \textbf{1919113} &                                                - &                                 \textbf{2095341} \\
%                                                   &                                                  &                                               80 &                                                - &                                                - &                                 \textbf{2576147} &                                                - &                                 \textbf{2848068} &                                                - &                                                - \\
\cline{2-10}
&
\multirow{2}{*}{Times}
%                                                                                                      &                                               10 &                                            2.18s &                                            4.41s &                                   \textbf{1.21s} &                                            4.82s &                                   \textbf{1.55s} &                                            7.50s &                                   \textbf{2.29s} \\
%                                                   &                                                  &                                               20 &                                              25s &                                              53s &                                     \textbf{16s} &                                              59s &                                     \textbf{21s} &                                              89s &                                     \textbf{24s} \\
%                                                   &                                                  &                                               30 &                                             117s &                                             266s &                                     \textbf{49s} &                                             292s &                                     \textbf{67s} &                                             435s &                                    \textbf{102s} \\
                                                                                                     &                                               40 &                                             285s &                                             989s &                                    \textbf{223s} &                                            1083s &                                    \textbf{180s} &                                            1563s &                                    \textbf{291s} \\
                                                   &                                                  &                                               50 &                                             535s &                                            2807s &                                    \textbf{338s} &                                            3216s &                                    \textbf{425s} &                                                - &                                    \textbf{703s} \\
%                                                   &                                                  &                                               60 &                                            1130s &                                                - &                                    \textbf{807s} &                                                - &                                    \textbf{762s} &                                                - &                                   \textbf{1300s} \\
%                                                   &                                                  &                                               70 &                                            2030s &                                                - &                                   \textbf{1280s} &                                                - &                                   \textbf{1529s} &                                                - &                                   \textbf{2138s} \\
%                                                   &                                                  &                                               80 &                                                - &                                                - &                                   \textbf{2050s} &                                                - &                                   \textbf{2654s} &                                                - &                                                - \\
\specialrule{.1em}{0em}{0em}
% pgsql_noassert.c
\multirow{4}{*}{\begin{tabular}{l}
\textbf{pgsql}\\
% lines of code: 61\\
% variables: 5\\
threads: 2
\end{tabular}}
&
\multirow{2}{*}{Traces}
%                                                                                                      &                                                1 &                                       \textbf{4} &                                       \textbf{4} &                                       \textbf{4} &                                       \textbf{4} &                                       \textbf{4} &                                      \textbf{10} &                                               28 \\
%                                                   &                                                  &                                                2 &                                      \textbf{85} &                                      \textbf{85} &                                      \textbf{85} &                                      \textbf{85} &                                      \textbf{85} &                                     \textbf{877} &                                             2917 \\
                                                                                                     &                                                3 &                                    \textbf{3906} &                                    \textbf{3906} &                                    \textbf{3906} &                                    \textbf{3906} &                                    \textbf{3906} &                                  \textbf{166666} &                                           555546 \\
                                                   &                                                  &                                                4 &                                  \textbf{335923} &                                  \textbf{335923} &                                  \textbf{335923} &                                  \textbf{335923} &                                  \textbf{335923} &                                                - &                                                - \\
%                                                   &                                                  &                                                5 &                                                - &                                                - &                                                - &                                                - &                                                - &                                                - &                                                - \\
\cline{2-10}
&
\multirow{2}{*}{Times}
%                                                                                                      &                                                1 &                                            0.07s &                                   \textbf{0.05s} &                                            0.07s &                                   \textbf{0.05s} &                                   \textbf{0.05s} &                                   \textbf{0.05s} &                                   \textbf{0.05s} \\
%                                                   &                                                  &                                                2 &                                            0.09s &                                            0.09s &                                   \textbf{0.08s} &                                            0.09s &                                   \textbf{0.06s} &                                   \textbf{0.89s} &                                            0.90s \\
                                                                                                     &                                                3 &                                            1.01s &                                            3.21s &                                   \textbf{0.91s} &                                            3.10s &                                   \textbf{0.56s} &                                             289s &                                    \textbf{127s} \\
                                                   &                                                  &                                                4 &                                             109s &                                             432s &                                     \textbf{90s} &                                             422s &                                     \textbf{70s} &                                                - &                                                - \\
%                                                   &                                                  &                                                5 &                                                - &                                                - &                                                - &                                                - &                                                - &                                                - &                                                - \\
\specialrule{.1em}{0em}{0em}
% poke.cc
\multirow{4}{*}{\begin{tabular}{l}
\textbf{poke}\\
% lines of code: 134\\
% variables: U+2\\
% locks: U+1\\
threads: U+8
\end{tabular}}
&
\multirow{2}{*}{Traces}
%                                                                                                      &                                                2 &                                     \textbf{104} &                                     \textbf{104} &                                     \textbf{104} &                                     \textbf{104} &                                     \textbf{104} &                                     \textbf{104} &                                     \textbf{104} \\
%                                                   &                                                  &                                                3 &                                     \textbf{584} &                                              711 &                                              711 &                                     \textbf{711} &                                     \textbf{711} &                                     \textbf{711} &                                     \textbf{711} \\
                                                                                                     &                                                4 &                                    \textbf{1946} &                                             2636 &                                             2636 &                                    \textbf{2636} &                                    \textbf{2636} &                                    \textbf{2636} &                                    \textbf{2636} \\
                                                   &                                                  &                                                5 &                                    \textbf{4903} &                                             7079 &                                             7079 &                                    \textbf{7079} &                                    \textbf{7079} &                                    \textbf{7079} &                                    \textbf{7079} \\
\cline{2-10}
&
\multirow{2}{*}{Times}
%                                                                                                      &                                                2 &                                   \textbf{0.16s} &                                            0.24s &                                   \textbf{0.16s} &                                            0.25s &                                   \textbf{0.23s} &                                            0.28s &                                   \textbf{0.13s} \\
%                                                   &                                                  &                                                3 &                                   \textbf{0.82s} &                                            1.64s &                                            1.04s &                                            1.69s &                                   \textbf{1.62s} &                                            2.06s &                                   \textbf{0.79s} \\
                                                                                                     &                                                4 &                                   \textbf{2.98s} &                                            7.61s &                                              11s &                                   \textbf{7.77s} &                                            9.57s &                                            9.94s &                                   \textbf{6.96s} \\
                                                   &                                                  &                                                5 &                                   \textbf{5.16s} &                                              26s &                                              58s &                                     \textbf{26s} &                                              50s &                                     \textbf{34s} &                                     \textbf{34s} \\
\specialrule{.1em}{0em}{0em}
% ra.c
\multirow{4}{*}{\begin{tabular}{l}
\textbf{ra}\\
% lines of code: 27\\
% variables: 1\\
threads: U
\end{tabular}}
&
\multirow{2}{*}{Traces}
                                                                                                     &                                                5 &                                    \textbf{1296} &                                    \textbf{1296} &                                            14400 &                                    \textbf{1296} &                                            14400 &                                    \textbf{1296} &                                            14400 \\
                                                   &                                                  &                                                6 &                                   \textbf{16807} &                                   \textbf{16807} &                                           518400 &                                   \textbf{16807} &                                           518400 &                                   \textbf{16807} &                                           518400 \\
%                                                   &                                                  &                                                7 &                                  \textbf{262144} &                                  \textbf{262144} &                                                - &                                  \textbf{262144} &                                                - &                                  \textbf{262144} &                                                - \\
%                                                   &                                                  &                                                8 &                                 \textbf{4782969} &                                                - &                                                - &                                                - &                                                - &                                                - &                                                - \\
\cline{2-10}
&
\multirow{2}{*}{Times}
                                                                                                     &                                                5 &                                   \textbf{0.39s} &                                            0.72s &                                            3.25s &                                   \textbf{0.73s} &                                            5.82s &                                   \textbf{0.79s} &                                            3.34s \\
                                                   &                                                  &                                                6 &                                   \textbf{4.94s} &                                              11s &                                              76s &                                     \textbf{11s} &                                             123s &                                     \textbf{12s} &                                             236s \\
%                                                   &                                                  &                                                7 &                                     \textbf{62s} &                                             238s &                                                - &                                    \textbf{234s} &                                                - &                                    \textbf{252s} &                                                - \\
%                                                   &                                                  &                                                8 &                                   \textbf{1484s} &                                                - &                                                - &                                                - &                                                - &                                                - &                                                - \\
\specialrule{.1em}{0em}{0em}
% race_parametric.c
\multirow{4}{*}{\begin{tabular}{l}
\textbf{race\_parametric}\\
% lines of code: 30\\
% variables: 3\\
threads: 2
\end{tabular}}
&
\multirow{2}{*}{Traces}
                                                                                                     &                                                6 &                                   \textbf{73789} &                                   \textbf{73789} &                                           372436 &                                   \textbf{73789} &                                           372436 &                                   \textbf{73789} &                                           372436 \\
                                                   &                                                  &                                                7 &                                  \textbf{616227} &                                  \textbf{616227} &                                          4027216 &                                  \textbf{616227} &                                          4027216 &                                  \textbf{616227} &                                          4027216 \\
%                                                   &                                                  &                                                8 &                                                - &                                                - &                                                - &                                                - &                                                - &                                                - &                                                - \\
\cline{2-10}
&
\multirow{2}{*}{Times}
                                                                                                     &                                                6 &                                     \textbf{47s} &                                             228s &                                             119s &                                    \textbf{187s} &                                             226s &                                             221s &                                    \textbf{115s} \\
                                                   &                                                  &                                                7 &                                    \textbf{630s} &                                            2566s &                                            1559s &                                            2090s &                                   \textbf{1931s} &                                            2388s &                                   \textbf{1395s} \\
%                                                   &                                                  &                                                8 &                                                - &                                                - &                                                - &                                                - &                                                - &                                                - &                                                - \\
\specialrule{.1em}{0em}{0em}
% readers.c
\multirow{4}{*}{\begin{tabular}{l}
\textbf{readers}\\
% lines of code: 39\\
% variables: U+2\\
threads: U+1
\end{tabular}}
&
\multirow{2}{*}{Traces}
                                                                                                     &                                               14 &                                   \textbf{16384} &                                   \textbf{16384} &                                   \textbf{16384} &                                   \textbf{16384} &                                   \textbf{16384} &                                   \textbf{16384} &                                   \textbf{16384} \\
                                                   &                                                  &                                               15 &                                   \textbf{32768} &                                   \textbf{32768} &                                   \textbf{32768} &                                   \textbf{32768} &                                   \textbf{32768} &                                   \textbf{32768} &                                   \textbf{32768} \\
\cline{2-10}
&
\multirow{2}{*}{Times}
                                                                                                     &                                               14 &                                            8.00s &                                              28s &                                   \textbf{7.77s} &                                              28s &                                   \textbf{8.63s} &                                              42s &                                   \textbf{7.26s} \\
                                                   &                                                  &                                               15 &                                              27s &                                              64s &                                     \textbf{15s} &                                              63s &                                     \textbf{18s} &                                              95s &                                     \textbf{16s} \\
\specialrule{.1em}{0em}{0em}
% redundant_co.c
\multirow{4}{*}{\begin{tabular}{l}
\textbf{redundant\_co}\\
% lines of code: 23\\
% variables: 1\\
threads: 2
\end{tabular}}
&
\multirow{2}{*}{Traces}
                                                                                                      &                                                5 &                                      \textbf{91} &                                      \textbf{91} &                                            16632 &                                      \textbf{91} &                                            16632 &                                      \textbf{91} &                                            16632 \\
                                                   &                                                  &                                               40 &                                    \textbf{4921} &                                    \textbf{4921} &                                                - &                                    \textbf{4921} &                                                - &                                    \textbf{4921} &                                                - \\
\cline{2-10}
&
\multirow{2}{*}{Times}
                                                                                                      &                                                5 &                                   \textbf{0.06s} &                                            0.07s &                                            1.04s &                                   \textbf{0.07s} &                                            2.17s &                                   \textbf{0.07s} &                                            2.49s \\
                                                   &                                                  &                                               40 &                                   \textbf{1.83s} &                                            6.05s &                                                - &                                   \textbf{6.09s} &                                                - &                                   \textbf{6.89s} &                                                - \\
\specialrule{.1em}{0em}{0em}
% seqlock.c
\multirow{4}{*}{\begin{tabular}{l}
\textbf{seqlock}\\
% lines of code: 944\\
% variables: 2\\
% locks: 1\\
threads: U+1
\end{tabular}}
&
\multirow{2}{*}{Traces}
                                                                                                     &                                                7 &                                  \textbf{181440} &                                  \textbf{181440} &                                  \textbf{181440} &                                  \textbf{181440} &                                  \textbf{181440} &                                  \textbf{181440} &                                  \textbf{181440} \\
                                                   &                                                  &                                                8 &                                 \textbf{1814400} &                                 \textbf{1814400} &                                 \textbf{1814400} &                                 \textbf{1814400} &                                 \textbf{1814400} &                                                - &                                 \textbf{1814400} \\
\cline{2-10}
&
\multirow{2}{*}{Times}
                                                                                                     &                                                7 &                                             112s &                                             235s &                                     \textbf{77s} &                                             239s &                                     \textbf{88s} &                                             307s &                                     \textbf{88s} \\
                                                   &                                                  &                                                8 &                                            1058s &                                            2741s &                                    \textbf{813s} &                                            2862s &                                    \textbf{967s} &                                                - &                                   \textbf{1089s} \\
\specialrule{.1em}{0em}{0em}
% seqlock_atomic.c
\multirow{4}{*}{\begin{tabular}{l}
\textbf{seqlock\_atomic}\\
% lines of code: 964\\
% variables: 2\\
% locks: 1\\
threads: U+1
\end{tabular}}
&
\multirow{2}{*}{Traces}
%                                                                                                      &                                                2 &                                      \textbf{92} &                                      \textbf{92} &                                      \textbf{92} &                                      \textbf{92} &                                      \textbf{92} &                                      \textbf{92} &                                      \textbf{92} \\
%                                                   &                                                  &                                                3 &                                    \textbf{1076} &                                    \textbf{1076} &                                    \textbf{1076} &                                    \textbf{1076} &                                    \textbf{1076} &                                    \textbf{1076} &                                    \textbf{1076} \\
                                                                                                     &                                                4 &                                   \textbf{29804} &                                   \textbf{29804} &                                   \textbf{29804} &                                   \textbf{29804} &                                   \textbf{29804} &                                   \textbf{29804} &                                   \textbf{29804} \\
                                                   &                                                  &                                                5 &                                  \textbf{605070} &                                  \textbf{605070} &                                  \textbf{605070} &                                  \textbf{605070} &                                  \textbf{605070} &                                  \textbf{605070} &                                  \textbf{605070} \\
%                                                   &                                                  &                                                6 &                                                - &                                                - &                                                - &                                                - &                                                - &                                                - &                                                - \\
\cline{2-10}
&
\multirow{2}{*}{Times}
%                                                                                                      &                                                2 &                                   \textbf{0.09s} &                                            0.11s &                                            0.13s &                                            0.11s &                                   \textbf{0.10s} &                                            0.12s &                                   \textbf{0.08s} \\
%                                                   &                                                  &                                                3 &                                   \textbf{0.39s} &                                            0.79s &                                            0.59s &                                            0.82s &                                   \textbf{0.62s} &                                            0.96s &                                   \textbf{0.33s} \\
                                                                                                     &                                                4 &                                              12s &                                              26s &                                   \textbf{9.74s} &                                              27s &                                     \textbf{10s} &                                              32s &                                     \textbf{13s} \\
                                                   &                                                  &                                                5 &                                             267s &                                             710s &                                    \textbf{254s} &                                             732s &                                    \textbf{271s} &                                             930s &                                    \textbf{362s} \\
%                                                   &                                                  &                                                6 &                                                - &                                                - &                                                - &                                                - &                                                - &                                                - &                                                - \\
\specialrule{.1em}{0em}{0em}
% spammer.c
\multirow{4}{*}{\begin{tabular}{l}
\textbf{spammer}\\
% lines of code: 24\\
% variables: U\\
threads: U
\end{tabular}}
&
\multirow{2}{*}{Traces}
%                                                                                                      &                                                1 &                                       \textbf{1} &                                       \textbf{1} &                                       \textbf{1} &                                       \textbf{1} &                                       \textbf{1} &                                       \textbf{1} &                                       \textbf{1} \\
%                                                   &                                                  &                                                2 &                                       \textbf{3} &                                       \textbf{3} &                                       \textbf{3} &                                       \textbf{4} &                                       \textbf{4} &                                       \textbf{4} &                                       \textbf{4} \\
%                                                   &                                                  &                                                3 &                                      \textbf{19} &                                      \textbf{19} &                                               63 &                                      \textbf{27} &                                              216 &                                      \textbf{27} &                                              216 \\
                                                                                                     &                                                4 &                                     \textbf{178} &                                     \textbf{178} &                                            19448 &                                     \textbf{256} &                                           331776 &                                     \textbf{256} &                                           331776 \\
%                                                   &                                                  &                                                5 &                                    \textbf{2166} &                                    \textbf{2166} &                                                - &                                    \textbf{3125} &                                                - &                                    \textbf{3125} &                                                - \\
%                                                   &                                                  &                                                6 &                                   \textbf{32301} &                                   \textbf{32301} &                                                - &                                   \textbf{46656} &                                                - &                                   \textbf{46656} &                                                - \\
                                                   &                                                  &                                                7 &                                  \textbf{569871} &                                  \textbf{569871} &                                                - &                                  \textbf{823543} &                                                - &                                  \textbf{823543} &                                                - \\
%                                                   &                                                  &                                                8 &                                                - &                                                - &                                                - &                                                - &                                                - &                                                - &                                                - \\
\cline{2-10}
&
\multirow{2}{*}{Times}
%                                                                                                      &                                                1 &                                   \textbf{0.05s} &                                   \textbf{0.05s} &                                   \textbf{0.05s} &                                   \textbf{0.05s} &                                            0.07s &                                   \textbf{0.05s} &                                   \textbf{0.05s} \\
%                                                   &                                                  &                                                2 &                                   \textbf{0.05s} &                                   \textbf{0.05s} &                                            0.06s &                                   \textbf{0.05s} &                                            0.06s &                                   \textbf{0.05s} &                                   \textbf{0.05s} \\
%                                                   &                                                  &                                                3 &                                   \textbf{0.05s} &                                   \textbf{0.05s} &                                   \textbf{0.05s} &                                   \textbf{0.06s} &                                            0.11s &                                   \textbf{0.06s} &                                            0.07s \\
                                                                                                     &                                                4 &                                            0.15s &                                   \textbf{0.14s} &                                            5.14s &                                   \textbf{0.17s} &                                              69s &                                   \textbf{0.25s} &                                              80s \\
%                                                   &                                                  &                                                5 &                                   \textbf{1.60s} &                                            1.64s &                                                - &                                   \textbf{2.17s} &                                                - &                                   \textbf{3.96s} &                                                - \\
%                                                   &                                                  &                                                6 &                                     \textbf{34s} &                                     \textbf{34s} &                                                - &                                     \textbf{46s} &                                                - &                                     \textbf{94s} &                                                - \\
                                                   &                                                  &                                                7 &                                            1001s &                                    \textbf{954s} &                                                - &                                   \textbf{1204s} &                                                - &                                   \textbf{3592s} &                                                - \\
%                                                   &                                                  &                                                8 &                                                - &                                                - &                                                - &                                                - &                                                - &                                                - &                                                - \\
\specialrule{.1em}{0em}{0em}
% writer_reader.c
\multirow{4}{*}{\begin{tabular}{l}
\textbf{writer\_reader}\\
% lines of code: 26\\
% variables: 1\\
threads: 2
\end{tabular}}
&
\multirow{2}{*}{Traces}
%                                                                                                      &                                                5 &                                     \textbf{252} &                                     \textbf{252} &                                     \textbf{252} &                                     \textbf{252} &                                     \textbf{252} &                                     \textbf{252} &                                     \textbf{252} \\
%                                                   &                                                  &                                               10 &                                  \textbf{184756} &                                  \textbf{184756} &                                  \textbf{184756} &                                  \textbf{184756} &                                  \textbf{184756} &                                  \textbf{184756} &                                  \textbf{184756} \\
                                                                                                     &                                               11 &                                  \textbf{705432} &                                  \textbf{705432} &                                  \textbf{705432} &                                  \textbf{705432} &                                  \textbf{705432} &                                  \textbf{705432} &                                  \textbf{705432} \\
                                                    &                                                 &                                               12 &                                 \textbf{2704156} &                                 \textbf{2704156} &                                 \textbf{2704156} &                                 \textbf{2704156} &                                 \textbf{2704156} &                                 \textbf{2704156} &                                 \textbf{2704156} \\
%                                                   &                                                  &                                               13 &                                \textbf{10400600} &                                \textbf{10400600} &                                \textbf{10400600} &                                \textbf{10400600} &                                \textbf{10400600} &                                                - &                                \textbf{10400600} \\
%                                                   &                                                  &                                               14 &                                                - &                                                - &                                                - &                                                - &                                                - &                                                - &                                                - \\
\cline{2-10}
&
\multirow{2}{*}{Times}
%                                                                                                      &                                                5 &                                   \textbf{0.06s} &                                            0.09s &                                   \textbf{0.06s} &                                            0.09s &                                   \textbf{0.06s} &                                            0.09s &                                   \textbf{0.06s} \\
%                                                   &                                                  &                                               10 &                                              20s &                                              45s &                                     \textbf{11s} &                                              48s &                                     \textbf{21s} &                                              54s &                                     \textbf{13s} \\
                                                                                                     &                                               11 &                                              82s &                                             193s &                                     \textbf{48s} &                                             201s &                                     \textbf{56s} &                                             231s &                                     \textbf{93s} \\
                                                    &                                                 &                                               12 &                                             411s &                                             797s &                                    \textbf{258s} &                                             842s &                                    \textbf{206s} &                                             963s &                                    \textbf{410s} \\
%                                                   &                                                  &                                               13 &                                            1254s &                                            3331s &                                    \textbf{814s} &                                            3502s &                                    \textbf{858s} &                                                - &                                    \textbf{979s} \\
%                                                   &                                                  &                                               14 &                                                - &                                                - &                                                - &                                                - &                                                - &                                                - &                                                - \\
\specialrule{.1em}{0em}{0em}
\end{tabular}
% \addtocounter{table}{-1}
% \vspace{-2mm}
\caption{Part2: Related papers and works, and synthetic benchmarks.}
\label{tab:all_other_2}
\end{table}

\begin{table}
\scriptsize
\newcolumntype{?}{!{\vrule width 1.5pt}}
\setlength{\extrarowheight}{.05em}
\begin{tabular}{? l | c | c ? r r r ? r r ? r r ? }
\specialrule{.15em}{0em}{0em}
\multicolumn{2}{?c|}{\multirow{2}{*}{\textbf{Benchmark}}} & \multirow{2}{*}{U} & \multicolumn{3}{c?}{\textbf{Sequential Consistency}} & \multicolumn{2}{c?}{\textbf{Total Store Order}} & \multicolumn{2}{c?}{\textbf{Partial Store Order}}\\
\cline{4-10}
\multicolumn{2}{?c|}{} & & $\ReadsFrom$ & $\DCTSOPSOM$ & $\Source$ & $\DCTSOPSOM$ & $\Source$ & $\DCTSOPSOM$ & $\Source$\\
\specialrule{.1em}{0em}{0em}
% 01_inc_true-unreach-call.c
\multirow{4}{*}{\begin{tabular}{l}
\textbf{01\_inc}\\
% lines of code: 45\\
% variables: 1\\
% locks: 1\\
threads: U
\end{tabular}}
&
\multirow{2}{*}{Traces}
%                                                                                                      &                                                3 &                                      \textbf{36} &                                      \textbf{36} &                                      \textbf{36} &                                      \textbf{36} &                                      \textbf{36} &                                      \textbf{36} &                                      \textbf{36} \\
%                                                   &                                                  &                                                4 &                                     \textbf{576} &                                     \textbf{576} &                                     \textbf{576} &                                     \textbf{576} &                                     \textbf{576} &                                     \textbf{576} &                                     \textbf{576} \\
                                                                                                     &                                                5 &                                   \textbf{14400} &                                   \textbf{14400} &                                   \textbf{14400} &                                   \textbf{14400} &                                   \textbf{14400} &                                   \textbf{14400} &                                   \textbf{14400} \\
                                                   &                                                  &                                                6 &                                  \textbf{518400} &                                  \textbf{518400} &                                  \textbf{518400} &                                  \textbf{518400} &                                  \textbf{518400} &                                  \textbf{518400} &                                  \textbf{518400} \\
%                                                   &                                                  &                                                7 &                                                - &                                                - &                                                - &                                                - &                                                - &                                                - &                                                - \\
\cline{2-10}
&
\multirow{2}{*}{Times}
%                                                                                                      &                                                3 &                                   \textbf{0.05s} &                                            0.06s &                                   \textbf{0.05s} &                                            0.06s &                                   \textbf{0.05s} &                                            0.06s &                                   \textbf{0.05s} \\
%                                                   &                                                  &                                                4 &                                            0.15s &                                            0.30s &                                   \textbf{0.11s} &                                            0.31s &                                   \textbf{0.17s} &                                            0.52s &                                   \textbf{0.18s} \\
                                                                                                     &                                                5 &                                   \textbf{3.08s} &                                            8.85s &                                            3.68s &                                            8.93s &                                   \textbf{3.90s} &                                            9.83s &                                   \textbf{2.02s} \\
                                                   &                                                  &                                                6 &                                             150s &                                             443s &                                    \textbf{101s} &                                             451s &                                    \textbf{100s} &                                             494s &                                     \textbf{96s} \\
%                                                   &                                                  &                                                7 &                                                - &                                                - &                                                - &                                                - &                                                - &                                                - &                                                - \\
\specialrule{.1em}{0em}{0em}
% 02_inc_cas_true-unreach-call.c
\multirow{4}{*}{\begin{tabular}{l}
\textbf{02\_inc\_cas}\\
% lines of code: 57\\
% variables: 1\\
% locks: 1\\
threads: U
\end{tabular}}
&
\multirow{2}{*}{Traces}
%                                                                                                      &                                                2 &                                       \textbf{8} &                                       \textbf{8} &                                       \textbf{8} &                                       \textbf{8} &                                       \textbf{8} &                                       \textbf{8} &                                       \textbf{8} \\
                                                                                                     &                                                3 &                                     \textbf{432} &                                     \textbf{432} &                                     \textbf{432} &                                     \textbf{432} &                                     \textbf{432} &                                     \textbf{432} &                                     \textbf{432} \\
                                                   &                                                  &                                                4 &                                  \textbf{159552} &                                  \textbf{159552} &                                  \textbf{159552} &                                  \textbf{159552} &                                  \textbf{159552} &                                  \textbf{159552} &                                  \textbf{159552} \\
%                                                   &                                                  &                                                5 &                                                - &                                                - &                                                - &                                                - &                                                - &                                                - &                                                - \\
\cline{2-10}
&
\multirow{2}{*}{Times}
%                                                                                                      &                                                2 &                                   \textbf{0.05s} &                                            0.06s &                                   \textbf{0.05s} &                                   \textbf{0.05s} &                                   \textbf{0.05s} &                                   \textbf{0.05s} &                                   \textbf{0.05s} \\
                                                                                                     &                                                3 &                                            0.18s &                                            0.22s &                                   \textbf{0.09s} &                                            0.22s &                                   \textbf{0.09s} &                                            0.24s &                                   \textbf{0.13s} \\
                                                   &                                                  &                                                4 &                                              40s &                                             112s &                                     \textbf{24s} &                                             114s &                                     \textbf{37s} &                                             127s &                                     \textbf{25s} \\
%                                                   &                                                  &                                                5 &                                                - &                                                - &                                                - &                                                - &                                                - &                                                - &                                                - \\
\specialrule{.1em}{0em}{0em}
% 03_incdec_true-unreach-call.c
\multirow{4}{*}{\begin{tabular}{l}
\textbf{03\_incdec}\\
% lines of code: 76\\
% variables: 1\\
% locks: 1\\
threads: U
\end{tabular}}
&
\multirow{2}{*}{Traces}
%                                                                                                      &                                                3 &                                      \textbf{24} &                                      \textbf{24} &                                      \textbf{24} &                                      \textbf{24} &                                      \textbf{24} &                                      \textbf{24} &                                      \textbf{24} \\
%                                                   &                                                  &                                                4 &                                    \textbf{1068} &                                    \textbf{1068} &                                    \textbf{1068} &                                    \textbf{1068} &                                    \textbf{1068} &                                    \textbf{1068} &                                    \textbf{1068} \\
                                                                                                     &                                                5 &                                   \textbf{14568} &                                   \textbf{14568} &                                   \textbf{14568} &                                   \textbf{14568} &                                   \textbf{14568} &                                   \textbf{14568} &                                   \textbf{14568} \\
                                                   &                                                  &                                                6 &                                 \textbf{2289708} &                                 \textbf{2289708} &                                 \textbf{2289708} &                                 \textbf{2289708} &                                 \textbf{2289708} &                                 \textbf{2289708} &                                 \textbf{2289708} \\
%                                                   &                                                  &                                                7 &                                                - &                                                - &                                                - &                                                - &                                                - &                                                - &                                                - \\
\cline{2-10}
&
\multirow{2}{*}{Times}
%                                                                                                      &                                                3 &                                            0.08s &                                            0.06s &                                   \textbf{0.05s} &                                   \textbf{0.06s} &                                            0.08s &                                            0.06s &                                   \textbf{0.05s} \\
%                                                   &                                                  &                                                4 &                                            0.46s &                                            0.61s &                                   \textbf{0.19s} &                                            0.62s &                                   \textbf{0.36s} &                                            0.78s &                                   \textbf{0.21s} \\
                                                                                                     &                                                5 &                                            3.67s &                                            9.30s &                                   \textbf{2.74s} &                                            9.42s &                                   \textbf{3.48s} &                                              12s &                                   \textbf{5.11s} \\
                                                   &                                                  &                                                6 &                                             866s &                                            2410s &                                    \textbf{469s} &                                            2393s &                                    \textbf{747s} &                                            3434s &                                    \textbf{662s} \\
%                                                   &                                                  &                                                7 &                                                - &                                                - &                                                - &                                                - &                                                - &                                                - &                                                - \\
\specialrule{.1em}{0em}{0em}
% 13_unverif_true-unreach-call.c
\multirow{4}{*}{\begin{tabular}{l}
\textbf{13\_unverif}\\
% lines of code: 41\\
% variables: 2\\
% locks: 1\\
threads: U
\end{tabular}}
&
\multirow{2}{*}{Traces}
%                                                                                                      &                                                2 &                                       \textbf{4} &                                       \textbf{4} &                                       \textbf{4} &                                       \textbf{4} &                                       \textbf{4} &                                       \textbf{4} &                                       \textbf{4} \\
%                                                   &                                                  &                                                3 &                                      \textbf{36} &                                      \textbf{36} &                                      \textbf{36} &                                      \textbf{36} &                                      \textbf{36} &                                      \textbf{36} &                                      \textbf{36} \\
%                                                   &                                                  &                                                4 &                                     \textbf{576} &                                     \textbf{576} &                                     \textbf{576} &                                     \textbf{576} &                                     \textbf{576} &                                     \textbf{576} &                                     \textbf{576} \\
                                                                                                     &                                                5 &                                   \textbf{14400} &                                   \textbf{14400} &                                   \textbf{14400} &                                   \textbf{14400} &                                   \textbf{14400} &                                   \textbf{14400} &                                   \textbf{14400} \\
                                                   &                                                  &                                                6 &                                  \textbf{518400} &                                  \textbf{518400} &                                  \textbf{518400} &                                  \textbf{518400} &                                  \textbf{518400} &                                  \textbf{518400} &                                  \textbf{518400} \\
\cline{2-10}
&
\multirow{2}{*}{Times}
%                                                                                                      &                                                2 &                                   \textbf{0.05s} &                                   \textbf{0.05s} &                                   \textbf{0.05s} &                                   \textbf{0.05s} &                                   \textbf{0.05s} &                                   \textbf{0.05s} &                                   \textbf{0.05s} \\
%                                                   &                                                  &                                                3 &                                   \textbf{0.05s} &                                            0.06s &                                            0.07s &                                            0.06s &                                   \textbf{0.05s} &                                            0.06s &                                   \textbf{0.05s} \\
%                                                   &                                                  &                                                4 &                                            0.19s &                                            0.32s &                                   \textbf{0.11s} &                                            0.33s &                                   \textbf{0.12s} &                                            0.36s &                                   \textbf{0.15s} \\
                                                                                                     &                                                5 &                                            3.20s &                                            9.31s &                                   \textbf{2.99s} &                                            9.47s &                                   \textbf{2.29s} &                                              10s &                                   \textbf{3.03s} \\
                                                   &                                                  &                                                6 &                                             158s &                                             464s &                                    \textbf{125s} &                                             467s &                                    \textbf{100s} &                                             518s &                                    \textbf{178s} \\
\specialrule{.1em}{0em}{0em}
% 18_read_write_lock_true-unreach-call.c
\multirow{4}{*}{\begin{tabular}{l}
\textbf{18\_read\_write\_lock}\\
% lines of code: 56\\
% variables: 4\\
% locks: 1\\
threads: U+2
\end{tabular}}
&
\multirow{2}{*}{Traces}
%                                                                                                      &                                                3 &                                     \textbf{440} &                                     \textbf{440} &                                     \textbf{440} &                                     \textbf{440} &                                     \textbf{440} &                                     \textbf{440} &                                     \textbf{440} \\
%                                                   &                                                  &                                                4 &                                    \textbf{2670} &                                    \textbf{2670} &                                    \textbf{2670} &                                    \textbf{2670} &                                    \textbf{2670} &                                    \textbf{2670} &                                    \textbf{2670} \\
%                                                   &                                                  &                                                5 &                                   \textbf{18732} &                                   \textbf{18732} &                                   \textbf{18732} &                                   \textbf{18732} &                                   \textbf{18732} &                                   \textbf{18732} &                                   \textbf{18732} \\
                                                                                                     &                                                6 &                                  \textbf{149912} &                                  \textbf{149912} &                                  \textbf{149912} &                                  \textbf{149912} &                                  \textbf{149912} &                                  \textbf{149912} &                                  \textbf{149912} \\
                                                   &                                                  &                                                7 &                                 \textbf{1349280} &                                 \textbf{1349280} &                                 \textbf{1349280} &                                 \textbf{1349280} &                                 \textbf{1349280} &                                 \textbf{1349280} &                                 \textbf{1349280} \\
\cline{2-10}
&
\multirow{2}{*}{Times}
%                                                                                                      &                                                3 &                                   \textbf{0.14s} &                                            0.27s &                                            0.19s &                                            0.27s &                                   \textbf{0.12s} &                                            0.33s &                                   \textbf{0.11s} \\
%                                                   &                                                  &                                                4 &                                   \textbf{1.05s} &                                            1.79s &                                            1.08s &                                            1.80s &                                   \textbf{0.68s} &                                            2.33s &                                   \textbf{0.90s} \\
%                                                   &                                                  &                                                5 &                                   \textbf{8.16s} &                                              16s &                                            9.09s &                                              16s &                                   \textbf{9.04s} &                                              21s &                                   \textbf{7.71s} \\
                                                                                                     &                                                6 &                                              69s &                                             160s &                                     \textbf{54s} &                                             159s &                                     \textbf{66s} &                                             218s &                                     \textbf{42s} \\
                                                   &                                                  &                                                7 &                                             805s &                                            1819s &                                    \textbf{603s} &                                            1834s &                                    \textbf{687s} &                                            2559s &                                    \textbf{536s} \\
\specialrule{.1em}{0em}{0em}
% 27_Boop_simple_vf_false-unreach-call.c
\multirow{4}{*}{\begin{tabular}{l}
\textbf{27\_Boop}\\
% lines of code: 74\\
% variables: 4\\
threads: U+1
\end{tabular}}
&
\multirow{2}{*}{Traces}
%                                                                                                      &                                                1 &                                       \textbf{2} &                                       \textbf{2} &                                       \textbf{2} &                                       \textbf{2} &                                       \textbf{2} &                                       \textbf{4} &                                                9 \\
                                                                                                     &                                                2 &                                     \textbf{165} &                                     \textbf{165} &                                              364 &                                     \textbf{205} &                                              536 &                                     \textbf{713} &                                            14604 \\
                                                   &                                                  &                                                3 &                                   \textbf{68083} &                                   \textbf{68083} &                                           966834 &                                  \textbf{100897} &                                          2157426 &                                  \textbf{447739} &                                                - \\
%                                                   &                                                  &                                                4 &                                                - &                                                - &                                                - &                                                - &                                                - &                                                - &                                                - \\
\cline{2-10}
&
\multirow{2}{*}{Times}
%                                                                                                      &                                                1 &                                            0.07s &                                   \textbf{0.05s} &                                   \textbf{0.05s} &                                   \textbf{0.05s} &                                   \textbf{0.05s} &                                   \textbf{0.05s} &                                   \textbf{0.05s} \\
                                                                                                     &                                                2 &                                            0.10s &                                            0.10s &                                   \textbf{0.08s} &                                            0.11s &                                   \textbf{0.10s} &                                   \textbf{0.36s} &                                            2.65s \\
                                                   &                                                  &                                                3 &                                     \textbf{14s} &                                              38s &                                             185s &                                     \textbf{52s} &                                             360s &                                    \textbf{347s} &                                                - \\
%                                                   &                                                  &                                                4 &                                                - &                                                - &                                                - &                                                - &                                                - &                                                - &                                                - \\
\specialrule{.1em}{0em}{0em}
% 27_Boop_simple_vf_false-unreach-call_4thr.c
\multirow{4}{*}{\begin{tabular}{l}
\textbf{27\_Boop4}\\
% lines of code: 74\\
% variables: 4\\
threads: 4
\end{tabular}}
&
\multirow{2}{*}{Traces}
%                                                                                                      &                                                1 &                                    \textbf{2902} &                                    \textbf{2902} &                                            21948 &                                    \textbf{3682} &                                            36588 &                                    \textbf{8233} &                                           572436 \\
%                                                   &                                                  &                                                2 &                                   \textbf{18040} &                                   \textbf{18040} &                                           178236 &                                   \textbf{24766} &                                           355716 &                                   \textbf{80680} &                                                - \\
%                                                   &                                                  &                                                3 &                                   \textbf{68083} &                                   \textbf{68083} &                                           966834 &                                  \textbf{100897} &                                          2157426 &                                  \textbf{447739} &                                                - \\
                                                                                                     &                                                4 &                                  \textbf{197260} &                                  \textbf{197260} &                                          3873348 &                                  \textbf{313336} &                                          9412428 &                                 \textbf{1807408} &                                                - \\
%                                                   &                                                  &                                                5 &                                  \textbf{483370} &                                  \textbf{483370} &                                         12577158 &                                  \textbf{816580} &                                                - &                                                - &                                                - \\
%                                                   &                                                  &                                                6 &                                 \textbf{1053088} &                                 \textbf{1053088} &                                                - &                                 \textbf{1878622} &                                                - &                                                - &                                                - \\
                                                   &                                                  &                                                7 &                                 \textbf{2101147} &                                 \textbf{2101147} &                                                - &                                 \textbf{3933691} &                                                - &                                                - &                                                - \\
\cline{2-10}
&
\multirow{2}{*}{Times}
%                                                                                                      &                                                1 &                                   \textbf{0.37s} &                                            1.22s &                                            1.74s &                                   \textbf{1.46s} &                                            6.18s &                                   \textbf{4.40s} &                                             169s \\
%                                                   &                                                  &                                                2 &                                   \textbf{2.27s} &                                            8.59s &                                              17s &                                     \textbf{11s} &                                              57s &                                     \textbf{51s} &                                                - \\
%                                                   &                                                  &                                                3 &                                     \textbf{12s} &                                              38s &                                             186s &                                     \textbf{52s} &                                             340s &                                    \textbf{338s} &                                                - \\
                                                                                                     &                                                4 &                                     \textbf{33s} &                                             124s &                                             550s &                                    \textbf{182s} &                                            2556s &                                   \textbf{1593s} &                                                - \\
%                                                   &                                                  &                                                5 &                                     \textbf{74s} &                                             336s &                                            1831s &                                    \textbf{521s} &                                                - &                                                - &                                                - \\
%                                                   &                                                  &                                                6 &                                    \textbf{274s} &                                             814s &                                                - &                                   \textbf{1335s} &                                                - &                                                - &                                                - \\
                                                   &                                                  &                                                7 &                                    \textbf{437s} &                                            1783s &                                                - &                                   \textbf{3016s} &                                                - &                                                - &                                                - \\
\specialrule{.1em}{0em}{0em}
% 30_Function_Pointer3_vs_true-unreach-call.c
\multirow{4}{*}{\begin{tabular}{l}
\textbf{30\_Function\_Pointer}\\
% lines of code: 67\\
% variables: 1\\
% locks: 1\\
threads: U+1
\end{tabular}}
&
\multirow{2}{*}{Traces}
                                                                                                     &                                                5 &                                   \textbf{30240} &                                   \textbf{30240} &                                   \textbf{30240} &                                   \textbf{30240} &                                   \textbf{30240} &                                   \textbf{30240} &                                   \textbf{30240} \\
                                                   &                                                  &                                                6 &                                  \textbf{665280} &                                  \textbf{665280} &                                  \textbf{665280} &                                  \textbf{665280} &                                  \textbf{665280} &                                  \textbf{665280} &                                  \textbf{665280} \\
%                                                   &                                                  &                                                7 &                                                - &                                                - &                                \textbf{17297280} &                                                - &                                                - &                                                - &                                \textbf{17297280} \\
%                                                   &                                                  &                                                8 &                                                - &                                                - &                                                - &                                                - &                                                - &                                                - &                                                - \\
\cline{2-10}
&
\multirow{2}{*}{Times}
                                                                                                     &                                                5 &                                            7.52s &                                              17s &                                   \textbf{5.74s} &                                              18s &                                   \textbf{7.23s} &                                              19s &                                   \textbf{7.62s} \\
                                                   &                                                  &                                                6 &                                             207s &                                             537s &                                    \textbf{114s} &                                             553s &                                    \textbf{157s} &                                             584s &                                    \textbf{174s} \\
%                                                   &                                                  &                                                7 &                                                - &                                                - &                                   \textbf{2839s} &                                                - &                                                - &                                                - &                                   \textbf{3318s} \\
%                                                   &                                                  &                                                8 &                                                - &                                                - &                                                - &                                                - &                                                - &                                                - &                                                - \\
\specialrule{.1em}{0em}{0em}
% 32_pthread5_vs_false-unreach-call.c
\multirow{4}{*}{\begin{tabular}{l}
\textbf{32\_pthread5}\\
% lines of code: 87\\
% variables: 4\\
% locks: 1\\
threads: U+2
\end{tabular}}
&
\multirow{2}{*}{Traces}
%                                                                                                      &                                                1 &                                      \textbf{20} &                                      \textbf{20} &                                               24 &                                      \textbf{20} &                                               24 &                                      \textbf{20} &                                               24 \\
                                                                                                     &                                                2 &                                    \textbf{1470} &                                    \textbf{1470} &                                             1890 &                                    \textbf{1470} &                                             1890 &                                    \textbf{1470} &                                             1890 \\
                                                   &                                                  &                                                3 &                                  \textbf{226800} &                                  \textbf{226800} &                                           302400 &                                  \textbf{226800} &                                           302400 &                                  \textbf{226800} &                                           302400 \\
%                                                   &                                                  &                                                4 &                                                - &                                                - &                                                - &                                                - &                                                - &                                                - &                                                - \\
\cline{2-10}
&
\multirow{2}{*}{Times}
%                                                                                                      &                                                1 &                                            0.06s &                                   \textbf{0.05s} &                                            0.08s &                                   \textbf{0.05s} &                                   \textbf{0.05s} &                                            0.06s &                                   \textbf{0.05s} \\
                                                                                                     &                                                2 &                                            0.55s &                                            0.75s &                                   \textbf{0.50s} &                                            0.75s &                                   \textbf{0.28s} &                                            1.00s &                                   \textbf{0.31s} \\
                                                   &                                                  &                                                3 &                                              79s &                                             172s &                                     \textbf{75s} &                                             170s &                                     \textbf{72s} &                                             253s &                                     \textbf{59s} \\
%                                                   &                                                  &                                                4 &                                                - &                                                - &                                                - &                                                - &                                                - &                                                - &                                                - \\
\specialrule{.1em}{0em}{0em}
% 40_barrier_vf_false-unreach-call.c
\multirow{4}{*}{\begin{tabular}{l}
\textbf{40\_barrier}\\
% lines of code: 54\\
% variables: 3\\
% locks: 1\\
threads: U
\end{tabular}}
&
\multirow{2}{*}{Traces}
%                                                                                                      &                                                3 &                                      \textbf{18} &                                      \textbf{18} &                                      \textbf{18} &                                      \textbf{18} &                                      \textbf{18} &                                      \textbf{18} &                                      \textbf{18} \\
%                                                   &                                                  &                                                4 &                                      \textbf{84} &                                      \textbf{84} &                                      \textbf{84} &                                      \textbf{84} &                                      \textbf{84} &                                      \textbf{84} &                                      \textbf{84} \\
%                                                   &                                                  &                                                5 &                                     \textbf{440} &                                     \textbf{440} &                                     \textbf{440} &                                     \textbf{440} &                                     \textbf{440} &                                     \textbf{440} &                                     \textbf{440} \\
                                                                                                     &                                                6 &                                    \textbf{2670} &                                    \textbf{2670} &                                    \textbf{2670} &                                    \textbf{2670} &                                    \textbf{2670} &                                    \textbf{2670} &                                    \textbf{2670} \\
                                                   &                                                  &                                                7 &                                   \textbf{18732} &                                   \textbf{18732} &                                   \textbf{18732} &                                   \textbf{18732} &                                   \textbf{18732} &                                   \textbf{18732} &                                   \textbf{18732} \\
\cline{2-10}
&
\multirow{2}{*}{Times}
%                                                                                                      &                                                3 &                                            0.08s &                                            0.07s &                                   \textbf{0.05s} &                                   \textbf{0.05s} &                                   \textbf{0.05s} &                                            0.07s &                                   \textbf{0.05s} \\
%                                                   &                                                  &                                                4 &                                            0.10s &                                            0.09s &                                   \textbf{0.06s} &                                            0.08s &                                   \textbf{0.06s} &                                            0.09s &                                   \textbf{0.06s} \\
%                                                   &                                                  &                                                5 &                                            0.25s &                                            0.34s &                                   \textbf{0.13s} &                                            0.32s &                                   \textbf{0.15s} &                                            0.43s &                                   \textbf{0.12s} \\
                                                                                                     &                                                6 &                                            1.12s &                                            2.23s &                                   \textbf{0.82s} &                                            2.27s &                                   \textbf{0.88s} &                                            3.09s &                                   \textbf{0.61s} \\
                                                   &                                                  &                                                7 &                                            8.78s &                                              20s &                                   \textbf{8.52s} &                                              20s &                                   \textbf{6.78s} &                                              28s &                                   \textbf{7.79s} \\
\specialrule{.1em}{0em}{0em}
% 45_monabsex1_vs_true-unreach-call.c
\multirow{4}{*}{\begin{tabular}{l}
\textbf{45\_monabsex1}\\
% lines of code: 24\\
% variables: 1\\
threads: U
\end{tabular}}
&
\multirow{2}{*}{Traces}
%                                                                                                      &                                                4 &                                     \textbf{125} &                                     \textbf{125} &                                              576 &                                     \textbf{125} &                                              576 &                                     \textbf{125} &                                              576 \\
%                                                   &                                                  &                                                5 &                                    \textbf{1296} &                                    \textbf{1296} &                                            14400 &                                    \textbf{1296} &                                            14400 &                                    \textbf{1296} &                                            14400 \\
                                                                                                     &                                                6 &                                   \textbf{16807} &                                   \textbf{16807} &                                           518400 &                                   \textbf{16807} &                                           518400 &                                   \textbf{16807} &                                           518400 \\
                                                   &                                                  &                                                7 &                                  \textbf{262144} &                                  \textbf{262144} &                                         25401600 &                                  \textbf{262144} &                                                - &                                  \textbf{262144} &                                                - \\
%                                                   &                                                  &                                                8 &                                 \textbf{4782969} &                                                - &                                                - &                                                - &                                                - &                                                - &                                                - \\
%                                                   &                                                  &                                                9 &                                                - &                                                - &                                                - &                                                - &                                                - &                                                - &                                                - \\
\cline{2-10}
&
\multirow{2}{*}{Times}
%                                                                                                      &                                                4 &                                   \textbf{0.06s} &                                            0.10s &                                            0.12s &                                   \textbf{0.09s} &                                            0.17s &                                   \textbf{0.12s} &                                   \textbf{0.12s} \\
%                                                   &                                                  &                                                5 &                                   \textbf{0.26s} &                                            0.58s &                                            2.01s &                                   \textbf{0.58s} &                                            3.44s &                                   \textbf{0.65s} &                                            2.73s \\
                                                                                                     &                                                6 &                                   \textbf{1.93s} &                                            9.02s &                                              49s &                                   \textbf{8.78s} &                                             130s &                                   \textbf{9.50s} &                                             129s \\
                                                   &                                                  &                                                7 &                                     \textbf{35s} &                                             180s &                                            3123s &                                    \textbf{181s} &                                                - &                                    \textbf{188s} &                                                - \\
%                                                   &                                                  &                                                8 &                                    \textbf{735s} &                                                - &                                                - &                                                - &                                                - &                                                - &                                                - \\
%                                                   &                                                  &                                                9 &                                                - &                                                - &                                                - &                                                - &                                                - &                                                - &                                                - \\
\specialrule{.1em}{0em}{0em}
% 46_monabsex2_vs_true-unreach-call.c
\multirow{4}{*}{\begin{tabular}{l}
\textbf{46\_monabsex2}\\
% lines of code: 22\\
% variables: 2\\
threads: U
\end{tabular}}
&
\multirow{2}{*}{Traces}
%                                                                                                      &                                                4 &                                     \textbf{125} &                                     \textbf{125} &                                              576 &                                     \textbf{125} &                                              576 &                                     \textbf{125} &                                              576 \\
%                                                   &                                                  &                                                5 &                                    \textbf{1296} &                                    \textbf{1296} &                                            14400 &                                    \textbf{1296} &                                            14400 &                                    \textbf{1296} &                                            14400 \\
                                                                                                     &                                                6 &                                   \textbf{16807} &                                   \textbf{16807} &                                           518400 &                                   \textbf{16807} &                                           518400 &                                   \textbf{16807} &                                           518400 \\
                                                   &                                                  &                                                7 &                                  \textbf{262144} &                                  \textbf{262144} &                                         25401600 &                                  \textbf{262144} &                                         25401600 &                                  \textbf{262144} &                                                - \\
%                                                   &                                                  &                                                8 &                                 \textbf{4782969} &                                                - &                                                - &                                                - &                                                - &                                                - &                                                - \\
%                                                   &                                                  &                                                9 &                                                - &                                                - &                                                - &                                                - &                                                - &                                                - &                                                - \\
\cline{2-10}
&
\multirow{2}{*}{Times}
%                                                                                                      &                                                4 &                                   \textbf{0.06s} &                                            0.09s &                                            0.08s &                                   \textbf{0.09s} &                                   \textbf{0.09s} &                                   \textbf{0.09s} &                                   \textbf{0.09s} \\
%                                                   &                                                  &                                                5 &                                   \textbf{0.18s} &                                            0.58s &                                            1.99s &                                   \textbf{0.58s} &                                            1.15s &                                   \textbf{0.71s} &                                            1.96s \\
                                                                                                     &                                                6 &                                   \textbf{1.97s} &                                            9.00s &                                              53s &                                   \textbf{8.79s} &                                              60s &                                     \textbf{11s} &                                             161s \\
                                                   &                                                  &                                                7 &                                     \textbf{41s} &                                             189s &                                            3229s &                                    \textbf{185s} &                                            3245s &                                    \textbf{251s} &                                                - \\
%                                                   &                                                  &                                                8 &                                   \textbf{1736s} &                                                - &                                                - &                                                - &                                                - &                                                - &                                                - \\
%                                                   &                                                  &                                                9 &                                                - &                                                - &                                                - &                                                - &                                                - &                                                - &                                                - \\
\specialrule{.1em}{0em}{0em}
% 47_ticket_lock_hc_backoff_vs_true-unreach-call.c
\multirow{4}{*}{\begin{tabular}{l}
\textbf{47\_ticket\_hc}\\
% lines of code: 73\\
% variables: 2\\
% locks: 1\\
threads: U
\end{tabular}}
&
\multirow{2}{*}{Traces}
                                                                                                      &                                                1 &                                       \textbf{1} &                                       \textbf{1} &                                       \textbf{1} &                                       \textbf{1} &                                       \textbf{1} &                                       \textbf{1} &                                       \textbf{1} \\
                                                   &                                                  &                                                2 &                                      \textbf{48} &                                      \textbf{48} &                                      \textbf{48} &                                      \textbf{48} &                                      \textbf{48} &                                      \textbf{48} &                                      \textbf{48} \\
%                                                   &                                                  &                                                3 &                                                - &                                                - &                                                - &                                                - &                                                - &                                                - &                                                - \\
\cline{2-10}
&
\multirow{2}{*}{Times}
                                                                                                      &                                                1 &                                   \textbf{0.05s} &                                   \textbf{0.05s} &                                   \textbf{0.05s} &                                   \textbf{0.05s} &                                            0.07s &                                   \textbf{0.05s} &                                            0.07s \\
                                                   &                                                  &                                                2 &                                   \textbf{0.06s} &                                            0.08s &                                   \textbf{0.06s} &                                            0.08s &                                   \textbf{0.06s} &                                   \textbf{0.08s} &                                   \textbf{0.08s} \\
%                                                   &                                                  &                                                3 &                                                - &                                                - &                                                - &                                                - &                                                - &                                                - &                                                - \\
\specialrule{.1em}{0em}{0em}
% 48_ticket_lock_low_contention_vs_true-unreach-call.c
\multirow{4}{*}{\begin{tabular}{l}
\textbf{48\_ticket\_low}\\
% lines of code: 58\\
% variables: 2\\
% locks: 1\\
threads: U
\end{tabular}}
&
\multirow{2}{*}{Traces}
%                                                                                                      &                                                2 &                                       \textbf{6} &                                       \textbf{6} &                                       \textbf{6} &                                       \textbf{6} &                                       \textbf{6} &                                       \textbf{6} &                                       \textbf{6} \\
                                                                                                     &                                                3 &                                     \textbf{204} &                                     \textbf{204} &                                     \textbf{204} &                                     \textbf{204} &                                     \textbf{204} &                                     \textbf{204} &                                     \textbf{204} \\
                                                   &                                                  &                                                4 &                                   \textbf{41400} &                                   \textbf{41400} &                                   \textbf{41400} &                                   \textbf{41400} &                                   \textbf{41400} &                                   \textbf{41400} &                                   \textbf{41400} \\
%                                                   &                                                  &                                                5 &                                                - &                                                - &                                                - &                                                - &                                                - &                                                - &                                                - \\
\cline{2-10}
&
\multirow{2}{*}{Times}
%                                                                                                      &                                                2 &                                   \textbf{0.05s} &                                   \textbf{0.05s} &                                   \textbf{0.05s} &                                   \textbf{0.05s} &                                   \textbf{0.05s} &                                   \textbf{0.05s} &                                   \textbf{0.05s} \\
                                                                                                     &                                                3 &                                            0.09s &                                            0.15s &                                   \textbf{0.08s} &                                            0.15s &                                   \textbf{0.08s} &                                            0.18s &                                   \textbf{0.08s} \\
                                                   &                                                  &                                                4 &                                              13s &                                              37s &                                   \textbf{8.33s} &                                              36s &                                   \textbf{9.85s} &                                              50s &                                     \textbf{11s} \\
%                                                   &                                                  &                                                5 &                                                - &                                                - &                                                - &                                                - &                                                - &                                                - &                                                - \\
\specialrule{.1em}{0em}{0em}
% fib_bench_false_join.c
\multirow{4}{*}{\begin{tabular}{l}
\textbf{fib\_bench}\\
% lines of code: 52\\
% variables: 2\\
threads: 2
\end{tabular}}
&
\multirow{2}{*}{Traces}
%                                                                                                      &                                                3 &                                     \textbf{141} &                                     \textbf{141} &                                     \textbf{141} &                                     \textbf{175} &                                     \textbf{175} &                                     \textbf{175} &                                     \textbf{175} \\
%                                                   &                                                  &                                                4 &                                    \textbf{1107} &                                    \textbf{1107} &                                    \textbf{1107} &                                    \textbf{1764} &                                    \textbf{1764} &                                    \textbf{1764} &                                    \textbf{1764} \\
%                                                   &                                                  &                                                5 &                                    \textbf{8953} &                                    \textbf{8953} &                                    \textbf{8953} &                                   \textbf{19404} &                                   \textbf{19404} &                                   \textbf{19404} &                                   \textbf{19404} \\
                                                                                                     &                                                6 &                                   \textbf{73789} &                                   \textbf{73789} &                                   \textbf{73789} &                                  \textbf{226512} &                                  \textbf{226512} &                                  \textbf{226512} &                                  \textbf{226512} \\
                                                   &                                                  &                                                7 &                                  \textbf{616227} &                                  \textbf{616227} &                                  \textbf{616227} &                                 \textbf{2760615} &                                 \textbf{2760615} &                                 \textbf{2760615} &                                 \textbf{2760615} \\
\cline{2-10}
&
\multirow{2}{*}{Times}
%                                                                                                      &                                                3 &                                            0.07s &                                            0.10s &                                   \textbf{0.06s} &                                            0.11s &                                   \textbf{0.06s} &                                            0.12s &                                   \textbf{0.08s} \\
%                                                   &                                                  &                                                4 &                                   \textbf{0.20s} &                                            0.51s &                                            0.24s &                                            0.77s &                                   \textbf{0.22s} &                                            0.91s &                                   \textbf{0.30s} \\
%                                                   &                                                  &                                                5 &                                            1.93s &                                            4.43s &                                   \textbf{0.90s} &                                            9.67s &                                   \textbf{2.47s} &                                              11s &                                   \textbf{2.18s} \\
                                                                                                     &                                                6 &                                              19s &                                              42s &                                     \textbf{15s} &                                             143s &                                     \textbf{44s} &                                             170s &                                     \textbf{43s} \\
                                                   &                                                  &                                                7 &                                             162s &                                             426s &                                     \textbf{84s} &                                            2157s &                                    \textbf{462s} &                                            2561s &                                    \textbf{489s} \\
\specialrule{.1em}{0em}{0em}
% fillarray_false-valid-deref.c
\multirow{4}{*}{\begin{tabular}{l}
\textbf{fillarray\_false}\\
% lines of code: 42\\
% variables: 14\\
threads: 2
\end{tabular}}
&
\multirow{2}{*}{Traces}
%                                                                                                      &                                                1 &                                      \textbf{17} &                                      \textbf{17} &                                               26 &                                      \textbf{17} &                                               28 &                                      \textbf{17} &                                               28 \\
%                                                   &                                                  &                                                2 &                                     \textbf{475} &                                     \textbf{475} &                                             1058 &                                     \textbf{475} &                                             1226 &                                     \textbf{475} &                                             1258 \\
                                                                                                     &                                                3 &                                   \textbf{14625} &                                   \textbf{14625} &                                            47892 &                                   \textbf{14625} &                                            59404 &                                   \textbf{14625} &                                            63088 \\
                                                   &                                                  &                                                4 &                                  \textbf{471821} &                                  \textbf{471821} &                                          2278732 &                                  \textbf{471821} &                                          3023380 &                                  \textbf{471821} &                                          3329934 \\
\cline{2-10}
&
\multirow{2}{*}{Times}
%                                                                                                      &                                                1 &                                   \textbf{0.05s} &                                   \textbf{0.05s} &                                   \textbf{0.05s} &                                   \textbf{0.05s} &                                            0.07s &                                            0.06s &                                   \textbf{0.05s} \\
%                                                   &                                                  &                                                2 &                                            0.19s &                                            0.30s &                                   \textbf{0.15s} &                                            0.30s &                                   \textbf{0.27s} &                                            0.42s &                                   \textbf{0.38s} \\
                                                                                                     &                                                3 &                                   \textbf{5.73s} &                                              12s &                                            6.18s &                                     \textbf{12s} &                                     \textbf{12s} &                                     \textbf{18s} &                                              39s \\
                                                   &                                                  &                                                4 &                                    \textbf{151s} &                                             553s &                                             331s &                                    \textbf{547s} &                                             778s &                                    \textbf{930s} &                                            2844s \\
\specialrule{.1em}{0em}{0em}
\end{tabular}
% \addtocounter{table}{-1}
% \vspace{-2mm}
\caption{Part1: SVCOMP benchmarks.}
\label{tab:all_svcomp_1}
\end{table}

\begin{table}
\scriptsize
\newcolumntype{?}{!{\vrule width 1.5pt}}
\setlength{\extrarowheight}{.05em}
\begin{tabular}{? l | c | c ? r r r ? r r ? r r ? }
\specialrule{.15em}{0em}{0em}
\multicolumn{2}{?c|}{\multirow{2}{*}{\textbf{Benchmark}}} & \multirow{2}{*}{U} & \multicolumn{3}{c?}{\textbf{Sequential Consistency}} & \multicolumn{2}{c?}{\textbf{Total Store Order}} & \multicolumn{2}{c?}{\textbf{Partial Store Order}}\\
\cline{4-10}
\multicolumn{2}{?c|}{} & & $\ReadsFrom$ & $\DCTSOPSOM$ & $\Source$ & $\DCTSOPSOM$ & $\Source$ & $\DCTSOPSOM$ & $\Source$\\
\specialrule{.1em}{0em}{0em}
% fillarray_true-valid-memsafety.c
\multirow{4}{*}{\begin{tabular}{l}
\textbf{fillarray\_true}\\
% lines of code: 65\\
% variables: 17\\
threads: 2
\end{tabular}}
&
\multirow{2}{*}{Traces}
                                                                                                      &                                                2 &                                       \textbf{9} &                                       \textbf{9} &                                               10 &                                    \textbf{3334} &                                             9820 &                                    \textbf{3334} &                                            10076 \\
                                                   &                                                  &                                                3 &                                      \textbf{11} &                                      \textbf{11} &                                               12 &                                  \textbf{131636} &                                           594054 &                                  \textbf{131636} &                                           630894 \\
\cline{2-10}
&
\multirow{2}{*}{Times}
                                                                                                      &                                                2 &                                   \textbf{0.05s} &                                   \textbf{0.05s} &                                   \textbf{0.05s} &                                            2.30s &                                   \textbf{1.86s} &                                   \textbf{3.86s} &                                            5.14s \\
                                                   &                                                  &                                                3 &                                   \textbf{0.05s} &                                   \textbf{0.05s} &                                   \textbf{0.05s} &                                    \textbf{139s} &                                             148s &                                    \textbf{257s} &                                             652s \\
\specialrule{.1em}{0em}{0em}
% fk2012_1p1c_true-unreach-call.c
\multirow{4}{*}{\begin{tabular}{l}
\textbf{fk2012\_1p1c}\\
% lines of code: 99\\
% variables: 1\\
% locks: 2\\
threads: 2
\end{tabular}}
&
\multirow{2}{*}{Traces}
                                                                                                     &                                               15 &                                 \textbf{1999336} &                                 \textbf{1999336} &                                          3236936 &                                 \textbf{1999336} &                                          3236936 &                                 \textbf{1999336} &                                          3236936 \\
                                                   &                                                  &                                               16 &                                 \textbf{2399193} &                                 \textbf{2399193} &                                          3884313 &                                 \textbf{2399193} &                                          3884313 &                                 \textbf{2399193} &                                          3884313 \\
\cline{2-10}
&
\multirow{2}{*}{Times}
                                                                                                     &                                               15 &                                            1032s &                                            2652s &                                    \textbf{632s} &                                            2474s &                                   \textbf{2208s} &                                            2780s &                                   \textbf{2437s} \\
                                                   &                                                  &                                               16 &                                            1476s &                                            3133s &                                    \textbf{888s} &                                            3059s &                                   \textbf{3006s} &                                            3410s &                                   \textbf{2792s} \\
\specialrule{.1em}{0em}{0em}
% fk2012_1p2c_true-unreach-call.c
\multirow{4}{*}{\begin{tabular}{l}
\textbf{fk2012\_1p2c}\\
% lines of code: 100\\
% variables: 1\\
% locks: 2\\
threads: 3
\end{tabular}}
&
\multirow{2}{*}{Traces}
%                                                                                                      &                                                2 &                                     \textbf{842} &                                     \textbf{842} &                                              952 &                                     \textbf{842} &                                              952 &                                     \textbf{842} &                                              952 \\
                                                                                                     &                                                3 &                                   \textbf{33886} &                                   \textbf{33886} &                                            42144 &                                   \textbf{33886} &                                            42144 &                                   \textbf{33886} &                                            42144 \\
                                                   &                                                  &                                                4 &                                  \textbf{888404} &                                  \textbf{888404} &                                          1217826 &                                  \textbf{888404} &                                          1217826 &                                  \textbf{888404} &                                          1217826 \\
%                                                   &                                                  &                                                5 &                                                - &                                                - &                                                - &                                                - &                                                - &                                                - &                                                - \\
\cline{2-10}
&
\multirow{2}{*}{Times}
%                                                                                                      &                                                2 &                                            0.34s &                                            0.39s &                                   \textbf{0.17s} &                                            0.39s &                                   \textbf{0.17s} &                                            0.44s &                                   \textbf{0.24s} \\
                                                                                                     &                                                3 &                                              15s &                                              19s &                                   \textbf{5.03s} &                                              19s &                                   \textbf{6.75s} &                                              21s &                                     \textbf{11s} \\
                                                   &                                                  &                                                4 &                                             523s &                                             704s &                                    \textbf{250s} &                                             718s &                                    \textbf{253s} &                                             787s &                                    \textbf{300s} \\
%                                                   &                                                  &                                                5 &                                                - &                                                - &                                                - &                                                - &                                                - &                                                - &                                                - \\
\specialrule{.1em}{0em}{0em}
% fk2012_2p1c_true-unreach-call.c
\multirow{4}{*}{\begin{tabular}{l}
\textbf{fk2012\_2p1c}\\
% lines of code: 100\\
% variables: 1\\
% locks: 2\\
threads: 3
\end{tabular}}
&
\multirow{2}{*}{Traces}
                                                                                                     &                                               12 &                                 \textbf{1250886} &                                 \textbf{1250886} &                                          1931566 &                                 \textbf{1250886} &                                          1931566 &                                 \textbf{1250886} &                                          1931566 \\
                                                   &                                                  &                                               13 &                                 \textbf{2059540} &                                 \textbf{2059540} &                                          3230710 &                                 \textbf{2059540} &                                          3230710 &                                 \textbf{2059540} &                                          3230710 \\
%                                                   &                                                  &                                               14 &                                                - &                                                - &                                 \textbf{5332560} &                                                - &                                                - &                                                - &                                                - \\
%                                                   &                                                  &                                               15 &                                                - &                                                - &                                 \textbf{6475232} &                                                - &                                                - &                                                - &                                                - \\
\cline{2-10}
&
\multirow{2}{*}{Times}
                                                                                                     &                                               12 &                                             913s &                                            1565s &                                    \textbf{406s} &                                            1542s &                                    \textbf{958s} &                                            1712s &                                   \textbf{1176s} \\
                                                   &                                                  &                                               13 &                                            1734s &                                            2770s &                                    \textbf{905s} &                                            2791s &                                   \textbf{1726s} &                                            3004s &                                   \textbf{1893s} \\
%                                                   &                                                  &                                               14 &                                                - &                                                - &                                   \textbf{1353s} &                                                - &                                                - &                                                - &                                                - \\
%                                                   &                                                  &                                               15 &                                                - &                                                - &                                   \textbf{2036s} &                                                - &                                                - &                                                - &                                                - \\
\specialrule{.1em}{0em}{0em}
% fk2012_2p2c_true-unreach-call.c
\multirow{4}{*}{\begin{tabular}{l}
\textbf{fk2012\_2p2c}\\
% lines of code: 101\\
% variables: 1\\
% locks: 2\\
threads: 4
\end{tabular}}
&
\multirow{2}{*}{Traces}
%                                                                                                      &                                                1 &                                      \textbf{68} &                                      \textbf{68} &                                      \textbf{68} &                                      \textbf{68} &                                      \textbf{68} &                                      \textbf{68} &                                      \textbf{68} \\
                                                                                                     &                                                2 &                                    \textbf{3556} &                                    \textbf{3556} &                                             3680 &                                    \textbf{3556} &                                             3680 &                                    \textbf{3556} &                                             3680 \\
                                                   &                                                  &                                                3 &                                  \textbf{129120} &                                  \textbf{129120} &                                           145068 &                                  \textbf{129120} &                                           145068 &                                  \textbf{129120} &                                           145068 \\
%                                                   &                                                  &                                                4 &                                                - &                                 \textbf{3005100} &                                          3660136 &                                 \textbf{3005100} &                                          3660136 &                                 \textbf{3005100} &                                          3660136 \\
\cline{2-10}
&
\multirow{2}{*}{Times}
%                                                                                                      &                                                1 &                                   \textbf{0.06s} &                                            0.08s &                                            0.08s &                                   \textbf{0.08s} &                                            0.09s &                                            0.08s &                                   \textbf{0.06s} \\
                                                                                                     &                                                2 &                                   \textbf{0.96s} &                                            2.14s &                                            1.04s &                                            2.16s &                                   \textbf{0.66s} &                                            2.39s &                                   \textbf{0.63s} \\
                                                   &                                                  &                                                3 &                                              52s &                                             104s &                                     \textbf{29s} &                                             104s &                                     \textbf{50s} &                                             114s &                                     \textbf{34s} \\
%                                                   &                                                  &                                                4 &                                                - &                                            3331s &                                    \textbf{734s} &                                            3308s &                                    \textbf{984s} &                                            3596s &                                   \textbf{1333s} \\
\specialrule{.1em}{0em}{0em}
% fkp2013_true-unreach-call.c
\multirow{4}{*}{\begin{tabular}{l}
\textbf{fkp2013}\\
% lines of code: 26\\
% variables: 1\\
threads: U+1
\end{tabular}}
&
\multirow{2}{*}{Traces}
%                                                                                                      &                                                3 &                                      \textbf{64} &                                      \textbf{64} &                                              144 &                                      \textbf{64} &                                              144 &                                      \textbf{64} &                                              144 \\
%                                                   &                                                  &                                                4 &                                     \textbf{625} &                                     \textbf{625} &                                             2880 &                                     \textbf{625} &                                             2880 &                                     \textbf{625} &                                             2880 \\
%                                                   &                                                  &                                                5 &                                    \textbf{7776} &                                    \textbf{7776} &                                            86400 &                                    \textbf{7776} &                                            86400 &                                    \textbf{7776} &                                            86400 \\
                                                                                                     &                                                6 &                                  \textbf{117649} &                                  \textbf{117649} &                                          3628800 &                                  \textbf{117649} &                                          3628800 &                                  \textbf{117649} &                                          3628800 \\
                                                   &                                                  &                                                7 &                                 \textbf{2097152} &                                 \textbf{2097152} &                                                - &                                 \textbf{2097152} &                                                - &                                 \textbf{2097152} &                                                - \\
%                                                   &                                                  &                                                8 &                                                - &                                                - &                                                - &                                                - &                                                - &                                                - &                                                - \\
\cline{2-10}
&
\multirow{2}{*}{Times}
%                                                                                                      &                                                3 &                                   \textbf{0.05s} &                                            0.06s &                                   \textbf{0.05s} &                                   \textbf{0.06s} &                                   \textbf{0.06s} &                                   \textbf{0.07s} &                                            0.09s \\
%                                                   &                                                  &                                                4 &                                   \textbf{0.10s} &                                            0.25s &                                            0.22s &                                   \textbf{0.25s} &                                            0.40s &                                   \textbf{0.27s} &                                            0.40s \\
%                                                   &                                                  &                                                5 &                                   \textbf{0.76s} &                                            3.20s &                                            5.89s &                                   \textbf{3.15s} &                                            6.78s &                                   \textbf{3.52s} &                                              19s \\
                                                                                                     &                                                6 &                                     \textbf{13s} &                                              64s &                                             418s &                                     \textbf{63s} &                                             437s &                                     \textbf{69s} &                                            1200s \\
                                                   &                                                  &                                                7 &                                    \textbf{391s} &                                            1575s &                                                - &                                   \textbf{1530s} &                                                - &                                   \textbf{1677s} &                                                - \\
%                                                   &                                                  &                                                8 &                                                - &                                                - &                                                - &                                                - &                                                - &                                                - &                                                - \\
\specialrule{.1em}{0em}{0em}
% fkp2014_true-unreach-call.c
\multirow{4}{*}{\begin{tabular}{l}
\textbf{fkp2014}\\
% lines of code: 36\\
% variables: 2\\
% locks: 1\\
threads: U
\end{tabular}}
&
\multirow{2}{*}{Traces}
%                                                                                                      &                                                2 &                                      \textbf{16} &                                      \textbf{16} &                                      \textbf{16} &                                      \textbf{16} &                                      \textbf{16} &                                      \textbf{16} &                                      \textbf{16} \\
                                                                                                     &                                                3 &                                    \textbf{1098} &                                    \textbf{1098} &                                    \textbf{1098} &                                    \textbf{1098} &                                    \textbf{1098} &                                    \textbf{1098} &                                    \textbf{1098} \\
                                                   &                                                  &                                                4 &                                  \textbf{207024} &                                  \textbf{207024} &                                  \textbf{207024} &                                  \textbf{207024} &                                  \textbf{207024} &                                  \textbf{207024} &                                  \textbf{207024} \\
\cline{2-10}
&
\multirow{2}{*}{Times}
%                                                                                                      &                                                2 &                                   \textbf{0.05s} &                                   \textbf{0.05s} &                                   \textbf{0.05s} &                                   \textbf{0.07s} &                                   \textbf{0.07s} &                                   \textbf{0.05s} &                                   \textbf{0.05s} \\
                                                                                                     &                                                3 &                                            0.24s &                                            0.54s &                                   \textbf{0.16s} &                                            0.55s &                                   \textbf{0.22s} &                                            1.05s &                                   \textbf{0.17s} \\
                                                   &                                                  &                                                4 &                                              49s &                                             157s &                                     \textbf{31s} &                                             157s &                                     \textbf{29s} &                                             194s &                                     \textbf{48s} \\
\specialrule{.1em}{0em}{0em}
% gcd_true-unreach-call_true-termination.c
\multirow{4}{*}{\begin{tabular}{l}
\textbf{gcd}\\
% lines of code: 141\\
% variables: 2\\
% locks: 1\\
threads: 2
\end{tabular}}
&
\multirow{2}{*}{Traces}
                                                                                                     &                                               25 &                                   \textbf{61302} &                                   \textbf{61302} &                                   \textbf{61302} &                                   \textbf{61302} &                                   \textbf{61302} &                                   \textbf{61302} &                                   \textbf{61302} \\
                                                   &                                                  &                                               30 &                                  \textbf{106262} &                                  \textbf{106262} &                                  \textbf{106262} &                                  \textbf{106262} &                                  \textbf{106262} &                                  \textbf{106262} &                                  \textbf{106262} \\
\cline{2-10}
&
\multirow{2}{*}{Times}
                                                                                                     &                                               25 &                                              62s &                                             154s &                                     \textbf{26s} &                                             154s &                                     \textbf{33s} &                                             165s &                                     \textbf{42s} \\
                                                   &                                                  &                                               30 &                                             105s &                                             366s &                                     \textbf{54s} &                                             368s &                                     \textbf{57s} &                                             383s &                                     \textbf{47s} \\
\specialrule{.1em}{0em}{0em}
% indexer_true-unreach-call.c
\multirow{4}{*}{\begin{tabular}{l}
\textbf{indexer}\\
% lines of code: 64\\
% variables: 128\\
% locks: 128\\
threads: U
\end{tabular}}
&
\multirow{2}{*}{Traces}
%                                                                                                      &                                               12 &                                       \textbf{8} &                                       \textbf{8} &                                       \textbf{8} &                                       \textbf{8} &                                       \textbf{8} &                                       \textbf{8} &                                       \textbf{8} \\
%                                                   &                                                  &                                               13 &                                      \textbf{64} &                                      \textbf{64} &                                      \textbf{64} &                                      \textbf{64} &                                      \textbf{64} &                                      \textbf{64} &                                      \textbf{64} \\
%                                                   &                                                  &                                               14 &                                     \textbf{512} &                                     \textbf{512} &                                     \textbf{512} &                                     \textbf{512} &                                     \textbf{512} &                                     \textbf{512} &                                     \textbf{512} \\
                                                                                                     &                                               15 &                                    \textbf{4096} &                                    \textbf{4096} &                                    \textbf{4096} &                                    \textbf{4096} &                                    \textbf{4096} &                                    \textbf{4096} &                                    \textbf{4096} \\
                                                   &                                                  &                                               16 &                                   \textbf{32768} &                                   \textbf{32768} &                                   \textbf{32768} &                                   \textbf{32768} &                                   \textbf{32768} &                                   \textbf{32768} &                                   \textbf{32768} \\
\cline{2-10}
&
\multirow{2}{*}{Times}
%                                                                                                      &                                               12 &                                            0.10s &                                            0.11s &                                   \textbf{0.07s} &                                            0.14s &                                   \textbf{0.10s} &                                            0.26s &                                   \textbf{0.07s} \\
%                                                   &                                                  &                                               13 &                                            0.37s &                                            0.65s &                                   \textbf{0.17s} &                                            0.68s &                                   \textbf{0.20s} &                                            2.13s &                                   \textbf{0.19s} \\
%                                                   &                                                  &                                               14 &                                   \textbf{1.71s} &                                            5.71s &                                            1.92s &                                            5.51s &                                   \textbf{2.13s} &                                              20s &                                   \textbf{1.23s} \\
                                                                                                     &                                               15 &                                              18s &                                              53s &                                     \textbf{16s} &                                              51s &                                     \textbf{16s} &                                             181s &                                     \textbf{15s} \\
                                                   &                                                  &                                               16 &                                             197s &                                             520s &                                    \textbf{114s} &                                             499s &                                    \textbf{131s} &                                            1772s &                                    \textbf{110s} \\
\specialrule{.1em}{0em}{0em}
% nondet-array_true-unreach-call.c
\multirow{4}{*}{\begin{tabular}{l}
\textbf{nondet-array}\\
% lines of code: 29\\
% variables: U+1\\
threads: U
\end{tabular}}
&
\multirow{2}{*}{Traces}
%                                                                                                      &                                                3 &                                      \textbf{27} &                                      \textbf{27} &                                               90 &                                      \textbf{27} &                                               90 &                                      \textbf{27} &                                               90 \\
%                                                   &                                                  &                                                4 &                                     \textbf{292} &                                     \textbf{292} &                                             2616 &                                     \textbf{292} &                                             2616 &                                     \textbf{292} &                                             2664 \\
%                                                   &                                                  &                                                5 &                                    \textbf{4185} &                                    \textbf{4185} &                                           128760 &                                    \textbf{4185} &                                           128760 &                                    \textbf{4185} &                                           136920 \\
                                                                                                     &                                                6 &                                   \textbf{75486} &                                   \textbf{75486} &                                          9854640 &                                   \textbf{75486} &                                          9854640 &                                   \textbf{75486} &                                                - \\
                                                   &                                                  &                                                7 &                                 \textbf{1649221} &                                 \textbf{1649221} &                                                - &                                 \textbf{1649221} &                                                - &                                 \textbf{1649221} &                                                - \\
%                                                   &                                                  &                                                8 &                                                - &                                                - &                                                - &                                                - &                                                - &                                                - &                                                - \\
\cline{2-10}
&
\multirow{2}{*}{Times}
%                                                                                                      &                                                3 &                                            0.07s &                                   \textbf{0.06s} &                                   \textbf{0.06s} &                                   \textbf{0.05s} &                                            0.06s &                                   \textbf{0.06s} &                                            0.07s \\
%                                                   &                                                  &                                                4 &                                   \textbf{0.14s} &                                            0.17s &                                            0.39s &                                   \textbf{0.17s} &                                            0.42s &                                   \textbf{0.20s} &                                            0.83s \\
%                                                   &                                                  &                                                5 &                                   \textbf{1.33s} &                                            2.26s &                                              34s &                                   \textbf{2.23s} &                                              22s &                                   \textbf{2.97s} &                                              73s \\
                                                                                                     &                                                6 &                                     \textbf{28s} &                                              52s &                                            2779s &                                     \textbf{50s} &                                            2532s &                                     \textbf{70s} &                                                - \\
                                                   &                                                  &                                                7 &                                    \textbf{521s} &                                            1512s &                                                - &                                   \textbf{1445s} &                                                - &                                   \textbf{2163s} &                                                - \\
%                                                   &                                                  &                                                8 &                                                - &                                                - &                                                - &                                                - &                                                - &                                                - &                                                - \\
\specialrule{.1em}{0em}{0em}
% nondet-loop-bound-variant_true-unreach-call.c
\multirow{4}{*}{\begin{tabular}{l}
\textbf{nondet-loop-variant}\\
% lines of code: 34\\
% variables: 2\\
% locks: 1\\
threads: U+1
\end{tabular}}
&
\multirow{2}{*}{Traces}
%                                                                                                      &                                                3 &                                      \textbf{24} &                                      \textbf{24} &                                      \textbf{24} &                                      \textbf{24} &                                      \textbf{24} &                                      \textbf{24} &                                      \textbf{24} \\
%                                                   &                                                  &                                                4 &                                     \textbf{120} &                                     \textbf{120} &                                     \textbf{120} &                                     \textbf{120} &                                     \textbf{120} &                                     \textbf{120} &                                     \textbf{120} \\
%                                                   &                                                  &                                                5 &                                     \textbf{720} &                                     \textbf{720} &                                     \textbf{720} &                                     \textbf{720} &                                     \textbf{720} &                                     \textbf{720} &                                     \textbf{720} \\
                                                                                                     &                                                6 &                                    \textbf{5040} &                                    \textbf{5040} &                                    \textbf{5040} &                                    \textbf{5040} &                                    \textbf{5040} &                                    \textbf{5040} &                                    \textbf{5040} \\
                                                   &                                                  &                                                7 &                                   \textbf{40320} &                                   \textbf{40320} &                                   \textbf{40320} &                                   \textbf{40320} &                                   \textbf{40320} &                                   \textbf{40320} &                                   \textbf{40320} \\
\cline{2-10}
&
\multirow{2}{*}{Times}
%                                                                                                      &                                                3 &                                            0.06s &                                            0.07s &                                   \textbf{0.05s} &                                            0.06s &                                   \textbf{0.05s} &                                            0.06s &                                   \textbf{0.05s} \\
%                                                   &                                                  &                                                4 &                                            0.07s &                                            0.11s &                                   \textbf{0.06s} &                                            0.10s &                                   \textbf{0.06s} &                                            0.11s &                                   \textbf{0.06s} \\
%                                                   &                                                  &                                                5 &                                            0.29s &                                            0.46s &                                   \textbf{0.13s} &                                            0.47s &                                   \textbf{0.19s} &                                            0.53s &                                   \textbf{0.13s} \\
                                                                                                     &                                                6 &                                            2.15s &                                            3.67s &                                   \textbf{1.18s} &                                            3.70s &                                   \textbf{1.16s} &                                            4.18s &                                   \textbf{1.33s} \\
                                                   &                                                  &                                                7 &                                              24s &                                              35s &                                     \textbf{11s} &                                              36s &                                   \textbf{8.78s} &                                              41s &                                   \textbf{7.09s} \\
\specialrule{.1em}{0em}{0em}
% pthread-demo-datarace_false-unreach-call.c
\multirow{4}{*}{\begin{tabular}{l}
\textbf{pthread-datarace}\\
% lines of code: 67\\
% variables: 1\\
threads: 2
\end{tabular}}
&
\multirow{2}{*}{Traces}
%                                                                                                      &                                                3 &                                     \textbf{192} &                                     \textbf{192} &                                              328 &                                     \textbf{192} &                                              328 &                                     \textbf{192} &                                              328 \\
%                                                   &                                                  &                                                4 &                                    \textbf{1500} &                                    \textbf{1500} &                                             3334 &                                    \textbf{1500} &                                             3334 &                                    \textbf{1500} &                                             3334 \\
%                                                   &                                                  &                                                5 &                                   \textbf{12092} &                                   \textbf{12092} &                                            34904 &                                   \textbf{12092} &                                            34904 &                                   \textbf{12092} &                                            34904 \\
                                                                                                     &                                                6 &                                   \textbf{99442} &                                   \textbf{99442} &                                           372436 &                                   \textbf{99442} &                                           372436 &                                   \textbf{99442} &                                           372436 \\
                                                   &                                                  &                                                7 &                                  \textbf{829168} &                                  \textbf{829168} &                                          4027216 &                                  \textbf{829168} &                                          4027216 &                                  \textbf{829168} &                                          4027216 \\
%                                                   &                                                  &                                                8 &                                 \textbf{6984234} &                                 \textbf{6984234} &                                                - &                                 \textbf{6984234} &                                                - &                                                - &                                                - \\
%                                                   &                                                  &                                                9 &                                                - &                                                - &                                                - &                                                - &                                                - &                                                - &                                                - \\
\cline{2-10}
&
\multirow{2}{*}{Times}
%                                                                                                      &                                                3 &                                   \textbf{0.07s} &                                            0.13s &                                            0.11s &                                            0.09s &                                   \textbf{0.08s} &                                            0.10s &                                   \textbf{0.08s} \\
%                                                   &                                                  &                                                4 &                                   \textbf{0.23s} &                                            0.41s &                                            0.29s &                                   \textbf{0.41s} &                                            0.61s &                                            0.47s &                                   \textbf{0.34s} \\
%                                                   &                                                  &                                                5 &                                   \textbf{2.34s} &                                            3.57s &                                            3.43s &                                   \textbf{3.46s} &                                            6.28s &                                            4.05s &                                   \textbf{3.60s} \\
                                                                                                     &                                                6 &                                     \textbf{19s} &                                              34s &                                              53s &                                     \textbf{34s} &                                              75s &                                     \textbf{40s} &                                              50s \\
                                                   &                                                  &                                                7 &                                    \textbf{115s} &                                             332s &                                             565s &                                    \textbf{334s} &                                             807s &                                    \textbf{389s} &                                             547s \\
%                                                   &                                                  &                                                8 &                                   \textbf{1209s} &                                            3230s &                                                - &                                   \textbf{3226s} &                                                - &                                                - &                                                - \\
%                                                   &                                                  &                                                9 &                                                - &                                                - &                                                - &                                                - &                                                - &                                                - &                                                - \\
\specialrule{.1em}{0em}{0em}
% queue_ok_true-unreach-call.c
\multirow{4}{*}{\begin{tabular}{l}
\textbf{queue\_ok}\\
% lines of code: 157\\
% variables: 17\\
% locks: 1\\
threads: U+1
\end{tabular}}
&
\multirow{2}{*}{Traces}
%                                                                                                      &                                                3 &                                      \textbf{24} &                                      \textbf{24} &                                      \textbf{24} &                                      \textbf{24} &                                      \textbf{24} &                                      \textbf{24} &                                      \textbf{24} \\
%                                                   &                                                  &                                                4 &                                     \textbf{120} &                                     \textbf{120} &                                     \textbf{120} &                                     \textbf{120} &                                     \textbf{120} &                                     \textbf{120} &                                     \textbf{120} \\
%                                                   &                                                  &                                                5 &                                     \textbf{720} &                                     \textbf{720} &                                     \textbf{720} &                                     \textbf{720} &                                     \textbf{720} &                                     \textbf{720} &                                     \textbf{720} \\
                                                                                                     &                                                6 &                                    \textbf{5040} &                                    \textbf{5040} &                                    \textbf{5040} &                                    \textbf{5040} &                                    \textbf{5040} &                                    \textbf{5040} &                                    \textbf{5040} \\
                                                   &                                                  &                                                7 &                                   \textbf{40320} &                                   \textbf{40320} &                                   \textbf{40320} &                                   \textbf{40320} &                                   \textbf{40320} &                                   \textbf{40320} &                                   \textbf{40320} \\
%                                                   &                                                  &                                                8 &                                  \textbf{362880} &                                                - &                                  \textbf{362880} &                                                - &                                  \textbf{362880} &                                                - &                                  \textbf{362880} \\
%                                                   &                                                  &                                                9 &                                                - &                                                - &                                                - &                                                - &                                                - &                                                - &                                                - \\
\cline{2-10}
&
\multirow{2}{*}{Times}
%                                                                                                      &                                                3 &                                   \textbf{0.07s} &                                            0.11s &                                            0.11s &                                            0.11s &                                   \textbf{0.07s} &                                            0.18s &                                   \textbf{0.10s} \\
%                                                   &                                                  &                                                4 &                                   \textbf{0.22s} &                                            0.48s &                                            0.24s &                                            0.47s &                                   \textbf{0.19s} &                                            1.01s &                                   \textbf{0.15s} \\
%                                                   &                                                  &                                                5 &                                            1.34s &                                            3.36s &                                   \textbf{1.18s} &                                            3.27s &                                   \textbf{1.20s} &                                            7.57s &                                   \textbf{0.57s} \\
                                                                                                     &                                                6 &                                            8.29s &                                              30s &                                   \textbf{4.93s} &                                              29s &                                   \textbf{7.76s} &                                              72s &                                   \textbf{8.08s} \\
                                                   &                                                  &                                                7 &                                              61s &                                             303s &                                     \textbf{52s} &                                             296s &                                     \textbf{62s} &                                             736s &                                     \textbf{50s} \\
%                                                   &                                                  &                                                8 &                                             678s &                                                - &                                    \textbf{580s} &                                                - &                                    \textbf{659s} &                                                - &                                    \textbf{470s} \\
%                                                   &                                                  &                                                9 &                                                - &                                                - &                                                - &                                                - &                                                - &                                                - &                                                - \\
\specialrule{.1em}{0em}{0em}
% qw2004_variant_true-unreach-call.c
\multirow{4}{*}{\begin{tabular}{l}
\textbf{qw2004}\\
% lines of code: 62\\
% variables: 4\\
% locks: 1\\
threads: U+1
\end{tabular}}
&
\multirow{2}{*}{Traces}
%                                                                                                      &                                                1 &                                       \textbf{5} &                                       \textbf{5} &                                       \textbf{5} &                                       \textbf{5} &                                       \textbf{5} &                                       \textbf{5} &                                       \textbf{5} \\
%                                                   &                                                  &                                                2 &                                      \textbf{50} &                                      \textbf{50} &                                      \textbf{50} &                                      \textbf{50} &                                      \textbf{50} &                                      \textbf{50} &                                                - \\
%                                                   &                                                  &                                                3 &                                     \textbf{912} &                                     \textbf{912} &                                     \textbf{912} &                                     \textbf{912} &                                     \textbf{912} &                                     \textbf{912} &                                                - \\
                                                                                                     &                                                4 &                                   \textbf{28152} &                                   \textbf{28152} &                                   \textbf{28152} &                                   \textbf{28152} &                                   \textbf{28152} &                                   \textbf{28152} &                                                - \\
                                                   &                                                  &                                                5 &                                 \textbf{1354920} &                                 \textbf{1354920} &                                 \textbf{1354920} &                                 \textbf{1354920} &                                 \textbf{1354920} &                                 \textbf{1354920} &                                                - \\
\cline{2-10}
&
\multirow{2}{*}{Times}
%                                                                                                      &                                                1 &                                   \textbf{0.05s} &                                   \textbf{0.05s} &                                   \textbf{0.05s} &                                   \textbf{0.05s} &                                   \textbf{0.05s} &                                   \textbf{0.05s} &                                   \textbf{0.05s} \\
%                                                   &                                                  &                                                2 &                                   \textbf{0.06s} &                                            0.08s &                                            0.07s &                                            0.07s &                                   \textbf{0.06s} &                                   \textbf{0.08s} &                                                - \\
%                                                   &                                                  &                                                3 &                                   \textbf{0.26s} &                                            0.64s &                                            0.30s &                                            0.64s &                                   \textbf{0.19s} &                                   \textbf{0.86s} &                                                - \\
                                                                                                     &                                                4 &                                              12s &                                              26s &                                   \textbf{9.15s} &                                              26s &                                   \textbf{5.44s} &                                     \textbf{36s} &                                                - \\
                                                   &                                                  &                                                5 &                                             610s &                                            1870s &                                    \textbf{310s} &                                            1893s &                                    \textbf{353s} &                                   \textbf{2578s} &                                                - \\
\specialrule{.1em}{0em}{0em}
% reorder_5_false-unreach-call.c
\multirow{4}{*}{\begin{tabular}{l}
\textbf{reorder\_5}\\
% lines of code: 65\\
% variables: 2\\
threads: U+1
\end{tabular}}
&
\multirow{2}{*}{Traces}
%                                                                                                      &                                                3 &                                      \textbf{64} &                                      \textbf{64} &                                             1248 &                                      \textbf{64} &                                             1248 &                                      \textbf{67} &                                             1656 \\
                                                                                                     &                                                4 &                                     \textbf{145} &                                     \textbf{145} &                                            40032 &                                     \textbf{145} &                                            40032 &                                     \textbf{149} &                                            54720 \\
                                                   &                                                  &                                               30 &                                   \textbf{54901} &                                   \textbf{54901} &                                                - &                                   \textbf{54901} &                                                - &                                   \textbf{54931} &                                                - \\
\cline{2-10}
&
\multirow{2}{*}{Times}
%                                                                                                      &                                                3 &                                   \textbf{0.06s} &                                            0.07s &                                            0.34s &                                   \textbf{0.07s} &                                            0.20s &                                   \textbf{0.08s} &                                            0.52s \\
                                                                                                     &                                                4 &                                   \textbf{0.07s} &                                            0.11s &                                            7.66s &                                   \textbf{0.11s} &                                            8.97s &                                   \textbf{0.13s} &                                              25s \\
                                                   &                                                  &                                               30 &                                     \textbf{61s} &                                             278s &                                                - &                                    \textbf{265s} &                                                - &                                    \textbf{430s} &                                                - \\
\specialrule{.1em}{0em}{0em}
% scull_true-unreach-call_Rloop.c
\multirow{4}{*}{\begin{tabular}{l}
\textbf{scull\_Rloop}\\
% lines of code: 389\\
% variables: 8\\
% locks: 1\\
threads: 3
\end{tabular}}
&
\multirow{2}{*}{Traces}
%                                                                                                      &                                                1 &                                    \textbf{2491} &                                    \textbf{2491} &                                             6992 &                                    \textbf{2491} &                                             6992 &                                    \textbf{6188} &                                            42636 \\
%                                                   &                                                  &                                                2 &                                   \textbf{27092} &                                   \textbf{27092} &                                            95118 &                                   \textbf{27092} &                                            95118 &                                   \textbf{77708} &                                           710439 \\
                                                                                                     &                                                3 &                                  \textbf{148684} &                                  \textbf{148684} &                                           617706 &                                  \textbf{148684} &                                           617706 &                                  \textbf{478024} &                                          5397158 \\
                                                   &                                                  &                                                4 &                                  \textbf{569409} &                                  \textbf{569409} &                                          2732933 &                                  \textbf{569409} &                                          2732933 &                                                - &                                                - \\
%                                                   &                                                  &                                                5 &                                 \textbf{1740498} &                                                - &                                                - &                                                - &                                                - &                                                - &                                                - \\
%                                                   &                                                  &                                                6 &                                 \textbf{4551055} &                                                - &                                                - &                                                - &                                                - &                                                - &                                                - \\
%                                                   &                                                  &                                                7 &                                                - &                                                - &                                                - &                                                - &                                                - &                                                - &                                                - \\
\cline{2-10}
&
\multirow{2}{*}{Times}
%                                                                                                      &                                                1 &                                   \textbf{0.88s} &                                            2.53s &                                            2.90s &                                   \textbf{2.48s} &                                            3.12s &                                     \textbf{10s} &                                              13s \\
%                                                   &                                                  &                                                2 &                                     \textbf{16s} &                                              35s &                                              39s &                                     \textbf{34s} &                                              54s &                                    \textbf{170s} &                                             304s \\
                                                                                                     &                                                3 &                                    \textbf{104s} &                                             242s &                                             338s &                                    \textbf{235s} &                                             350s &                                   \textbf{1324s} &                                            2574s \\
                                                   &                                                  &                                                4 &                                    \textbf{508s} &                                            1123s &                                            1152s &                                   \textbf{1080s} &                                            1550s &                                                - &                                                - \\
%                                                   &                                                  &                                                5 &                                   \textbf{1173s} &                                                - &                                                - &                                                - &                                                - &                                                - &                                                - \\
%                                                   &                                                  &                                                6 &                                   \textbf{3594s} &                                                - &                                                - &                                                - &                                                - &                                                - &                                                - \\
%                                                   &                                                  &                                                7 &                                                - &                                                - &                                                - &                                                - &                                                - &                                                - &                                                - \\
\specialrule{.1em}{0em}{0em}
\end{tabular}
% \addtocounter{table}{-1}
% \vspace{-2mm}
\caption{Part2: SVCOMP benchmarks.}
\label{tab:all_svcomp_2}
\end{table}

\begin{table}
\scriptsize
\newcolumntype{?}{!{\vrule width 1.5pt}}
\setlength{\extrarowheight}{.01em}
\begin{tabular}{? l | c | c ? r r r ? r r ? r r ? }
\specialrule{.15em}{0em}{0em}
\multicolumn{2}{?c|}{\multirow{2}{*}{\textbf{Benchmark}}} & \multirow{2}{*}{U} & \multicolumn{3}{c?}{\textbf{Sequential Consistency}} & \multicolumn{2}{c?}{\textbf{Total Store Order}} & \multicolumn{2}{c?}{\textbf{Partial Store Order}}\\
\cline{4-10}
\multicolumn{2}{?c|}{} & & $\ReadsFrom$ & $\DCTSOPSOM$ & $\Source$ & $\DCTSOPSOM$ & $\Source$ & $\DCTSOPSOM$ & $\Source$\\
\specialrule{.1em}{0em}{0em}
% scull_true-unreach-call_Wloop.c
\multirow{4}{*}{\begin{tabular}{l}
\textbf{scull\_Wloop}\\
% lines of code: 389\\
% variables: 8\\
% locks: 1\\
threads: 3
\end{tabular}}
&
\multirow{2}{*}{Traces}
%                                                                                                      &                                                1 &                                    \textbf{2491} &                                    \textbf{2491} &                                             6992 &                                    \textbf{2491} &                                             6992 &                                    \textbf{6188} &                                            42636 \\
                                                                                                     &                                                2 &                                   \textbf{32305} &                                   \textbf{32305} &                                           117149 &                                   \textbf{32305} &                                           117149 &                                  \textbf{106220} &                                           929708 \\
%                                                   &                                                  &                                                3 &                                  \textbf{203866} &                                  \textbf{203866} &                                           870773 &                                  \textbf{203866} &                                           870773 &                                  \textbf{823076} &                                                - \\
                                                   &                                                  &                                                4 &                                  \textbf{874282} &                                  \textbf{874282} &                                          4241286 &                                  \textbf{874282} &                                          4241286 &                                                - &                                                - \\
%                                                   &                                                  &                                                5 &                                 \textbf{2940616} &                                                - &                                                - &                                                - &                                                - &                                                - &                                                - \\
%                                                   &                                                  &                                                6 &                                                - &                                                - &                                                - &                                                - &                                                - &                                                - &                                                - \\
\cline{2-10}
&
\multirow{2}{*}{Times}
%                                                                                                      &                                                1 &                                   \textbf{0.87s} &                                            2.55s &                                            3.03s &                                            2.49s &                                   \textbf{1.88s} &                                     \textbf{10s} &                                              23s \\
                                                                                                     &                                                2 &                                     \textbf{16s} &                                              44s &                                              62s &                                              43s &                                     \textbf{38s} &                                    \textbf{244s} &                                             534s \\
%                                                   &                                                  &                                                3 &                                    \textbf{123s} &                                             372s &                                             408s &                                    \textbf{348s} &                                             360s &                                   \textbf{2393s} &                                                - \\
                                                   &                                                  &                                                4 &                                    \textbf{669s} &                                            1896s &                                            2130s &                                   \textbf{1835s} &                                            2070s &                                                - &                                                - \\
%                                                   &                                                  &                                                5 &                                   \textbf{2764s} &                                                - &                                                - &                                                - &                                                - &                                                - &                                                - \\
%                                                   &                                                  &                                                6 &                                                - &                                                - &                                                - &                                                - &                                                - &                                                - &                                                - \\
\specialrule{.1em}{0em}{0em}
% scull_true-unreach-call_loop.c
\multirow{4}{*}{\begin{tabular}{l}
\textbf{scull\_loop}\\
% lines of code: 391\\
% variables: 8\\
% locks: 1\\
threads: 3
\end{tabular}}
&
\multirow{2}{*}{Traces}
                                                                                                      &                                                1 &                                    \textbf{2491} &                                    \textbf{2491} &                                             6992 &                                    \textbf{2491} &                                             6992 &                                    \textbf{6188} &                                            42636 \\
                                                   &                                                  &                                                2 &                                  \textbf{749811} &                                  \textbf{749811} &                                          3157281 &                                  \textbf{749811} &                                          3157281 &                                                - &                                                - \\
%                                                   &                                                  &                                                3 &                                                - &                                                - &                                                - &                                                - &                                                - &                                                - &                                                - \\
\cline{2-10}
&
\multirow{2}{*}{Times}
                                                                                                      &                                                1 &                                   \textbf{0.89s} &                                            2.52s &                                            2.66s &                                   \textbf{2.46s} &                                            3.13s &                                   \textbf{9.98s} &                                              22s \\
                                                   &                                                  &                                                2 &                                    \textbf{419s} &                                            1354s &                                            1489s &                                   \textbf{1301s} &                                            1499s &                                                - &                                                - \\
%                                                   &                                                  &                                                3 &                                                - &                                                - &                                                - &                                                - &                                                - &                                                - &                                                - \\
\specialrule{.1em}{0em}{0em}
% sigma_false-unreach-call.c
\multirow{4}{*}{\begin{tabular}{l}
\textbf{sigma}\\
% lines of code: 36\\
% variables: U+1\\
threads: U
\end{tabular}}
&
\multirow{2}{*}{Traces}
%                                                                                                      &                                                4 &                                      \textbf{64} &                                      \textbf{64} &                                              105 &                                      \textbf{64} &                                              105 &                                      \textbf{64} &                                              105 \\
%                                                   &                                                  &                                                5 &                                     \textbf{424} &                                     \textbf{424} &                                              945 &                                     \textbf{424} &                                              945 &                                     \textbf{424} &                                              945 \\
%                                                   &                                                  &                                                6 &                                    \textbf{3358} &                                    \textbf{3358} &                                            10395 &                                    \textbf{3358} &                                            10395 &                                    \textbf{3358} &                                            10395 \\
                                                                                                     &                                                7 &                                   \textbf{30952} &                                   \textbf{30952} &                                           135135 &                                   \textbf{30952} &                                           135135 &                                   \textbf{30952} &                                           135135 \\
                                                   &                                                  &                                                8 &                                  \textbf{325488} &                                  \textbf{325488} &                                          2027025 &                                  \textbf{325488} &                                          2027025 &                                  \textbf{325488} &                                          2027025 \\
%                                                   &                                                  &                                                9 &                                 \textbf{3845724} &                                                - &                                                - &                                                - &                                                - &                                                - &                                                - \\
%                                                   &                                                  &                                               10 &                                                - &                                                - &                                                - &                                                - &                                                - &                                                - &                                                - \\
\cline{2-10}
&
\multirow{2}{*}{Times}
%                                                                                                      &                                                4 &                                   \textbf{0.06s} &                                            0.08s &                                   \textbf{0.06s} &                                            0.08s &                                   \textbf{0.07s} &                                            0.08s &                                   \textbf{0.06s} \\
%                                                   &                                                  &                                                5 &                                   \textbf{0.18s} &                                            0.29s &                                            0.19s &                                            0.28s &                                   \textbf{0.26s} &                                            0.35s &                                   \textbf{0.20s} \\
%                                                   &                                                  &                                                6 &                                   \textbf{1.15s} &                                            2.44s &                                            1.86s &                                   \textbf{2.42s} &                                            3.82s &                                   \textbf{2.74s} &                                            4.00s \\
                                                                                                     &                                                7 &                                   \textbf{7.14s} &                                              27s &                                              29s &                                     \textbf{27s} &                                              44s &                                     \textbf{30s} &                                              35s \\
                                                   &                                                  &                                                8 &                                     \textbf{91s} &                                             373s &                                             532s &                                    \textbf{364s} &                                             658s &                                    \textbf{407s} &                                             722s \\
%                                                   &                                                  &                                                9 &                                   \textbf{1688s} &                                                - &                                                - &                                                - &                                                - &                                                - &                                                - \\
%                                                   &                                                  &                                               10 &                                                - &                                                - &                                                - &                                                - &                                                - &                                                - &                                                - \\
\specialrule{.1em}{0em}{0em}
% singleton_with-uninit-problems_true-unreach-call.c
\multirow{4}{*}{\begin{tabular}{l}
\textbf{singleton}\\
% lines of code: 43\\
% variables: 1\\
threads: U+1
\end{tabular}}
&
\multirow{2}{*}{Traces}
%                                                                                                      &                                                5 &                                       \textbf{5} &                                       \textbf{5} &                                              120 &                                       \textbf{5} &                                              120 &                                       \textbf{5} &                                              120 \\
%                                                   &                                                  &                                                8 &                                       \textbf{8} &                                       \textbf{8} &                                            40320 &                                       \textbf{8} &                                            40320 &                                       \textbf{8} &                                            40320 \\
                                                                                                     &                                                9 &                                       \textbf{9} &                                       \textbf{9} &                                           362880 &                                       \textbf{9} &                                           362880 &                                       \textbf{9} &                                           362880 \\
                                                   &                                                  &                                               10 &                                      \textbf{10} &                                      \textbf{10} &                                          3628800 &                                      \textbf{10} &                                          3628800 &                                      \textbf{10} &                                                - \\
\cline{2-10}
&
\multirow{2}{*}{Times}
%                                                                                                      &                                                5 &                                   \textbf{0.05s} &                                   \textbf{0.05s} &                                            0.07s &                                   \textbf{0.06s} &                                            0.07s &                                   \textbf{0.05s} &                                            0.07s \\
%                                                   &                                                  &                                                8 &                                   \textbf{0.05s} &                                            0.06s &                                            7.81s &                                   \textbf{0.06s} &                                            8.87s &                                   \textbf{0.06s} &                                              58s \\
                                                                                                     &                                                9 &                                   \textbf{0.06s} &                                   \textbf{0.06s} &                                             142s &                                   \textbf{0.06s} &                                              88s &                                   \textbf{0.06s} &                                             649s \\
                                                   &                                                  &                                               10 &                                            0.08s &                                   \textbf{0.06s} &                                             977s &                                   \textbf{0.06s} &                                            1083s &                                   \textbf{0.06s} &                                                - \\
\specialrule{.1em}{0em}{0em}
% sssc12_true-unreach-call.c
\multirow{4}{*}{\begin{tabular}{l}
\textbf{sssc12}\\
% lines of code: 60\\
% variables: 102\\
% locks: 1\\
threads: U
\end{tabular}}
&
\multirow{2}{*}{Traces}
%                                                                                                      &                                                4 &                                      \textbf{24} &                                      \textbf{24} &                                      \textbf{24} &                                      \textbf{24} &                                      \textbf{24} &                                      \textbf{24} &                                      \textbf{24} \\
%                                                   &                                                  &                                                5 &                                     \textbf{120} &                                     \textbf{120} &                                     \textbf{120} &                                     \textbf{120} &                                     \textbf{120} &                                     \textbf{120} &                                     \textbf{120} \\
%                                                   &                                                  &                                                6 &                                     \textbf{720} &                                     \textbf{720} &                                     \textbf{720} &                                     \textbf{720} &                                     \textbf{720} &                                     \textbf{720} &                                     \textbf{720} \\
                                                                                                     &                                                7 &                                    \textbf{5040} &                                    \textbf{5040} &                                    \textbf{5040} &                                    \textbf{5040} &                                    \textbf{5040} &                                    \textbf{5040} &                                    \textbf{5040} \\
                                                   &                                                  &                                                8 &                                   \textbf{40320} &                                   \textbf{40320} &                                   \textbf{40320} &                                   \textbf{40320} &                                   \textbf{40320} &                                   \textbf{40320} &                                   \textbf{40320} \\
\cline{2-10}
&
\multirow{2}{*}{Times}
%                                                                                                      &                                                4 &                                   \textbf{0.06s} &                                            0.07s &                                   \textbf{0.06s} &                                            0.07s &                                   \textbf{0.06s} &                                            0.08s &                                   \textbf{0.06s} \\
%                                                   &                                                  &                                                5 &                                            0.11s &                                            0.18s &                                   \textbf{0.10s} &                                            0.18s &                                   \textbf{0.16s} &                                            0.26s &                                   \textbf{0.16s} \\
%                                                   &                                                  &                                                6 &                                            0.45s &                                            1.13s &                                   \textbf{0.41s} &                                            1.13s &                                   \textbf{0.83s} &                                            2.02s &                                   \textbf{0.71s} \\
                                                                                                     &                                                7 &                                   \textbf{3.46s} &                                              11s &                                            4.08s &                                              11s &                                   \textbf{5.31s} &                                              22s &                                   \textbf{3.31s} \\
                                                   &                                                  &                                                8 &                                     \textbf{36s} &                                             124s &                                              49s &                                             122s &                                     \textbf{35s} &                                             275s &                                     \textbf{36s} \\
\specialrule{.1em}{0em}{0em}
% sssc12_variant_true-unreach-call.c
\multirow{4}{*}{\begin{tabular}{l}
\textbf{sssc12\_variant}\\
% lines of code: 61\\
% variables: 102\\
% locks: 1\\
threads: U
\end{tabular}}
&
\multirow{2}{*}{Traces}
%                                                                                                      &                                                4 &                                      \textbf{24} &                                      \textbf{24} &                                      \textbf{24} &                                      \textbf{24} &                                      \textbf{24} &                                      \textbf{24} &                                      \textbf{24} \\
%                                                   &                                                  &                                                5 &                                     \textbf{120} &                                     \textbf{120} &                                     \textbf{120} &                                     \textbf{120} &                                     \textbf{120} &                                     \textbf{120} &                                     \textbf{120} \\
%                                                   &                                                  &                                                6 &                                     \textbf{720} &                                     \textbf{720} &                                     \textbf{720} &                                     \textbf{720} &                                     \textbf{720} &                                     \textbf{720} &                                     \textbf{720} \\
                                                                                                     &                                                7 &                                    \textbf{5040} &                                    \textbf{5040} &                                    \textbf{5040} &                                    \textbf{5040} &                                    \textbf{5040} &                                    \textbf{5040} &                                    \textbf{5040} \\
                                                   &                                                  &                                                8 &                                   \textbf{40320} &                                   \textbf{40320} &                                   \textbf{40320} &                                   \textbf{40320} &                                   \textbf{40320} &                                   \textbf{40320} &                                   \textbf{40320} \\
\cline{2-10}
&
\multirow{2}{*}{Times}
%                                                                                                      &                                                4 &                                   \textbf{0.06s} &                                            0.07s &                                   \textbf{0.06s} &                                            0.07s &                                   \textbf{0.06s} &                                            0.08s &                                   \textbf{0.06s} \\
%                                                   &                                                  &                                                5 &                                            0.15s &                                            0.19s &                                   \textbf{0.11s} &                                            0.19s &                                   \textbf{0.11s} &                                            0.28s &                                   \textbf{0.11s} \\
%                                                   &                                                  &                                                6 &                                            0.83s &                                            1.30s &                                   \textbf{0.48s} &                                            1.25s &                                   \textbf{0.53s} &                                            2.25s &                                   \textbf{0.82s} \\
                                                                                                     &                                                7 &                                            4.47s &                                              13s &                                   \textbf{3.87s} &                                              12s &                                   \textbf{6.89s} &                                              24s &                                   \textbf{8.07s} \\
                                                   &                                                  &                                                8 &                                     \textbf{61s} &                                             149s &                                              62s &                                             134s &                                     \textbf{59s} &                                             309s &                                     \textbf{46s} \\
\specialrule{.1em}{0em}{0em}
% stack_true-unreach-call.c
\multirow{4}{*}{\begin{tabular}{l}
\textbf{stack}\\
% lines of code: 104\\
% variables: U+1\\
% locks: 1\\
threads: 2
\end{tabular}}
&
\multirow{2}{*}{Traces}
                                                                                                     &                                               10 &                                  \textbf{184756} &                                  \textbf{184756} &                                  \textbf{184756} &                                  \textbf{184756} &                                  \textbf{184756} &                                  \textbf{184756} &                                  \textbf{184756} \\
                                                   &                                                  &                                               11 &                                  \textbf{705432} &                                  \textbf{705432} &                                  \textbf{705432} &                                  \textbf{705432} &                                  \textbf{705432} &                                  \textbf{705432} &                                  \textbf{705432} \\
%                                                   &                                                  &                                               12 &                                 \textbf{2704156} &                                                - &                                 \textbf{2704156} &                                                - &                                 \textbf{2704156} &                                                - &                                 \textbf{2704156} \\
\cline{2-10}
&
\multirow{2}{*}{Times}
                                                                                                     &                                               10 &                                             117s &                                             275s &                                     \textbf{86s} &                                             281s &                                     \textbf{87s} &                                             371s &                                     \textbf{78s} \\
                                                   &                                                  &                                               11 &                                    \textbf{453s} &                                            1211s &                                             462s &                                            1216s &                                    \textbf{357s} &                                            1573s &                                    \textbf{340s} \\
%                                                   &                                                  &                                               12 &                                            2205s &                                                - &                                   \textbf{1299s} &                                                - &                                   \textbf{1789s} &                                                - &                                   \textbf{1409s} \\
\specialrule{.1em}{0em}{0em}
\end{tabular}
% \addtocounter{table}{-1}
% \vspace{-2mm}
\caption{Part3: SVCOMP benchmarks.}
\label{tab:all_svcomp_3}
\end{table}

\begin{table}[h]
% \vspace{-4mm}
\scriptsize
\newcolumntype{?}{!{\vrule width 1.5pt}}
\setlength{\extrarowheight}{.02em}
\begin{tabular}{? l | c | c ? r r r ? r r ? r r ? }
\specialrule{.15em}{0em}{0em}
\multicolumn{2}{?c|}{\multirow{2}{*}{\textbf{Benchmark}}} & \multirow{2}{*}{U} & \multicolumn{3}{c?}{\textbf{Sequential Consistency}} & \multicolumn{2}{c?}{\textbf{Total Store Order}} & \multicolumn{2}{c?}{\textbf{Partial Store Order}}\\
\cline{4-10}
\multicolumn{2}{?c|}{} & & $\ReadsFrom$ & $\DCTSOPSOM$ & $\Source$ & $\DCTSOPSOM$ & $\Source$ & $\DCTSOPSOM$ & $\Source$\\
\specialrule{.1em}{0em}{0em}
% X2Tv1.c
\multirow{4}{*}{\begin{tabular}{l}
\textbf{X2Tv1}\\
% lines of code: 54\\
% variables: 2\\
threads: 2
\end{tabular}}
&
\multirow{2}{*}{Traces}
%                                                                                                      &                                                3 &                                     \textbf{520} &                                     \textbf{520} &                                     \textbf{520} &                                    \textbf{4691} &                                    \textbf{4691} &                                    \textbf{4691} &                                    \textbf{4691} \\
                                                                                                     &                                                4 &                                    \textbf{3851} &                                    \textbf{3851} &                                    \textbf{3851} &                                  \textbf{129946} &                                  \textbf{129946} &                                  \textbf{129946} &                                  \textbf{129946} \\
                                                   &                                                  &                                                5 &                                   \textbf{30737} &                                   \textbf{30737} &                                   \textbf{30737} &                                 \textbf{3976753} &                                 \textbf{3976753} &                                 \textbf{3976753} &                                 \textbf{3976753} \\
%                                                   &                                                  &                                                6 &                                  \textbf{253042} &                                  \textbf{253042} &                                  \textbf{253042} &                                                - &                                                - &                                                - &                                                - \\
%                                                   &                                                  &                                                7 &                                 \textbf{2115302} &                                 \textbf{2115302} &                                 \textbf{2115302} &                                                - &                                                - &                                                - &                                                - \\
\cline{2-10}
&
\multirow{2}{*}{Times}
%                                                                                                      &                                                3 &                                            0.16s &                                            0.25s &                                   \textbf{0.10s} &                                            1.85s &                                   \textbf{1.19s} &                                            2.17s &                                   \textbf{0.70s} \\
                                                                                                     &                                                4 &                                            1.18s &                                            1.89s &                                   \textbf{0.94s} &                                              64s &                                     \textbf{24s} &                                              75s &                                     \textbf{36s} \\
                                                   &                                                  &                                                5 &                                              10s &                                              18s &                                   \textbf{8.17s} &                                            2581s &                                    \textbf{812s} &                                            3086s &                                    \textbf{891s} \\
%                                                   &                                                  &                                                6 &                                     \textbf{68s} &                                             177s &                                              72s &                                                - &                                                - &                                                - &                                                - \\
%                                                   &                                                  &                                                7 &                                             640s &                                            1755s &                                    \textbf{499s} &                                                - &                                                - &                                                - &                                                - \\
\specialrule{.1em}{0em}{0em}
% X2Tv10.c
\multirow{4}{*}{\begin{tabular}{l}
\textbf{X2Tv10}\\
% lines of code: 89\\
% variables: 2\\
threads: 2
\end{tabular}}
&
\multirow{2}{*}{Traces}
%                                                                                                      &                                                1 &                                      \textbf{64} &                                      \textbf{64} &                                      \textbf{64} &                                      \textbf{78} &                                      \textbf{78} &                                      \textbf{78} &                                      \textbf{78} \\
                                                                                                     &                                                2 &                                    \textbf{5079} &                                    \textbf{5079} &                                    \textbf{5079} &                                   \textbf{16282} &                                   \textbf{16282} &                                   \textbf{16282} &                                   \textbf{16282} \\
                                                   &                                                  &                                                3 &                                  \textbf{308433} &                                  \textbf{308433} &                                  \textbf{308433} &                                 \textbf{4225551} &                                 \textbf{4225551} &                                 \textbf{4225551} &                                 \textbf{4225551} \\
%                                                   &                                                  &                                                4 &                                                - &                                                - &                                                - &                                                - &                                                - &                                                - &                                                - \\
\cline{2-10}
&
\multirow{2}{*}{Times}
%                                                                                                      &                                                1 &                                            0.09s &                                   \textbf{0.07s} &                                            0.08s &                                   \textbf{0.07s} &                                            0.09s &                                            0.07s &                                   \textbf{0.06s} \\
                                                                                                     &                                                2 &                                   \textbf{0.87s} &                                            2.22s &                                            1.11s &                                            7.09s &                                   \textbf{2.77s} &                                            8.13s &                                   \textbf{4.35s} \\
                                                   &                                                  &                                                3 &                                     \textbf{62s} &                                             192s &                                     \textbf{62s} &                                            2915s &                                   \textbf{1156s} &                                            3496s &                                   \textbf{1047s} \\
%                                                   &                                                  &                                                4 &                                                - &                                                - &                                                - &                                                - &                                                - &                                                - &                                                - \\
\specialrule{.1em}{0em}{0em}
% X2Tv2.c
\multirow{4}{*}{\begin{tabular}{l}
\textbf{X2Tv2}\\
% lines of code: 63\\
% variables: 2\\
threads: 2
\end{tabular}}
&
\multirow{2}{*}{Traces}
%                                                                                                      &                                                1 &                                      \textbf{42} &                                      \textbf{42} &                                      \textbf{42} &                                      \textbf{48} &                                      \textbf{48} &                                      \textbf{48} &                                      \textbf{48} \\
                                                                                                     &                                                2 &                                    \textbf{1293} &                                    \textbf{1293} &                                    \textbf{1293} &                                    \textbf{4338} &                                    \textbf{4338} &                                    \textbf{4338} &                                    \textbf{4338} \\
                                                   &                                                  &                                                3 &                                   \textbf{69316} &                                   \textbf{69316} &                                   \textbf{69316} &                                  \textbf{931349} &                                  \textbf{931349} &                                  \textbf{931349} &                                  \textbf{931349} \\
%                                                   &                                                  &                                                4 &                                 \textbf{3552837} &                                 \textbf{3552837} &                                 \textbf{3552837} &                                                - &                                                - &                                                - &                                                - \\
\cline{2-10}
&
\multirow{2}{*}{Times}
%                                                                                                      &                                                1 &                                   \textbf{0.05s} &                                            0.06s &                                   \textbf{0.05s} &                                   \textbf{0.06s} &                                            0.08s &                                            0.06s &                                   \textbf{0.05s} \\
                                                                                                     &                                                2 &                                            0.36s &                                            0.52s &                                   \textbf{0.19s} &                                            1.68s &                                   \textbf{1.09s} &                                            1.92s &                                   \textbf{1.16s} \\
                                                   &                                                  &                                                3 &                                              16s &                                              36s &                                   \textbf{9.48s} &                                             535s &                                    \textbf{197s} &                                             634s &                                    \textbf{194s} \\
%                                                   &                                                  &                                                4 &                                             901s &                                            2537s &                                    \textbf{666s} &                                                - &                                                - &                                                - &                                                - \\
\specialrule{.1em}{0em}{0em}
% X2Tv3.c
\multirow{4}{*}{\begin{tabular}{l}
\textbf{X2Tv3}\\
% lines of code: 62\\
% variables: 2\\
threads: 2
\end{tabular}}
&
\multirow{2}{*}{Traces}
%                                                                                                      &                                                1 &                                      \textbf{26} &                                      \textbf{26} &                                      \textbf{26} &                                      \textbf{28} &                                      \textbf{28} &                                      \textbf{28} &                                      \textbf{28} \\
                                                                                                     &                                                2 &                                    \textbf{1030} &                                    \textbf{1030} &                                    \textbf{1030} &                                    \textbf{2486} &                                    \textbf{2486} &                                    \textbf{2486} &                                    \textbf{2486} \\
                                                   &                                                  &                                                3 &                                   \textbf{33866} &                                   \textbf{33866} &                                   \textbf{33866} &                                  \textbf{290984} &                                  \textbf{290984} &                                  \textbf{290984} &                                  \textbf{290984} \\
%                                                   &                                                  &                                                4 &                                 \textbf{1100628} &                                 \textbf{1100628} &                                 \textbf{1100628} &                                                - &                                                - &                                                - &                                                - \\
\cline{2-10}
&
\multirow{2}{*}{Times}
%                                                                                                      &                                                1 &                                   \textbf{0.05s} &                                   \textbf{0.05s} &                                            0.07s &                                   \textbf{0.05s} &                                   \textbf{0.05s} &                                            0.06s &                                   \textbf{0.05s} \\
                                                                                                     &                                                2 &                                            0.32s &                                            0.42s &                                   \textbf{0.19s} &                                            0.97s &                                   \textbf{0.34s} &                                            1.12s &                                   \textbf{0.67s} \\
                                                   &                                                  &                                                3 &                                              10s &                                              17s &                                   \textbf{5.56s} &                                             160s &                                     \textbf{79s} &                                             187s &                                     \textbf{71s} \\
%                                                   &                                                  &                                                4 &                                             306s &                                             749s &                                    \textbf{214s} &                                                - &                                                - &                                                - &                                                - \\
\specialrule{.1em}{0em}{0em}
% X2Tv4.c
\multirow{4}{*}{\begin{tabular}{l}
\textbf{X2Tv4}\\
% lines of code: 41\\
% variables: 3\\
threads: 2
\end{tabular}}
&
\multirow{2}{*}{Traces}
%                                                                                                      &                                                2 &                                     \textbf{187} &                                     \textbf{187} &                                              204 &                                     \textbf{339} &                                              418 &                                     \textbf{423} &                                              512 \\
                                                                                                     &                                                3 &                                    \textbf{2907} &                                    \textbf{2907} &                                             3164 &                                   \textbf{11210} &                                            16127 &                                   \textbf{20804} &                                            29165 \\
                                                   &                                                  &                                                4 &                                   \textbf{46275} &                                   \textbf{46275} &                                            50340 &                                  \textbf{425612} &                                           724832 &                                 \textbf{1292944} &                                          2121092 \\
%                                                   &                                                  &                                                5 &                                  \textbf{753131} &                                  \textbf{753131} &                                           819052 &                                                - &                                                - &                                                - &                                                - \\
%                                                   &                                                  &                                                6 &                                                - &                                                - &                                \textbf{13538204} &                                                - &                                                - &                                                - &                                                - \\
\cline{2-10}
&
\multirow{2}{*}{Times}
%                                                                                                      &                                                2 &                                            0.11s &                                            0.11s &                                   \textbf{0.07s} &                                            0.15s &                                   \textbf{0.10s} &                                            0.22s &                                   \textbf{0.13s} \\
                                                                                                     &                                                3 &                                            0.82s &                                            1.35s &                                   \textbf{0.37s} &                                            4.77s &                                   \textbf{2.50s} &                                              12s &                                   \textbf{7.11s} \\
                                                   &                                                  &                                                4 &                                              15s &                                              27s &                                   \textbf{6.65s} &                                             242s &                                    \textbf{172s} &                                            1047s &                                    \textbf{754s} \\
%                                                   &                                                  &                                                5 &                                             217s &                                             564s &                                    \textbf{131s} &                                                - &                                                - &                                                - &                                                - \\
%                                                   &                                                  &                                                6 &                                                - &                                                - &                                   \textbf{2545s} &                                                - &                                                - &                                                - &                                                - \\
\specialrule{.1em}{0em}{0em}
% X2Tv5.c
\multirow{4}{*}{\begin{tabular}{l}
\textbf{X2Tv5}\\
% lines of code: 53\\
% variables: 3\\
threads: 2
\end{tabular}}
&
\multirow{2}{*}{Traces}
%                                                                                                      &                                                2 &                                     \textbf{198} &                                     \textbf{198} &                                     \textbf{198} &                                     \textbf{367} &                                     \textbf{367} &                                     \textbf{367} &                                     \textbf{367} \\
                                                                                                     &                                                3 &                                    \textbf{2985} &                                    \textbf{2985} &                                    \textbf{2985} &                                   \textbf{12145} &                                   \textbf{12145} &                                   \textbf{12145} &                                   \textbf{12145} \\
                                                   &                                                  &                                                4 &                                   \textbf{46161} &                                   \textbf{46161} &                                   \textbf{46161} &                                  \textbf{459543} &                                  \textbf{459543} &                                  \textbf{459543} &                                  \textbf{459543} \\
%                                                   &                                                  &                                                5 &                                  \textbf{730647} &                                  \textbf{730647} &                                  \textbf{730647} &                                                - &                                                - &                                                - &                                                - \\
%                                                   &                                                  &                                                6 &                                                - &                                                - &                                \textbf{11755440} &                                                - &                                                - &                                                - &                                                - \\
\cline{2-10}
&
\multirow{2}{*}{Times}
%                                                                                                      &                                                2 &                                            0.08s &                                            0.11s &                                   \textbf{0.07s} &                                            0.16s &                                   \textbf{0.15s} &                                            0.18s &                                   \textbf{0.09s} \\
                                                                                                     &                                                3 &                                            0.89s &                                            1.35s &                                   \textbf{0.41s} &                                            5.11s &                                   \textbf{1.72s} &                                            6.02s &                                   \textbf{1.85s} \\
                                                   &                                                  &                                                4 &                                              16s &                                              25s &                                   \textbf{9.83s} &                                             257s &                                    \textbf{112s} &                                             307s &                                    \textbf{103s} \\
%                                                   &                                                  &                                                5 &                                             233s &                                             513s &                                    \textbf{133s} &                                                - &                                                - &                                                - &                                                - \\
%                                                   &                                                  &                                                6 &                                                - &                                                - &                                   \textbf{2580s} &                                                - &                                                - &                                                - &                                                - \\
\specialrule{.1em}{0em}{0em}
% X2Tv6.c
\multirow{4}{*}{\begin{tabular}{l}
\textbf{X2Tv6}\\
% lines of code: 73\\
% variables: 3\\
threads: 2
\end{tabular}}
&
\multirow{2}{*}{Traces}
%                                                                                                      &                                                1 &                                      \textbf{25} &                                      \textbf{25} &                                      \textbf{25} &                                      \textbf{52} &                                      \textbf{52} &                                      \textbf{62} &                                      \textbf{62} \\
                                                                                                     &                                                2 &                                     \textbf{718} &                                     \textbf{718} &                                     \textbf{718} &                                    \textbf{4568} &                                    \textbf{4568} &                                    \textbf{9361} &                                    \textbf{9361} \\
                                                   &                                                  &                                                3 &                                   \textbf{20371} &                                   \textbf{20371} &                                   \textbf{20371} &                                  \textbf{417726} &                                  \textbf{417726} &                                 \textbf{1927527} &                                 \textbf{1927527} \\
%                                                   &                                                  &                                                4 &                                  \textbf{596354} &                                  \textbf{596354} &                                  \textbf{596354} &                                                - &                                                - &                                                - &                                                - \\
%                                                   &                                                  &                                                5 &                                                - &                                                - &                                                - &                                                - &                                                - &                                                - &                                                - \\
\cline{2-10}
&
\multirow{2}{*}{Times}
%                                                                                                      &                                                1 &                                   \textbf{0.05s} &                                   \textbf{0.05s} &                                            0.07s &                                            0.06s &                                   \textbf{0.05s} &                                            0.07s &                                   \textbf{0.05s} \\
                                                                                                     &                                                2 &                                            0.16s &                                            0.34s &                                   \textbf{0.12s} &                                            1.85s &                                   \textbf{0.57s} &                                            4.74s &                                   \textbf{1.37s} \\
                                                   &                                                  &                                                3 &                                   \textbf{4.04s} &                                              12s &                                            4.77s &                                             249s &                                     \textbf{74s} &                                            1573s &                                    \textbf{630s} \\
%                                                   &                                                  &                                                4 &                                             216s &                                             476s &                                    \textbf{162s} &                                                - &                                                - &                                                - &                                                - \\
%                                                   &                                                  &                                                5 &                                                - &                                                - &                                                - &                                                - &                                                - &                                                - &                                                - \\
\specialrule{.1em}{0em}{0em}
\end{tabular}
% \addtocounter{table}{-1}
% \vspace{-2mm}
\caption{Part1: Mutual exclusion benchmarks.}
\label{tab:all_crit_1}
\end{table}

\begin{table}
\scriptsize
\newcolumntype{?}{!{\vrule width 1.5pt}}
\setlength{\extrarowheight}{.056em}
\begin{tabular}{? l | c | c ? r r r ? r r ? r r ? }
\specialrule{.15em}{0em}{0em}
\multicolumn{2}{?c|}{\multirow{2}{*}{\textbf{Benchmark}}} & \multirow{2}{*}{U} & \multicolumn{3}{c?}{\textbf{Sequential Consistency}} & \multicolumn{2}{c?}{\textbf{Total Store Order}} & \multicolumn{2}{c?}{\textbf{Partial Store Order}}\\
\cline{4-10}
\multicolumn{2}{?c|}{} & & $\ReadsFrom$ & $\DCTSOPSOM$ & $\Source$ & $\DCTSOPSOM$ & $\Source$ & $\DCTSOPSOM$ & $\Source$\\
\specialrule{.1em}{0em}{0em}
% X2Tv7.c
\multirow{4}{*}{\begin{tabular}{l}
\textbf{X2Tv7}\\
% lines of code: 81\\
% variables: 2\\
threads: 2
\end{tabular}}
&
\multirow{2}{*}{Traces}
                                                                                                      &                                                3 &                                     \textbf{573} &                                     \textbf{573} &                                     \textbf{573} &                                   \textbf{17803} &                                   \textbf{17803} &                                   \textbf{17803} &                                   \textbf{17803} \\
                                                   &                                                  &                                                4 &                                    \textbf{2383} &                                    \textbf{2383} &                                    \textbf{2383} &                                  \textbf{514323} &                                  \textbf{514323} &                                  \textbf{514323} &                                  \textbf{514323} \\
\cline{2-10}
&
\multirow{2}{*}{Times}
                                                                                                      &                                                3 &                                            0.14s &                                            0.27s &                                   \textbf{0.11s} &                                            7.68s &                                   \textbf{2.44s} &                                            8.83s &                                   \textbf{5.29s} \\
                                                   &                                                  &                                                4 &                                            0.47s &                                            1.20s &                                   \textbf{0.36s} &                                             296s &                                     \textbf{86s} &                                             348s &                                    \textbf{130s} \\
\specialrule{.1em}{0em}{0em}
% X2Tv8.c
\multirow{4}{*}{\begin{tabular}{l}
\textbf{X2Tv8}\\
% lines of code: 62\\
% variables: 2\\
threads: 2
\end{tabular}}
&
\multirow{2}{*}{Traces}
%                                                                                                      &                                                1 &                                      \textbf{16} &                                      \textbf{16} &                                               18 &                                      \textbf{21} &                                               25 &                                      \textbf{31} &                                               38 \\
                                                                                                     &                                                2 &                                     \textbf{394} &                                     \textbf{394} &                                              441 &                                    \textbf{2031} &                                             2987 &                                    \textbf{5785} &                                             9234 \\
                                                   &                                                  &                                                3 &                                    \textbf{8434} &                                    \textbf{8434} &                                             9894 &                                  \textbf{220505} &                                           411677 &                                 \textbf{1449109} &                                          3012426 \\
%                                                   &                                                  &                                                4 &                                  \textbf{186040} &                                  \textbf{186040} &                                           228417 &                                                - &                                                - &                                                - &                                                - \\
%                                                   &                                                  &                                                5 &                                 \textbf{4192466} &                                                - &                                          5391534 &                                                - &                                                - &                                                - &                                                - \\
\cline{2-10}
&
\multirow{2}{*}{Times}
%                                                                                                      &                                                1 &                                            0.08s &                                   \textbf{0.05s} &                                   \textbf{0.05s} &                                   \textbf{0.05s} &                                   \textbf{0.05s} &                                   \textbf{0.06s} &                                   \textbf{0.06s} \\
                                                                                                     &                                                2 &                                            0.13s &                                            0.20s &                                   \textbf{0.09s} &                                            0.80s &                                   \textbf{0.39s} &                                            2.83s &                                   \textbf{1.58s} \\
                                                   &                                                  &                                                3 &                                            2.54s &                                            4.65s &                                   \textbf{1.16s} &                                             124s &                                     \textbf{92s} &                                            1112s &                                    \textbf{806s} \\
%                                                   &                                                  &                                                4 &                                              43s &                                             133s &                                     \textbf{35s} &                                                - &                                                - &                                                - &                                                - \\
%                                                   &                                                  &                                                5 &                                            1390s &                                                - &                                    \textbf{964s} &                                                - &                                                - &                                                - &                                                - \\
\specialrule{.1em}{0em}{0em}
% X2Tv9.c
\multirow{4}{*}{\begin{tabular}{l}
\textbf{X2Tv9}\\
% lines of code: 59\\
% variables: 2\\
threads: 2
\end{tabular}}
&
\multirow{2}{*}{Traces}
%                                                                                                      &                                                1 &                                      \textbf{17} &                                      \textbf{17} &                                      \textbf{17} &                                      \textbf{18} &                                      \textbf{18} &                                      \textbf{18} &                                      \textbf{18} \\
%                                                   &                                                  &                                                2 &                                     \textbf{361} &                                     \textbf{361} &                                     \textbf{361} &                                     \textbf{731} &                                     \textbf{731} &                                     \textbf{731} &                                     \textbf{731} \\
                                                                                                     &                                                3 &                                    \textbf{7304} &                                    \textbf{7304} &                                    \textbf{7304} &                                   \textbf{38778} &                                   \textbf{38778} &                                   \textbf{38778} &                                   \textbf{38778} \\
                                                   &                                                  &                                                4 &                                  \textbf{153725} &                                  \textbf{153725} &                                  \textbf{153725} &                                 \textbf{2340172} &                                 \textbf{2340172} &                                 \textbf{2340172} &                                 \textbf{2340172} \\
%                                                   &                                                  &                                                5 &                                 \textbf{3324991} &                                 \textbf{3324991} &                                 \textbf{3324991} &                                                - &                                                - &                                                - &                                                - \\
\cline{2-10}
&
\multirow{2}{*}{Times}
%                                                                                                      &                                                1 &                                   \textbf{0.05s} &                                            0.06s &                                   \textbf{0.05s} &                                            0.06s &                                   \textbf{0.05s} &                                   \textbf{0.05s} &                                   \textbf{0.05s} \\
%                                                   &                                                  &                                                2 &                                            0.10s &                                            0.17s &                                   \textbf{0.08s} &                                            0.29s &                                   \textbf{0.12s} &                                            0.32s &                                   \textbf{0.13s} \\
                                                                                                     &                                                3 &                                            1.23s &                                            3.21s &                                   \textbf{0.79s} &                                              17s &                                   \textbf{9.01s} &                                              20s &                                   \textbf{5.16s} \\
                                                   &                                                  &                                                4 &                                     \textbf{29s} &                                              87s &                                     \textbf{29s} &                                            1528s &                                    \textbf{399s} &                                            1765s &                                    \textbf{486s} \\
%                                                   &                                                  &                                                5 &                                             918s &                                            2484s &                                    \textbf{534s} &                                                - &                                                - &                                                - &                                                - \\
\specialrule{.1em}{0em}{0em}
% bakery_2thr.c
\multirow{4}{*}{\begin{tabular}{l}
\textbf{bakery}\\
% lines of code: 87\\
% variables: 4\\
threads: 2
\end{tabular}}
&
\multirow{2}{*}{Traces}
                                                                                                      &                                                1 &                                      \textbf{41} &                                      \textbf{41} &                                      \textbf{41} &                                      \textbf{73} &                                      \textbf{73} &                                      \textbf{77} &                                      \textbf{77} \\
                                                   &                                                  &                                                2 &                                    \textbf{7795} &                                    \textbf{7795} &                                    \textbf{7795} &                                   \textbf{25127} &                                   \textbf{25127} &                                   \textbf{28749} &                                   \textbf{28749} \\
%                                                   &                                                  &                                                3 &                                 \textbf{1574565} &                                 \textbf{1574565} &                                 \textbf{1574565} &                                                - &                                 \textbf{9297249} &                                                - &                                                - \\
\cline{2-10}
&
\multirow{2}{*}{Times}
                                                                                                      &                                                1 &                                   \textbf{0.06s} &                                            0.07s &                                   \textbf{0.06s} &                                            0.08s &                                   \textbf{0.06s} &                                            0.10s &                                   \textbf{0.07s} \\
                                                   &                                                  &                                                2 &                                            2.16s &                                            5.63s &                                   \textbf{1.72s} &                                              19s &                                     \textbf{11s} &                                              31s &                                     \textbf{13s} \\
%                                                   &                                                  &                                                3 &                                             687s &                                            1714s &                                    \textbf{478s} &                                                - &                                   \textbf{3533s} &                                                - &                                                - \\
\specialrule{.1em}{0em}{0em}
% bakery_3thr.c
\multirow{4}{*}{\begin{tabular}{l}
\textbf{bakery3}\\
% lines of code: 87\\
% variables: 6\\
threads: 3
\end{tabular}}
&
\multirow{2}{*}{Traces}
                                                                                                      &                                                1 &                                    \textbf{5296} &                                    \textbf{5296} &                                    \textbf{5296} &                                   \textbf{16315} &                                   \textbf{16315} &                                   \textbf{17588} &                                   \textbf{17588} \\
                                                   &                                                  &                                                2 &                                                - &                                                - &                                                - &                                                - &                                                - &                                                - &                                                - \\
\cline{2-10}
&
\multirow{2}{*}{Times}
                                                                                                      &                                                1 &                                            1.72s &                                            4.64s &                                   \textbf{1.63s} &                                              14s &                                   \textbf{6.33s} &                                              24s &                                   \textbf{7.37s} \\
                                                   &                                                  &                                                2 &                                                - &                                                - &                                                - &                                                - &                                                - &                                                - &                                                - \\
\specialrule{.1em}{0em}{0em}
% burns.c
\multirow{4}{*}{\begin{tabular}{l}
\textbf{burns}\\
% lines of code: 66\\
% variables: 2\\
threads: 2
\end{tabular}}
&
\multirow{2}{*}{Traces}
                                                                                                      &                                                2 &                                     \textbf{342} &                                     \textbf{342} &                                     \textbf{342} &                                    \textbf{2989} &                                    \textbf{2989} &                                    \textbf{2989} &                                    \textbf{2989} \\
                                                   &                                                  &                                                3 &                                    \textbf{6887} &                                    \textbf{6887} &                                    \textbf{6887} &                                  \textbf{213915} &                                  \textbf{213915} &                                  \textbf{213915} &                                  \textbf{213915} \\
%                                                   &                                                  &                                                4 &                                  \textbf{140380} &                                  \textbf{140380} &                                  \textbf{140380} &                                                - &                                                - &                                                - &                                                - \\
%                                                   &                                                  &                                                5 &                                 \textbf{2916980} &                                 \textbf{2916980} &                                 \textbf{2916980} &                                                - &                                                - &                                                - &                                                - \\
%                                                   &                                                  &                                                6 &                                                - &                                                - &                                                - &                                                - &                                                - &                                                - &                                                - \\
\cline{2-10}
&
\multirow{2}{*}{Times}
                                                                                                      &                                                2 &                                   \textbf{0.11s} &                                            0.19s &                                            0.16s &                                            1.23s &                                   \textbf{0.51s} &                                            1.38s &                                   \textbf{0.47s} \\
                                                   &                                                  &                                                3 &                                   \textbf{1.64s} &                                            3.74s &                                            2.23s &                                             117s &                                     \textbf{43s} &                                             135s &                                     \textbf{59s} \\
%                                                   &                                                  &                                                4 &                                     \textbf{41s} &                                             101s &                                              56s &                                                - &                                                - &                                                - &                                                - \\
%                                                   &                                                  &                                                5 &                                            1188s &                                            2629s &                                    \textbf{905s} &                                                - &                                                - &                                                - &                                                - \\
%                                                   &                                                  &                                                6 &                                                - &                                                - &                                                - &                                                - &                                                - &                                                - &                                                - \\
\specialrule{.1em}{0em}{0em}
% burns_3thr.c
\multirow{4}{*}{\begin{tabular}{l}
\textbf{burns3}\\
% lines of code: 66\\
% variables: 3\\
threads: 3
\end{tabular}}
&
\multirow{2}{*}{Traces}
                                                                                                      &                                                1 &                                     \textbf{849} &                                     \textbf{849} &                                     \textbf{849} &                                   \textbf{23502} &                                   \textbf{23502} &                                   \textbf{23502} &                                   \textbf{23502} \\
                                                   &                                                  &                                                2 &                                 \textbf{1490331} &                                 \textbf{1490331} &                                 \textbf{1490331} &                                                - &                                                - &                                                - &                                                - \\
%                                                   &                                                  &                                                3 &                                                - &                                                - &                                                - &                                                - &                                                - &                                                - &                                                - \\
\cline{2-10}
&
\multirow{2}{*}{Times}
                                                                                                      &                                                1 &                                            0.21s &                                            0.44s &                                   \textbf{0.18s} &                                              11s &                                   \textbf{7.74s} &                                              12s &                                   \textbf{8.83s} \\
                                                   &                                                  &                                                2 &                                    \textbf{458s} &                                            1287s &                                             485s &                                                - &                                                - &                                                - &                                                - \\
%                                                   &                                                  &                                                3 &                                                - &                                                - &                                                - &                                                - &                                                - &                                                - &                                                - \\
\specialrule{.1em}{0em}{0em}
% dekker.c
\multirow{4}{*}{\begin{tabular}{l}
\textbf{dekker}\\
% lines of code: 82\\
% variables: 3\\
threads: 2
\end{tabular}}
&
\multirow{2}{*}{Traces}
%                                                                                                      &                                                2 &                                     \textbf{143} &                                     \textbf{143} &                                     \textbf{143} &                                     \textbf{146} &                                     \textbf{146} &                                     \textbf{198} &                                              252 \\
%                                                   &                                                  &                                                3 &                                    \textbf{1946} &                                    \textbf{1946} &                                    \textbf{1946} &                                    \textbf{1977} &                                    \textbf{1977} &                                    \textbf{3135} &                                             4400 \\
                                                                                                     &                                                4 &                                   \textbf{28595} &                                   \textbf{28595} &                                   \textbf{28595} &                                   \textbf{29044} &                                   \textbf{29044} &                                   \textbf{53349} &                                            83333 \\
                                                   &                                                  &                                                5 &                                  \textbf{435245} &                                  \textbf{435245} &                                  \textbf{435245} &                                  \textbf{441810} &                                  \textbf{441810} &                                  \textbf{947754} &                                          1636946 \\
%                                                   &                                                  &                                                6 &                                 \textbf{6745775} &                                                - &                                 \textbf{6745775} &                                                - &                                 \textbf{6845369} &                                                - &                                                - \\
%                                                   &                                                  &                                                7 &                                                - &                                                - &                                                - &                                                - &                                                - &                                                - &                                                - \\
\cline{2-10}
&
\multirow{2}{*}{Times}
%                                                                                                      &                                                2 &                                   \textbf{0.06s} &                                            0.09s &                                            0.07s &                                   \textbf{0.09s} &                                   \textbf{0.09s} &                                            0.12s &                                   \textbf{0.09s} \\
%                                                   &                                                  &                                                3 &                                            0.51s &                                            0.75s &                                   \textbf{0.24s} &                                            0.76s &                                   \textbf{0.27s} &                                            1.61s &                                   \textbf{1.01s} \\
                                                                                                     &                                                4 &                                            6.18s &                                              13s &                                   \textbf{3.88s} &                                              13s &                                   \textbf{6.03s} &                                              35s &                                     \textbf{17s} \\
                                                   &                                                  &                                                5 &                                             128s &                                             259s &                                     \textbf{56s} &                                             262s &                                     \textbf{80s} &                                             808s &                                    \textbf{289s} \\
%                                                   &                                                  &                                                6 &                                            1869s &                                                - &                                    \textbf{937s} &                                                - &                                   \textbf{1274s} &                                                - &                                                - \\
%                                                   &                                                  &                                                7 &                                                - &                                                - &                                                - &                                                - &                                                - &                                                - &                                                - \\
\specialrule{.1em}{0em}{0em}
% dijkstra.c
\multirow{4}{*}{\begin{tabular}{l}
\textbf{dijkstra}\\
% lines of code: 87\\
% variables: 5\\
threads: 2
\end{tabular}}
&
\multirow{2}{*}{Traces}
                                                                                                      &                                                2 &                                     \textbf{319} &                                     \textbf{319} &                                     \textbf{319} &                                     \textbf{540} &                                     \textbf{540} &                                    \textbf{9961} &                                    \textbf{9961} \\
%                                                   &                                                  &                                                3 &                                   \textbf{10151} &                                   \textbf{10151} &                                   \textbf{10151} &                                   \textbf{32993} &                                   \textbf{32993} &                                                - &                                                - \\
                                                   &                                                  &                                                4 &                                  \textbf{353859} &                                  \textbf{353859} &                                  \textbf{353859} &                                 \textbf{2196640} &                                 \textbf{2196640} &                                                - &                                                - \\
%                                                   &                                                  &                                                5 &                                                - &                                                - &                                                - &                                                - &                                                - &                                                - &                                                - \\
\cline{2-10}
&
\multirow{2}{*}{Times}
                                                                                                      &                                                2 &                                   \textbf{0.15s} &                                            0.27s &                                            0.25s &                                            0.42s &                                   \textbf{0.22s} &                                              13s &                                   \textbf{6.29s} \\
%                                                   &                                                  &                                                3 &                                            3.57s &                                            9.82s &                                   \textbf{2.77s} &                                              31s &                                     \textbf{11s} &                                                - &                                                - \\
                                                   &                                                  &                                                4 &                                             207s &                                             462s &                                    \textbf{157s} &                                            2927s &                                   \textbf{1234s} &                                                - &                                                - \\
%                                                   &                                                  &                                                5 &                                                - &                                                - &                                                - &                                                - &                                                - &                                                - &                                                - \\
\specialrule{.1em}{0em}{0em}
% dijkstra_3thr.c
\multirow{4}{*}{\begin{tabular}{l}
\textbf{dijkstra3}\\
% lines of code: 87\\
% variables: 7\\
threads: 3
\end{tabular}}
&
\multirow{2}{*}{Traces}
                                                                                                      &                                                1 &                                     \textbf{741} &                                     \textbf{741} &                                             1192 &                                    \textbf{1120} &                                             1934 &                                   \textbf{12328} &                                            19330 \\
                                                   &                                                  &                                                2 &                                                - &                                                - &                                                - &                                                - &                                                - &                                                - &                                                - \\
\cline{2-10}
&
\multirow{2}{*}{Times}
                                                                                                      &                                                1 &                                   \textbf{0.38s} &                                            0.66s &                                            0.52s &                                            0.93s &                                   \textbf{0.78s} &                                     \textbf{19s} &                                              36s \\
                                                   &                                                  &                                                2 &                                                - &                                                - &                                                - &                                                - &                                                - &                                                - &                                                - \\
\specialrule{.1em}{0em}{0em}
% kessels.c
\multirow{4}{*}{\begin{tabular}{l}
\textbf{kessels}\\
% lines of code: 42\\
% variables: 2\\
threads: 2
\end{tabular}}
&
\multirow{2}{*}{Traces}
                                                                                                      &                                                2 &                                     \textbf{624} &                                     \textbf{624} &                                     \textbf{624} &                                    \textbf{3779} &                                    \textbf{3779} &                                    \textbf{3779} &                                    \textbf{3779} \\
                                                   &                                                  &                                                3 &                                   \textbf{13856} &                                   \textbf{13856} &                                   \textbf{13856} &                                  \textbf{356844} &                                  \textbf{356844} &                                  \textbf{356844} &                                  \textbf{356844} \\
%                                                   &                                                  &                                                4 &                                  \textbf{323400} &                                  \textbf{323400} &                                  \textbf{323400} &                                                - &                                                - &                                                - &                                                - \\
%                                                   &                                                  &                                                5 &                                 \textbf{7763704} &                                                - &                                 \textbf{7763704} &                                                - &                                                - &                                                - &                                                - \\
\cline{2-10}
&
\multirow{2}{*}{Times}
                                                                                                      &                                                2 &                                            0.15s &                                            0.27s &                                   \textbf{0.11s} &                                            1.36s &                                   \textbf{0.84s} &                                            1.58s &                                   \textbf{0.51s} \\
                                                   &                                                  &                                                3 &                                            2.96s &                                            6.61s &                                   \textbf{1.71s} &                                             177s &                                     \textbf{48s} &                                             212s &                                     \textbf{83s} \\
%                                                   &                                                  &                                                4 &                                     \textbf{70s} &                                             199s &                                     \textbf{70s} &                                                - &                                                - &                                                - &                                                - \\
%                                                   &                                                  &                                                5 &                                            1965s &                                                - &                                   \textbf{1761s} &                                                - &                                                - &                                                - &                                                - \\
\specialrule{.1em}{0em}{0em}
% lamport.c
\multirow{4}{*}{\begin{tabular}{l}
\textbf{lamport}\\
% lines of code: 76\\
% variables: 4\\
threads: 2
\end{tabular}}
&
\multirow{2}{*}{Traces}
%                                                                                                      &                                                1 &                                      \textbf{22} &                                      \textbf{22} &                                               32 &                                      \textbf{23} &                                               44 &                                      \textbf{49} &                                               96 \\
                                                                                                     &                                                2 &                                    \textbf{1456} &                                    \textbf{1456} &                                             3940 &                                    \textbf{2449} &                                            10652 &                                  \textbf{100521} &                                           473670 \\
                                                   &                                                  &                                                3 &                                  \textbf{130024} &                                  \textbf{130024} &                                           741370 &                                  \textbf{367113} &                                          3887642 &                                                - &                                                - \\
%                                                   &                                                  &                                                4 &                                                - &                                                - &                                                - &                                                - &                                                - &                                                - &                                                - \\
\cline{2-10}
&
\multirow{2}{*}{Times}
%                                                                                                      &                                                1 &                                   \textbf{0.05s} &                                   \textbf{0.05s} &                                            0.08s &                                   \textbf{0.05s} &                                            0.06s &                                   \textbf{0.07s} &                                            0.09s \\
                                                                                                     &                                                2 &                                   \textbf{0.51s} &                                            0.77s &                                            0.86s &                                   \textbf{1.25s} &                                            1.86s &                                     \textbf{94s} &                                             456s \\
                                                   &                                                  &                                                3 &                                     \textbf{41s} &                                             105s &                                             128s &                                    \textbf{301s} &                                             766s &                                                - &                                                - \\
%                                                   &                                                  &                                                4 &                                                - &                                                - &                                                - &                                                - &                                                - &                                                - &                                                - \\
\specialrule{.1em}{0em}{0em}
% peterson.c
\multirow{4}{*}{\begin{tabular}{l}
\textbf{peterson}\\
% lines of code: 61\\
% variables: 3\\
threads: 2
\end{tabular}}
&
\multirow{2}{*}{Traces}
%                                                                                                      &                                                2 &                                     \textbf{129} &                                     \textbf{129} &                                              136 &                                     \textbf{341} &                                              706 &                                     \textbf{399} &                                              844 \\
                                                                                                     &                                                3 &                                    \textbf{1609} &                                    \textbf{1609} &                                             1686 &                                    \textbf{9251} &                                            29546 &                                   \textbf{15925} &                                            48500 \\
                                                   &                                                  &                                                4 &                                   \textbf{20161} &                                   \textbf{20161} &                                            21120 &                                  \textbf{263916} &                                          1359890 &                                  \textbf{744888} &                                          3286070 \\
%                                                   &                                                  &                                                5 &                                  \textbf{256457} &                                  \textbf{256457} &                                           268706 &                                                - &                                                - &                                                - &                                                - \\
%                                                   &                                                  &                                                6 &                                 \textbf{3303617} &                                 \textbf{3303617} &                                          3462008 &                                                - &                                                - &                                                - &                                                - \\
%                                                   &                                                  &                                                7 &                                                - &                                                - &                                                - &                                                - &                                                - &                                                - &                                                - \\
\cline{2-10}
&
\multirow{2}{*}{Times}
%                                                                                                      &                                                2 &                                   \textbf{0.06s} &                                            0.09s &                                   \textbf{0.06s} &                                   \textbf{0.14s} &                                            0.20s &                                            0.20s &                                   \textbf{0.15s} \\
                                                                                                     &                                                3 &                                            0.27s &                                            0.68s &                                   \textbf{0.19s} &                                   \textbf{3.55s} &                                            6.73s &                                   \textbf{8.13s} &                                            9.87s \\
                                                   &                                                  &                                                4 &                                            3.48s &                                              10s &                                   \textbf{2.99s} &                                    \textbf{138s} &                                             314s &                                    \textbf{537s} &                                             710s \\
%                                                   &                                                  &                                                5 &                                              50s &                                             169s &                                     \textbf{37s} &                                                - &                                                - &                                                - &                                                - \\
%                                                   &                                                  &                                                6 &                                             793s &                                            2646s &                                    \textbf{491s} &                                                - &                                                - &                                                - &                                                - \\
%                                                   &                                                  &                                                7 &                                                - &                                                - &                                                - &                                                - &                                                - &                                                - &                                                - \\
\specialrule{.1em}{0em}{0em}
% peterson_fischer.c
\multirow{4}{*}{\begin{tabular}{l}
\textbf{pet\_fischer}\\
% lines of code: 57\\
% variables: 2\\
threads: 2
\end{tabular}}
&
\multirow{2}{*}{Traces}
%                                                                                                      &                                                1 &                                      \textbf{49} &                                      \textbf{49} &                                      \textbf{49} &                                      \textbf{65} &                                      \textbf{65} &                                      \textbf{65} &                                      \textbf{65} \\
                                                                                                     &                                                2 &                                    \textbf{4386} &                                    \textbf{4386} &                                    \textbf{4386} &                                   \textbf{13895} &                                   \textbf{13895} &                                   \textbf{13895} &                                   \textbf{13895} \\
                                                   &                                                  &                                                3 &                                  \textbf{430004} &                                  \textbf{430004} &                                  \textbf{430004} &                                 \textbf{3786571} &                                 \textbf{3786571} &                                 \textbf{3786571} &                                 \textbf{3786571} \\
%                                                   &                                                  &                                                4 &                                                - &                                                - &                                                - &                                                - &                                                - &                                                - &                                                - \\
\cline{2-10}
&
\multirow{2}{*}{Times}
%                                                                                                      &                                                1 &                                   \textbf{0.05s} &                                            0.06s &                                   \textbf{0.05s} &                                            0.07s &                                   \textbf{0.05s} &                                            0.07s &                                   \textbf{0.05s} \\
                                                                                                     &                                                2 &                                   \textbf{0.66s} &                                            1.74s &                                            0.85s &                                            5.49s &                                   \textbf{1.70s} &                                            6.35s &                                   \textbf{3.60s} \\
                                                   &                                                  &                                                3 &                                              79s &                                             236s &                                     \textbf{59s} &                                            2342s &                                    \textbf{618s} &                                            2707s &                                    \textbf{719s} \\
%                                                   &                                                  &                                                4 &                                                - &                                                - &                                                - &                                                - &                                                - &                                                - &                                                - \\
\specialrule{.1em}{0em}{0em}
% szymanski.c
\multirow{4}{*}{\begin{tabular}{l}
\textbf{szymanski}\\
% lines of code: 84\\
% variables: 2\\
threads: 2
\end{tabular}}
&
\multirow{2}{*}{Traces}
                                                                                                      &                                                1 &                                     \textbf{103} &                                     \textbf{103} &                                     \textbf{103} &                                     \textbf{497} &                                     \textbf{497} &                                     \textbf{497} &                                     \textbf{497} \\
                                                   &                                                  &                                                2 &                                    \textbf{1991} &                                    \textbf{1991} &                                    \textbf{1991} &                                  \textbf{284015} &                                  \textbf{284015} &                                  \textbf{284015} &                                  \textbf{284015} \\
%                                                   &                                                  &                                                3 &                                   \textbf{27951} &                                   \textbf{27951} &                                   \textbf{27951} &                                                - &                                                - &                                                - &                                                - \\
%                                                   &                                                  &                                                4 &                                  \textbf{396583} &                                  \textbf{396583} &                                  \textbf{396583} &                                                - &                                                - &                                                - &                                                - \\
%                                                   &                                                  &                                                5 &                                 \textbf{5746703} &                                                - &                                 \textbf{5746703} &                                                - &                                                - &                                                - &                                                - \\
\cline{2-10}
&
\multirow{2}{*}{Times}
                                                                                                      &                                                1 &                                            0.08s &                                            0.08s &                                   \textbf{0.05s} &                                            0.17s &                                   \textbf{0.12s} &                                            0.19s &                                   \textbf{0.09s} \\
                                                   &                                                  &                                                2 &                                            0.53s &                                            0.80s &                                   \textbf{0.20s} &                                             132s &                                     \textbf{52s} &                                             155s &                                     \textbf{63s} \\
%                                                   &                                                  &                                                3 &                                            8.86s &                                              15s &                                   \textbf{6.44s} &                                                - &                                                - &                                                - &                                                - \\
%                                                   &                                                  &                                                4 &                                             112s &                                             308s &                                     \textbf{57s} &                                                - &                                                - &                                                - &                                                - \\
%                                                   &                                                  &                                                5 &                                            2013s &                                                - &                                   \textbf{1218s} &                                                - &                                                - &                                                - &                                                - \\
\specialrule{.1em}{0em}{0em}
% tsay.c
\multirow{4}{*}{\begin{tabular}{l}
\textbf{tsay}\\
% lines of code: 52\\
% variables: 2\\
threads: 2
\end{tabular}}
&
\multirow{2}{*}{Traces}
                                                                                                      &                                                1 &                                      \textbf{45} &                                      \textbf{45} &                                      \textbf{45} &                                      \textbf{63} &                                      \textbf{63} &                                      \textbf{63} &                                      \textbf{63} \\
                                                   &                                                  &                                                2 &                                    \textbf{7469} &                                    \textbf{7469} &                                    \textbf{7469} &                                   \textbf{21597} &                                   \textbf{21597} &                                   \textbf{21597} &                                   \textbf{21597} \\
%                                                   &                                                  &                                                3 &                                 \textbf{1414576} &                                 \textbf{1414576} &                                 \textbf{1414576} &                                                - &                                 \textbf{8987324} &                                                - &                                 \textbf{8987324} \\
%                                                   &                                                  &                                                4 &                                                - &                                                - &                                                - &                                                - &                                                - &                                                - &                                                - \\
\cline{2-10}
&
\multirow{2}{*}{Times}
                                                                                                      &                                                1 &                                            0.06s &                                            0.06s &                                   \textbf{0.05s} &                                            0.06s &                                   \textbf{0.05s} &                                   \textbf{0.07s} &                                            0.08s \\
                                                   &                                                  &                                                2 &                                            1.94s &                                            3.13s &                                   \textbf{0.87s} &                                            8.72s &                                   \textbf{2.54s} &                                              10s &                                   \textbf{3.52s} \\
%                                                   &                                                  &                                                3 &                                             346s &                                             870s &                                    \textbf{228s} &                                                - &                                   \textbf{1551s} &                                                - &                                   \textbf{1961s} \\
%                                                   &                                                  &                                                4 &                                                - &                                                - &                                                - &                                                - &                                                - &                                                - &                                                - \\
\specialrule{.1em}{0em}{0em}
\end{tabular}
% \addtocounter{table}{-1}
% \vspace{-2mm}
\caption{Part2: Mutual exclusion benchmarks.}
\label{tab:all_crit_2}
\end{table}

\begin{table}
\scriptsize
\newcolumntype{?}{!{\vrule width 1.5pt}}
\setlength{\extrarowheight}{.058em}
\begin{tabular}{? l | c | c ? r r r ? r r ? r r ? }
\specialrule{.15em}{0em}{0em}
\multicolumn{2}{?c|}{\multirow{2}{*}{\textbf{Benchmark}}} & \multirow{2}{*}{U} & \multicolumn{3}{c?}{\textbf{Sequential Consistency}} & \multicolumn{2}{c?}{\textbf{Total Store Order}} & \multicolumn{2}{c?}{\textbf{Partial Store Order}}\\
\cline{4-10}
\multicolumn{2}{?c|}{} & & $\ReadsFrom$ & $\DCTSOPSOM$ & $\Source$ & $\DCTSOPSOM$ & $\Source$ & $\DCTSOPSOM$ & $\Source$\\
\specialrule{.1em}{0em}{0em}
% bin_noconsec_bu.c
\multirow{4}{*}{\begin{tabular}{l}
\textbf{bin\_noconsec\_bu3}\\
% lines of code: 60\\
% variables: 2U+1\\
threads: 3
\end{tabular}}
&
\multirow{2}{*}{Traces}
                                                                                                     &                                                6 &                                   \textbf{11875} &                                   \textbf{11875} &                                           110446 &                                   \textbf{36288} &                                           710000 &                                  \textbf{155648} &                                          5467500 \\
                                                   &                                                  &                                                7 &                                   \textbf{59375} &                                   \textbf{59375} &                                           773122 &                                  \textbf{217728} &                                          7100000 &                                 \textbf{1245184} &                                                - \\
%                                                   &                                                  &                                                8 &                                  \textbf{296875} &                                  \textbf{296875} &                                          5411854 &                                 \textbf{1306368} &                                                - &                            \textcolor{red}{(-6)} &                                                - \\
%                                                   &                                                  &                                                9 &                                 \textbf{1484375} &                                 \textbf{1484375} &                                                - &                            \textcolor{red}{(-6)} &                                                - &                            \textcolor{red}{(-6)} &                                                - \\
\cline{2-10}
&
\multirow{2}{*}{Times}
                                                                                                     &                                                6 &                                   \textbf{5.15s} &                                            9.49s &                                              38s &                                     \textbf{29s} &                                             192s &                                    \textbf{201s} &                                            2276s \\
                                                   &                                                  &                                                7 &                                     \textbf{41s} &                                              56s &                                             207s &                                    \textbf{201s} &                                            2323s &                                   \textbf{2002s} &                                                - \\
%                                                   &                                                  &                                                8 &                                    \textbf{206s} &                                             317s &                                            1843s &                                   \textbf{1670s} &                                                - &                            \textcolor{red}{(-6)} &                                                - \\
%                                                   &                                                  &                                                9 &                                    \textbf{929s} &                                            1800s &                                                - &                            \textcolor{red}{(-6)} &                                                - &                            \textcolor{red}{(-6)} &                                                - \\
\specialrule{.1em}{0em}{0em}
% bin_noconsec_td.c
\multirow{4}{*}{\begin{tabular}{l}
\textbf{bin\_noconsec\_td3}\\
% lines of code: 50\\
% variables: 2U+1\\
threads: 3
\end{tabular}}
&
\multirow{2}{*}{Traces}
                                                                                                     &                                                6 &                                    \textbf{5308} &                                    \textbf{5308} &                                            50960 &                                   \textbf{29794} &                                           661015 &                                   \textbf{69741} &                                          3141122 \\
%                                                   &                                                  &                                                7 &                                   \textbf{26455} &                                   \textbf{26455} &                                           372788 &                                  \textbf{190729} &                                                - &                                  \textbf{422768} &                                                - \\
                                                   &                                                  &                                                8 &                                   \textbf{88294} &                                   \textbf{88294} &                                          1664672 &                                  \textbf{939466} &                                                - &                                                - &                                                - \\
%                                                   &                                                  &                                                9 &                                  \textbf{440167} &                                  \textbf{440167} &                                                - &                            \textcolor{red}{(-6)} &                                                - &                                                - &                                                - \\
%                                                   &                                                  &                                               10 &                                 \textbf{1469128} &                                 \textbf{1469128} &                                                - &                            \textcolor{red}{(-6)} &                                                - &                                                - &                                                - \\
%                                                   &                                                  &                                               11 &                                                - &                            \textcolor{red}{(-6)} &                                                - &                            \textcolor{red}{(-6)} &                                                - &                                                - &                                                - \\
\cline{2-10}
&
\multirow{2}{*}{Times}
                                                                                                     &                                                6 &                                   \textbf{2.30s} &                                            5.62s &                                              17s &                                     \textbf{33s} &                                             398s &                                    \textbf{124s} &                                            2456s \\
%                                                   &                                                  &                                                7 &                                     \textbf{13s} &                                              33s &                                             150s &                                    \textbf{251s} &                                                - &                                    \textbf{903s} &                                                - \\
                                                   &                                                  &                                                8 &                                     \textbf{53s} &                                             129s &                                             833s &                                   \textbf{1435s} &                                                - &                                                - &                                                - \\
%                                                   &                                                  &                                                9 &                                    \textbf{371s} &                                             730s &                                                - &                            \textcolor{red}{(-6)} &                                                - &                                                - &                                                - \\
%                                                   &                                                  &                                               10 &                                   \textbf{1335s} &                                            2776s &                                                - &                            \textcolor{red}{(-6)} &                                                - &                                                - &                                                - \\
%                                                   &                                                  &                                               11 &                                                - &                            \textcolor{red}{(-6)} &                                                - &                            \textcolor{red}{(-6)} &                                                - &                                                - &                                                - \\
\specialrule{.1em}{0em}{0em}
% coin_change_all_bu.c
\multirow{4}{*}{\begin{tabular}{l}
\textbf{coin\_all\_bu3}\\
% lines of code: 46\\
% variables: 3U+1\\
threads: 3
\end{tabular}}
&
\multirow{2}{*}{Traces}
%                                                                                                      &                                                2 &                                      \textbf{99} &                                      \textbf{99} &                                              416 &                                      \textbf{99} &                                              416 &                                      \textbf{99} &                                              416 \\
                                                                                                     &                                                3 &                                    \textbf{2673} &                                    \textbf{2673} &                                            26624 &                                    \textbf{2673} &                                            26624 &                                    \textbf{2673} &                                            26624 \\
                                                   &                                                  &                                                4 &                                   \textbf{96294} &                                   \textbf{96294} &                                          1704560 &                                   \textbf{98307} &                                          1855568 &                                  \textbf{104247} &                                          2342912 \\
%                                                   &                                                  &                                                5 &                                 \textbf{4623982} &                            \textcolor{red}{(-6)} &                                                - &                            \textcolor{red}{(-6)} &                                                - &                            \textcolor{red}{(-6)} &                                                - \\
\cline{2-10}
&
\multirow{2}{*}{Times}
%                                                                                                      &                                                2 &                                            0.10s &                                   \textbf{0.08s} &                                            0.10s &                                   \textbf{0.08s} &                                            0.10s &                                   \textbf{0.09s} &                                            0.14s \\
                                                                                                     &                                                3 &                                   \textbf{0.85s} &                                            1.18s &                                            4.13s &                                   \textbf{1.19s} &                                            4.62s &                                   \textbf{1.75s} &                                            7.32s \\
                                                   &                                                  &                                                4 &                                     \textbf{31s} &                                              57s &                                             476s &                                     \textbf{59s} &                                             423s &                                    \textbf{101s} &                                             888s \\
%                                                   &                                                  &                                                5 &                                   \textbf{2360s} &                            \textcolor{red}{(-6)} &                                                - &                            \textcolor{red}{(-6)} &                                                - &                            \textcolor{red}{(-6)} &                                                - \\
\specialrule{.1em}{0em}{0em}
% coin_change_all_bu_3thr.c
\multirow{4}{*}{\begin{tabular}{l}
\textbf{coin\_all\_bu4}\\
% lines of code: 47\\
% variables: 3U+1\\
threads: 4
\end{tabular}}
&
\multirow{2}{*}{Traces}
                                                                                                      &                                                1 &                                       \textbf{4} &                                       \textbf{4} &                                               24 &                                       \textbf{4} &                                               24 &                                       \textbf{4} &                                               24 \\
                                                   &                                                  &                                                2 &                                    \textbf{6400} &                                    \textbf{6400} &                                           264600 &                                    \textbf{6400} &                                           264600 &                                    \textbf{6400} &                                           264600 \\
%                                                   &                                                  &                                                3 &                                 \textbf{6400000} &                                                - &                                                - &                                                - &                                                - &                                                - &                                                - \\
\cline{2-10}
&
\multirow{2}{*}{Times}
                                                                                                      &                                                1 &                                   \textbf{0.05s} &                                   \textbf{0.05s} &                                            0.08s &                                   \textbf{0.05s} &                                   \textbf{0.05s} &                                   \textbf{0.05s} &                                   \textbf{0.05s} \\
                                                   &                                                  &                                                2 &                                   \textbf{1.87s} &                                            2.66s &                                              44s &                                   \textbf{2.68s} &                                              43s &                                   \textbf{3.65s} &                                              73s \\
%                                                   &                                                  &                                                3 &                                   \textbf{2584s} &                                                - &                                                - &                                                - &                                                - &                                                - &                                                - \\
\specialrule{.1em}{0em}{0em}
% coin_change_all_td.c
\multirow{4}{*}{\begin{tabular}{l}
\textbf{coin\_all\_td3}\\
% lines of code: 58\\
% variables: 3U+1\\
threads: 3
\end{tabular}}
&
\multirow{2}{*}{Traces}
                                                                                                     &                                               11 &                                    \textbf{1771} &                                    \textbf{1771} &                                            94484 &                                   \textbf{10216} &                                          1898566 &                                  \textbf{184561} &                                                - \\
                                                   &                                                  &                                               16 &                                   \textbf{37171} &                                   \textbf{37171} &                                                - &                                  \textbf{590116} &                                                - &                                                - &                                                - \\
\cline{2-10}
&
\multirow{2}{*}{Times}
                                                                                                     &                                               11 &                                   \textbf{2.04s} &                                            3.06s &                                              50s &                                     \textbf{21s} &                                            1277s &                                    \textbf{667s} &                                                - \\
                                                   &                                                  &                                               16 &                                     \textbf{46s} &                                             107s &                                                - &                                   \textbf{2073s} &                                                - &                                                - &                                                - \\
\specialrule{.1em}{0em}{0em}
% coin_change_all_td_3thr.c
\multirow{4}{*}{\begin{tabular}{l}
\textbf{coin\_all\_td4}\\
% lines of code: 60\\
% variables: 3U+1\\
threads: 4
\end{tabular}}
&
\multirow{2}{*}{Traces}
                                                                                                      &                                                4 &                                     \textbf{946} &                                     \textbf{946} &                                           234984 &                                     \textbf{946} &                                           234984 &                                     \textbf{946} &                                           234984 \\
%                                                   &                                                  &                                                5 &                                    \textbf{4298} &                                    \textbf{4298} &                                          2489448 &                                    \textbf{6350} &                                                - &                                    \textbf{9614} &                                                - \\
%                                                   &                                                  &                                                6 &                                    \textbf{6304} &                                    \textbf{6304} &                                          6406248 &                                   \textbf{10692} &                                                - &                                   \textbf{17848} &                                                - \\
                                                   &                                                  &                                                7 &                                  \textbf{110182} &                                  \textbf{110182} &                                                - &                                  \textbf{321298} &                                                - &                                                - &                                                - \\
%                                                   &                                                  &                                                8 &                                  \textbf{238746} &                                  \textbf{238746} &                                                - &                                                - &                                                - &                                                - &                                                - \\
%                                                   &                                                  &                                                9 &                                  \textbf{777776} &                                  \textbf{777776} &                                                - &                                                - &                                                - &                                                - &                                                - \\
%                                                   &                                                  &                                               10 &                                 \textbf{2550286} &                                                - &                                                - &                                                - &                                                - &                                                - &                                                - \\
\cline{2-10}
&
\multirow{2}{*}{Times}
                                                                                                      &                                                4 &                                   \textbf{0.59s} &                                            0.74s &                                             121s &                                   \textbf{0.76s} &                                              78s &                                   \textbf{0.99s} &                                             117s \\
%                                                   &                                                  &                                                5 &                                   \textbf{3.11s} &                                            4.36s &                                             996s &                                   \textbf{6.83s} &                                                - &                                     \textbf{15s} &                                                - \\
%                                                   &                                                  &                                                6 &                                   \textbf{5.17s} &                                            7.31s &                                            3199s &                                     \textbf{14s} &                                                - &                                     \textbf{33s} &                                                - \\
                                                   &                                                  &                                                7 &                                     \textbf{75s} &                                             176s &                                                - &                                    \textbf{601s} &                                                - &                                                - &                                                - \\
%                                                   &                                                  &                                                8 &                                    \textbf{202s} &                                             441s &                                                - &                                                - &                                                - &                                                - &                                                - \\
%                                                   &                                                  &                                                9 &                                    \textbf{623s} &                                            1614s &                                                - &                                                - &                                                - &                                                - &                                                - \\
%                                                   &                                                  &                                               10 &                                   \textbf{2585s} &                                                - &                                                - &                                                - &                                                - &                                                - &                                                - \\
\specialrule{.1em}{0em}{0em}
% coin_change_min_bu.c
\multirow{4}{*}{\begin{tabular}{l}
\textbf{coin\_min\_bu3}\\
% lines of code: 44\\
% variables: U+1\\
threads: 3
\end{tabular}}
&
\multirow{2}{*}{Traces}
                                                                                                     &                                               14 &                                  \textbf{229550} &                                  \textbf{229550} &                                          1713984 &                                  \textbf{271319} &                                          2930352 &                                  \textbf{322959} &                                          6133248 \\
                                                   &                                                  &                                               15 &                                  \textbf{918794} &                                  \textbf{918794} &                                          6865920 &                                 \textbf{1146047} &                                                - &                                 \textbf{1399489} &                                                - \\
%                                                   &                                                  &                                               16 &                                 \textbf{3678760} &                            \textcolor{red}{(-6)} &                                                - &                            \textcolor{red}{(-6)} &                                                - &                            \textcolor{red}{(-6)} &                                                - \\
%                                                   &                                                  &                                               17 &                            \textcolor{red}{(-6)} &                            \textcolor{red}{(-6)} &                                                - &                            \textcolor{red}{(-6)} &                                                - &                            \textcolor{red}{(-6)} &                                                - \\
\cline{2-10}
&
\multirow{2}{*}{Times}
                                                                                                     &                                               14 &                                    \textbf{136s} &                                             220s &                                             481s &                                    \textbf{261s} &                                             839s &                                    \textbf{546s} &                                            3501s \\
                                                   &                                                  &                                               15 &                                    \textbf{444s} &                                             951s &                                            2168s &                                   \textbf{1214s} &                                                - &                                   \textbf{2647s} &                                                - \\
%                                                   &                                                  &                                               16 &                                   \textbf{2154s} &                            \textcolor{red}{(-6)} &                                                - &                            \textcolor{red}{(-6)} &                                                - &                            \textcolor{red}{(-6)} &                                                - \\
%                                                   &                                                  &                                               17 &                            \textcolor{red}{(-6)} &                            \textcolor{red}{(-6)} &                                                - &                            \textcolor{red}{(-6)} &                                                - &                            \textcolor{red}{(-6)} &                                                - \\
\specialrule{.1em}{0em}{0em}
% coin_change_min_bu_3thr.c
\multirow{4}{*}{\begin{tabular}{l}
\textbf{coin\_min\_bu4}\\
% lines of code: 45\\
% variables: U+1\\
threads: 4
\end{tabular}}
&
\multirow{2}{*}{Traces}
%                                                                                                      &                                                6 &                                      \textbf{64} &                                      \textbf{64} &                                              600 &                                      \textbf{64} &                                              600 &                                      \textbf{64} &                                              600 \\
%                                                   &                                                  &                                                7 &                                     \textbf{640} &                                     \textbf{640} &                                            12600 &                                     \textbf{640} &                                            12600 &                                     \textbf{640} &                                            12600 \\
%                                                   &                                                  &                                                8 &                                    \textbf{6400} &                                    \textbf{6400} &                                           264600 &                                    \textbf{6400} &                                           264600 &                                    \textbf{6400} &                                           264600 \\
                                                                                                     &                                                9 &                                   \textbf{64000} &                                   \textbf{64000} &                                          5556600 &                                   \textbf{64000} &                                          5556600 &                                   \textbf{64000} &                                          5556600 \\
                                                   &                                                  &                                               10 &                                  \textbf{640000} &                                  \textbf{640000} &                                                - &                                  \textbf{640000} &                                                - &                                  \textbf{640000} &                                                - \\
%                                                   &                                                  &                                               11 &                                                - &                                                - &                                                - &                                                - &                                                - &                                                - &                                                - \\
\cline{2-10}
&
\multirow{2}{*}{Times}
%                                                                                                      &                                                6 &                                   \textbf{0.06s} &                                            0.07s &                                            0.12s &                                   \textbf{0.07s} &                                            0.13s &                                   \textbf{0.09s} &                                            0.18s \\
%                                                   &                                                  &                                                7 &                                   \textbf{0.21s} &                                            0.36s &                                            2.00s &                                   \textbf{0.36s} &                                            2.17s &                                   \textbf{0.84s} &                                            3.63s \\
%                                                   &                                                  &                                                8 &                                   \textbf{1.68s} &                                            3.96s &                                              56s &                                   \textbf{3.95s} &                                              55s &                                   \textbf{6.29s} &                                              83s \\
                                                                                                     &                                                9 &                                     \textbf{19s} &                                              48s &                                            1589s &                                     \textbf{47s} &                                            1328s &                                     \textbf{78s} &                                            2050s \\
                                                   &                                                  &                                               10 &                                    \textbf{224s} &                                             575s &                                                - &                                    \textbf{575s} &                                                - &                                    \textbf{983s} &                                                - \\
%                                                   &                                                  &                                               11 &                                                - &                                                - &                                                - &                                                - &                                                - &                                                - &                                                - \\
\specialrule{.1em}{0em}{0em}
% coin_change_min_td.c
\multirow{4}{*}{\begin{tabular}{l}
\textbf{coin\_min\_td3}\\
% lines of code: 53\\
% variables: U+1\\
threads: 3
\end{tabular}}
&
\multirow{2}{*}{Traces}
                                                                                                     &                                               14 &                                   \textbf{86091} &                                   \textbf{86091} &                                           500260 &                                  \textbf{252661} &                                          3589906 &                                  \textbf{458256} &                                                - \\
                                                   &                                                  &                                               15 &                                  \textbf{326976} &                                  \textbf{326976} &                                          1902262 &                                 \textbf{1328496} &                                                - &                                                - &                                                - \\
%                                                   &                                                  &                                               16 &                                 \textbf{1399676} &                                 \textbf{1399676} &                                                - &                            \textcolor{red}{(-6)} &                                                - &                                                - &                                                - \\
%                                                   &                                                  &                                               17 &                                                - &                            \textcolor{red}{(-6)} &                                                - &                            \textcolor{red}{(-6)} &                                                - &                                                - &                                                - \\
\cline{2-10}
&
\multirow{2}{*}{Times}
                                                                                                     &                                               14 &                                     \textbf{60s} &                                             177s &                                             211s &                                    \textbf{523s} &                                            1897s &                                   \textbf{1433s} &                                                - \\
                                                   &                                                  &                                               15 &                                    \textbf{379s} &                                             754s &                                             887s &                                   \textbf{3066s} &                                                - &                                                - &                                                - \\
%                                                   &                                                  &                                               16 &                                   \textbf{1323s} &                                            3543s &                                                - &                            \textcolor{red}{(-6)} &                                                - &                                                - &                                                - \\
%                                                   &                                                  &                                               17 &                                                - &                            \textcolor{red}{(-6)} &                                                - &                            \textcolor{red}{(-6)} &                                                - &                                                - &                                                - \\
\specialrule{.1em}{0em}{0em}
% coin_change_min_td_3thr.c
\multirow{4}{*}{\begin{tabular}{l}
\textbf{coin\_min\_td4}\\
% lines of code: 59\\
% variables: U+1\\
threads: 4
\end{tabular}}
&
\multirow{2}{*}{Traces}
%                                                                                                      &                                                6 &                                      \textbf{64} &                                      \textbf{64} &                                              600 &                                      \textbf{64} &                                              600 &                                      \textbf{64} &                                              600 \\
%                                                   &                                                  &                                                7 &                                     \textbf{238} &                                     \textbf{238} &                                             5448 &                                     \textbf{238} &                                             5448 &                                     \textbf{238} &                                             5448 \\
%                                                   &                                                  &                                                8 &                                    \textbf{1790} &                                    \textbf{1790} &                                            91592 &                                    \textbf{1890} &                                           144488 &                                    \textbf{2504} &                                           184008 \\
                                                                                                     &                                                9 &                                   \textbf{16682} &                                   \textbf{16682} &                                          1470312 &                                   \textbf{19502} &                                          4482536 &                                   \textbf{30736} &                                                - \\
                                                   &                                                  &                                               10 &                                  \textbf{230402} &                                  \textbf{230402} &                                                - &                                  \textbf{332182} &                                                - &                                  \textbf{478292} &                                                - \\
%                                                   &                                                  &                                               11 &                                                - &                                                - &                                                - &                                                - &                                                - &                                                - &                                                - \\
\cline{2-10}
&
\multirow{2}{*}{Times}
%                                                                                                      &                                                6 &                                   \textbf{0.07s} &                                            0.09s &                                            0.15s &                                   \textbf{0.08s} &                                            0.21s &                                   \textbf{0.09s} &                                            0.27s \\
%                                                   &                                                  &                                                7 &                                   \textbf{0.17s} &                                            0.21s &                                            1.27s &                                   \textbf{0.22s} &                                            1.37s &                                   \textbf{0.45s} &                                            2.12s \\
%                                                   &                                                  &                                                8 &                                   \textbf{1.15s} &                                            1.91s &                                              26s &                                   \textbf{1.98s} &                                              47s &                                   \textbf{3.80s} &                                              87s \\
                                                                                                     &                                                9 &                                     \textbf{13s} &                                              24s &                                             521s &                                     \textbf{28s} &                                            1869s &                                     \textbf{66s} &                                                - \\
                                                   &                                                  &                                               10 &                                    \textbf{188s} &                                             741s &                                                - &                                    \textbf{627s} &                                                - &                                   \textbf{1372s} &                                                - \\
%                                                   &                                                  &                                               11 &                                                - &                                                - &                                                - &                                                - &                                                - &                                                - &                                                - \\
\specialrule{.1em}{0em}{0em}
% rod_cut_bu.c
\multirow{4}{*}{\begin{tabular}{l}
\textbf{rod\_cut\_bu3}\\
% lines of code: 43\\
% variables: U+1\\
threads: 3
\end{tabular}}
&
\multirow{2}{*}{Traces}
%                                                                                                      &                                                3 &                                     \textbf{282} &                                     \textbf{282} &                                              852 &                                     \textbf{297} &                                              960 &                                     \textbf{297} &                                             1040 \\
%                                                   &                                                  &                                                4 &                                    \textbf{1680} &                                    \textbf{1680} &                                             5100 &                                    \textbf{2123} &                                             7228 &                                    \textbf{2323} &                                             9376 \\
%                                                   &                                                  &                                                5 &                                   \textbf{10068} &                                   \textbf{10068} &                                            30588 &                                   \textbf{16815} &                                            59304 &                                   \textbf{22463} &                                           101952 \\
                                                                                                     &                                                6 &                                   \textbf{60396} &                                   \textbf{60396} &                                           183516 &                                  \textbf{143259} &                                           518676 &                                  \textbf{259857} &                                          1302112 \\
                                                   &                                                  &                                                7 &                                  \textbf{362364} &                                  \textbf{362364} &                                          1101084 &                                 \textbf{1289881} &                                          4765876 &                             - &                                                - \\  % \textcolor{red}{(-6)}
%                                                   &                                                  &                                                8 &                                 \textbf{2174172} &                            \textcolor{red}{(-6)} &                                          6606492 &                            \textcolor{red}{(-6)} &                                                - &                            \textcolor{red}{(-6)} &                                                - \\
\cline{2-10}
&
\multirow{2}{*}{Times}
%                                                                                                      &                                                3 &                                   \textbf{0.16s} &                                            0.17s &                                            0.20s &                                   \textbf{0.18s} &                                            0.31s &                                   \textbf{0.34s} &                                            0.39s \\
%                                                   &                                                  &                                                4 &                                            1.06s &                                   \textbf{1.01s} &                                            1.21s &                                   \textbf{1.28s} &                                            2.17s &                                   \textbf{1.99s} &                                            4.19s \\
%                                                   &                                                  &                                                5 &                                            8.13s &                                   \textbf{7.91s} &                                            9.16s &                                     \textbf{13s} &                                              22s &                                     \textbf{25s} &                                              59s \\
                                                                                                     &                                                6 &                                     \textbf{38s} &                                              62s &                                              71s &                                    \textbf{154s} &                                             235s &                                    \textbf{401s} &                                             996s \\
                                                   &                                                  &                                                7 &                                    \textbf{439s} &                                             484s &                                             488s &                                   \textbf{1858s} &                                            2668s &                            - &                                                - \\  % \textcolor{red}{(-6)}
%                                                   &                                                  &                                                8 &                                   \textbf{2973s} &                            \textcolor{red}{(-6)} &                                            3543s &                            \textcolor{red}{(-6)} &                                                - &                            \textcolor{red}{(-6)} &                                                - \\
\specialrule{.1em}{0em}{0em}
% rod_cut_bu_3thr.c
\multirow{4}{*}{\begin{tabular}{l}
\textbf{rod\_cut\_bu4}\\
% lines of code: 44\\
% variables: U+1\\
threads: 4
\end{tabular}}
&
\multirow{2}{*}{Traces}
%                                                                                                      &                                                1 &                                      \textbf{64} &                                      \textbf{64} &                                              600 &                                      \textbf{64} &                                              600 &                                      \textbf{64} &                                              600 \\
                                                                                                     &                                                2 &                                    \textbf{2008} &                                    \textbf{2008} &                                            33912 &                                    \textbf{2008} &                                            33912 &                                    \textbf{2008} &                                            33912 \\
                                                   &                                                  &                                                3 &                                  \textbf{106500} &                                  \textbf{106500} &                                          2246424 &                                  \textbf{135988} &                                          3354504 &                                  \textbf{151720} &                                          4080168 \\
%                                                   &                                                  &                                                4 &                                                - &                                                - &                                                - &                                                - &                                                - &                                                - &                                                - \\
\cline{2-10}
&
\multirow{2}{*}{Times}
%                                                                                                      &                                                1 &                                            0.10s &                                   \textbf{0.07s} &                                            0.13s &                                   \textbf{0.11s} &                                            0.19s &                                   \textbf{0.10s} &                                            0.26s \\
                                                                                                     &                                                2 &                                   \textbf{0.65s} &                                            0.94s &                                            6.95s &                                   \textbf{0.97s} &                                            7.32s &                                   \textbf{1.26s} &                                              12s \\
                                                   &                                                  &                                                3 &                                     \textbf{42s} &                                              75s &                                             618s &                                     \textbf{94s} &                                             985s &                                    \textbf{148s} &                                            1975s \\
%                                                   &                                                  &                                                4 &                                                - &                                                - &                                                - &                                                - &                                                - &                                                - &                                                - \\
\specialrule{.1em}{0em}{0em}
% rod_cut_td.c
\multirow{4}{*}{\begin{tabular}{l}
\textbf{rod\_cut\_td3}\\
% lines of code: 56\\
% variables: U+1\\
threads: 3
\end{tabular}}
&
\multirow{2}{*}{Traces}
%                                                                                                      &                                                4 &                                     \textbf{206} &                                     \textbf{206} &                                              978 &                                     \textbf{241} &                                             1468 &                                     \textbf{286} &                                             1832 \\
%                                                   &                                                  &                                                5 &                                     \textbf{896} &                                     \textbf{896} &                                             4422 &                                    \textbf{1496} &                                            10330 &                                    \textbf{1936} &                                            15034 \\
%                                                   &                                                  &                                                6 &                                    \textbf{4191} &                                    \textbf{4191} &                                            20942 &                                   \textbf{11446} &                                            84362 &                                   \textbf{17066} &                                           152234 \\
                                                                                                     &                                                7 &                                   \textbf{20336} &                                   \textbf{20336} &                                           102128 &                                   \textbf{99281} &                                           762942 &                                  \textbf{184701} &                                          1837610 \\
                                                   &                                                  &                                                8 &                                  \textbf{101001} &                                  \textbf{101001} &                                           508646 &                                  \textbf{938731} &                                                - &                                                - &                                                - \\
%                                                   &                                                  &                                                9 &                                  \textbf{510281} &                                  \textbf{510281} &                                          2574752 &                                                - &                                                - &                                                - &                                                - \\
\cline{2-10}
&
\multirow{2}{*}{Times}
%                                                                                                      &                                                4 &                                   \textbf{0.13s} &                                            0.19s &                                            0.30s &                                   \textbf{0.25s} &                                            0.59s &                                   \textbf{0.33s} &                                            0.89s \\
%                                                   &                                                  &                                                5 &                                   \textbf{0.75s} &                                            0.92s &                                            1.51s &                                   \textbf{1.57s} &                                            4.45s &                                   \textbf{2.62s} &                                            9.67s \\
%                                                   &                                                  &                                                6 &                                   \textbf{3.58s} &                                            5.67s &                                            8.83s &                                     \textbf{16s} &                                              46s &                                     \textbf{31s} &                                             132s \\
                                                                                                     &                                                7 &                                     \textbf{17s} &                                              37s &                                              55s &                                    \textbf{185s} &                                             528s &                                    \textbf{465s} &                                            2271s \\
                                                   &                                                  &                                                8 &                                    \textbf{128s} &                                             248s &                                             324s &                                   \textbf{2330s} &                                                - &                                                - &                                                - \\
%                                                   &                                                  &                                                9 &                                    \textbf{828s} &                                            1657s &                                            2002s &                                                - &                                                - &                                                - &                                                - \\
\specialrule{.1em}{0em}{0em}
% rod_cut_td_3thr.c
\multirow{4}{*}{\begin{tabular}{l}
\textbf{rod\_cut\_td4}\\
% lines of code: 69\\
% variables: U+1\\
threads: 4
\end{tabular}}
&
\multirow{2}{*}{Traces}
%                                                                                                      &                                                1 &                                      \textbf{64} &                                      \textbf{64} &                                              600 &                                      \textbf{64} &                                              600 &                                      \textbf{64} &                                              600 \\
%                                                   &                                                  &                                                2 &                                     \textbf{238} &                                     \textbf{238} &                                             5448 &                                     \textbf{238} &                                             5448 &                                     \textbf{238} &                                             5448 \\
                                                                                                     &                                                3 &                                    \textbf{1790} &                                    \textbf{1790} &                                            91592 &                                    \textbf{1890} &                                           144488 &                                    \textbf{2504} &                                           184008 \\
                                                   &                                                  &                                                4 &                                   \textbf{33550} &                                   \textbf{33550} &                                          2459640 &                                   \textbf{62748} &                                                - &                                  \textbf{103622} &                                                - \\
%                                                   &                                                  &                                                5 &                                  \textbf{990990} &                                  \textbf{990990} &                                                - &                                                - &                                                - &                                                - &                                                - \\
%                                                   &                                                  &                                                6 &                                                - &                                                - &                                                - &                                                - &                                                - &                                                - &                                                - \\
\cline{2-10}
&
\multirow{2}{*}{Times}
%                                                                                                      &                                                1 &                                            0.08s &                                   \textbf{0.07s} &                                            0.15s &                                   \textbf{0.10s} &                                            0.16s &                                   \textbf{0.14s} &                                            0.30s \\
%                                                   &                                                  &                                                2 &                                            0.19s &                                   \textbf{0.16s} &                                            1.20s &                                   \textbf{0.19s} &                                            1.32s &                                   \textbf{0.29s} &                                            2.13s \\
                                                                                                     &                                                3 &                                   \textbf{1.26s} &                                            1.36s &                                              29s &                                   \textbf{1.49s} &                                              46s &                                   \textbf{2.52s} &                                              90s \\
                                                   &                                                  &                                                4 &                                     \textbf{19s} &                                              43s &                                            1003s &                                     \textbf{78s} &                                                - &                                    \textbf{174s} &                                                - \\
%                                                   &                                                  &                                                5 &                                    \textbf{907s} &                                            2122s &                                                - &                                                - &                                                - &                                                - &                                                - \\
%                                                   &                                                  &                                                6 &                                                - &                                                - &                                                - &                                                - &                                                - &                                                - &                                                - \\
\specialrule{.1em}{0em}{0em}
% subs_inc_long_bu.c
\multirow{4}{*}{\begin{tabular}{l}
\textbf{lis\_bu3}\\
% lines of code: 47\\
% variables: U+1\\
threads: 3
\end{tabular}}
&
\multirow{2}{*}{Traces}
                                                                                                     &                                                7 &                                  \textbf{103260} &                                  \textbf{103260} &                                           429632 &                                  \textbf{165105} &                                           975040 &                                  \textbf{229965} &                                          1862144 \\
                                                   &                                                  &                                                8 &                                  \textbf{325740} &                                  \textbf{325740} &                                          1744064 &                                  \textbf{596475} &                                          4678656 &                                  \textbf{977685} &                                                - \\
%                                                   &                                                  &                                                9 &                                  \textbf{993180} &                                  \textbf{993180} &                                                - &                                 \textbf{1989465} &                                                - &                            \textcolor{red}{(-6)} &                                                - \\
%                                                   &                                                  &                                               10 &                            \textcolor{red}{(-6)} &                            \textcolor{red}{(-6)} &                                                - &                            \textcolor{red}{(-6)} &                                                - &                            \textcolor{red}{(-6)} &                                                - \\
\cline{2-10}
&
\multirow{2}{*}{Times}
                                                                                                     &                                                7 &                                     \textbf{76s} &                                              99s &                                             191s &                                    \textbf{168s} &                                             511s &                                    \textbf{338s} &                                            2140s \\
                                                   &                                                  &                                                8 &                                             366s &                                    \textbf{352s} &                                             913s &                                    \textbf{672s} &                                            2782s &                                   \textbf{1675s} &                                                - \\
%                                                   &                                                  &                                                9 &                                   \textbf{1159s} &                                            1203s &                                                - &                                   \textbf{2563s} &                                                - &                            \textcolor{red}{(-6)} &                                                - \\
%                                                   &                                                  &                                               10 &                            \textcolor{red}{(-6)} &                            \textcolor{red}{(-6)} &                                                - &                            \textcolor{red}{(-6)} &                                                - &                            \textcolor{red}{(-6)} &                                                - \\
\specialrule{.1em}{0em}{0em}
% subs_inc_long_bu_3thr.c
\multirow{4}{*}{\begin{tabular}{l}
\textbf{lis\_bu4}\\
% lines of code: 48\\
% variables: U+1\\
threads: 4
\end{tabular}}
&
\multirow{2}{*}{Traces}
%                                                                                                      &                                                1 &                                     \textbf{100} &                                     \textbf{100} &                                              882 &                                     \textbf{100} &                                              882 &                                     \textbf{100} &                                              882 \\
%                                                   &                                                  &                                                2 &                                    \textbf{1000} &                                    \textbf{1000} &                                            18522 &                                    \textbf{1000} &                                            18522 &                                    \textbf{1000} &                                            18522 \\
                                                                                                     &                                                3 &                                   \textbf{28900} &                                   \textbf{28900} &                                          1024002 &                                   \textbf{28900} &                                          1024002 &                                   \textbf{28900} &                                          1024002 \\
                                                   &                                                  &                                                4 &                                 \textbf{1504200} &                                 \textbf{1504200} &                                                - &                                 \textbf{1863700} &                                                - &                                 \textbf{2059000} &                                                - \\
%                                                   &                                                  &                                                5 &                                                - &                                                - &                                                - &                                                - &                                                - &                                                - &                                                - \\
\cline{2-10}
&
\multirow{2}{*}{Times}
%                                                                                                      &                                                1 &                                   \textbf{0.08s} &                                   \textbf{0.08s} &                                            0.20s &                                   \textbf{0.17s} &                                            0.26s &                                   \textbf{0.17s} &                                            0.42s \\
%                                                   &                                                  &                                                2 &                                   \textbf{0.38s} &                                            0.50s &                                            4.19s &                                   \textbf{0.59s} &                                            4.59s &                                   \textbf{0.79s} &                                            7.27s \\
                                                                                                     &                                                3 &                                     \textbf{16s} &                                              18s &                                             307s &                                     \textbf{19s} &                                             335s &                                     \textbf{25s} &                                             545s \\
                                                   &                                                  &                                                4 &                                    \textbf{898s} &                                            1451s &                                                - &                                   \textbf{1755s} &                                                - &                                   \textbf{2831s} &                                                - \\
%                                                   &                                                  &                                                5 &                                                - &                                                - &                                                - &                                                - &                                                - &                                                - &                                                - \\
\specialrule{.1em}{0em}{0em}
\end{tabular}
% \addtocounter{table}{-1}
% \vspace{-2mm}
\caption{Dynamic programming benchmarks.}
\label{tab:all_dyn}
\end{table}

%\fi

\Paragraph{Related-work benchmarks and synthetic benchmarks.}
In Tables~\ref{tab:all_other_1}~and~\ref{tab:all_other_2} we present benchmarks collected from
previous SMC works, namely~\citet{Abdulla19,Abdulla2015,Huang16,Chatterjee19}.
The benchmarks contain several examples originating from industrial code,
such as \texttt{parker} and \texttt{pgsql}. Further there are
several synthetic benchmarks, such as \texttt{spammer} and \texttt{overtake}.

\Paragraph{SVCOMP.}
In Tables~\ref{tab:all_svcomp_1},~\ref{tab:all_svcomp_2}~and~\ref{tab:all_svcomp_3} we present our
results on SVCOMP concurrency benchmarks.

\Paragraph{Mutual exclusion benchmarks.}
In Tables~\ref{tab:all_crit_1}~and~\ref{tab:all_crit_2} we present
our results for mutual-exclusion algorithms from the literature.
We include the classical solutions, and novel solutions presented
by~\citet{Correia16} (prefixed with \texttt{X2Tv}).

\Paragraph{Benchmarks on dynamic programming.}
Finally, we present
% We present one further category of benchmarks, namely,
benchmarks that perform parallel dynamic programming tasks,
introduced by~\citet{Chatterjee19}.
%\cref{tab:exp_dynamic1} and \cref{tab:exp_dynamic2} present the results.
\cref{tab:all_dyn} presents the results.

\end{document}